%% file: FROC-Bianco-etal-2025-01-31.tex
\documentclass[12pt]{article}   
\usepackage{amsfonts}

\usepackage{geometry}

\usepackage{setspace}
\usepackage{amsfonts} 
\usepackage{graphicx}
\usepackage{color}
\usepackage{rotating}
\usepackage{amsmath}
\usepackage{amssymb}
\usepackage{amsthm}
\usepackage{float}
 \usepackage{url}
\usepackage{subfigure}

\usepackage{natbib} 

\usepackage{enumerate}
\usepackage{enumitem}

\usepackage{dsfont}
 
 \usepackage{algorithm}
\usepackage{algorithmic}
 
\usepackage{array,arydshln}

\newtheorem{theorem}{Theorem}[section]
\newtheorem{lemma}[theorem]{Lemma}

\newtheorem{proposition}[theorem]{Proposition} 
\newtheorem{remark}{Remark}[section]

\usepackage{soul} 

\usepackage{hyperref} 

 \usepackage{xcolor}
\hypersetup{
    colorlinks,
    linkcolor= {blue!90!black}, 
    citecolor={blue!60!black},
    urlcolor={blue!80!black}
}

\include{definiciones}

\begin{document}

 \title{\textbf{ROC curve analysis for functional markers}}
\author{Ana M. Bianco$^a$, Graciela Boente$^b$  and  Juan Carlos Pardo-Fern\'andez$^c$ \\
\small $^a$ Instituto de C\'alculo,  
\small Facultad de Ciencias Exactas y Naturales, \\
\small Universidad de Buenos Aires and CONICET, Argentina\\
\small $^a$ email: abianco@dm.uba.ar\\
\small $^b$ Departamento de Matem\'aticas and Instituto de C\'alculo,  
\small Facultad de Ciencias Exactas y Naturales,\\
 \small Universidad de Buenos Aires and CONICET, Argentina\\
\small $^c$  Centro de Investigaci\'{o}n e Tecnolox\'{\i}a Matem\'{a}tica de Galicia (CITMAga) and \\
\small Departamento de Estat{\'i}stica e Investigaci\'on Operativa,  Universidade de Vigo, Spain    
}
\date{}
\maketitle

\begin{abstract}
Functional markers have become a  frequent tool in medical diagnosis. In this paper, we aim to define an index allowing to discriminate between populations when the observations are functional data belonging to a Hilbert space. We discuss some of the problems arising when estimating  optimal directions defined to maximize the area under the curve of a projection index and we construct the corresponding ROC curve. We also go one step further and consider the case of possibly different covariance operators, for which we recommend a quadratic discrimination rule. Consistency results are derived for both linear and quadratic indexes under mild conditions. The results of our  numerical experiments allow to see the advantages of the quadratic rule when the populations have different covariance operators. We also illustrate the considered methods on a real data set on cardiotoxicity related to therapies on breast cancer patients.   
 \end{abstract}

\noi \textbf{Keywords:} Consistency; Discriminating index; Functional data; Optimal directions; ROC curve.

\section{Introduction}{\label{sec:intro}}
In applied sciences, there is a permanent search for better and better tools of diagnosis and screening of different diseases. A key--point is the evaluation of the performance of such developments. The Receiver Operating Characteristic curve (ROC curve) is a very well--accepted graphical technique to  assess  the accuracy of a diagnostic test based on a continuous marker. The  use of ROC curves  is extensive in medical and pharmacological investigations, but they are also employed in  completely different scenarios, such as for the evaluation of a machine learning process, see  \citet{Krzanowski:Hand:2009} where more applications can be found. 

For the sole purpose of introducing the concepts related to ROC curves, we will focus on medical diagnosis, where there are two groups, corresponding to diseased and healthy populations, and the aim is to classify a new subject in one of these groups according to the outcome of a continuous biomarker. In this context, two essential concepts appear concerning the errors one can make:  the \textsl{sensitivity}, related to the ability of correctly detecting diseased people, and the \textsl{specificity}, that involves the skill of correctly assigning a subject to the healthy group. Thus, a ROC curve for a test based on a continuous marker is a graphical representation  of the \textsl{sensitivity} against the complementary of the \textsl{specificity} (that is, $1-$specificity) computed from the classification rule that assigns a subject to the diseased group if the biomarker is greater than a critical value $c$ and to the healthy group, otherwise,   as  the threshold $c$ varies.

In order to compact the information about the discriminatory  performance of the diagnostic test several summary  indexes  were introduced. The classification accuracy is frequently measured through the area under the curve (AUC), which can be interpreted as the average sensitivity for all specificity values. The Youden index, $\YI$, is another  global measure that is extensively used in the literature and corresponds to the maximum difference between the ROC curve and the identity function. Estimation and inference methods regarding ROC curve and related summary measures are very well-studied in the univariate setting. \citet{pepe:2003} and \citet{zhou:etal:2011} provide a comprehensive review concerning both theoretical and practical aspects of ROC curves based on univariate markers.


For some diseases,  it is necessary to combine several biomarkers  in order to get a diagnosis tool that improves the performance of each single marker on its own. In such cases, a global diagnostic measure is desirable to achieve a better classification rule. \citet{perez:2020}  reviews some proposals given to summarize the joint information provided by several biomarkers, including methods based on linear and quadratic discrimination rules.  

  In recent years with the evolution of biomedical technology,  data with more and more complex structure are collected. This is the case of functional data that consist of curves varying over time or any other continuum and thus,  valued in an infinite--dimensional space, usually   a metric or semi--metric, and in some cases,  a Hilbert space.   Henceforth, we will assume that the functional data are measured over time. In fact, functional markers  have been increasingly used in clinical studies to diagnose diseases. To summarize the curves in a univariate biomarker usual practices are to consider the maximum or minimum values, the time to the maximum or the  integral of the curve over the time  range.  However, specifically designed techniques should be employed to analyse this kind of data in order to take advantage of their potential and, as extensively discussed, the infinite--dimensional structure  should be taken into account when considering any estimation procedure, see \citet{wang:etal:2016}.   For an overview on functional data analysis we refer among others to \citet{ramsay2005functional}, \citet{ferraty:vieu:2006}, \citet{horvath:kokoska:2012} and \citet{hsing:eubank:2015}. 

 Our contribution is oriented to situations such as the one described in Section \ref{sec:realdata}, where we address a study on breast cancer patients with high  levels of the protein human epidermal growth factor receptor 2 (HER2). These women have a better response to drugs that target the HER2 protein, but  this kind of therapies may have side effects such as  cardiotoxicity. In order to prevent therapy-related cardiac dysfunction (CTRCD),  it is recommended to follow--up the appearance of CTRCD through cardiac imaging tests such as the Tissue Doppler Imaging (TDI), an echocardiographic technique that reflects the myocardial motion.  TDI is processed so as to obtain  a functional datum used to study the heart status. The aim is to evaluate the performance of this functional biomarker, which is displayed in Figure \ref{fig:ciclos}, to distinguish between patients with CTRCD from those who do not suffer from this condition. Gray curves  correspond to  patients without CTRCD (CTRCD$=0$) and aquamarine ones to women with CTRCD (CTRCD$=1$). 
At a first glance, the two groups look different since the cycles of patients without CTRCD are more spread, while the data of women that experienced CTRCD are more concentrated in the central area. Besides, the means of the two groups are similar, except for cycle values between 0.4 and 0.8. It is worth mentioning that the data set  is available at \citet{Pineiro:etal:2023}, where a thorough description is given.

\begin{figure}[ht!]
	\begin{center}
		\footnotesize
		\includegraphics[scale=0.45]{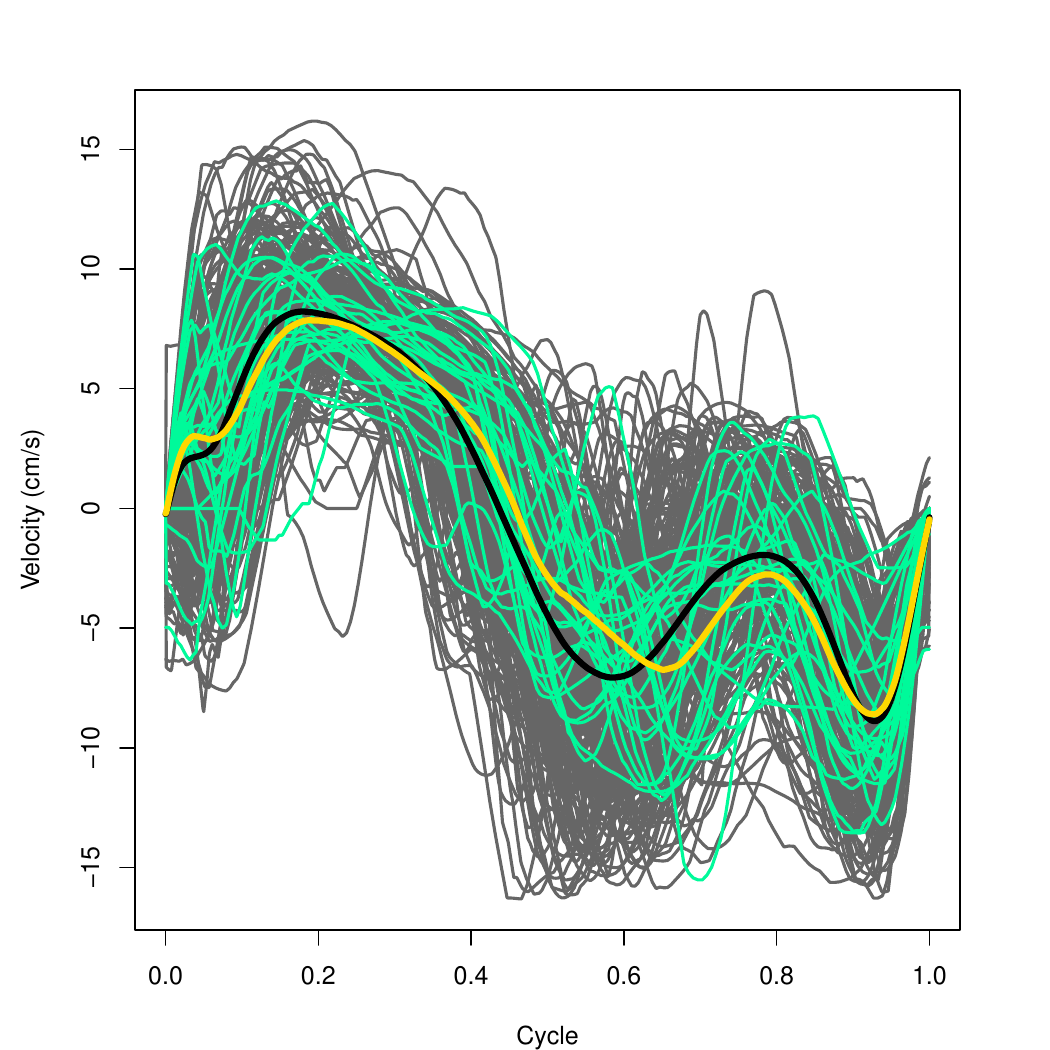}
		\vskip-0.1in 
		\caption{\label{fig:ciclos} Cardiotoxicity data. Cycles of 270 patients during the follow--up. The cycles of patients with   CTRCD$=0$ are displayed as gray lines, while those with CTRCD$=1$ are represented in aquamarine. The lines in black and gold correspond  to the mean of   CTRCD$=0$  and  CTRCD$=1$, respectively.}
		\end{center} 
\end{figure}

Functional data, that become more frequent every day in clinical research, pose  different challenges to Statistics, in particular,  to ROC curve estimation.  In this direction, some developments were done to extend the existing  methodology to the functional setting. For instance,  proposals for the induced ROC  considering a univariate marker and a functional covariate  are  considered by \citet{inacio:etal:2012} and \citet{inacio:etal:2016}. However, in some situations, the biomarker itself is a functional data usually discretely recorded. This is often the case in  longitudinal studies and we refer to \citet{liu:wu:2003} and \citet{liu:etal:2005}, who  propose  a generalized mixed model to  predict the condition (healthy or diseased) based on the observed values of the biomarker.  \citet{Haben:etal:2019} consider a dynamic scoring prediction rule and then several extensions of the ROC curve are introduced. When the biomarker is discretely recorded with some possible noise, smoothing techniques, such as smoothing splines and kernel smoothing, can be employed, see \citet{wang:etal:2016}.

 In the context of functional biomarker, \citet{estevezperez:vieu:2021}  define a functional version of the ROC curve by properly ranking  the sample of functional data via the projection over a selected subset $\itE$ of the functional space indexed by a real number $\theta\in [0,1]$. The functional ROC curve is then defined as the ROC of the projections over $\itE$. As mentioned therein, their proposal is appropriate when differences between healthy and diseased individuals arise in their mean, but not on the covariances. Later,  \citet{jang:manatunga:2022} propose a  wide class of  features to summarize the functional biomarker into an univariate one, including an integral-type that can also be viewed as the inner $L^2-$product between the functional data and the constant function that corresponds to the average value of the functional marker. Among these features, they consider the average velocity and average acceleration as well as the maximum and minimum of the biomarker curve mentioned above.  The ROC curves and the corresponding AUC are then constructed using the  summary functional  and the smoothed trajectories.

 In this paper, we follow a unifying approach.  
We first review in Section \ref{sec:multiva} the basic notions related to ROC curves and the situation of multivariate biomarkers, where methods based on searching for linear combinations maximizing the AUC were given. 
Later on, in Section \ref{sec:functional} we adapt these revisited ideas to the functional framework assuming that the trajectories belong to a Hilbert space and we discuss the issues involved in the estimation of the optimal directions. We also go one step forward, by considering the case of possibly different covariance operators. Section \ref{sec:consist} is devoted to the study of asymptotic properties, such as the uniform consistency of the proposed ROC curve estimators and the strong consistency of the related  estimators of the AUC and the Youden index. In Section \ref{sec:monte} a thorough numerical experiment is performed, while in Section \ref{sec:realdata} the analysis of the real data set mentioned above illustrates the application of the studied procedures. Some final comments are provided in Section \ref{sec:finalcoment}. Proofs are relegated to the Appendix.


\section{Preliminaries}{\label{sec:multiva}}

\subsection{Basic notions and multivariate biomarkers} 

 We begin by introducing basic concepts related to ROC curves in the simple case of a  univariate biomarker. Assume that $Y$ is a continuous biomarker  and let  $c$ be a threshold value. Thus, we consider the classification rule that assigns a subject to the  diseased group when $Y \ge c$ and to the healthy one, otherwise.  Moreover, denote $Y_{D}$ the marker in the diseased population and $F_{D}$ its distribution, while $Y_{H}$ and $F_{H}$ stand  for the marker and its distribution in the healthy population. As usual in this context, the populations are assumed to be independent. Our interest focuses on the evolution of the pairs $\{(1-F_{H}(c),1-F_{D}(c)) \}$ as the threshold $ c \in \real$ varies. This leads to the usual  ROC curve formula given by $\ROC(p)= 1-F_{D}\{F_{H}^{-1}(1-p)\}$, $p \in (0,1) $. 

 An extensively used model is the  binormal  one, that assumes that in both populations the marker is normally distributed, i.e., $F_j\sim N(\mu_j, \sigma_j)$, for $j=D,H$. In this case, the distributions are characterized by the means $\mu_{j}$ and the standard deviations $\sigma_{j}$, $j=D,H$, leading to
 $\ROC(p) = \Phi\left\{({\mu_{H}-\mu_{D}})/{\sigma_{D}}+ {\sigma_{H}}\Phi^{-1}(p)/{\sigma_{D}} \right\}$, where   $\Phi$ stands for the cumulative distribution function of a standard normal variable. 

 In order to have a global measure of the accuracy of the test, different indexes have been introduced. The  area under the ROC curve (AUC) is the most popular and is defined as
$\AUC= \int_0^1 \ROC(p) dp$. 
Straightforward calculus enables to prove that $\AUC= \prob(Y_{D} > Y_{H})$ and this is why values of AUC near to 1 are related to high diagnostic accuracy of the biomarker. The Youden index, $\YI$, is also very well known. It measures the proximity  of the ROC curve to the identity function, thinking of the identity as the ROC of useless marker, and is defined as $\YI=\max_{0 < p < 1 } \{\ROC(p)-p\}$. It is worth mentioning that when $Y_D$ is stochastically greater than $Y_H$, then $\ROC(p) \ge p$, so the area under the curve is greater or equal than 0.5 and  the Youden index is non--negative. 

 Now, let us consider the multivariate case where $\bx\in \real^k$ is the multivariate biomarker  used for the diagnostic of a given disease. From now on,  $\bx_D \sim F_D$  and    $\bx_{H} \sim F_{H}$ stand for the independent biomarkers over the diseased and healthy populations, respectively.  A well--known discrimination rule is the  linear one, which intends to project the data over a given direction chosen to differentiate both groups. For a given $\bbe\in \real^k$ denote $Y_{j,\bbech}=\bbe\trasp \bx_j$, for $j=D,H$, and let  $F_{j, \bbech}$ their respective distribution functions. Then, a ROC curve can be constructed for each fixed $\bbe\ne \bcero_k$ as ${\ROC}_{\bbech}(p)= 1-F_{D, \bbech}\left\{F_{H,\bbech}^{-1}\left(1-p\right)\right\}$, $p \in (0,1) $.
${\ROC}_{\bbech}$ allows to evaluate the capability of the projected biomarker to distinguish between the two groups. Among others,  \citet{ma:huang:2005} and \citet{pepe:cai:Longton:2006}  provide methods to choose the best direction by means of the  area under the curve $\AUC(\bbe)=\int_0^1 {\ROC}_{\bbech}(p) dp= \prob(Y_{D,\bbech} >Y_{H,\bbech})$. In Section \ref{sec:binormal}, we review the construction of the optimal direction on a population level for normally distributed biomarkers, to understand how these ideas can be extended to the functional case, as presented in Section \ref{sec:linealfun}. 
  
\subsection{The binormal case}{\label{sec:binormal}}
The binormal model has been extensively considered for univariate  biomarkers as a way to  supply  a simple  parametric approach
to  ROC curves estimation.  Hence, if the practitioner suspects that this model  gives  a suitable approximation, they can choose the threshold constants according to that belief, providing a semiparametric framework that will indeed keep consistency under the suspected model. In the multivariate setting, the binormal model also provides a simple expression for both the $\ROC_{\bbech}(p)$ and  the $\AUC(\bbe)$.

Let us assume that $\bx_j \sim N(\bmu_j, \bSi_j)$, hence $Y_{j,\bbech}\sim  N(\bbe\trasp\bmu_j, \bbe\trasp\bSi_j\bbe)$. Denote $z_p$  the value  such that $\Phi(z_p)=1-p$. As it is well known,  in such framework, the quantile and distribution functions of   $Y_{H,\bbech}$ and  $Y_{D,\bbech}$, respectively have an explicit expression given by
 $$F_{H,\bbech}^{-1}\left(1-p\right)=\bbe\trasp \bmu_{H}+ \left({\bbe\trasp \bSi_H \bbe}\right)^{1/2}  z_{p}\qquad \mbox{and}\qquad F_{D, \bbech}(u)= \Phi\left\{\frac{u-\bbe\trasp \bmu_{D}}{ \left({\bbe\trasp \bSi_D \bbe}\right)^{1/2}}\right\}\,,$$
which  allow  to express the related ${\ROC}_{\bbech}$ curve using the cumulative standard normal distribution as
$${\ROC}_{\bbech}(p)=  1-F_{D, \bbech}\left\{\bbe\trasp \bmu_{H}+ \left({\bbe\trasp \bSi_H \bbe}\right)^{1/2}  z_{p}\right\}=1- \Phi\left\{\frac{\bbe\trasp\left( \bmu_{H}-\bmu_{D}\right)+ \left({\bbe\trasp \bSi_H \bbe}\right)^{1/2}  z_{p} }{ \left({\bbe\trasp \bSi_D \bbe}\right)^{1/2}}\right\}\,.$$
Furthermore, taking into account that $Y_{D,\bbech}-Y_{H,\bbech} \sim N\left(\mu_{\bbech}, \sigma_{\bbech}^2\right)$, with $\mu_{\bbech}=\bbe\trasp (\bmu_D-\bmu_H)$ and $\sigma_{\bbech}^2=\bbe\trasp\left(\bSi_D+\bSi_H\right)\bbe$, an explicit expression may also be given for the $\AUC(\bbe)$ as
 $$\AUC(\bbe)=\prob(Y_{D,\bbech}-Y_{H,\bbech}>0)=\Phi\left[\frac{\bbe\trasp (\bmu_D-\bmu_H)}{\left\{\bbe\trasp\left(\bSi_D+\bSi_H\right)\bbe\right\}^{ 1/2}}\right]\,.$$
Note that when $\mu_D=\mu_H$ the $\AUC(\bbe)$ equals 1/2 for any $\bbe$ and the problem of maximizing $\AUC(\bbe)$ is vacuous. Furthermore, if $\Sigma_H=\Sigma_D$, then, the function  ${\ROC}_{\bbech}(p)$ is constant and equal to the identity function, since both populations cannot be distinguished.

\citet{ma:huang:2007} propose to estimate $\bbe$ as the value maximizing a smooth estimator, $\widehat{\AUC}(\bbe)$, of $\AUC(\bbe)$. Taking into account that the area under the curve and ${\ROC}_{\bbech}$ remain unchanged if $\bbe$ is multiplied by a positive constant, \citet{ma:huang:2007} suggest to maximize  $\widehat{\AUC}(\bbe)$ over the values   $\bbe$ such that the first coordinate equals 1. When constructing the smooth estimator  $\widehat{\AUC}(\bbe)$ the indicator function is approximated by the sigmoid function.  
 
It is clear that the population counterparts of the estimators defined in   \citet{ma:huang:2007} correspond to the value $\bbe$ maximizing $\AUC(\bbe)$ over the directions with its first component equal to 1, or equivalently to the maximizer of 
\begin{equation}
\label{eq:eleRp}
L(\bbe)=\frac{ \bbe\trasp (\bmu_D-\bmu_H) }{\left\{\bbe\trasp\left(\bSi_D+\bSi_H\right)\bbe\right\}^{1/2}}\,.
\end{equation}
The Cauchy--Schwartz inequality implies that  any scalar multiple of $\bbe_0=\left(\bSi_D+\bSi_H\right)^{-1}\left(\bmu_{D}-  \bmu_H\right)$ maximizes $|L(\bbe)|$, so $\bbe_0$ (and any its positive multiples) maximizes  $\AUC(\bbe)$ leading to
\begin{equation*}
\max_{\bbech}\AUC(\bbe)=\AUC(\bbe_0)=\Phi\left[ \left\{\left(\bmu_{D}-  \bmu_H\right)\trasp\left(\bSi_D+\bSi_H\right)^{-1}\left(\bmu_{D}-  \bmu_H\right)\right\}^{1/2}\right]\,.
\end{equation*} 
Note that when $\bSi_D=\bSi_H$, the Fisher discriminating direction is obtained.

 Define $\itA_{\bbech}(p)=\{\bx \in \real^k: \bbe\trasp \bx \ge \bbe\trasp \bmu_{H}+ \left({\bbe\trasp \bSi_H \bbe}\right)^{1/2} \; z_{p}\}$,
 so that $\prob\{\bx_H \in \itA_{\bbech}(p)\} =p$. Then, using that $\bx_D \sim N(\bmu_D, \bSi_D)$ and denoting as $Z\sim N(0,1)$, we get that
\begin{align}
\prob\left(\bx_D \in \itA_{\bbech}(p)\right)&= \prob\left\{\bbe\trasp \bx_D \ge \bbe\trasp \bmu_{H}+ \left({\bbe\trasp \bSi_H \bbe}\right)^{1/2}  z_{p}\right\}
\nonumber\\
& = \prob\left\{Z \ge \frac{\bbe\trasp \left(\bmu_{H}-  \bmu_D\right)+ \left({\bbe\trasp \bSi_H \bbe}\right)^{1/2}  z_{p}}{\left({\bbe\trasp \bSi_D \bbe}\right)^{1/2} }\right\}
\nonumber\\
&= 1- \Phi\left\{\frac{\bbe\trasp \left(\bmu_{H}-  \bmu_D\right)+ \left({\bbe\trasp \bSi_H \bbe}\right)^{1/2}  z_{p}}{\left({\bbe\trasp \bSi_D \bbe}\right)^{1/2} }\right\}
\nonumber\\
& = \Phi\left\{\frac{\bbe\trasp \left(\bmu_{D}-  \bmu_H\right)- \left({\bbe\trasp \bSi_H \bbe}\right)^{1/2}  z_{p}}{\left({\bbe\trasp \bSi_D \bbe}\right)^{1/2} }\right\}\,,
\label{eq:expresion1}
\end{align}
 which implies that ${\ROC}_{\bbech}(p)=\prob\left\{\bx_D \in \itA_{\bbech}(p)\right\}$. 

From now on, to simplify the discussion below, assume  that   $\bSi_j=\bSi$, for $j=D, H$. 

 A global ROC curve has been defined as $\ROC(p)=\sup_{\|\bbech\|=1} \prob\left\{\bx_D \in \itA_{\bbech}(p)\right\}$.
Using \eqref{eq:expresion1} and that  $\bSi_j=\bSi$, we obtain that the global ROC curve equals
\begin{align*}
\ROC(p) 
&=\sup_{\|\bbech\|=1}  \Phi\left\{\frac{\bbe\trasp \left(\bmu_{D}-  \bmu_H\right)- \left({\bbe\trasp \bSi_H \bbe}\right)^{1/2}  z_{p}}{\left({\bbe\trasp \bSi_D \bbe}\right)^{1/2} }  \right\}\\
&=\sup_{\|\bbech\|=1} \Phi\left\{\frac{\bbe\trasp \left(\bmu_{D}-  \bmu_H\right)}{\left({\bbe\trasp \bSi  \bbe}\right)^{1/2}}-  \; z_{p} \right\}=\sup_{\|\bbech\|=1}  \Phi\left\{\sqrt{2}\,L(\bbe)-  \; z_{p}  \right\}\,.
\end{align*}
Hence, taking into account that $\bbe_0$ maximizes $L(\bbe)$, we obtain that the supremum is a maximum and is attained at $\bbe_0= \bSi^{-1}\left(\bmu_{D}-  \bmu_H\right)$. Therefore, the global ROC is given by
\begin{align*}
\ROC(p)& =   {\ROC}_{\bbech_0}(p)= \Phi\left[\left\{\left(\bmu_{D}-  \bmu_H\right)\trasp \bSi^{-1}\left(\bmu_{D}-  \bmu_H\right)\right\}^{1/2}-  \; z_{p}\right]  \\
&=1-\Phi\left[z_p\,-\,\left\{\left(\bmu_{D}-  \bmu_H\right)\trasp \bSi^{-1}\left(\bmu_{D}-  \bmu_H\right)\right\}^{1/2}\right]  \,,
\end{align*}
meaning that the global ROC curve is the ROC curve associated to the optimal direction with respect to the area under the curve.

Beyond the AUC, the Youden index,   which measures the difference between the ROC curve and the identity function, may also be maximized to obtain the associated optimal direction. In this case, the induced optimality problem searches for the direction $\bbe$ such that 
$$\YI(\bbe)=\max_{c\in \real}\left|\Phi\left\{\frac{c-\bbe\trasp \bmu_D}{\left({\bbe\trasp \bSi   \bbe}\right)^{1/2}}\right\}-\Phi\left\{\frac{c-\bbe\trasp \bmu_H}{\left({\bbe\trasp \bSi   \bbe}\right)^{1/2}}\right\}\right|=\max_{c\in \real} \Delta_{\bbech}(c)$$
is maximum. 
Straightforward calculations relegated to the Appendix allow to show that
\begin{equation}
\argmax_{\|\bbe\|=1} \YI(\bbe)= \frac{\bSi^{-1} \left(\bmu_{D}-  \bmu_H\right)}{\left\{\left(\bmu_{D}-  \bmu_H\right)\trasp \bSi^{-1}\left(\bmu_{D}-  \bmu_H\right)\right\}^{1/2}}\;.
\label{eq:optYIbeta}
\end{equation}
Hence, if the covariance matrices are equal, the value $\bbe$ maximizing $\YI(\bbe)$ is proportional to $\bSi^{-1}(\bmu_{D}-  \bmu_H)$ and coincides with the direction yielding  the maximum AUC.

 \section{Functional setting}{\label{sec:functional}}
In this section, we consider functional biomarkers belonging to a separable Hilbert space $\itH$ which, without loss of generality may be assumed to be $L^2(0,1)$. More precisely, we assume that each functional biomarker $X_j\in \itH$, $j=D,H$, and denote as $P_j$ the probability measure related to $X_j$. From now on,  $\|\cdot\|$ and $\langle \cdot, \cdot \rangle$ stand for the norm and the inner product in $\itH$.

Let $\Upsilon: \itH\to \real$  be an operator used as discrimination index to classify a new observation to one of the two classes. Clearly, from this index a related ROC curve may be defined as
\begin{equation}
\label{eq:ROCUpsi}
{\ROC}_{\Upsilon}(p)=  1- F_{\Upsilon,D}\left\{ F_{\Upsilon,H}^{-1}\left(1-p\right)\right\}, \quad p \in (0,1) \,, 
\end{equation}
where $F_{\Upsilon,D}$ and $F_{\Upsilon,H}^{-1}$  stand  for the distribution function and quantile function of  $Y_{\Upsilon,j}=\Upsilon(X_j)$,   $j= D,H$, respectively, that is, $F_{\Upsilon,j}(u)= P_j\left\{\Upsilon(X_j)\le u\right\}$. The related area under the curve is then obtained as
${\AUC}_{\Upsilon}= \prob\left(Y_{\Upsilon,D} > Y_{\Upsilon,H}\right)$.
When independent samples  $X_{j,i}$, $i=1,\ldots, n_j$,  $j= D,H$, are available,  estimators of  $\ROC_{\Upsilon}$ and $\AUC_{\Upsilon}$ may be obtained if the discriminating index has a closed form, as it is the case for integral  of the biomarker curve, its  maximum and/or minimum and other features described in \citet{jang:manatunga:2022}.  In this case, defining, $Y_{j,i}=\Upsilon(X_{j,i})$, $i= 1, \ldots , n_j$,  $j= D,H$, the ROC estimator is obtained as 
\begin{equation}
\label{eq:wROCUpsi}
\widehat{\ROC}(p)=\widehat{\ROC}_{\Upsilon}(p)= 1-\wF_{\Upsilon,D}\left\{\wF_{\Upsilon,H}^{-1}\left(1-p\right)\right\}, \quad p \in (0,1) \,, 
\end{equation}
where $\wF_{\Upsilon,D}$ and $\wF_{\Upsilon,H}^{-1}$  stand  for the empirical distribution function and quantile function of the samples $Y_{j,i}$, $  i=1, \ldots, n_j$,  $j= D,H$, respectively. From the univariate case, the $\YI$ estimator can be obtained just by plugging--in the formula $\widehat{\YI}=\widehat{\YI}_{\Upsilon}=\max_{0 < p < 1 } \{\widehat{\ROC}_{\Upsilon}(p)-p\} $,
 while the AUC estimator is defined as
 $$\widehat{\AUC}=\widehat{\AUC}_{\Upsilon}= \frac{1}{n_D n_H}\sum_{i=1}^{n_D}\sum_{\ell=1}^{n_H} \indica_{\{Y_{D,i} > Y_{H,\ell} \}}\,.$$

In the previous situation, the operator chosen to construct the univariate marker is completely known. However, we would like to explore more general situations where $\Upsilon$ may depend on unknown parameters that must be estimated from the data. This may lead to more complex indexes that require, for instance, the estimation of the mean or the covariance operators from both populations. 
Denote by $\wUps$ the predicted discrimination index, and let  $\wY_{j,i}=\wUps(X_{j,i})$, $i=1, \ldots, n_j$,  $j= D,H$. Then, the ROC and AUC estimators may be defined as
\begin{align}
\label{eq:wROCwUps}
\widehat{\ROC}(p) =\widehat{\ROC}_{\wUps}(p)& = 1-\wF_{\wUps,D}\left\{\wF_{\wUps,H}^{-1}\left(1-p\right)\right\}, \quad p \in (0,1) \,, \\
\widehat{\AUC}=\widehat{\AUC}_{\wUps}& =  \frac{1}{n_D n_H} \sum_{i=1}^{n_D}\sum_{\ell=1}^{n_H}  \indica_{\{\wY_{D,i} > \wY_{H,\ell}\}}\,,
\label{eq:wAUCwUps}
\end{align}
where $\wF_{\wUps,D}$ and $\wF_{\wUps,H}^{-1}$  stand now for the empirical distribution function and quantile function of the predictors $\wY_{j,i}=\wUps(X_{j,i})$, $i=1, \ldots, n_j$,  $j= D,H$, respectively.  We proceed similarly with the estimation of  the Youden index.

\subsection{Linear discriminating index: Maximizing the AUC}{\label{sec:linealfun}}
In this section, we will work at a population level and define a linear operator $\Upsilon=\Upsilon_\beta$ leading to the maximum AUC when both populations have different mean functions. For the sake of simplicity, as in the multivariate setting, we will consider the situation where $X_j$ are Gaussian processes with mean $\mu_j$ and covariance operator   $\Gamma_j$, denoted $\itG(\mu_j, \Gamma_j)$.  Then, if $\beta \in \itH$, $\|\beta\|=1$, and we denote $Y_{j,\beta}=\Upsilon_\beta(X_j)=\langle \beta, X_j\rangle$, for $j=D,H$, and by $F_{j, \beta}$ their distribution functions, we have that $Y_{j,\beta} \sim N(\langle \beta, \mu_j\rangle, \langle \beta, \Gamma_j \beta\rangle)$ whenever $\beta\notin \mbox{Ker}(\Gamma_j)$. 
It is worth mentioning that if the kernel of $\Gamma_j$ does not reduce to $0$, for any $\beta\in \mbox{Ker}(\Gamma_j)$, we have that
$\prob(Y_{j,\beta}=\langle \beta, \mu_j\rangle)=1$. Let us consider the case where $\mu_j \in \mbox{Ker}(\Gamma_j)^\bot$, then if $\Gamma_D=\Gamma_H=\Gamma$, $\mbox{Ker}(\Gamma )\ne \{0\}$ and $\beta\in \mbox{Ker}(\Gamma)$, we have that $\prob(Y_{j,\beta}=0 )=1$, for $j=D,H$, so the two populations cannot be distinguished in such directions.

As in the multivariate case for each fixed $\beta\in \itH$, $\|\beta\|=1$,  $\beta \notin  \mbox{Ker}(\Gamma_D) \cap  \mbox{Ker}(\Gamma_H)$, the AUC associated to the new independent biomarkers  $Y_{j,\beta} $ is given by
$$\AUC(\beta)= \prob(Y_{D,\beta} > Y_{H,\beta})=\Phi\left[\frac{\langle \beta, \mu_D-\mu_H\rangle}{\left\{ \langle \beta, (\Gamma_H+\Gamma_D) \beta\rangle\right\}^{1/2}}\right] \,.$$
As in the finite--dimensional case, if $\mu_H=\mu_D$ and $\Gamma_H=\Gamma_D$, the projected univariate biomarkers $Y_{j,\beta}$ cannot be distinguished. Furthermore, if $\mu_D=\mu_H$ but $\Gamma_H\ne \Gamma_D$ using a linear rule will not be a suitable choice in practice since  $\AUC(\beta)=1/2$, for any $\beta\in \itH$, so the problem of searching for a direction maximizing the AUC is not adequate   and other rules  should be employed. For that reason, along this section  we will assume $\mu_D\ne \mu_H$.

Then, if we denote $\Gamma_{\ave}=(\Gamma_H+\Gamma_D)/2$, the element $\beta_0$ maximizing $\AUC(\beta)$ may be obtained up to a positive constant as $\beta_0= \argmax_{\beta\ne 0} L(\beta)$, where 
\begin{equation}
\label{eq:Lave}
L(\beta)= \frac{\langle \beta, \mu_D-\mu_H\rangle}{\left(2 \langle \beta,\Gamma_{\ave} \beta\rangle\right)^{1/2}}\,.
\end{equation}
It is worth mentioning the analogy with the expression given in \eqref{eq:eleRp}. 
In practice, the population parameters in equation \eqref{eq:Lave} need to be estimated as will be described in the next Section.

Unlike the multivariate case and as in canonical correlation, this maximization problem poses some challenges due to its infinite--dimensional structure. The major problem is due to the fact that the operator $\Gamma_{\ave}=(\Gamma_H+\Gamma_D)/2$ is compact and then it does not have an inverse.
To clarify the difficulties and the relation to canonical correlation, let   $X$ be such that  $X \mid G=1\sim X_D$ and $X \mid G=0\sim X_H$ with $G$ a binary variable indicating the group membership, such that $\prob(G=1)=\pi_D$ and denote $\pi_H=1-\pi_D$.

\begin{lemma}{\label{lema:Cancor}}
Let $L_{\pool}(\beta)$ be defined as 
\begin{equation}
\label{eq:Lpool}
L_{\pool}(\beta)=\frac{\langle \beta, \mu_D-\mu_H\rangle}{\left(2 \langle \beta,\Gamma_{\pool} \beta\rangle\right)^{1/2}}\,,
\end{equation}
where  $\Gamma_{\pool}= \pi_D  \Gamma_D+ \pi_H \Gamma_H$. Then, if $\mu_D\ne \mu_H$
\begin{itemize}
\item[a)] for any $\beta \notin  \mbox{Ker}(\Gamma_D) \cap  \mbox{Ker}(\Gamma_H)$,
$\mbox{corr}^2(\langle \beta, X\rangle, G)= \pi_D\pi_H \left[1-  \{1 +L_{\pool}^2(\beta) \}^{-1}\right]\,.$
\item[b)] When $\pi_D=1/2$ or when  $\Gamma_H =\Gamma_D$, the problem of maximizing the AUC   and that of maximizing $\mbox{corr}(\langle \beta, X\rangle, G)$ coincide.
\end{itemize}
\end{lemma}

In the sequel $\Gamma$ will denote either $\Gamma_{\pool}$ or $\Gamma_{\ave}$. Note that when $\Gamma_H =\Gamma_D$, then $\Gamma_{\pool}=\Gamma_{\ave}=\Gamma$. The following proposition provides an explicit expression for the direction maximizing $L_{\pool}^2(\beta)$ and/or the AUC.
 
From now on denote $\lambda_1 \ge \lambda_2\ge \dots $ the eigenvalues of $\Gamma$ and $\phi_j$ the corresponding eigenfunctions.

\begin{proposition} \label{prop:expresionbeta}
Let us assume that the eigenvalues $\lambda_j$  of $\Gamma$ are all positive and define the linear space
$ \itR(\Gamma)=\left \{y\in \itH:\; \sum_{\ell \ge 1}  {\lambda_{\ell}^{-2}}\; \langle y, \phi_{\ell}\rangle^2 <\infty \right\}$,
and the inverse of $\Gamma:\itR(\Gamma)\to \itH$ as
$\Gamma^{-1 } (y)=\sum_{\ell \ge 1} {\lambda_{\ell}^{-1}}\; \langle y, \phi_{\ell}\rangle\; \phi_{\ell}$,  for any $y\in \itR(\Gamma)$.
Furthermore, assume that $\mu_D-\mu_H\in \itR(\Gamma)$ and $\mu_D\ne \mu_H$.
Then, if  $\Gamma=\Gamma_{\pool}$  and $\beta_0$ stands for the value maximizing $L_{\pool}^2(\beta)$ or if $\Gamma=\Gamma_{\ave}$  and $\beta_0$ stands for the value maximizing the AUC, then $\beta_0=  ({\pi_D\pi_H})^{1/2}\,\Gamma^{-1}\left(\mu_D-\mu_H\right)$.
\end{proposition}

\begin{remark}{\label{remark:bielip}}
It is worth mentioning that similar arguments enable to extend the results of Proposition \ref{prop:expresionbeta} to settings  more general than the Gaussian. In fact, straightforward arguments allow to show that $\beta_0$ still maximizes the AUC if 
the biomarkers have an elliptical distribution with scatter operators, as defined in \citet{bali:boente:2009} and studied in \citet{boente:salibian:tyler:2014}. In what follows we briefly present the arguments leading to this conclusion. 

 To state the definition of elliptical distributions in a functional setting, we will first remind the basic concept of elliptical distributions in $\real^k$. Recall that a random vector $\bz \in \real^k$ is said to have a $k$--dimensional spherical distribution if its distribution is invariant under
orthogonal transformations. In general, the characteristic function of a spherically distributed $\bx \in
\real^k$ is of the form $\psi_{\bx}(\bt_k) = \varphi(\bt_k\trasp\bt_k)$ for $\bt_k \in
\real^k$, and any distribution in $\real^k$ having a characteristic function of
this form is a spherical distribution. Hence, we can denote a spherically distributed vector as  $\bx \sim {\itS}_k(\varphi)$, which is convenient since, for $\bx\trasp =
(\bx_1\trasp,\bx_2\trasp)$ with $\bx_1 \in \real^m$, we have that $\bx_1  \sim {\itS}_m(\varphi)$. 

 Elliptical distributions in $\real^k$ correspond to
distributions obtained from affine transformations of spherically distributed
random vectors in $\real^k$. More precisely, for a given  matrix $\bA\in \real^{k \times k}$ and a vector
$\bmu \in \real^k$, the distribution of $\bx = \bA\bz + \bmu$ when $\bz \sim
{\itS}_k(\varphi)$ is said to have an elliptical distribution, denoted  
$\bx \sim {\itE}_{k}(\bmu, \bSi,\varphi)$, where $\bSi = \bA\bA\trasp$.  
When first moment exists,  $\esp(\bX) = \bmu$. Furthermore, when second moments exist then the  covariance
matrix  of $\bx$ is proportional to $\bSi$. 

 It is easy to see that the
characteristic function of $\bx$  equals $\psi_{\bx}(\bt) = \exp(i\bt\trasp\bmu) \varphi(\bt\trasp \bSi \bt)$. Hence,  the scatter matrix $\bSi$
  is confounded with the function $\varphi$ in the sense that, for any $c > 0$,
${\itE}_k(\mu, \bSi,\varphi) \sim {\itE}_k(\mu, c \bSi, \varphi_c)$ where $\varphi_c(w) = \varphi(w/c)$. For that reason, henceforth, we will assume that the characteristic function $\varphi$ is chosen so that the covariance matrix of $\bx$ equals $\bSi$.  

 An important property of elliptical distributions, is that the sum of independent elliptical random vectors with the same scatter matrix  $\bSi$ is elliptical, see \citet{hult:lindskog:2002} and  \citet{frahm:2004}. This fact is important in what follows, since it implies that if $\bx_D\sim {\itE}_k(\bmu_D, \bSi,\varphi_D)$ and $\bx_H  \sim {\itE}_k(\bmu_H, \bSi,\varphi_H)$ are independent, then $   \bx_D-\bx_H  \sim \itE_k(\bmu_D-\bmu_H,\bSi, \varphi_D\times\varphi_H)$ and $ \bbe\trasp (  \bx_D-\bx_H)  \sim  \bbe\trasp (\bmu_D-\bmu_H ) + z (\bbe\trasp \bSi  \bbe)^{1/2}$, where the random variable $z$  has a symmetric distribution $G_0$ with characteristic function   $\varphi_z(t)=\varphi_D(t^2)\;\varphi_H(t^2)$.

 \citet{bali:boente:2009} define elliptical distributed random elements    in a separable Hilbert-space $\itH$ as follows. The random element $X$ is said to have an elliptical distribution $\itE(\mu, \Gamma, \varphi)$ with parameters $\mu \in \itH$ and $\Gamma: \itH \to \itH$  a self-adjoint, positive semi--definite and compact operator if   and only if for any linear and bounded operator $A: \itH \to \real^k$, we have that $AX\sim {\itE}_{k}(A\mu, \,
A \Gamma A^*, \, \varphi)$ where $A^*: \real^k\rightarrow \itH$ stands for 
the adjoint operator of $A$. As noted in \citet{boente:salibian:tyler:2014},  $X\sim \itE(\mu, \Gamma,\varphi)$ if and only if  $\langle a, X \rangle \sim
{\itE}_{1}(\langle a, \mu \rangle, \langle a, \Gamma a \rangle,\varphi)$ for
all $a \in {\itH}$. 

 The above discussion implies that, if $X_j\sim \itE(\mu_j, \Gamma, \varphi_j)$, for $j=D,H$, then the projected biomarkers  AUC  can   be expressed as
$\AUC(\beta)= \prob(Y_{D,\beta} > Y_{H,\beta})= G_0\left\{ {\langle \beta, \mu_D-\mu_H\rangle}/{\left( \langle \beta,  \Gamma  \beta\rangle\right)^{1/2}}\right\} \,,$
with $G_0$ a distribution function symmetric around $0$, that is, $G_0(t)=1-G_0(-t)$. Therefore, the value maximizing the AUC is still the one maximizing $L_{\pool}(\beta)$ in \eqref{eq:Lpool} and is given in Proposition \ref{prop:expresionbeta}. 
\end{remark}

\subsection{Estimating the linear projection  index}{\label{sec:estlinfun}} 
 As mentioned above, general discriminating indexes require the estimation of unknown parameters. In this section, we focus on the linear index $\Upsilon(X)=\langle \beta_0, X\rangle$, where $\beta_0$ maximizes the AUC as given in Proposition \ref{prop:expresionbeta}.
 
  Suppose that independent samples  $X_{j,i}$, $i=1, \ldots, n_j$,  $j= D,H$, are available. As shown by \citet{leurgans:etal:1993} for the case of canonical correlation between two functional random elements, the sample maximum correlation can attain the value $1$, for proper directions. The same arises in the present situation, where  the sample version of the AUC may lead to values close to 1 for proper directions, meaning that the maximizer of $\wL(\beta)$ or equivalently of $\widehat{\AUC}(\beta)$ will not provide relevant information.

As in functional canonical correlation, the problem may be  overcome using increasing finite--dimensional linear spaces and/or penalizations.  More precisely, let $D\alpha=\alpha^{\prime\prime}$  and $\Psi(\alpha)=\|D\alpha\|$, then we can penalize the  denominator in $L(\beta)$ to define
$$\wL_{\lambda}(\beta)= \frac{\langle \beta, \wmu_D-\wmu_H\rangle}{\left\{ \langle \beta, \wGamma \beta\rangle+\lambda\Psi(\beta)\right\}^{1/2}}\,,$$
where $ \wmu_D$, $\wmu_H$ are estimators of $\mu_D$ and $\mu_H$ and $\wGamma$ is an estimator of $\Gamma_{\ave}$, such as their sample versions. A typical choice for $ \wmu_j$ is the sample mean $\overline{X}_j=(1/n_j) \sum_{i=1}^{n_j} X_{j,i} $. When both samples have the same covariance matrix a possible estimator for $\Gamma_{\ave}=\Gamma_{\pool}$  is the pooled covariance operator given by
$  \wGamma_{\pool}= ({n_D}/{n}) \wGamma_D+  ({n_H}/{n}) \wGamma_H $, with $n=n_D+n_H$ and
$\wGamma_j=(1/{n_j}) \sum_{i=1}^{n_j} (X_{j,i}- \overline{X}_j)\otimes  (X_{j,i}- \overline{X}_j)$,
 while if  the practitioner suspects that $\Gamma_D\ne \Gamma_H$, it is better to choose $\wGamma=\wGamma_{\ave}$, where
$\wGamma_{\ave}=  ( \wGamma_D+   \wGamma_H )/2$.
  We seek for the values  $\wbeta$ maximizing $\wL_{\lambda}(\beta)$ over the set $\{\beta \in \itH_k: \|\beta\|=1\}$  with $\itH_k$ a finite--dimensional linear space of dimension $k$, such as the one spanned by the first elements of the Fourier basis. An adaptive basis, as the one spanned by the first eigenfunctions of the pooled sample operator may also be  chosen. Note that if no dimension reduction is performed, that is,  when $\wbeta$ maximizes $\wL_{\lambda}(\beta)$ over   the unit ball in $ \itH$,   the procedure corresponds to the optimal scoring approach to discriminant analysis described in \citet{ramsay2005functional}.

As mentioned above, once the direction $\wbeta$ is obtained, the estimated linear discrimination index, denoted $\wUps_{\lin}$, equals $\wUps(X)=  \wUps_{\lin}(X)=\Upsilon_{\wbeta}(X)=\langle \wbeta, X\rangle $ may be constructed, leading to the real--valued predictors
 $\wY_{j,i}=\wUps(X_{j,i})=\Upsilon_{\wbeta}(X_{j,i})=\langle \wbeta, X_{j,i}\rangle$, $i=1, \ldots, n_j$,  $j= D,H$ from which the ROC curve estimator may be constructed as 
\begin{equation}
\widehat{\ROC}(p)=\widehat{\ROC}_{\Upsilon_{\wbeta}}(p)= 1-\wF_{D}\left\{\wF_{H}^{-1}\left(1-p\right)\right\}, \quad p \in (0,1) \,, 
 \label{eq:ROC:LIN}
\end{equation}
where $\wF_{D}=\wF_{\Upsilon_{\wbeta},D}$ and $\wF_{H}^{-1}=\wF_{\Upsilon_{\wbeta},H }^{-1}$  stand  for the empirical distribution function and quantile function of the predictors $\wY_{j,i}$, $i=1, \ldots, n_j$,  $j= D,H$, respectively.


\subsection{Definition of a quadratic discrimination index}{\label{sec:quadfun}}
 In the situations where the covariance operators  differ between populations,  some improvements to the linear index defined above may be obtained in terms of the AUC by considering alternative indexes. As mentioned above, the situation is even worst when population means are equal in which case, the linear rule is an improper one. For instance, as it is well known, for multivariate normally distributed biomarkers the quadratic discriminating rule offers a procedure with better classification rates than the linear one  when the covariance matrices of both populations are different from each other, in particular for unbalanced samples. We will explore a generalization  of these ideas in the functional framework.

Measuring differences between covariance operators or even between covariance matrices is difficult, we refer to \citet{Flury:1988} who mentioned that \textsl{\lq\lq In contrast to the univariate situation, inequality is not just inequality--there are indeed many ways in which
covariance matrices can differ\rq\rq}. For that reason, in the functional literature some parsimonious models have been considered, including models with proportional covariance operators or models assuming that both covariance operators share the same eigenfunctions. In this section, we focus on these settings and we will use the basis of principal directions to reduce the dimension, even when any basis can be chosen to project the data and construct the quadratic index when we suspect that differences between covariance operators arise.

A natural extension of functional principal components to several populations, which  corresponds to the generalization of the common principal components model introduced by \citet{Flury:1984} to the functional setting, is to assume that the covariance operators $\Gamma_j$ have common eigenfunctions $\phi_\ell$ but possible different eigenvalues $\lambda_{j,\ell}$, i.e.,
\begin{equation}
\Gamma_j = \sum_{\ell=1}^\infty \lambda_{j,\ell}\; \phi_\ell \otimes \phi_\ell\;,  
\label{eq:fcpc}
\end{equation}
where, to identify the directions, we assume that the eigenvalues of the healthy population are ordered in decreasing order, that is,  $\lambda_{H,1}\geq\lambda_{H,2}\ge \cdots \geq\lambda_{H,\ell}\geq \lambda_{H,\ell+1}\geq \cdots$. This model is usually denoted the functional common principal component  model (\textsc{fcpc}) and provides a framework for analysing different population data that share their main modes of variation $\phi_1, \phi_2, \dots$ using a parsimonious approach.  When the eigenvalues preserve the order across populations,   i.e.,  if   
\begin{equation}
\lambda_{j,1}\geq\lambda_{j,2}\ge \cdots \geq\lambda_{j,\ell}\geq \lambda_{j,\ell+1}\geq\cdots\, , \mbox{ for } j=D,H\,,
\label{eq:ordenlamdas}
\end{equation} 
as assumed, for instance, in  \citet{Benko:Hardle:2005}  and \citet{boente:rodriguez:sued:2010}, the common directions will represent, as in the one--population setting, the main modes of variation for each population. Furthermore, in such a situation, the operators $\Gamma_{\pool}$ and $\Gamma_{\ave}$ will also have the same principal directions and in the same order.

However, if   the largest $k$ eigenvalues do not preserve the order among populations, that is,  if we only have
$\lambda_{j,\ell}\geq \lambda_{j,k+1}\geq \lambda_{j,k+2}\ge \dots \ge 0$, for  $j=D,H$ and $\ell= 1, \ldots, k$,
the eigenfunctions $\phi_1, \dots, \phi_k$ represent the modes of variation that are common to each group, even when the ordering across  groups changes. As  mentioned in \citet{Coffey:etal:2011}, the eigenvalues $\lambda_{j, \ell}$, $ \ell= 1, \ldots,  k$, determine the order of the common directions in each group and may allow to study the differences in the distribution of the variation across groups. The  functional common principal component  model may be used to reduce the dimensionality of the data, retaining the maximum variability present in each of the populations. 

Assume that the  covariance operators of both populations satisfy a \textsc{fcpc} model and that \eqref{eq:ordenlamdas} holds, then   the first $k$ principal directions provide a natural linear space where the data projection may be performed. More precisely, define the $k-$dimensional vectors 
$\bx_{j}=(\langle \phi_1, X_{j}\rangle, \dots,$ $ \langle \phi_k, X_{j}\rangle)\trasp$,   $j= D,H$, 
where $\phi_\ell$, $\ell=1,\dots k$, stand for the first $k$ common principal directions. Usually in practice the dimension $k$ is chosen so as to explain a given percentage of the total variability measured through the eigenvalues of an estimator of $\Gamma_{\pool}$. 

Note that when $ X_{j}$, $j=D,H$, are Gaussian processes with mean $\mu_j$ and covariance operators $\Gamma_j$ satisfying \eqref{eq:fcpc}, then $\bx_{j,i}\sim N(\bmu_j, \bSi_j)$ with $\bmu_{j}=(\mu_{j,1}, \dots, \mu_{j,k})\trasp$, $\mu_{j,\ell}=\langle \phi_\ell, \mu_j\rangle$ and $\bSi_j= \diag( \lambda_{j,1},\dots,\lambda_{j,k}) $, meaning that $\bx_{j,1}$, $j=D,H$, fulfil the common principal components model (\textsc{cpc}) considered in \citet{Flury:1984}. From these projections, if $\lambda_{j,k}>0$, for $j=D,H$, a quadratic discriminant index may be defined as 
$$\Upsilon_{\cuad}(X)=\Upsilon_{\bLamch_0, \balfach_0 }(X)=\, - \bx\trasp \bLam_0 \bx + 2 \balfa_0\trasp \bx\,,$$
where $\bx=(\langle \phi_1, X \rangle, \dots, \langle \phi_k, X \rangle)\trasp$, $\bLam_0=\bSi_D^{-1}-\bSi_H^{-1}$ and $\balfa_0=  \bSi_D^{-1} \bmu_D-\bSi_H^{-1} \bmu_H $. Note that  under the \textsc{fcpc} model
\begin{align*}
\bLam_0 & = \diag\left(\frac{\lambda_{H,1}-\lambda_{D,1}}{\lambda_{D,1}\, \lambda_{H,1}},\dots,\frac{\lambda_{H,k}-\lambda_{D,k}}{\lambda_{D,k}\, \lambda_{H,k}}\right)=\diag(\Lambda_1, \dots, \Lambda_k)\\
\balfa_0 & =  \left(\begin{array}{c}
\dst\frac{\mu_{D,1}\, \lambda_{H,1}-\mu_{H,1}\, \lambda_{D,1}}{\lambda_{D,1}\, \lambda_{H,1}}\\
\vdots\\
\dst\frac{\mu_{D,k}\, \lambda_{H,k}-\mu_{H,k}\, \lambda_{D,k}}{\lambda_{D,k}\, \lambda_{H,k}}
\end{array}\right) = \left(\begin{array}{c}
\dst\frac{\mu_{D,1}}{\lambda_{D,1}}-\frac{\mu_{H,1}}{ \lambda_{H,1}}\\
\vdots\\
\dst\frac{\mu_{D,k}}{\lambda_{D,k}}-\frac{\mu_{H,k}}{ \lambda_{H,k}}
\end{array}\right)  \,.
\end{align*}
 Besides, to be precise we should write $\Upsilon_{\cuad}(X)=\Upsilon_{k,\cuad}(X)$ since it depends on the number of selected common directions. However, we have decided to omit its dependence on $k$ to simplify the notation. 

 Under mild assumptions, Proposition \ref{prop:cuad} below provides an expression of the quadratic discriminating rule which suggests an asymptotic expression for it, as the number of principal directions increases.

\begin{proposition}{\label{prop:cuad}}
Assume that $\Gamma_j$ satisfy \eqref{eq:fcpc} and \eqref{eq:ordenlamdas}, with $\lambda_{j,\ell}>0$, for $j=D,H$ and $\ell\ge 1$, and that $\mu_j\in \itR(\Gamma_j)$. Denote
$\alpha_0=\Gamma_D^{-1} \mu_D-\Gamma_H^{-1} \mu_H$ and let $A:\itH\to \real^k$ stand for the projection operator
$Ay=\left(\langle \phi_1, y\rangle, \dots, \langle \phi_k, y\rangle\right)\trasp$, for any $y\in \itH$.  Then,
\begin{itemize}
\item[a)] $\balfa_0\trasp \bx=\langle A^{*} A \alpha_0, X\rangle = \sum_{\ell=1}^k \langle  \alpha_0, \phi_\ell \rangle \; \langle \phi_\ell, X \rangle$ where $A^{*}:\real^k \to \itH$ stands for the adjoint operator of $A$ given by $ A^{*} \bu= \sum_{\ell=1}^k u_\ell \phi_\ell$.
\item[b)] $\bx\trasp \bLam_0 \bx = \|A \Gamma_D^{-1/2} X\|^2- \|A \Gamma_H^{-1/2} X\|^2$, if $X\in \itR(\Gamma_D^{1/2})\cap \itR(\Gamma_H^{1/2})$. Hence, 
$$\Upsilon_{\cuad}(X)=\, - \left(\|A \Gamma_D^{-1/2} X\|^2 - \|A \Gamma_H^{-1/2} X\|^2 \right) +2\, \langle A^{*} A \alpha_0, X\rangle\,.$$ 
Then, for any $X\in \itR(\Gamma_D^{1/2})\cap \itR(\Gamma_H^{1/2})$, if the number $k$ of principal directions increases to infinity the index $\Upsilon_{\cuad}(X)$ converges to $\Upsilon (X)=\, - \left(\|  \Gamma_D^{-1/2} X\|^2 - \|  \Gamma_H^{-1/2} X\|^2 \right) +2\, \langle  \alpha_0, X\rangle$.
\end{itemize}
\end{proposition}
Two facts should be highlighted regarding Proposition \ref{prop:cuad}. On the one hand, note that from the Karhunen--Lo\`eve expansion of the Gaussian processes $X_j$, we get that $\prob\{X_j-\mu_j\in \itR(\Gamma_j^{1/2})\}=0$. Effectively, the mentioned expansion leads to $X_j-\mu_j=\sum_{\ell\ge 1} \xi_{j,\ell} \; \phi_\ell$ where $\xi_{j,\ell}$ are independent and $\xi_{j,\ell}\sim N(0, \lambda_{j,\ell})$ which implies that $\langle X_j, \phi_\ell\rangle^2/\lambda_{j,\ell}=\xi_{j,\ell}^2/\lambda_{j,\ell} \sim \chi^2_1$, meaning that $X_j-\mu_j\notin \itR(\Gamma_j^{1/2})$ with probability 1.
Thus, the linear space  where $\Upsilon $ is defined is too small to define a proper discriminating rule. In this sense,  $\Upsilon_{\cuad} =\Upsilon_{k,\cuad} $ circumvents the \textsl{curse of dimensionality} imposed by the infinite--dimensional structure of the model and provides a finite--dimensional approximation of $\Upsilon$ whose range is the whole space $\itH$ ensuring the definition of a proper rule. On the other hand, the requirement that  $\lambda_{j,\ell}>0$, for $j=D,H$ and $\ell\ge 1$, is needed to guarantee that   $\Gamma_j$ has an inverse over $\itR(\Gamma_j)$. Moreover, if $\lambda_{j,\ell}=0$, for  $\ell\ge k_0$, then the quadratic rule $\Upsilon_{k,\cuad} $ may only be defined when $k< k_0$, since otherwise the matrix $\bSi_j$ will be singular.

Clearly, the principal directions are unknown and must be estimated from the sample. These estimators denoted  $\wphi_\ell$, $\ell=1,\dots k$, may be obtained, for instance, as the eigenfunctions related to the largest $k$ eigenvalues of the sample pooled covariance operator, see for instance, \citet{boente:rodriguez:sued:2010}.  We then may define the $k-$dimensional vectors 
$\wbx_{j,i}=(\langle \wphi_1, X_{j,i}\rangle, \dots, \langle \wphi_k, X_{j,i}\rangle)\trasp$, $i=1, \ldots, n_j$,  $j= D,H$, which provide predictors of the finite--dimensional vectors $\bx_{j,i}=(\langle \phi_1, X_{j,i}\rangle, \dots, \langle \phi_k, X_{j,i}\rangle)\trasp$.

In this case, the quadratic discrimination rule provides a better approach to classify a new observation to each group. For that reason, we propose to consider as estimated quadratic index
\begin{equation}
\label{eq:UPCUAD}
\wUps_{\cuad}(X)= \Upsilon_{\wbLamch, \wbalfach  }(X)=\, - \bx\trasp \wbLam \bx + 2 \wbalfa\trasp \bx\,,
\end{equation}
where $\bx=(\langle \wphi_1, X \rangle, \dots, \langle \wphi_k, X \rangle)\trasp$, $\wbLam=\wbSi_D^{-1}-\wbSi_H^{-1}$, $\wbalfa=  \wbSi_D^{-1} \wbmu_D-\wbSi_H^{-1} \wbmu_H $ where $\wbmu_j$  and $\wbSi_j$ stand for  the sample mean and covariance matrix of $\{\bx_{j,i}\}_{i=1}^{n_j}$, $j=D,H$, respectively. 

Again, the estimated ROC curve may be  defined as
\begin{equation}
\widehat{\ROC}(p)=\widehat{\ROC}_{\wUps_{\cuadch}}(p)= 1-\wF_{D}\left\{\wF_{H}^{-1}\left(1-p\right)\right\}, \quad p \in (0,1) \,, 
\label{eq:ROC:CUAD}
\end{equation}
where $\wF_{D}=\wF_{\wUps_{\cuadch},D}$ and $\wF_{H}^{-1}=\wF_{\wUps_{\cuadch},H}^{-1}$  denote  the empirical distribution function and quantile function of the predictors $\wY_{j,i}=\wUps_{\cuad}(X_{j,i})$, $i=1, \ldots, n_j$,  $j= D,H$, respectively. 

As mentioned in \citet{Flury:Schmid:1992}, for the multivariate setting, the quadratic rule may be adapted to the setting of CPC or proportional models estimating the parameters  under these constraints. This procedure may lead to more stable estimations than those obtained using the sample covariance estimators $\wbSi_j$, specially for small samples, leading to better rates of misclassification and probably to higher AUC values, when the most
parsimonious among all the correct models is used for discrimination. We leave the  interesting  topic of comparing the performance of the ROC curve estimators obtained using constrained estimators   for future research.

 \section{Consistency results}{\label{sec:consist}}
 In this section, we establish consistency results for some discriminating indexes under mild assumptions. In particular, our results include  cases where the discriminating index depends on a known operator, as those considered in \citet{jang:manatunga:2022}. We also analyse  the situation where  the index $\Upsilon(X)$ is linear and defined by means of a coefficient $\beta_0$ that may be known as in the case of $\Upsilon(X)=\int_{\itI} X(t) dt$ where $\beta_0\equiv 1$ or unknown as it arises when considering, for instance, $\Upsilon(X)=\langle \mu_D-\mu_H, X\rangle$  where $\beta_0=\mu_D-\mu_H$ or $\Upsilon(X)=\langle \beta_{0}, X\rangle$, where $\beta_{0}$ maximizes $L_{\pool}(\beta)$. Besides, we   study the situation of the quadratic discrimination rule defined in Section \ref{sec:quadfun}, when the number of principal directions is fixed.

   Let $X_j\sim P_j$, $j=D,H$, be a multivariate or functional biomarker. Consider  $\Upsilon: \itH\to \real$ a  discrimination index used to define the ROC curve  $ {\ROC}_{\Upsilon}$ as in \eqref{eq:ROCUpsi}, where for simplicity we denote $F_j=F_{\Upsilon,j}$    the distribution function  of  $Y_j= \Upsilon(X_j)$,   $j= D,H$. Recall that the related area under the curve equals $ {\AUC}_{\Upsilon}= \prob\left(Y_{\Upsilon,D} > Y_{\Upsilon,H}\right)$. 
   
  The estimator of  the ${\ROC}$ curve is obtained in many situations by means of an estimated discrimination index $\wUps$, unless $\Upsilon$ is completely known in which case in what follows  $\wUps=\Upsilon$. This estimator is based on   independent samples  $X_{j,i}$, $i=1, \ldots, n_j$,  $j= D,H$,  which allow to define $\widehat{\ROC} =\widehat{\ROC}_{\wUps} $ and $\widehat{\AUC}=\widehat{\AUC}_{\wUps}$ as in \eqref{eq:wROCwUps} and \eqref{eq:wAUCwUps}, respectively. For simplicity, we label $\wF_{j}=\wF_{\wUps,j }$  for the empirical distribution function  of the predictors $\wY_{j,i}=\wUps(X_{j,i})$, $i=1, \ldots, n_j$,  $j= D,H$.   It is worth noting that the asymptotic results we derive are based on the assumption that sample sizes grow to infinity, that is $n_j \to \infty$,  $j= D,H$.

In order to derive consistency results for the ROC curve estimators $\widehat{\ROC}_{\wUps} $, we will need the following assumptions.
  
\begin{enumerate}[label = \textbf{A\arabic*}]
\item 	\label{ass:A1}   $F_H=F_{H,\Upsilon} : \real\to  (0, 1)$ has density $f_{H,\Upsilon}$ such that $f_{H,\Upsilon}(u)>0$, for all $u\in \real$.
 
\item	\label{ass:A2} $F_D=F_{D,\Upsilon} : \real\to  (0, 1)$  is continuous.

\item 	\label{ass:A3} $\|\wF_{j}-F_j\|_{\infty}\convpp 0$, $j=D,H$.
\end{enumerate}

 Assumptions   \ref{ass:A1} and \ref{ass:A2} are the usual assumptions required to obtain consistency when using a  univariate biomarker. Some situations where \ref{ass:A3} holds are discussed in Remark \ref{remark:remarkGC} and below.
   
Theorem \ref{theo:consist.1}  below states the consistency for  the  studied estimators of the ROC curve. 

\begin{theorem}  \label{theo:consist.1} 
Let  $\{X_{j,i} \}_{i=1}^{n_j}$, $j=D, H$, be independent observations   and   let  $\wUps$ be an estimator  of $\Upsilon$. Assume that Assumptions   \ref{ass:A1}--\ref{ass:A3}  hold. Then,       
\begin{enumerate}
\item[a)] $\sup_{0<p<1} |\widehat{\ROC}(p)- {\ROC} (p)|\convpp 0$,
\item[b)] $ \widehat{\AUC}   \convpp  {\AUC}  $,
\item[c)] $ \widehat{\YI}  \convpp  {\YI}  $.
\end{enumerate}
\end{theorem}

\begin{remark}{\label{remark:remarkGC}}
Note that    the Glivenko--Cantelli Theorem entails that Assumption \ref{ass:A3} holds when $\Upsilon$ is known. Hence, when the discriminating  index is linear $\Upsilon(X)=\langle \beta_0, X\rangle$ where $\beta_0$ is known, Theorem \ref{theo:consist.1} provides mild conditions ensuring consistency of the estimated ROC curve. 

More precisely, if the distribution function   $F_{j,\Upsilon}$ of $Y_j=\Upsilon(X_j)$ are continuous, for $j=D,H$, and $F_{H,\Upsilon}$ has density $f_{H,\Upsilon}$ such that $f_{H,\Upsilon}(u)>0$, for all $u\in \real$, then the conclusion of Theorem \ref{theo:consist.1} holds for the estimator given in \eqref{eq:wROCUpsi}. 

Theorem \ref{theo:consist.1}  also provides consistency results for the discriminating indexes defined in \citet{jang:manatunga:2022}, $\Upsilon(X)=\int_{\itI} X(t) dt$, $\Upsilon(X)=\int_{\itI} X^{\prime}(t) dt=X(b)-X(a)$, $\Upsilon(X)=\int_{\itI} X^{\prime\,\prime}(t) dt=X^{\prime}(b)-X^{\prime}(a)$  as well as  $\Upsilon(X)=\max_{t\in \itI} X(t)$ and  $\Upsilon(X)=\min_{t\in \itI} X(t)$.
\end{remark}

By means of Theorem \ref{theo:consist.1}, Theorems \ref{theo:consist.2}  and \ref{theo:consist.3} below state consistency results when the discriminating index is linear. More precisely,  let $\beta_0$ be the direction in which we project in order to determine the ROC curve, i.e., we assume that $\Upsilon(X)=\langle \beta_0, X\rangle$.
Furthermore, denote $\Upsilon_{\beta}(X)=\langle X, \beta\rangle$ and  $F_{\beta,j}=F_{\Upsilon_{\beta},j}$  the distribution of $\Upsilon_{\beta}(X_j)$ when $X_j\sim P_j$, while for simplicity, as above, we will call $F_{j}$ that corresponding to $F_{\beta_0 , j}$.
   
From now on,  $\wbeta$ stands for an estimator of $\beta_0$ which  we will assume to be consistent.  We denote $\wF_{\beta,j }=\wF_{\Upsilon_{\beta},j}$ the empirical distribution related of $\Upsilon_{\beta}(X_{j,i})=\langle X_{j,i}, \beta\rangle$, $i=1, \ldots, n_j$, $j=D,H$, while, to abbreviate,  $\wF_j$ stands for $\wF_{\wbeta,j}$. The ROC curve associated to $\Upsilon(X)=\Upsilon_{\beta_0}(X)=\langle \beta_{0}, X\rangle$ can be written as $\ROC(p)= 1-F_{D}\{F_{H}^{-1}(1-p))\}$, while its estimator equals $\widehat{\ROC}(p) = 1-\wF_{D}\left\{   \wF_{H}^{-1}(1-p) \right\}\,.$  Without loss of generality, we may assume that $\|\beta_0\|=1$ and $\|\wbeta\|=1$.

Assumptions \ref{ass:A1} and \ref{ass:A2} state that $F_{\beta_0,j} $  are continuous for $j=D,H$ and $F_{\beta_0,H} $ has a strictly positive density $f_{\beta_0 ,H}$. The key point is to provide conditions ensuring that Assumption \ref{ass:A3} holds.

In the finite--dimensional setting, Theorem \ref{theo:consist.2} ensures that   Assumption \ref{ass:A3} is a consequence of the strong consistency of $\wbeta$, while Theorem \ref{theo:consist.3} extends the result to the case of an infinite--dimensional biomarker requiring a finite expansion for the possible estimators $\wbeta$. In Theorem \ref{theo:consist.2} below, to  strengthen  the fact that we are dealing with finite--dimensional observations, the biomarkers and the index coefficient are indicated with boldface.

\begin{theorem}  \label{theo:consist.2} 
Assume that $\itH=\real^k$ and let  $\{\bx_{j,i} \}_{i=1}^{n_j}$, $j=D, H$, be independent observations in $\real^k$. Let  $\wbbe$ be a strongly consistent estimator  of $\bbe_{0}$ and assume that $F_{\bbech_0,j} $  are continuous for $j=D,H$. Then, 
\begin{itemize}
\item[a)] Assumption \ref{ass:A3}  holds.  
\item[b)] If in addition   $F_{\bbech_0,H} $ has density $f_{\bbech_0 ,H}$ such that $f_{\bbech_0,H}(u)>0$, for all $u\in \real$, then the conclusion of Theorem \ref{theo:consist.1} holds.
\end{itemize}
\end{theorem}

Let us consider now the situation of a separable Hilbert space, where the estimator $\wbeta$ is obtained using finite--dimensional candidates obtained from a fixed basis. More precisely, assume that  
\begin{equation}
\wbeta=\sum_{s=1}^k \wb_s \phi_s\,,
\label{eq:wbeta}
\end{equation}
 where the coefficients $\wb_s$, $s=1,\dots, k$, are data--dependent and the dimension $k=k_n$ increases with the sample size $n=n_D+n_H$.
   
\begin{theorem}  \label{theo:consist.3} 
Let  $\{X_{j,i} \}_{i=1}^{n_j}$, $j=D, H$, be independent observations in a separable Hilbert space $\itH$,  let  $\wbeta$ in \eqref{eq:wbeta}  be  an estimator  of $\beta_{0}$ such that $  \| \wbeta - \beta_{0}  \|\convpp 0$ and $k_n/n\to 0$ and assume that $F_{\beta_0,j} $  are continuous for $j=D,H$. Then,  
\begin{itemize}
\item[a)] Assumption \ref{ass:A3}  holds.  
\item[b)] If in addition   $F_{\beta_0,H} $ has density $f_{\beta_0 ,H}$ such that $f_{\beta_0,H}(u)>0$, for all $u\in \real$, then the conclusion of Theorem \ref{theo:consist.1} holds.  
\end{itemize}
\end{theorem}

To derive consistency results for the ROC curve associated to $\wUps_{\cuad}$  given in \eqref{eq:UPCUAD}, we will consider general quadratic indexes defined as  $\Upsilon_{\bLamch, \balfach} (X)= \,-\, \bx\trasp \bLam  \bx +  \balfa \trasp \bx$ where $\bx=A(X)=(\langle X, \phi_1\rangle, \dots,\langle X, \phi_k\rangle)\trasp$, $\balfa \in \real^k$ and $\bLam \in \real^{k\times k}$ and $\{\phi_\ell\}_{\ell \ge 1}$ is an orthonormal basis of $\itH$.   To avoid complicating the notation, we have omitted the index $k$ related to the number of principal directions chosen which will be assumed to be fixed. As above, we assume that 
the index $\Upsilon$ used to construct the ROC curve equals  $\Upsilon=\Upsilon_{\bLamch_0, \balfach_0}$ for some squared matrix $\bLam_0$ and vector $\balfa_0$. We also assume that estimators $(\wbLam, \wbalfa)$ of $(\bLam_0, \balfa_0)$ are available and that  the estimated ROC curve is defined  as in \eqref{eq:wROCwUps} through  the predictors $\wY_{j,i}=\wUps(X_{j,i})$, $i=1, \ldots, n_j$,  $j= D,H$, where   $\wUps=\Upsilon_{\wbLamch, \wbalfach}$. As for the linear index,   Assumption \ref{ass:A3} will follow from the consistency of $(\wbLam, \wbalfa)$. 

Again, Assumptions \ref{ass:A1} and \ref{ass:A2} state conditions on the behaviour of the distribution functions $F_j=F_{j, \Upsilon_{\bLamch_0, \balfach_0}}$ of $Y_j=\Upsilon(X_j)=\Upsilon_{\bLamch_0, \balfach_0}(X_j)$ which are the usual requirements to establish consistency for univariate biomarkers.
Hence, we have the following result.

\begin{theorem}  \label{theo:consist.4} 
Let  $\{X_{j,i} \}_{i=1}^{n_j}$, $j=D, H$, be independent observations in a separable Hilbert space $\itH$,   let  $(\wbLam, \wbalfa)$ be strongly consistent estimators   of  $(\bLam_0, \balfa_0)$ and assume that $F_j$ are continuous. Then,  
\begin{itemize}
\item[a)] Assumption \ref{ass:A3}  holds.  
\item[b)] If in addition  $F_{H}= F_{H, \Upsilon_{\bLamch_0, \balfach_0}}$ has density $f_{H}$ such that $f_{H}(u)>0$, for all $u\in \real$,   the conclusion of Theorem \ref{theo:consist.1} holds.  
\end{itemize}
\end{theorem}

\section{Monte Carlo Study}{\label{sec:monte}}
In this section, we report the results of a numerical study performed with the aim of comparing different rules used to define an estimator of the ROC curve. In particular, we consider the estimators defined through \eqref{eq:ROC:LIN} and \eqref{eq:ROC:CUAD}. The estimator defined in \eqref{eq:ROC:LIN} may also be used when $\beta(t)=1$ for all $t$, which leads to the ROC curve based on the  integral-type  index labelled $\Upsilon_{\inte}(X)=\int_{\itI} X(t) dt$.  Beyond these discriminating indexes, the maximum and the minimum of each trajectory are also used   as discriminating  rule, that is, the  ROC curve  is based on  $\Upsilon_{\maxi}(X_{j,i})$  or  $\Upsilon_{\mini}(X_{j,i})$,   $i=1, \ldots, n_j$,  $j= D,H$, where $\Upsilon_{\maxi}(X)=\max_{t\in \itI} X(t)$ and  $\Upsilon_{\mini}(X)=\min_{t\in \itI} X(t)$. The operators $\Upsilon_{\inte}$,  $\Upsilon_{\maxi}$ and $\Upsilon_{\mini}$ used in \citet{jang:manatunga:2022} do not involve unknown parameters. 

In contrast, we include in our simulations some operators that require estimation of population quantities.
First, we consider the non--necessarily optimal but natural linear rule defined as $\wUps_{\media}(X)=  \langle \overline{X}_D-\overline{X}_H, X\rangle$ which corresponds to an estimator of the linear discriminating rule based on  the mean difference $\Upsilon_{\media}(X)=\langle \mu_D-\mu_H, X\rangle$.
Secondly, we   consider the estimator of the optimal linear discriminating rule defined in Section \ref{sec:estlinfun} labelled $\wUps_{\lin}(X)=\langle \wbeta , X\rangle$, where the   $\wbeta $ is computed as follows.  
Let $\wGamma_{\ave}$ be the operator  $\wGamma_{\ave}=(\wGamma_D+   \wGamma_H)/2$ and denote $\itH_k$   the linear space spanned by the first $k$ eigenfunctions of the pooled covariance operator
$\wGamma_{\pool}=({n_D}/{n})  \wGamma_D + ({n_H}/{n}) \wGamma_H$, related to its largest eigenvalues, with $\wGamma_j$  the sample covariance operators of the $j-$th sample, $j=D,H$ and $n=n_D+n_H$.   The estimator $\wbeta$ is obtained by maximizing $\wL (\beta)= {\langle \beta, \wmu_D-\wmu_H\rangle}/{\left( \langle \beta, \wGamma_{\ave} \,\beta\rangle \right)^{1/2}}$,
over $\{\beta \in \itH_k: \|\beta\|=1\}$.
In all cases, the dimension $k$ is chosen as the smallest dimension ensuring that at least a 95\% of the total variability is explained, measured through the eigenvalues of $\wGamma_{\pool}$. 

Finally, we also include  in the comparison the quadratic discriminating rule $\wUps_{\cuad}$ defined in \eqref{eq:UPCUAD}.
The results for the AUC and ROC curve estimators obtained using each discriminating index are indicated in all Tables and Figures through the considered index.
We  consider  several frameworks including equal or different covariance operators as well as equal or different mean functions.

Sections \ref{sec:monte-gaussian}  and \ref{sec:desbalance} report the obtained results for different Gaussian processes,   when differences between mean curves are present and for equal mean functions. For all scenarios, we performed $1000$ replications where the trajectories were recorded over a grid of $100$ points equally spaced on $[0,1]$.  The estimated  ROC curves were  computed over a grid of $101$ points equally spaced on  $[0,1]$. In Section  \ref{sec:monte-gaussian}, we consider equal sample sizes for both populations, $n_D=n_H=300$, while Section \ref{sec:desbalance} is concerned with an unbalanced design, $n_D=30$ and $n_H=250$ and a setting with proportional covariance operators.


\subsection{Numerical results for balanced designs}{\label{sec:monte-gaussian}}
In this section we analyse the performance of the discriminating indexes mentioned above for different Gaussian processes.  In all settings $\mu_H(t)=0$, so that the possible differences between mean functions is given through $\mu_D$, which varies across scenarios.
We describe below the different scenarios considered which are labelled \textbf{PROP}, \textbf{CPC} and \textbf{DIFF}.

\begin{itemize}
\item  \textbf{PROP}: This scenario corresponds to the  case  of proportional covariance operators $\Gamma_D=\rho \Gamma_H$. When  $\rho=1$ it corresponds to the situation of equal covariance operators. In this last setting, the linear discriminating rule $\wUps_{\lin}$  usually improves the quadratic one in the multivariate setting for different mean functions. 
We consider the case of equal and different mean functions, that will be labelled \textbf{P0} and \textbf{P1}. 
\begin{itemize}
\item[$\star$]\textbf{P0}: Under this setting, $\mu_D(t)=\mu_H(t)=0$ and $\Gamma_D=\rho \Gamma_H$ with $\rho=2$.
\item[$\star$]\textbf{P1}: In this case, $\mu_D(t)=2 \,\sin(\pi\; t)$ and $\Gamma_D=\rho \Gamma_H$ with $\rho=1$ and $2$.
\end{itemize}

In these two scenarios both populations have the same underlying Gaussian distribution up to changes in the mean and/or covariance operators.

Two possible Gaussian processes were selected: the Brownian motion  and   a random Gaussian process with exponential kernel, labelled Exponential Variogram in all Tables and Figures. For the former,   covariance kernel for the healthy population equals $\gamma_H(s,t)=\min(s,t)$, while for the latter it corresponds to  $\gamma_H(s,t)=\exp(-|s-t|/\theta)$ with $\theta=0.2$.   

\item  \textbf{CPC}: This scenario corresponds to the situation where the covariance operators satisfy a \textsc{fcpc} model. The sample $X_{H,i}$, $i=1, \ldots, n_H$, was generated as a Brownian motion with kernel $\gamma_H(s,t)=\min(s,t)$. Recall that the eigenfunctions of the Brownian covariance operator are  $\phi_\ell(t)=\sqrt{2}\,\sin\{(2\,\ell-1)\,\pi\,t/2\}$ with related principal values $\lambda_{H,\ell} = 10\left[ {2}/\{(2\,\ell-1)\pi\}\right]^2$.

The diseased population is  a finite--range one, generated as $X_{D,i} = \mu_D + Z_{1,i} \phi_1 + Z_{2, i} \phi_2 + Z_{3,i}\phi_3$, 
where   $Z_{\ell,i} \sim N(0, \lambda_{D,\ell})$, with two possible choices for the variances of the scores.
\begin{itemize}
\item[$\star$]\textbf{C1}: In the first one, $\lambda_{D,1} = 2$, $\lambda_{D,2}=0.3$ and $\lambda_{D,3}=0.05$. In this situation the order between eigenvalues is preserved across populations.
\item[$\star$]\textbf{C2}: For the second choice, $\lambda_{D,1} = 0.3$, $\lambda_{D,2}=2$ and $\lambda_{D,3}=0.05$. Note that, in this setting, the order between eigenvalues is not preserved, meaning that the vectors $\bx_j=(\langle \phi_1,X_{j,1}\rangle, \langle \phi_2,X_{j,1}\rangle,\langle \phi_3,X_{j,1}\rangle)\trasp$ will be normally distributed with diagonal covariance matrices $\diag (\lambda_{j,1},\lambda_{j,2},\lambda_{j,3})$, but the order between the first two principal axes is reversed between populations.
\end{itemize} 

For each of the above described schemes, we allow for two different mean settings. On the one hand, we labelled with a $0$ after its identifier, that is,  according to the choice of $\lambda_{D,\ell}$, for $\ell=1,2,3$,  as \textbf{C10} or \textbf{C20} the situation where  $\mu_D(t)=0$, for all $t\in (0,1)$. On the other hand, the cases \textbf{C11} or \textbf{C21} correspond to the situation where the diseased population has   mean $\mu_D(t)=3 \,\sin(\pi\; t)$.

\item  \textbf{DIFF}: We also  consider  a situation  where the processes have different covariance operators and do not share their eigenfunctions. Hence, this framework is not included in the proportional or \textsc{fcpc} models described above. Scheme  \textbf{DIFF} includes two different settings, but in both of them, the sample $X_{D,i}$, $i=1, \ldots, n_D$, was generated as a Brownian motion with kernel $\gamma_H(s,t)=\min(s,t)$, whereas the sample  $X_{H,i}$, $i=1, \ldots, n_H$, was generated as described below.
\begin{itemize}
\item[$\star$]\textbf{D1}: Under \textbf{D1}, the healthy population corresponds to an Ornstein Uhlenbeck process with mean  $\mu_H\equiv 0$ and covariance kernel $\gamma_H(s,t)=\{{1}/({2\theta})\}\, \exp\{-\theta(s+t)\} \left[\exp\{2\,\theta(s+t)\}-1\right]\,,$   
with $\theta=1/3$. 
\item[$\star$]\textbf{D2}:  In this framework, $X_{H,i}$ is distributed as  a centered  Exponential Variogram, that is,  $\gamma_H(s,t)=\exp(-|s-t|/\theta)$ with $\theta=0.2$ and $\mu_H(t)=0$, for all $t$.   
\end{itemize} 
As above, we include two different mean scenarios: the one labelled a $0$ after its identifier, that is, as \textbf{D10} or \textbf{D20} corresponds to  the situation where  $\mu_H(t)=\mu_D(t)=0$, for all $t\in (0,1)$. In contrast, the cases \textbf{D11} or \textbf{D21} correspond to the situation where $\mu_H(t)=0$ and  the diseased population has   mean $\mu_D(t)=2 \,\sin(\pi\; t)$. 

\end{itemize} 

\begin{figure}[ht]
	\begin{center}
		\footnotesize
\renewcommand{\arraystretch}{0.1}
		\begin{tabular}{ccc}
		\textbf{P1}, $\rho=1$  & \textbf{P1}, $\rho=2$  &	\textbf{P0}\\
	\multicolumn{3}{c}{Brownian Motion} \\[-2ex]
			\includegraphics[scale=0.25]{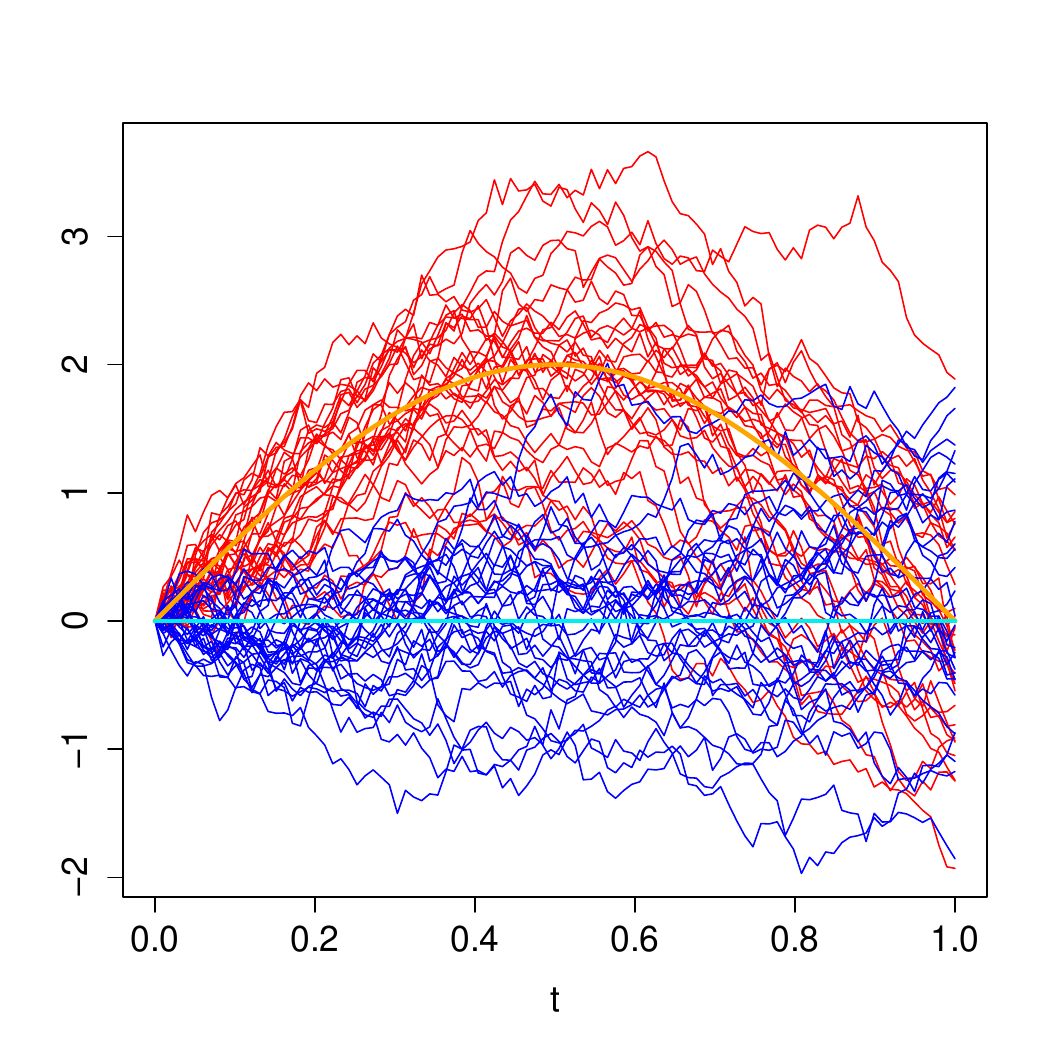} & 
			\includegraphics[scale=0.25]{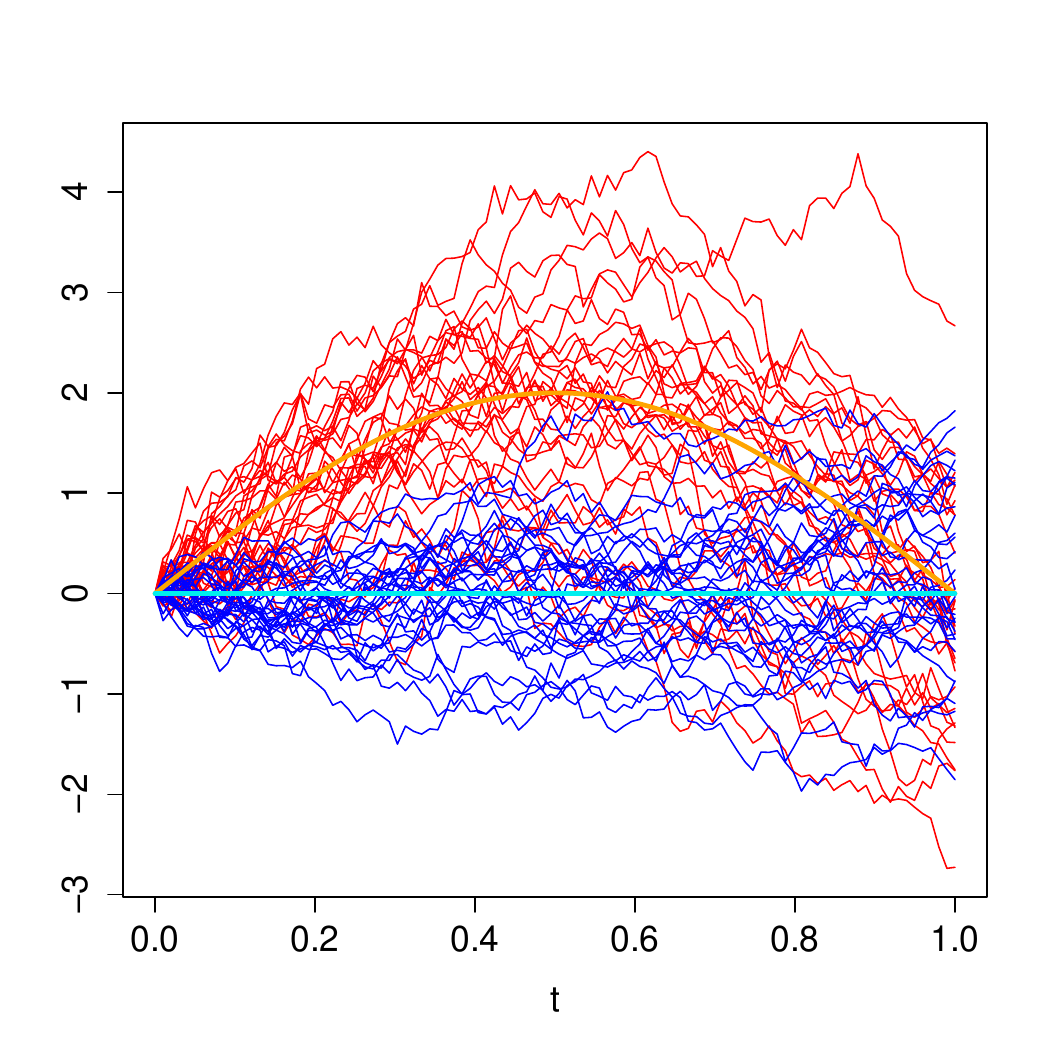} & 
			\includegraphics[scale=0.25]{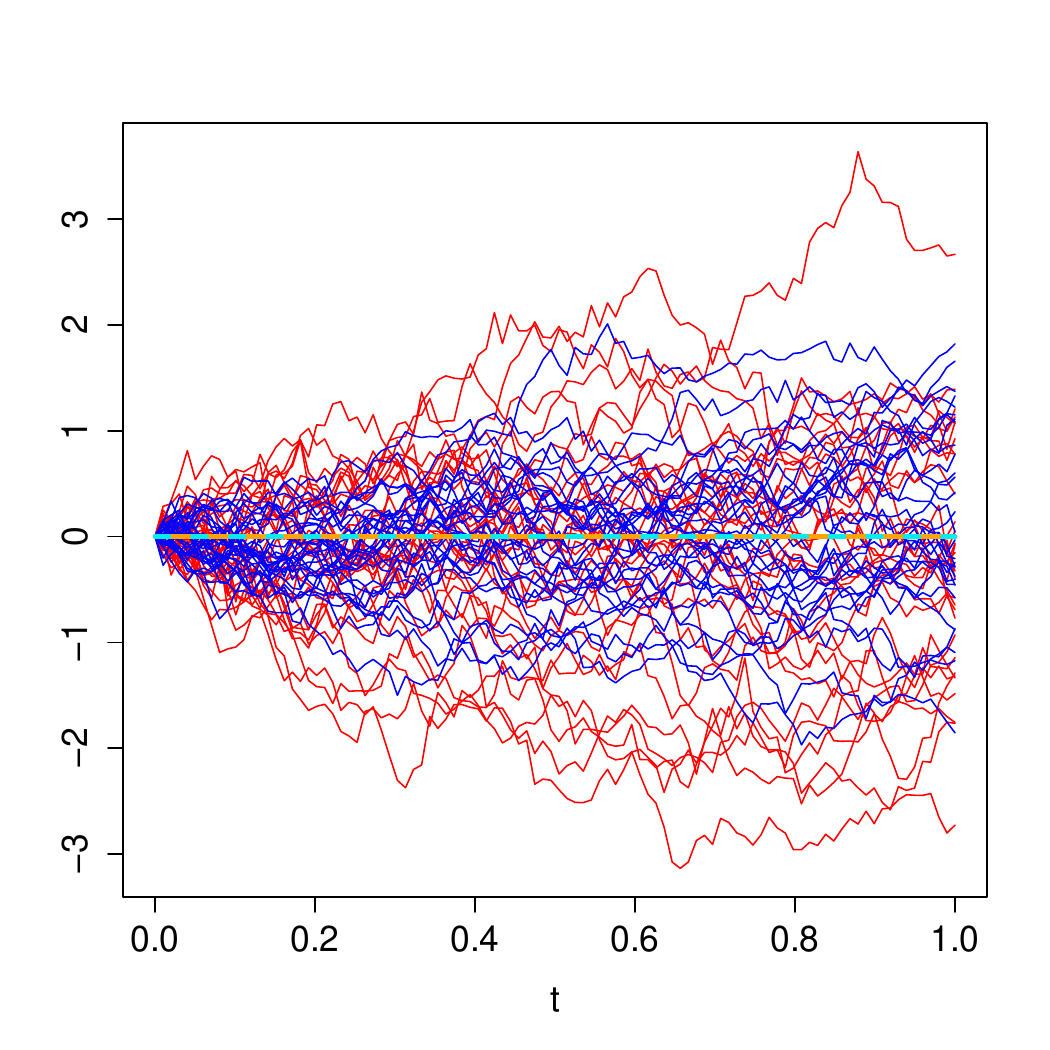}\\
			\multicolumn{3}{c}{Exponential Variogram}\\[-2ex]
			\includegraphics[scale=0.25]{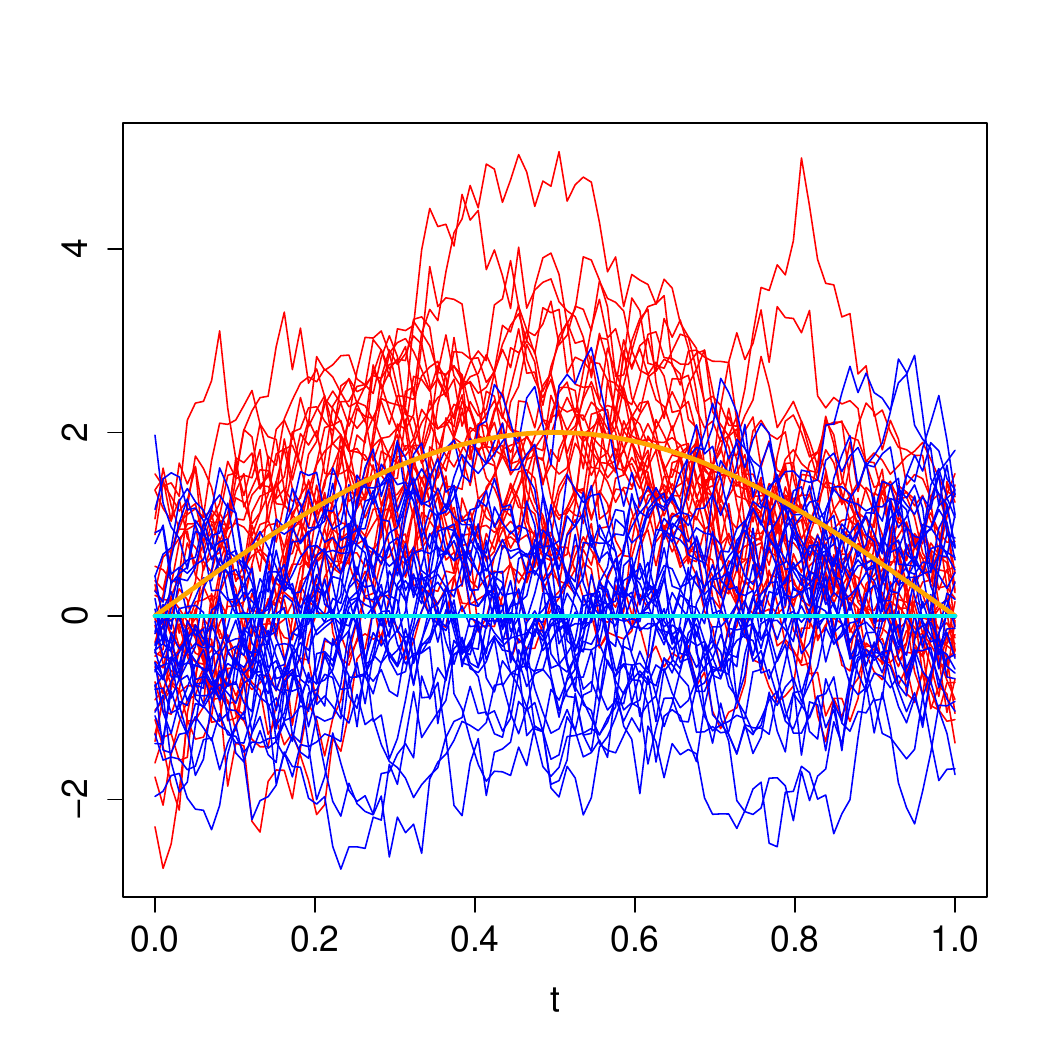} & 
			\includegraphics[scale=0.25]{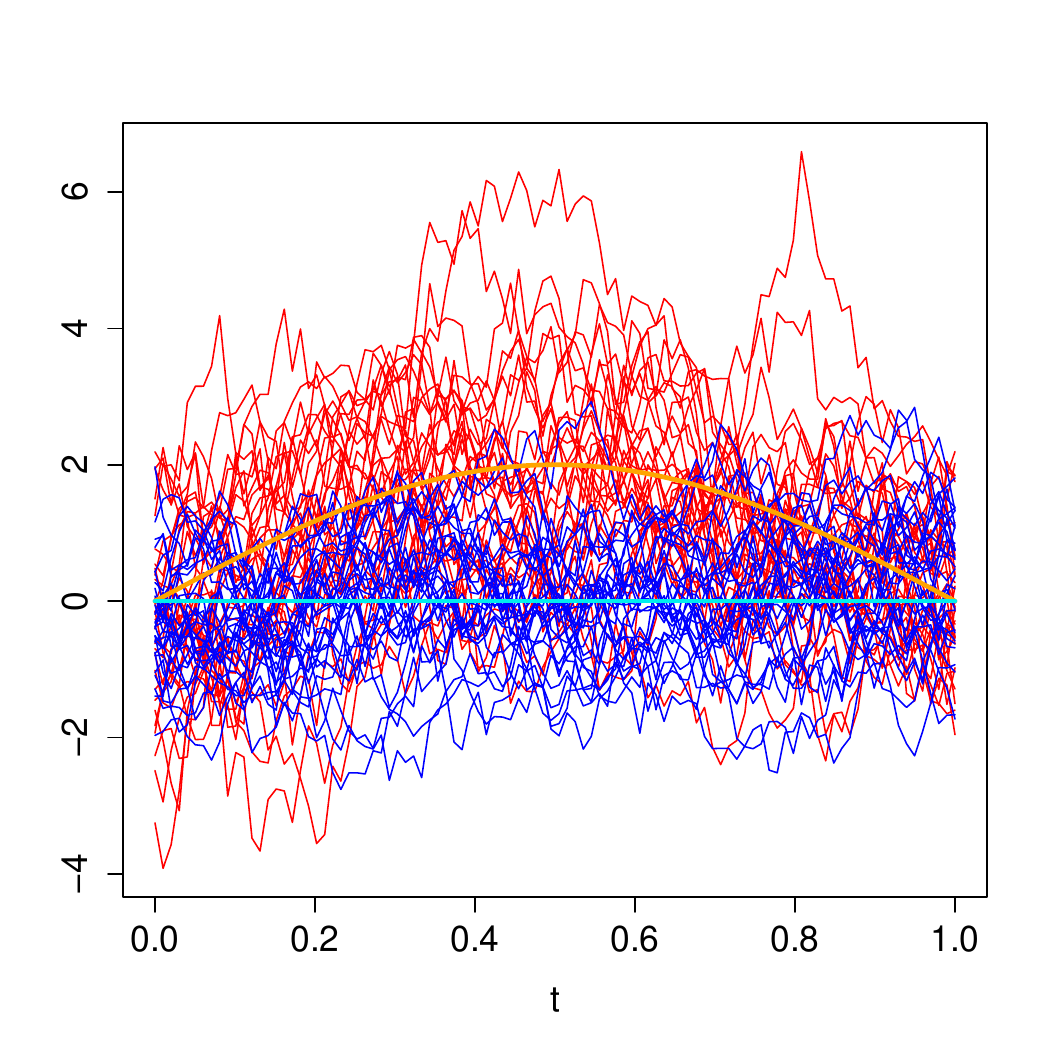}   & 
			\includegraphics[scale=0.25]{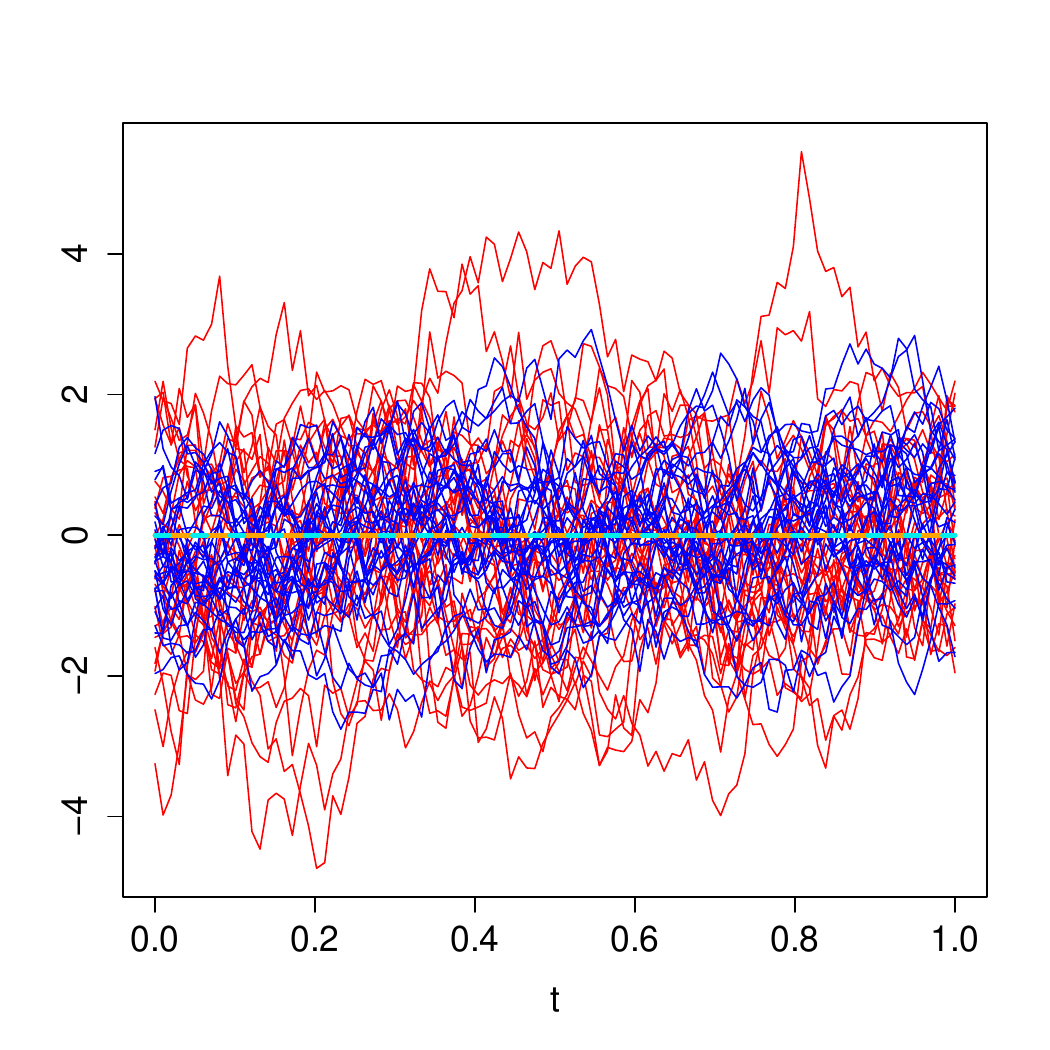}   
		\end{tabular}
		\vskip-0.1in 
		\caption{\label{fig:datos-propor} Data sets for proportional models (scheme \textbf{PROP}). Blue and red lines correspond to $X_{H,i}$ and $X_{D,i}$, respectively, while the true mean functions $\mu_H$ and $\mu_D$ are depicted in cyan and orange lines, respectively.}
	\end{center} 
\end{figure}

Figures \ref{fig:datos-propor} and \ref{fig:datos-cpc} depict 30 trajectories generated for the healthy and diseased populations in blue and red lines, respectively, under the schemes \textbf{PROP} and \textbf{CPC}, while Figure \ref{fig:datos-distintas} displays some of the generated trajectories for each scenario in \textbf{DIFF}. 
Figures \ref{fig:datos-propor} and \ref{fig:datos-cpc}  reveal that when  $\mu_H=\mu_D=0$, the differences between the two underlying distributions are more difficult to detect under the proportional model with $\rho=2$ than under the \textsc{fcpc} model. Under the considered   \textsc{fcpc} model, the smoothness and the differences in the range of the trajectories allow to distinguish the two populations. We do not consider the scheme,  $\mu_D(t)=\mu_H(t)=0$ and $\Gamma_D=\rho\; \Gamma_H$ with $\rho=1$, since in this case both populations have the same distribution. For that reason in Table  \ref{tab:summary-propor} below the cells corresponding to \textbf{P0} and $\rho=1$ are empty. Besides, when both population means are equal, under scenario \textbf{D20} the populations may be easily discriminated by looking at their behaviour at $t=0$, in contrast scheme  \textbf{D10} seems more challenging.

\begin{figure}[ht!]
	\begin{center}
		\footnotesize
\renewcommand{\arraystretch}{0.1}
\begin{tabular}{cc} 
 \textbf{C11} &  \textbf{C10} \\[-2ex]
\includegraphics[scale=0.25]{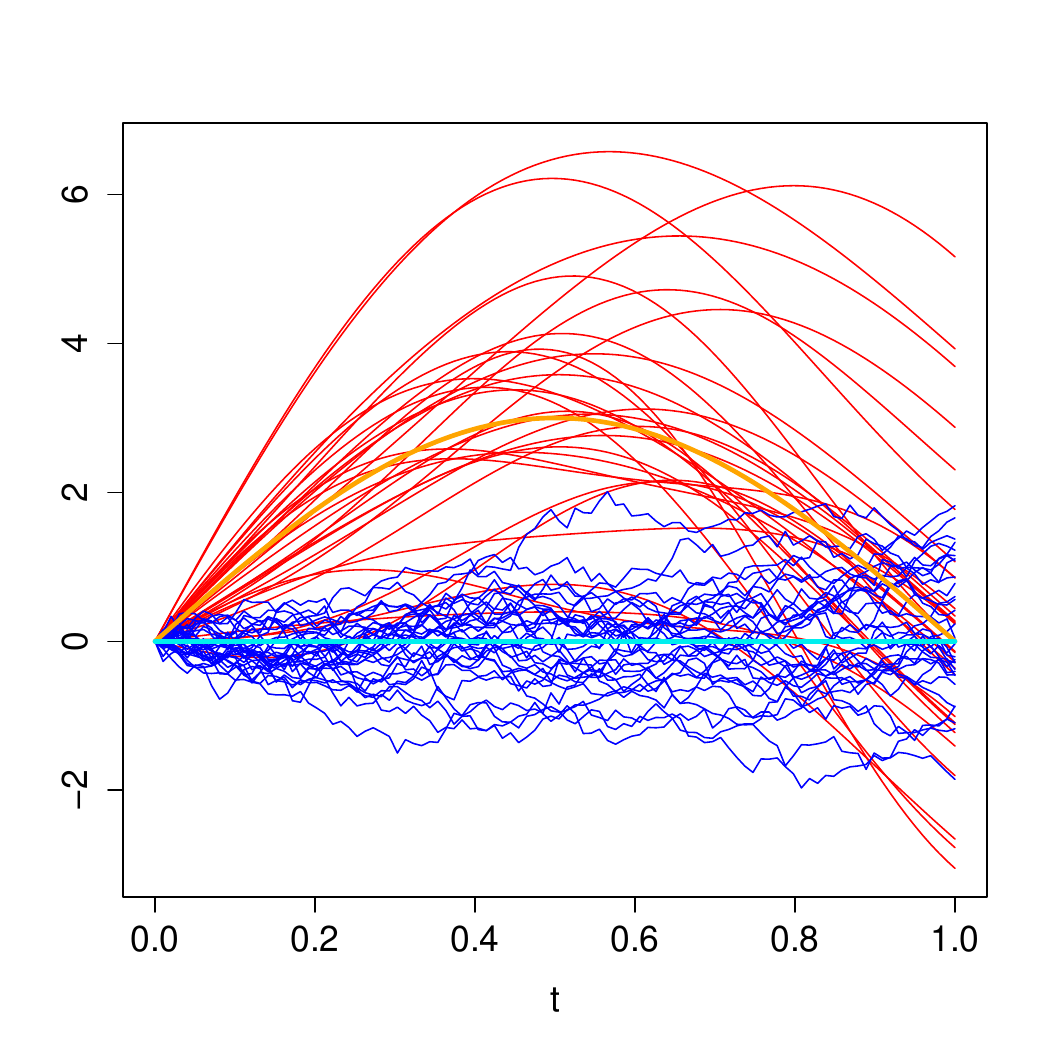} & 
\includegraphics[scale=0.25]{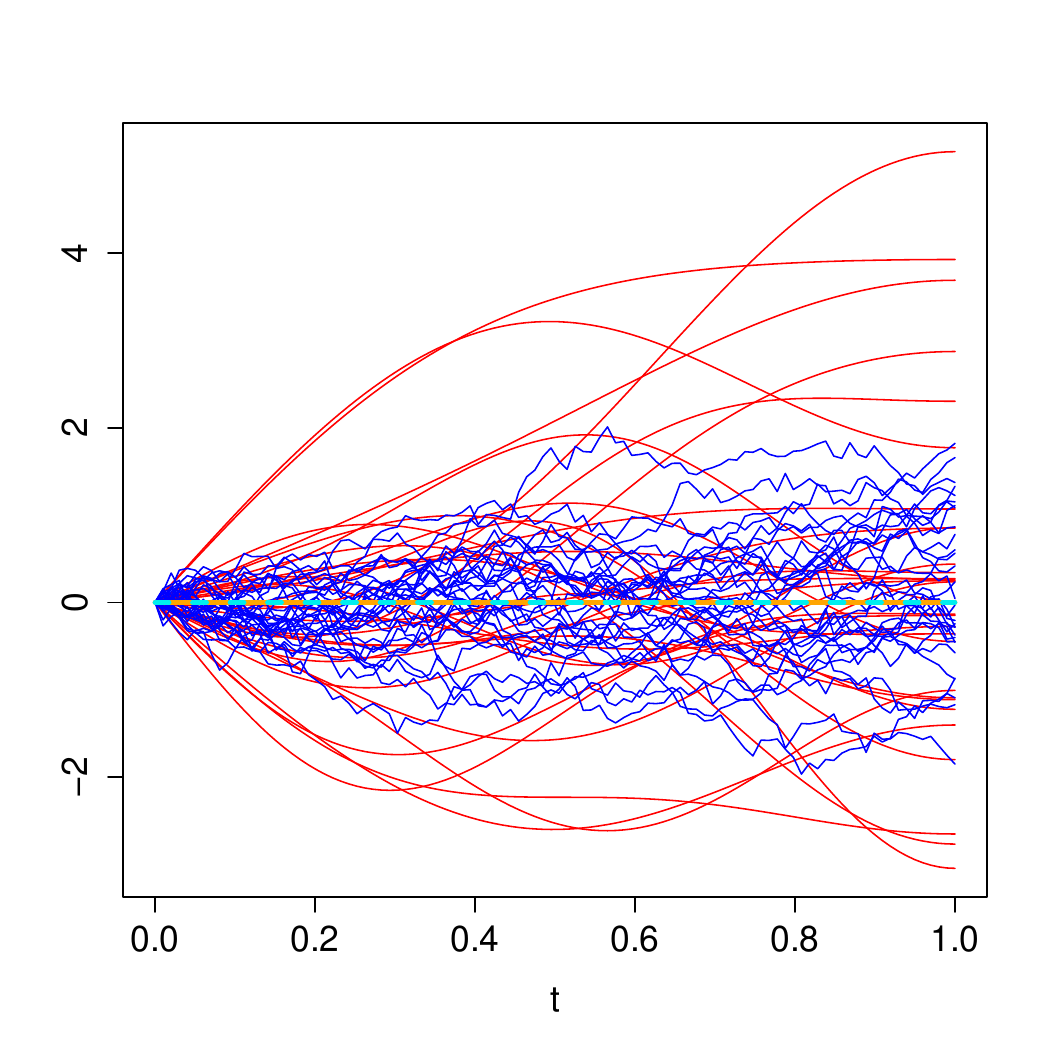} 
				\\
 \textbf{C21} &  \textbf{C20} \\[-2ex]
\includegraphics[scale=0.25]{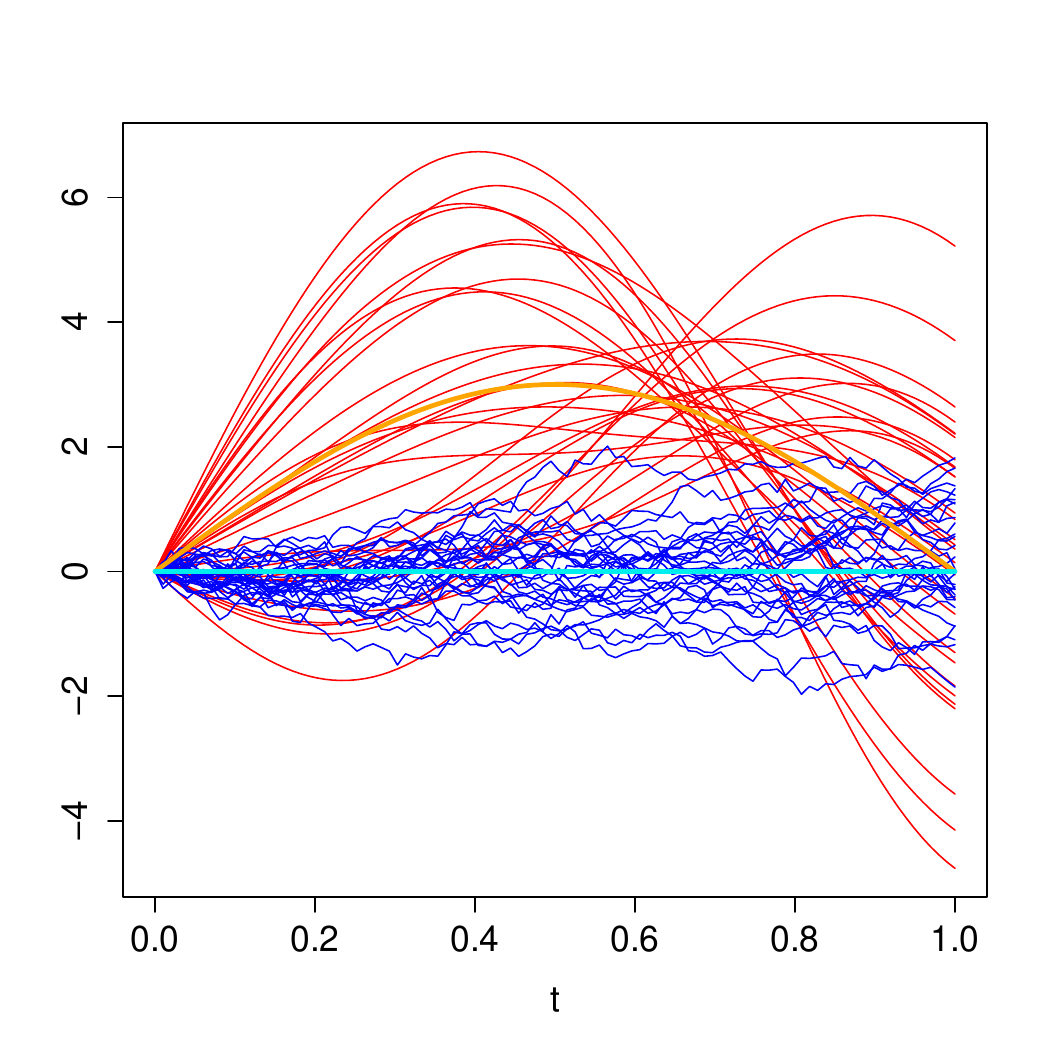} & 
\includegraphics[scale=0.25]{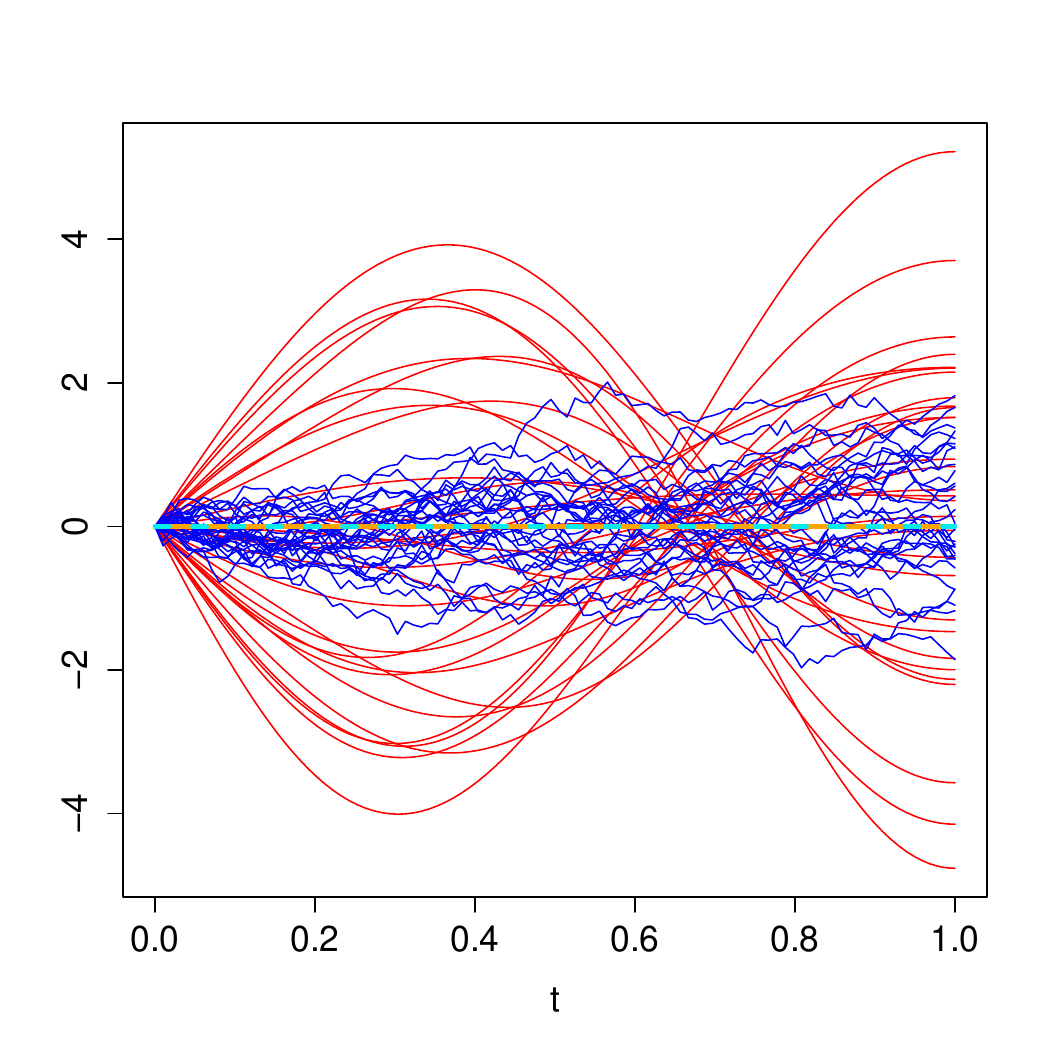}
\end{tabular}
\vskip-0.1in 
\caption{\label{fig:datos-cpc} Data sets for under a \textsc{fcpc} model (scheme \textbf{CPC}). Blue and red lines correspond to $X_{H,i}$ and $X_{D,i}$, respectively, while the true mean functions $\mu_H$ and $\mu_D$ are depicted in cyan and orange lines, respectively.} 
	\end{center} 
\end{figure}

\begin{figure}[ht!]
	\begin{center}
		\footnotesize
\renewcommand{\arraystretch}{0.1}
		\begin{tabular}{cc}
	\textbf{D11} & \textbf{D10} \\[-2ex]
			\includegraphics[scale=0.25]{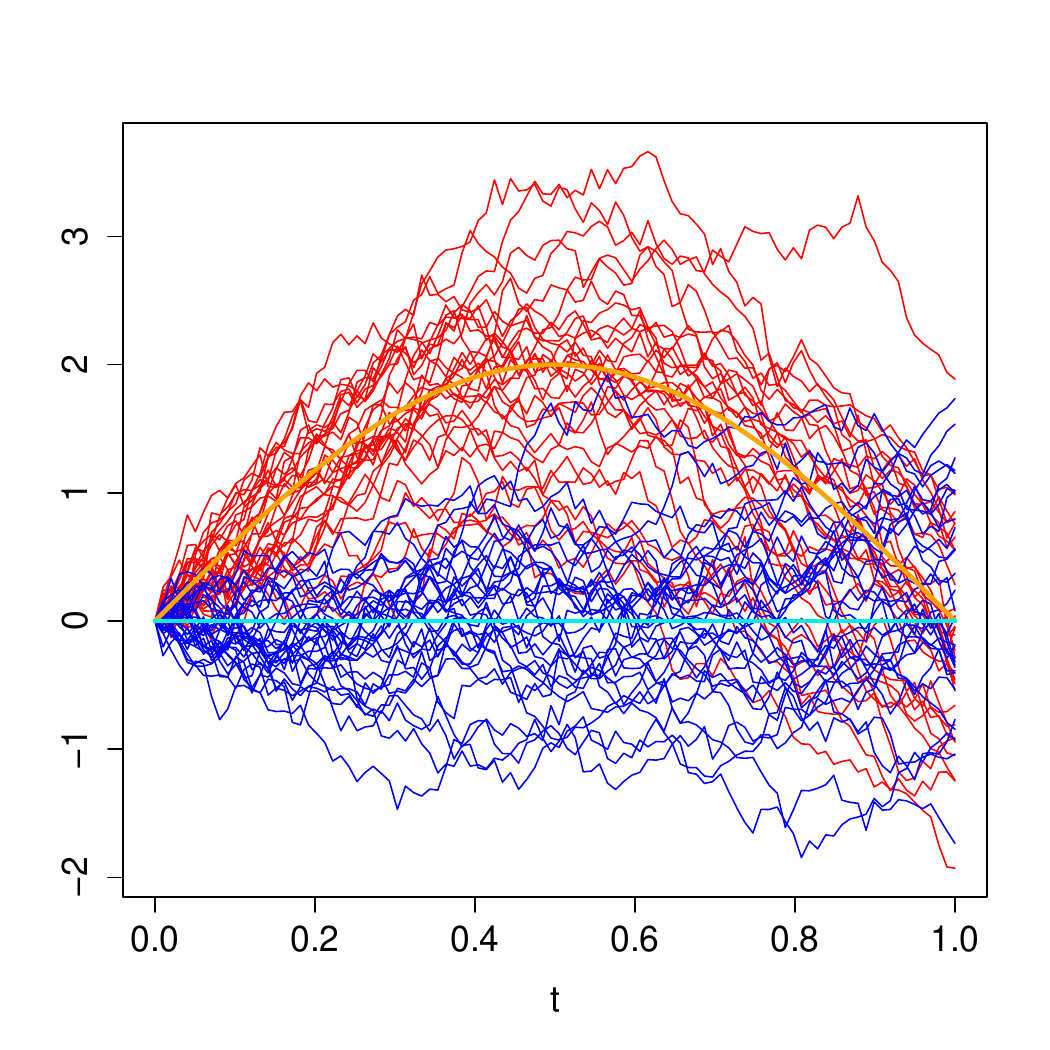} & 
 			\includegraphics[scale=0.25]{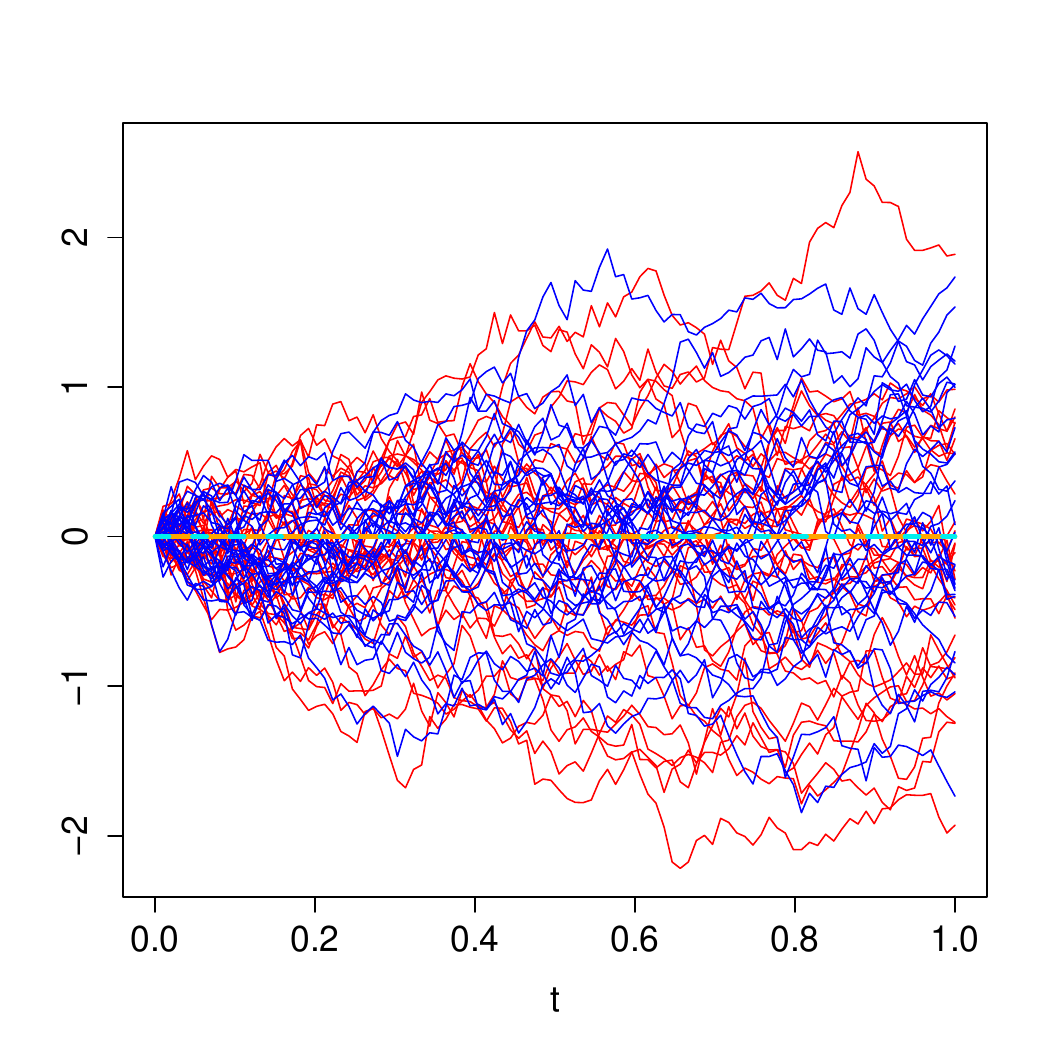}\\
 		\textbf{D21} & \textbf{D20} \\[-2ex]
		\includegraphics[scale=0.25]{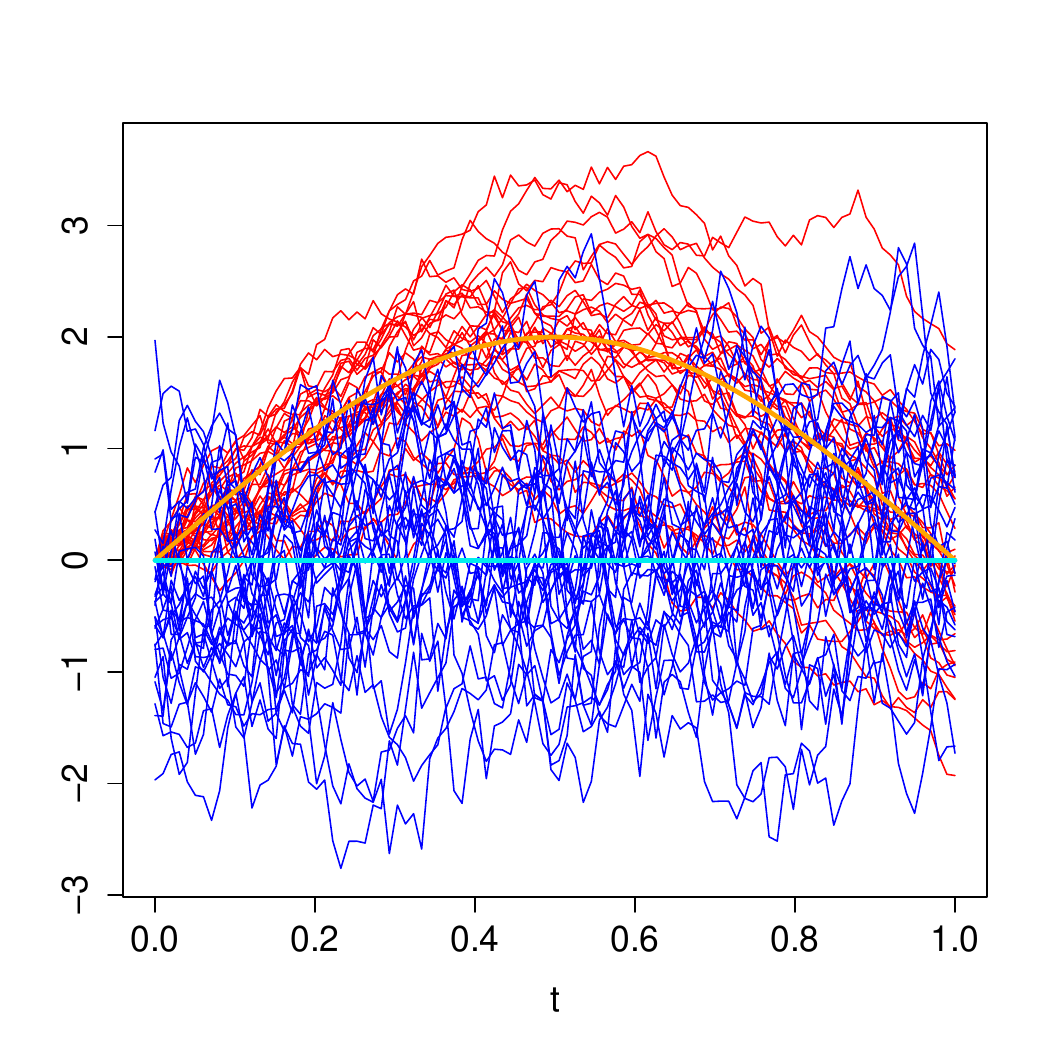}  & 
 		\includegraphics[scale=0.25]{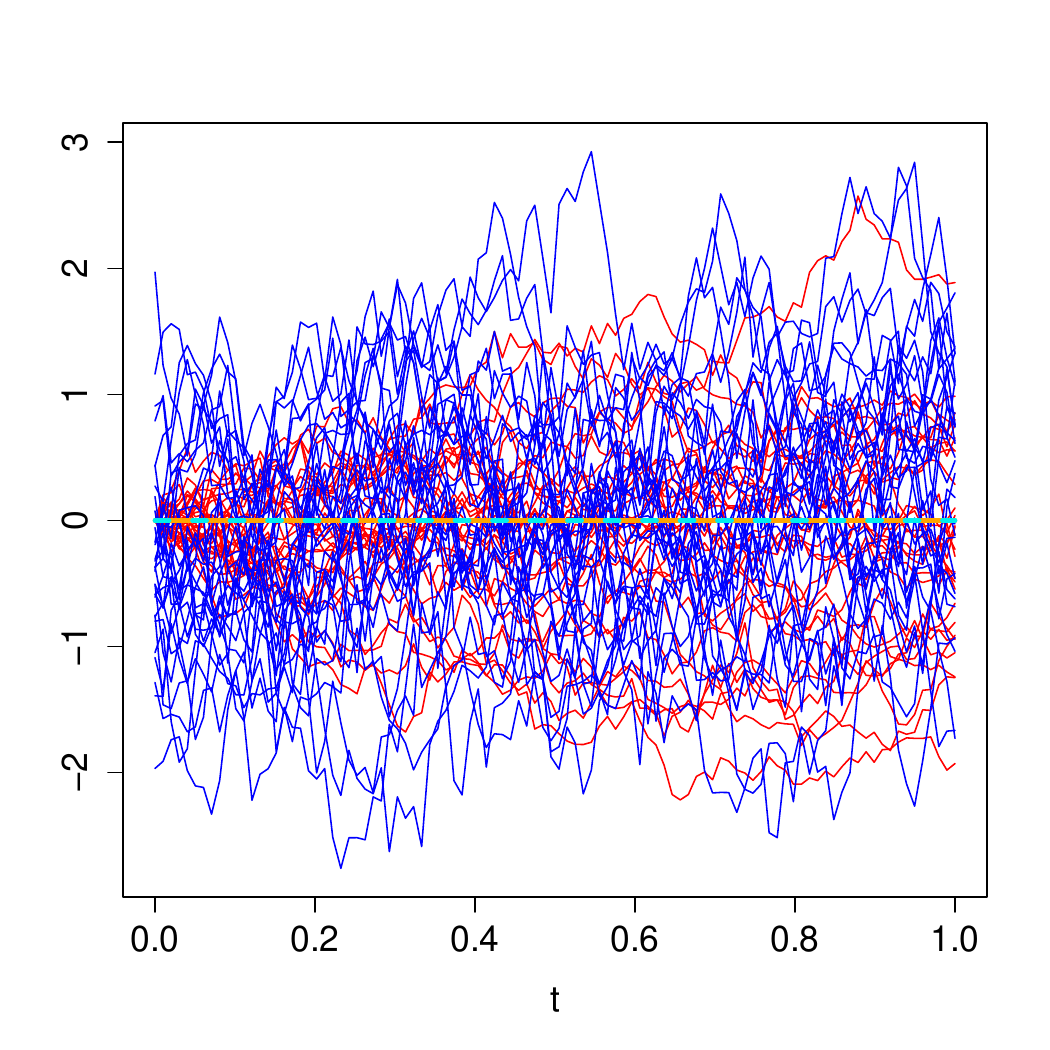}  
 		\end{tabular}
		\vskip-0.1in 
		\caption{\label{fig:datos-distintas} Data sets for under scheme \textbf{DIFF}. The left panel corresponds to the case $\mu_H(t)=0$ and $\mu_D(t)=3 \,\sin(\pi  t)$ and the right one to $\mu_D(t)=\mu_H(t)=0$. Blue and red lines correspond to $X_{H,i}$ and $X_{D,i}$, respectively, while the true mean functions $\mu_H$ and $\mu_D$ are depicted in cyan and orange lines, respectively.}
			\end{center} 
\end{figure}

Tables  \ref{tab:summary-propor} and \ref{tab:summary-cpc} report  the mean and standard deviations over 1000 replications  of the AUC  estimators  under scenarios \textbf{PROP} and \textbf{CPC},  while Table  \ref{tab:summary-distintas} displays the same summary measures under scenario \textbf{DIFF}. To facilitate the reading we indicate in boldface the largest value attained for the mean    of the AUC and in italic, the second largest. The obtained results reveal that in all cases, the best performance is obtained by the quadratic rule, $\wUps_{\cuad}$, followed in most situations by the linear one induced by the coefficient maximizing $\AUC(\beta)$ as defined in Section \ref{sec:linealfun}, that is, by $\wUps_{\lin}$. It is worth mentioning that under the proportional model   when the means of both populations are equal or when means are different but the underlying process is an Exponential Variogram,  the rule  $\Upsilon_{\maxi}$  based on the maximum value of the trajectory achieves a larger mean value of the AUC estimators  than $\wUps_{\lin}$, a fact that is clearly revealed in Figure \ref{fig:boxplot-PROP2}   that presents the  boxplots of the AUC estimators. When the mean functions are equal, the discriminating indexes based on a linear rule, $\Upsilon_{\inte}$, $\wUps_{\media}$ and $\wUps_{\lin}$, barely exceed an average estimated AUC of 0.5. The rule based on the maximum has a much better performance under \textbf{C20} than under \textbf{C10}, while under a proportional model with equal means (\textbf{P0}),  the differences between populations seem  to be more easily detected by   $\Upsilon_{\maxi}$ for the Exponential Variogram process than for the Brownian one.

\begin{table}[ht!]
\begin{center}
\footnotesize
\renewcommand{\arraystretch}{1.2}
\setlength{\tabcolsep}{2pt}
\caption{ \label{tab:summary-propor} Mean and standard deviation of the $\widehat{\AUC}$, under scenario \textbf{PROP}, that is, under a proportional model  $\gamma_D(t,s)=\rho\;\gamma_H(t,s)$ with equal (\textbf{P0}) or different mean functions (\textbf{P1}), $\mu_H(t)=0$ and $\mu_D(t)=2\, \sin(\pi  t)$). In all cases, $n_H=n_D=300$.}
  \begin{tabular}{c c cccccc@{\extracolsep{1cm}}    c@{\extracolsep{3pt}}ccccc}
   \hline \\[-2ex]
$\hskip0.1in\rho\hskip0.1in$  & & $\Upsilon_{\maxi}$ & $\Upsilon_{\mini}$ & $\Upsilon_{\inte}$ & $\wUps_{\media}$ & $\wUps_{\lin}$  & $\wUps_{\cuad}$
& $\Upsilon_{\maxi}$ & $\Upsilon_{\mini}$ & $\Upsilon_{\inte}$ & $\wUps_{\media}$ & $\wUps_{\lin}$  & $\wUps_{\cuad}$
\\
\hline
& & \multicolumn{6}{c}{\textbf{P1}} & \multicolumn{6}{c}{\textbf{P0}}\\ \hline 
& & \multicolumn{12}{c}{Brownian Motion}\\
\hline 
 1 & Mean & 0.9520 & 0.6963 & 0.9389 & 0.9653 & \textit{0.9892} & \textbf{0.9987} 
 		&  & &  &  &  & 
 \\
  & SD & 0.0083 & 0.0219 & 0.0094 & 0.0057 & 0.0004 & 0.0007 
  		&  & &  &   &  & \\
\hline 
 2 & Mean & 0.9434 & 0.6183 & 0.8965 & 0.9309 & \textit{0.9845} & \textbf{0.9945} 
          & \textit{0.5977} & 0.4025 & 0.4998 & 0.5307 & 0.5465 & \textbf{0.7648} 
 \\
  & SD & 0.0091 & 0.0237 & 0.0133 & 0.0090 & 0.0017 & 0.0022 
       & 0.0239 & 0.0233 & 0.0244 & 0.0136 & 0.0165 & 0.0229
 \\
\hline 
& & \multicolumn{12}{c}{Exponential Variogram}\\
\hline 
 1 & Mean & 0.9200 & 0.7653 & 0.9426 & 0.9616 & \textit{0.9644} &  \textbf{0.9809} 
 			&  & &  &   &  & \\
  & SD & 0.0108 & 0.0191 & 0.0090 & 0.0068 & 0.0051 &  0.0045 
  		&  & &  &   &  & \\
\hline 
 2 & Mean & \textit{0.9520} & 0.5511 & 0.9011 & 0.9259 & 0.9349 &  \textbf{0.9905} 
 	      & \textit{0.7179} & 0.2823 & 0.5000 & 0.5575 & 0.6005 & \textbf{0.9627}
 \\
  & SD & 0.0082 & 0.0242 & 0.0128 & 0.0107 & 0.0088 & 0.0031 
       & 0.0212 & 0.0211 & 0.0244 & 0.0134 & 0.0165 & 0.0069
 \\
\hline
\end{tabular} 
\end{center}
\end{table}


\begin{table}[ht!]
\begin{center}
\footnotesize
\renewcommand{\arraystretch}{1.2}
\setlength{\tabcolsep}{3pt}
\caption{\label{tab:summary-cpc}  Mean and standard deviation of the $\widehat{\AUC}$, under scenario   \textbf{CPC}, which corresponds to processes with different mean functions $\mu_H(t)=0$ and $\mu_D(t)=3\, \sin(\pi  t)$ (\textbf{C11} and \textbf{C21}) or equal means   (\textbf{C10} and \textbf{C20}). In all cases,  $n_H=n_D=300$.} 
\begin{tabular}{ c  cccccc@{\extracolsep{1cm}}    c@{\extracolsep{3pt}}ccccc}
 \hline \\[-2ex]
 & $\Upsilon_{\maxi}$ & $\Upsilon_{\mini}$ & $\Upsilon_{\inte}$ & $\wUps_{\media}$ & $\wUps_{\lin}$  & $\wUps_{\cuad}$
& $\Upsilon_{\maxi}$ & $\Upsilon_{\mini}$ & $\Upsilon_{\inte}$ & $\wUps_{\media}$ & $\wUps_{\lin}$  & $\wUps_{\cuad}$
\\\hline
   & \multicolumn{12}{c}{$\lambda_D=(2,0.30,0.05)\trasp$}\\
\hline 
& \multicolumn{6}{c}{\textbf{C11}} & \multicolumn{6}{c}{\textbf{C10}} \\
\hline
 Mean & 0.9417 & 0.6103 & 0.9099 & 0.9417 & \textit{0.9853} & \textbf{0.9905} 
      & 0.5248 & 0.4723 & 0.4985 & 0.5266 & \textit{0.5291} & \textbf{0.8531}
  \\
 SD & 0.0101 & 0.0257 & 0.0132 & 0.0093 & 0.0049 & 0.0038 
    & 0.0248 & 0.0253 & 0.0247 & 0.0154 & 0.0174 & 0.0160
\\
\hline 
   & \multicolumn{12}{c}{$\lambda_D=(0.30,2, 0.05)\trasp$}\\
\hline
& \multicolumn{6}{c}{\textbf{C21}} & \multicolumn{6}{c}{\textbf{C20}} \\
\hline
 Mean & \textit{0.9922} & 0.5355 & 0.9856 & 0.9810 & 0.9881 & \textbf{0.9966} 
      & \textit{0.7340} & 0.2647 & 0.4991 & 0.5296 & 0.5295 & \textbf{0.9090}
 \\
 SD & 0.0024 & 0.0250 & 0.0036 & 0.0046 & 0.0031 & 0.0013 
    & 0.0214 & 0.0207 & 0.0239 & 0.0166 & 0.0166 & 0.0126
 \\
\hline
\end{tabular} 
\end{center}
\end{table}

\begin{table}[ht!]
\begin{center}
\footnotesize
\renewcommand{\arraystretch}{1.2}
\setlength{\tabcolsep}{3pt}
\caption{\label{tab:summary-distintas}  Mean and standard deviation of the $\widehat{\AUC}$, under scenario   \textbf{DIFF},  when different covariance operators are considered. Under \textbf{D11} and \textbf{D21},  $\mu_H(t)=0$ and $\mu_D(t)=2\, \sin(\pi  t)$, while for the schemes \textbf{D10} and \textbf{D20} $\mu_D=\mu_H\equiv 0$. In all cases,  $n_H=n_D=300$.} 
 \begin{tabular}{c  cccccc@{\extracolsep{1cm}}    c@{\extracolsep{3pt}}ccccc}
  \hline \\[-2ex]
 & $\Upsilon_{\maxi}$ & $\Upsilon_{\mini}$ & $\Upsilon_{\inte}$ & $\wUps_{\media}$ & $\wUps_{\lin}$  & $\wUps_{\cuad}$
& $\Upsilon_{\maxi}$ & $\Upsilon_{\mini}$ & $\Upsilon_{\inte}$ & $\wUps_{\media}$ & $\wUps_{\lin}$  & $\wUps_{\cuad}$
\\\hline
 & \multicolumn{6}{c}{\textbf{D11}} & \multicolumn{6}{c}{\textbf{D10}}\\
   \hline
 Mean &  0.9607 & 0.6956 & 0.9454 & \textit{0.9695} & \textbf{0.9988}   & \textbf{0.9988} 
  	  &  0.5048 & 0.4953 & 0.4996 & 0.5314 & \textit{0.5486} & \textbf{0.5989}
 \\  
 SD &   0.0072 & 0.0221 & 0.0088 & 0.0053 & 0.0007   & 0.0007
    &   0.0244 & 0.0238 & 0.0242 & 0.0130 & 0.0163 & 0.0187 
 \\ 
\hline
  & \multicolumn{6}{c}{\textbf{D21}} & \multicolumn{6}{c}{\textbf{D20}}\\
\hline
 Mean &   0.7370 & 0.9123 & 0.9407 & 0.9634 & \textit{0.9844}   & \textbf{1.0000}
      &   0.1486 & 0.8512 & 0.4996 & 0.5491 & \textit{0.5864} & \textbf{1.0000} 
  \\ 
 SD &   0.0207 & 0.0118 & 0.0092 & 0.0064 & 0.0042    & 0.0000
    & 0.0154 & 0.0159 & 0.0242 & 0.0140 & 0.0180 & 0.0002
\\
\hline
\end{tabular} 
\end{center}
\end{table}

\begin{figure}[ht!]
 \begin{center}
 \footnotesize
 \renewcommand{\arraystretch}{0.2}
\begin{tabular}{cc}
   Brownian Motion  &  Exponential Variogram \\[-2ex]
    \multicolumn{2}{c}{\textbf{P1}}\\[-2ex]
 \includegraphics[scale=0.35]{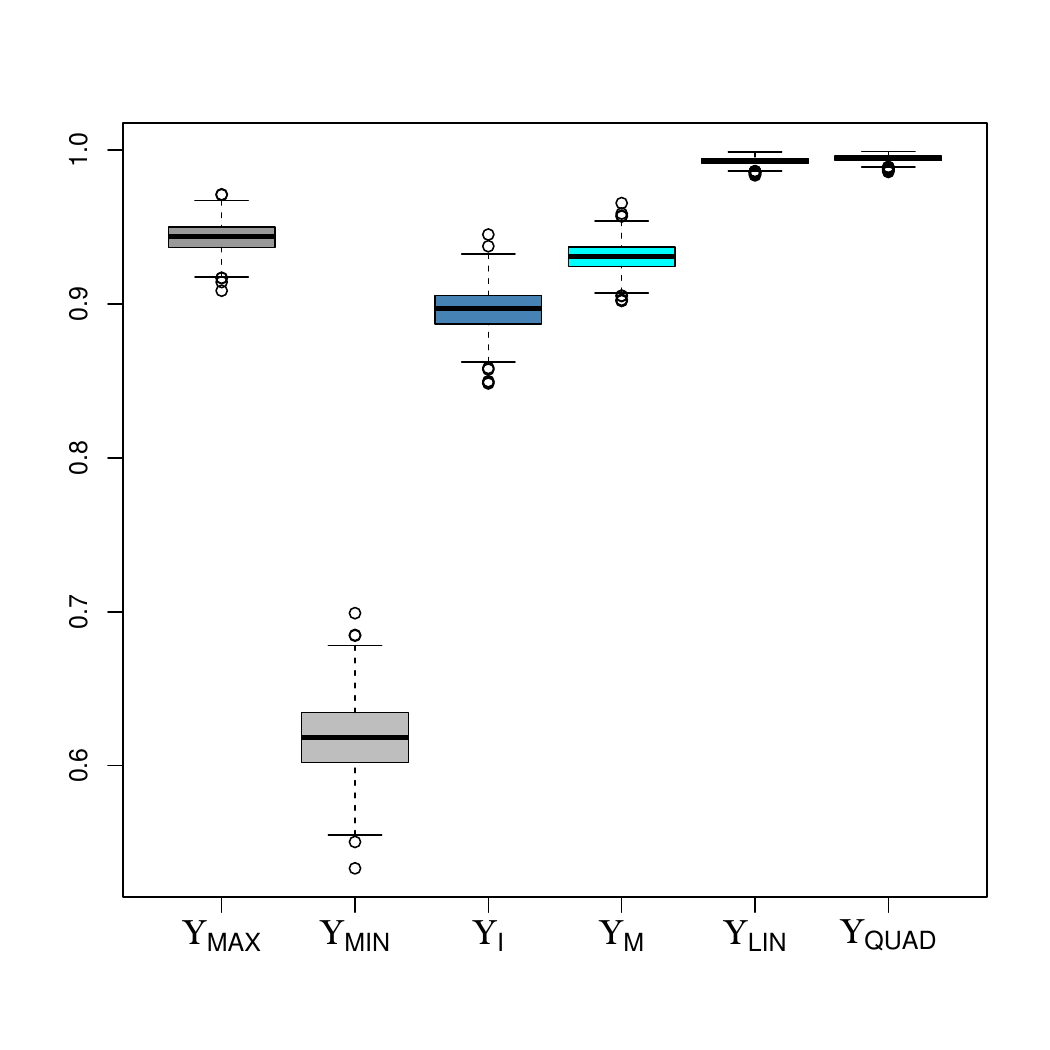}
& \includegraphics[scale=0.35]{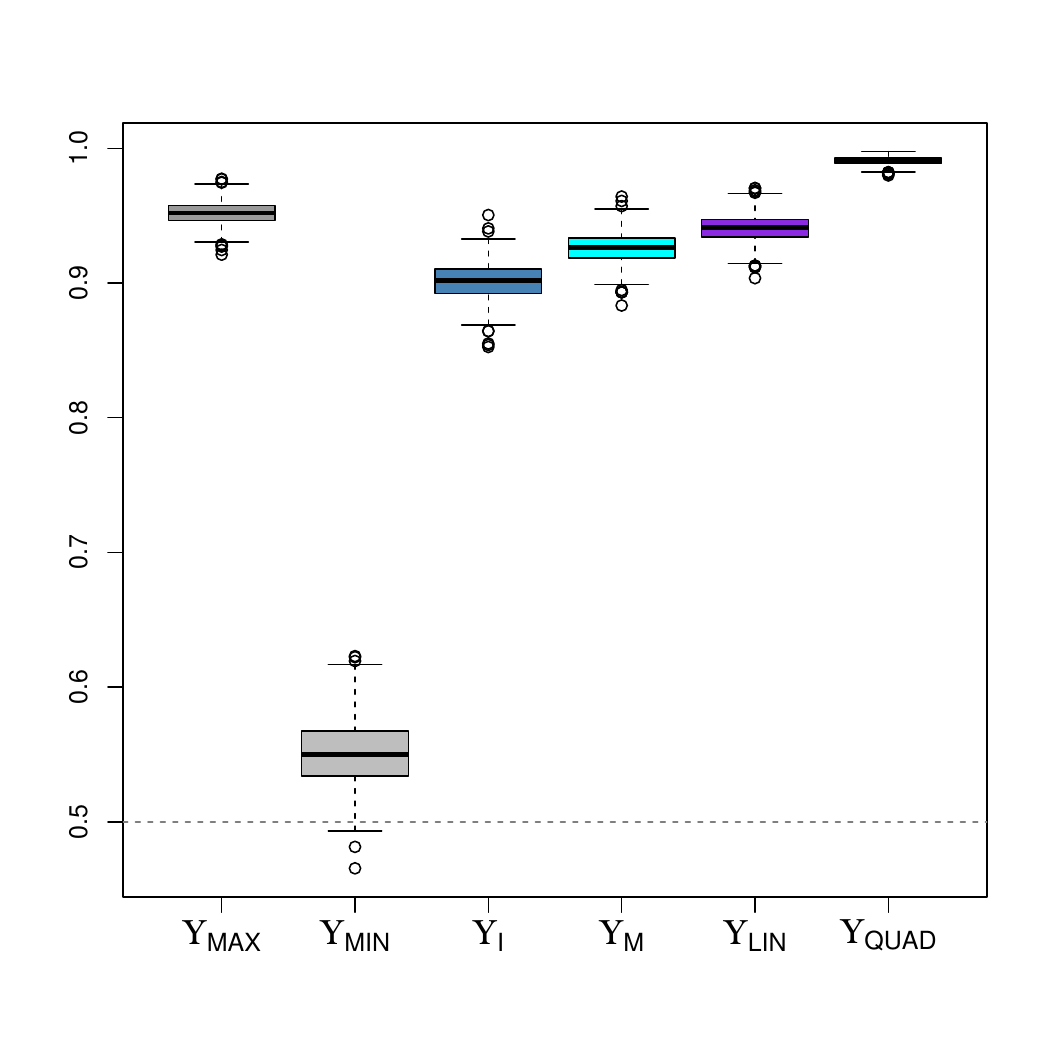}\\[-2ex]

   \multicolumn{2}{c}{\textbf{P0}}\\[-2ex]
 \includegraphics[scale=0.35]{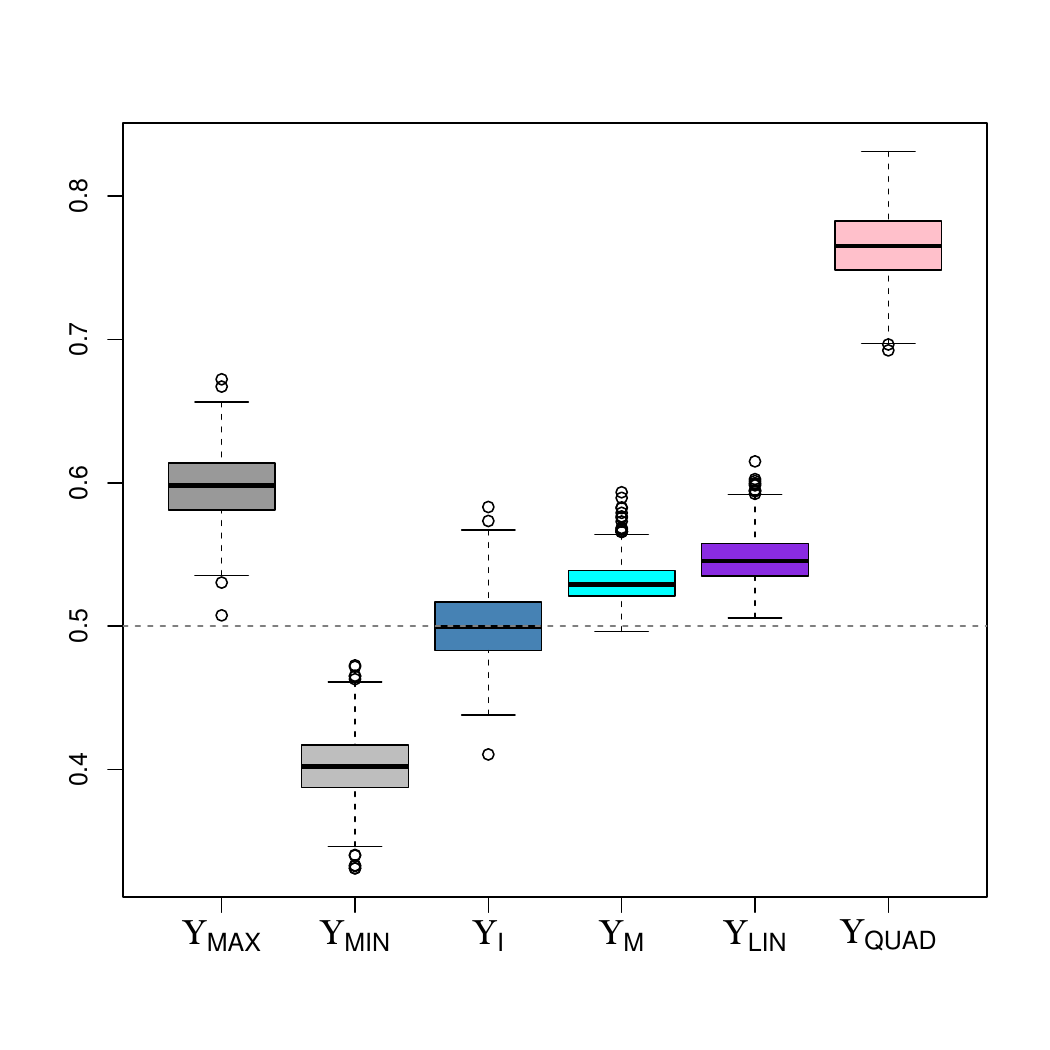}
& \includegraphics[scale=0.35]{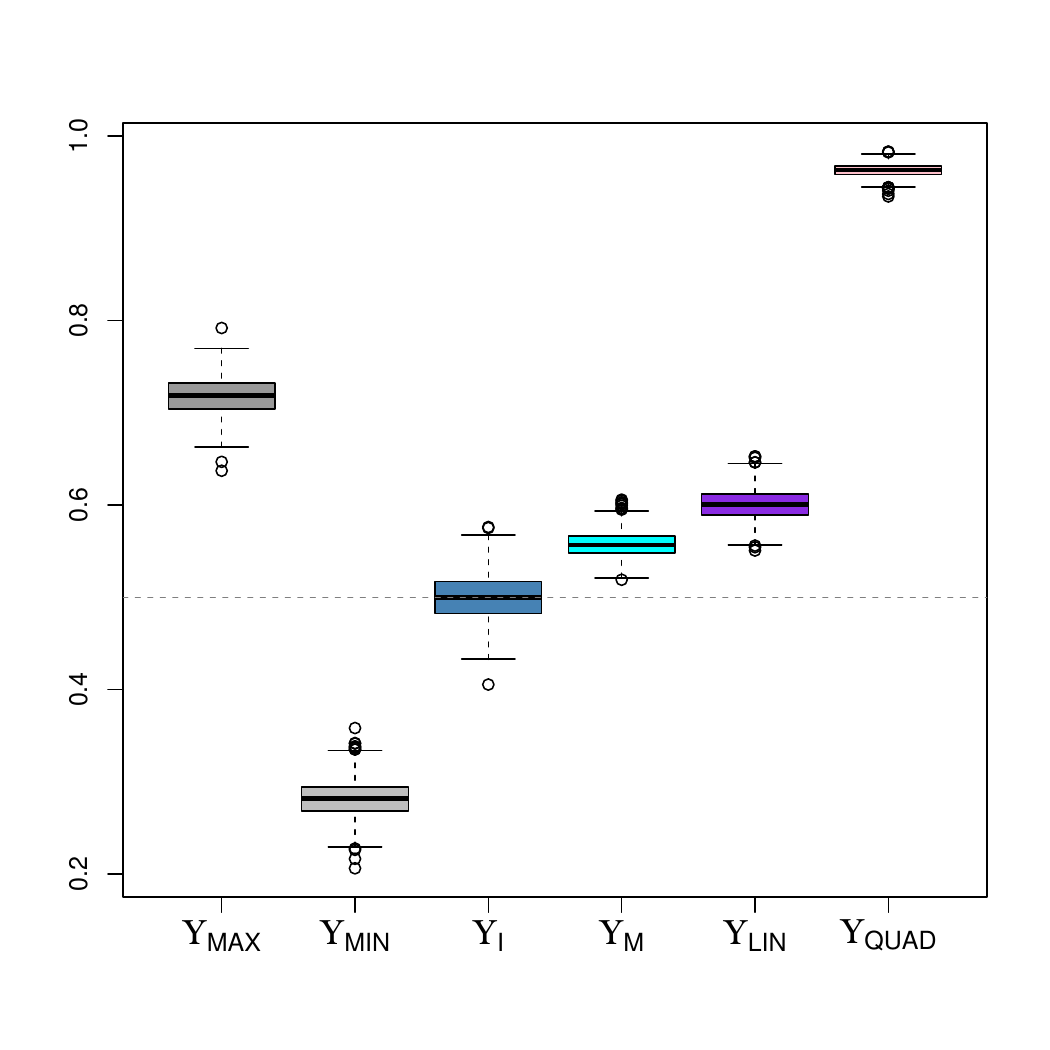}
\end{tabular}
\vskip-0.1in
\caption{Boxplots of the estimators of the AUC under scenario  \textbf{PROP} with  $\rho=2$.  The horizontal dashed line, when appearing, indicates 0.5.}
\label{fig:boxplot-PROP2}
\end{center} 
\end{figure}

  Figures \ref{fig:boxplot-PROP2}, \ref{fig:boxplot-CPC} and  \ref{fig:boxplot-DIFF}   present  the  boxplots of the AUC estimators, under the proportional model, the \textsc{fcpc} one and under scheme \textbf{DIFF}, respectively. In particular, Figures \ref{fig:boxplot-CPC} and  \ref{fig:boxplot-DIFF}   highlight the performance differences as the models vary. It is evident from these plots that scheme \textbf{D10} is the more challenging one and only for $\wUps_{\cuad}$ most estimators are larger than 0.55.    There are several simulation scenarios where the obtained AUCs are clearly below 0.5. For example, $\Upsilon_{\mini}$ achieves $0.2823$ and $0.2647$ under \textbf{P0}  (Exponential Variogram, $\rho=2$) and \textbf{C20}, respectively. This means in fact that, if the roles of the healthy and diseased populations are interchanged, the corresponding rule would achieve AUCs of $1-0.2823$ and $1-0.2647$, respectively. These values are similar to the ones of $\Upsilon_{\maxi}$ under the same simulation scenarios. Analogous comments can be done for $\Upsilon_{\maxi}$ under \textbf{D20}, which yields an AUC of $0.1486$.   In summary, in most scenarios, the quadratic rule outperforms the considered competitors.

\begin{figure}[ht!]
 \begin{center}
 \footnotesize
 \renewcommand{\arraystretch}{0.2}
\begin{tabular}{cc}
 $\lambda_D=(2,0.30,0.05)\trasp$ & $\lambda_D=(0.30,2,0.05)\trasp$\\ 
 
  \textbf{C11}  &   \textbf{C21}   \\[-2ex]
 
 \includegraphics[scale=0.35]{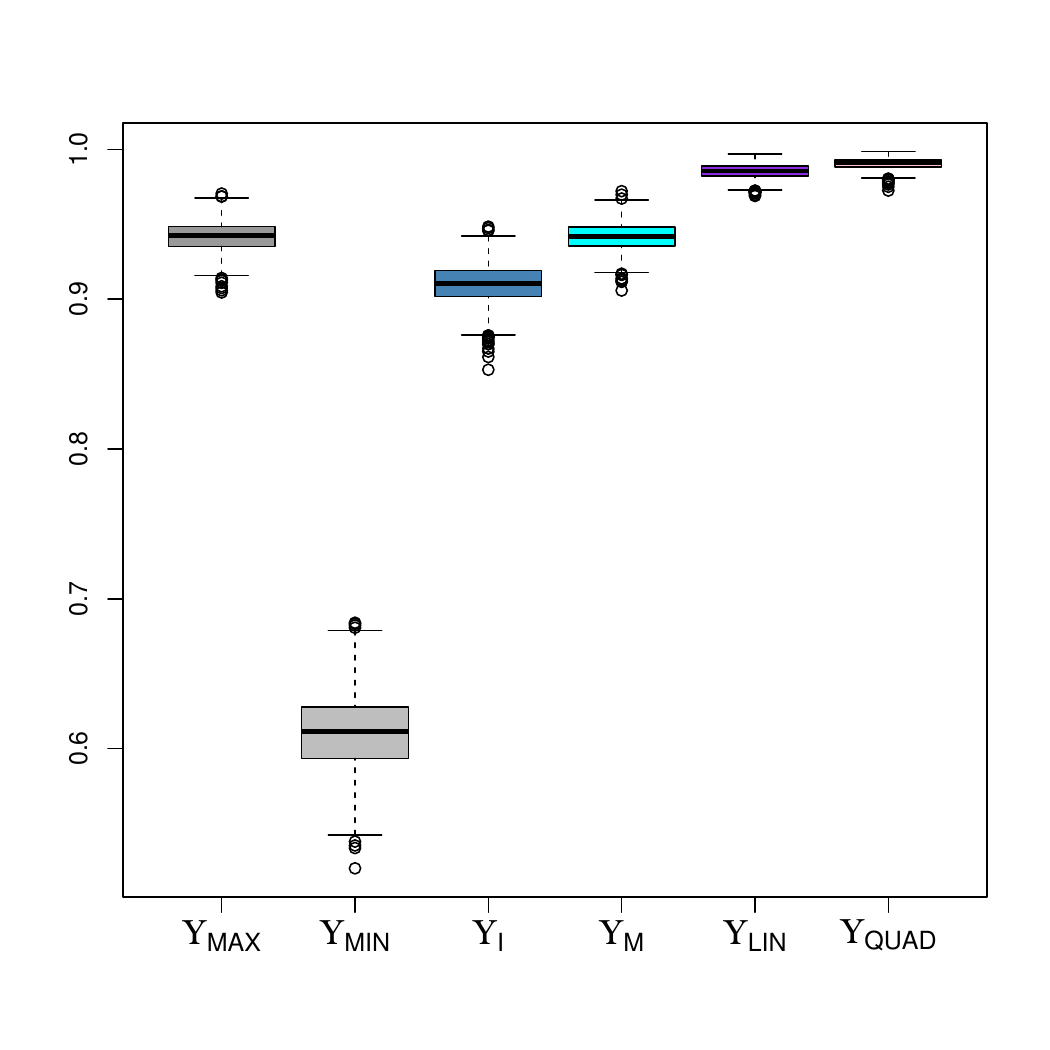}
& 
 \includegraphics[scale=0.35]{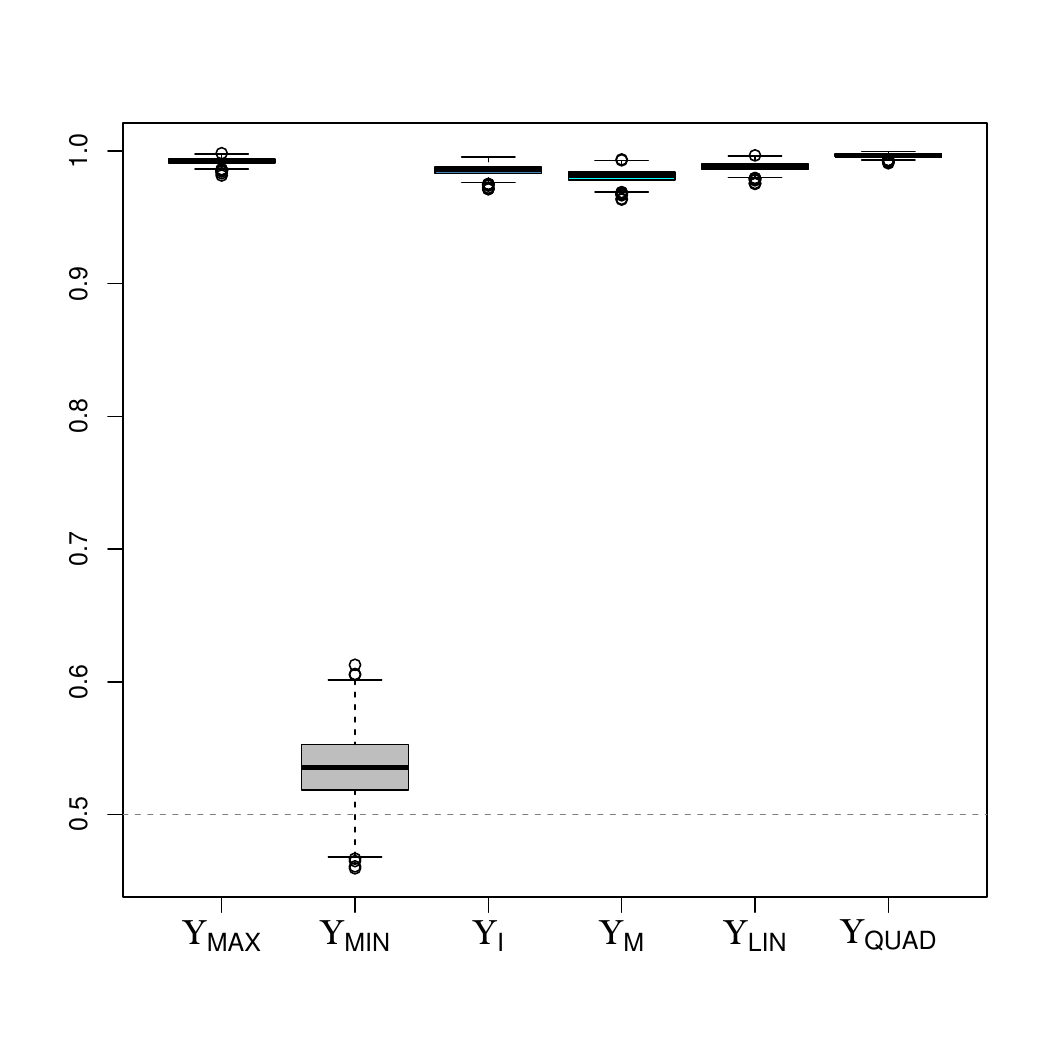}\\[-2ex]

 \textbf{C10} & \textbf{C20} \\[-2ex]
 \includegraphics[scale=0.35]{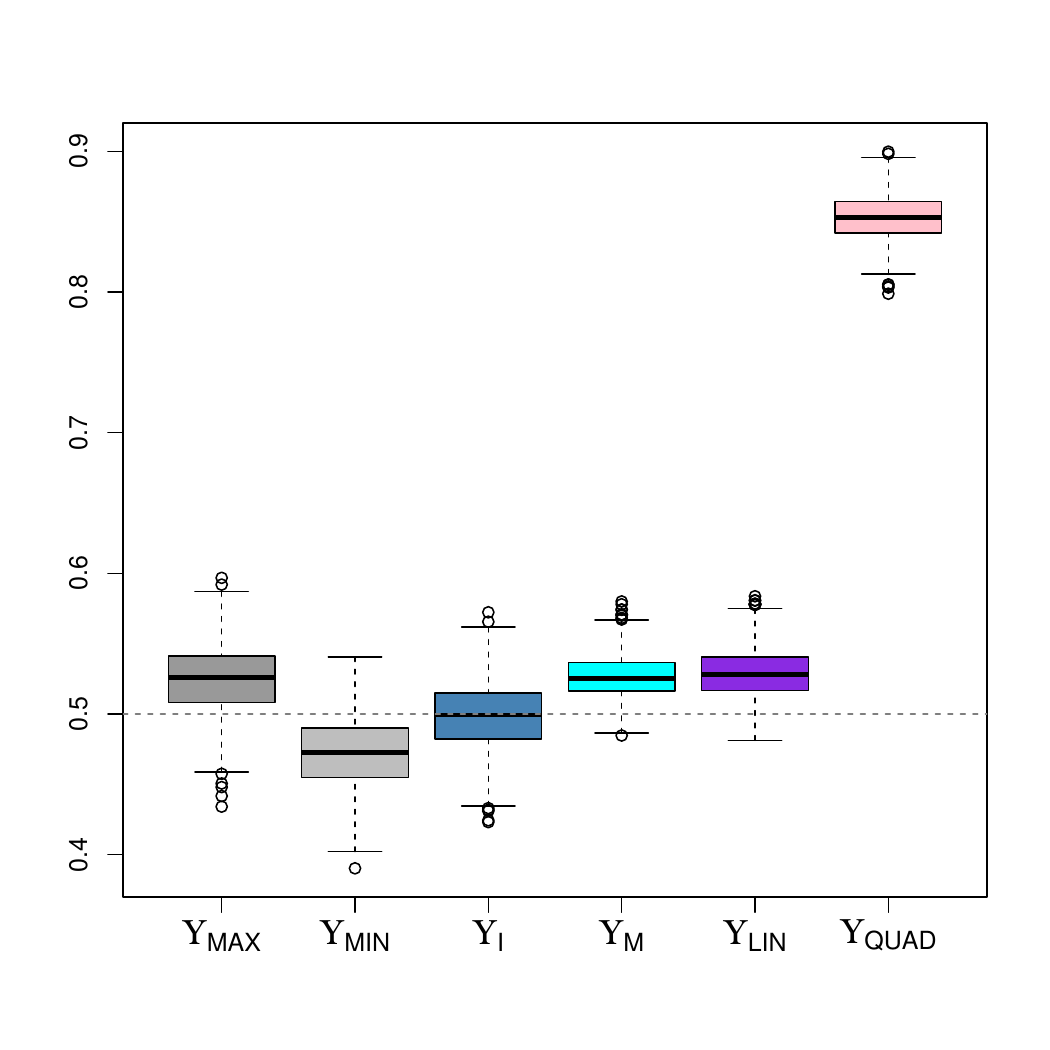}
& \includegraphics[scale=0.35]{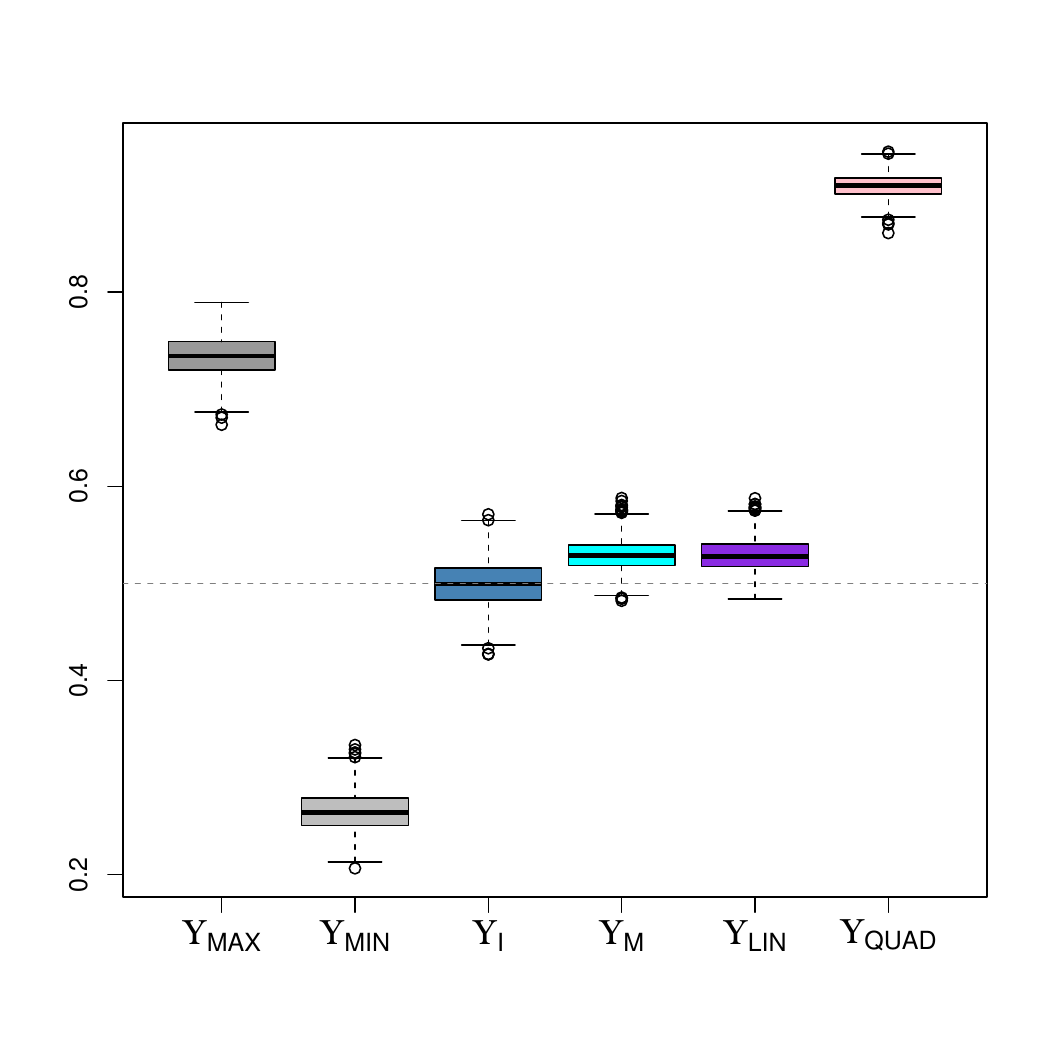}
\end{tabular}
\vskip-0.1in
\caption{Boxplots of the estimators of the AUC under scenario \textbf{CPC}. The horizontal dashed line, when appearing, indicates 0.5.}
\label{fig:boxplot-CPC}
\end{center} 
\end{figure}

\begin{figure}[ht!]
 \begin{center}
 \footnotesize
 \renewcommand{\arraystretch}{0.2}
\begin{tabular}{cc}
  
  \textbf{D11}  &   \textbf{D21}   \\[-2ex]
 
 \includegraphics[scale=0.35]{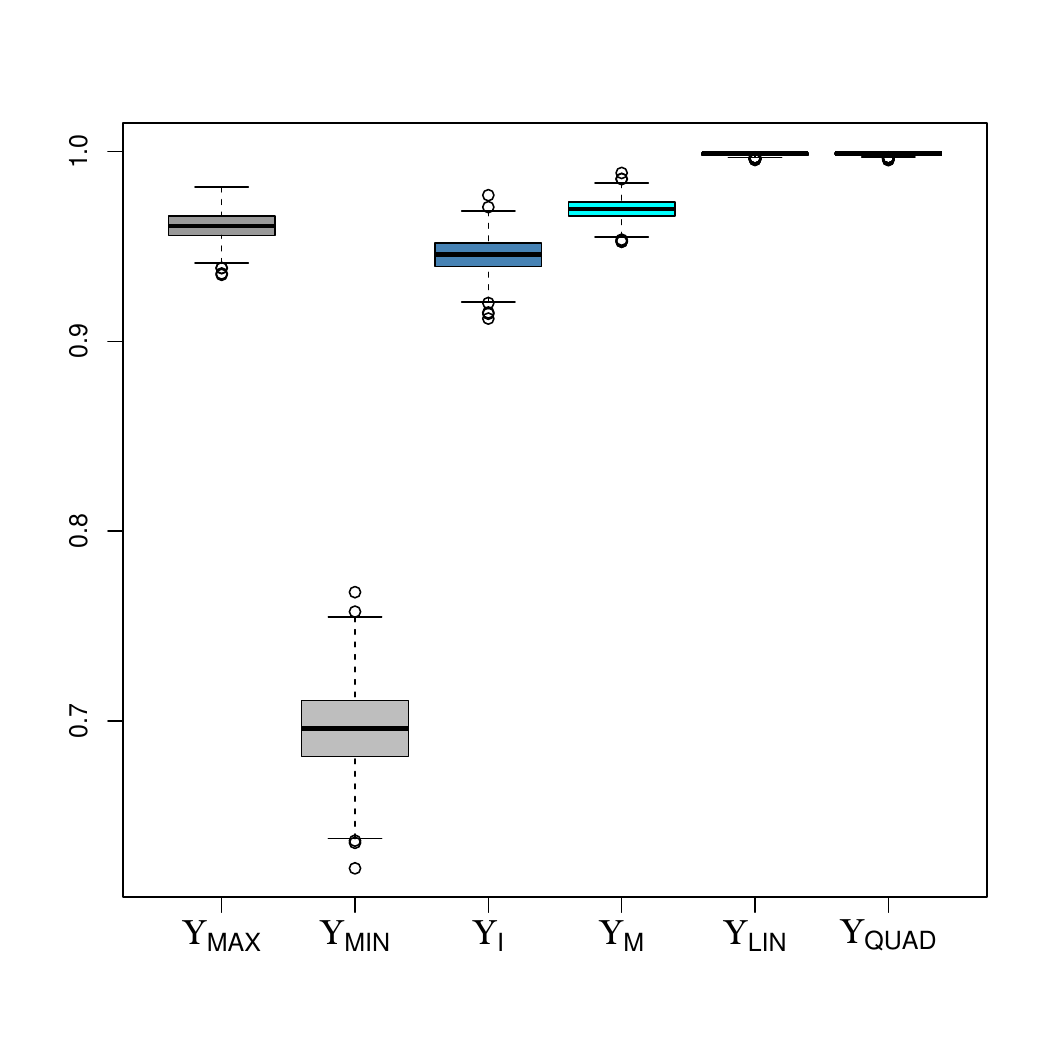}
& 
 \includegraphics[scale=0.35]{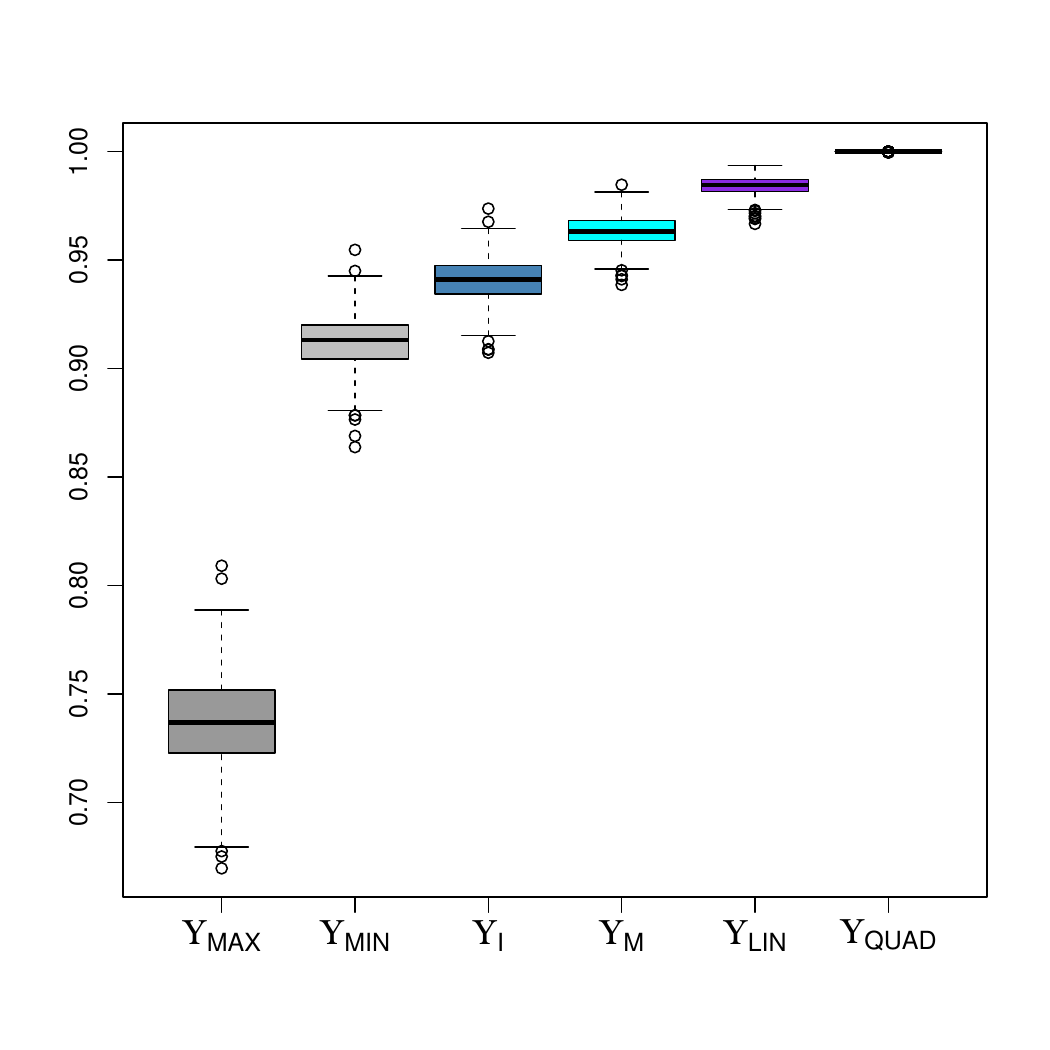}\\[-2ex]

 \textbf{D10} & \textbf{D20} \\[-2ex]
 \includegraphics[scale=0.35]{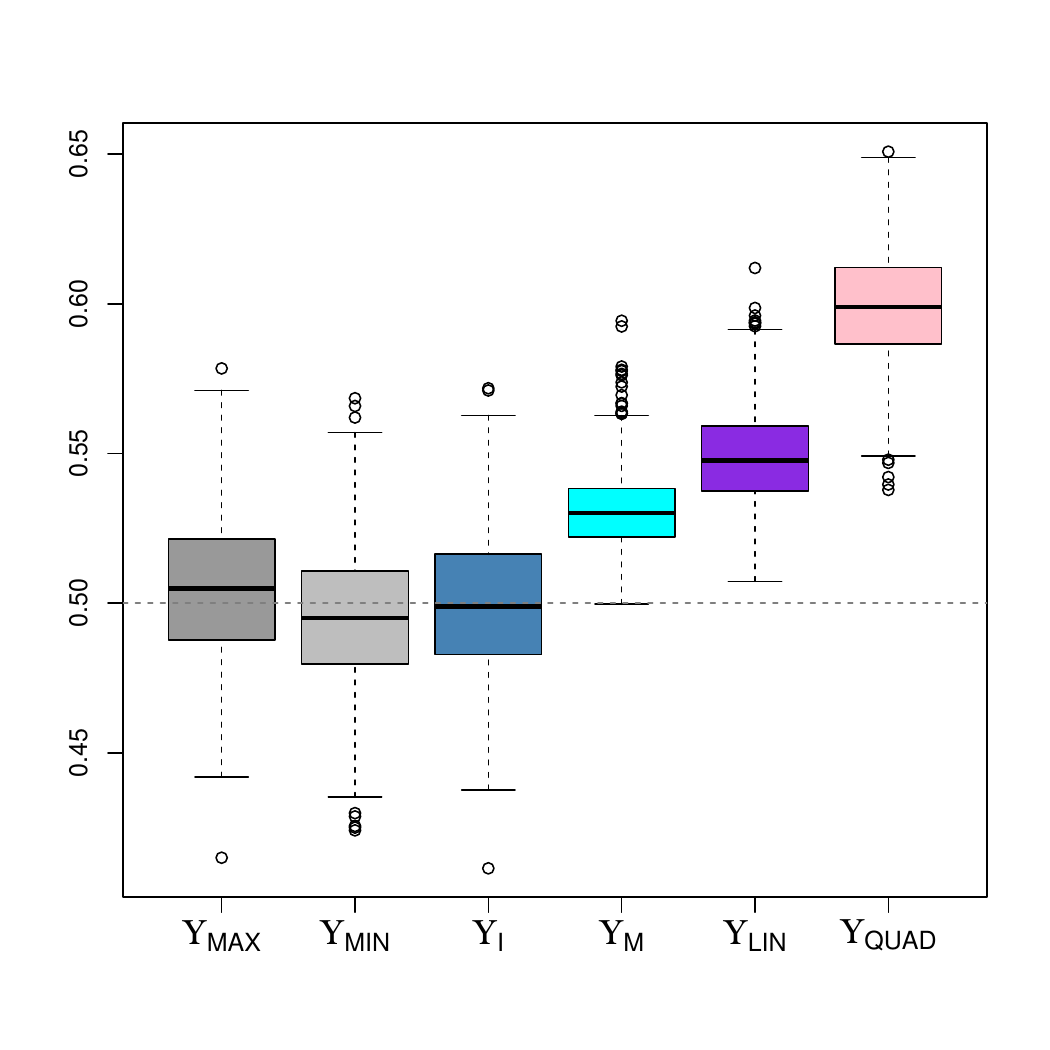}
& \includegraphics[scale=0.35]{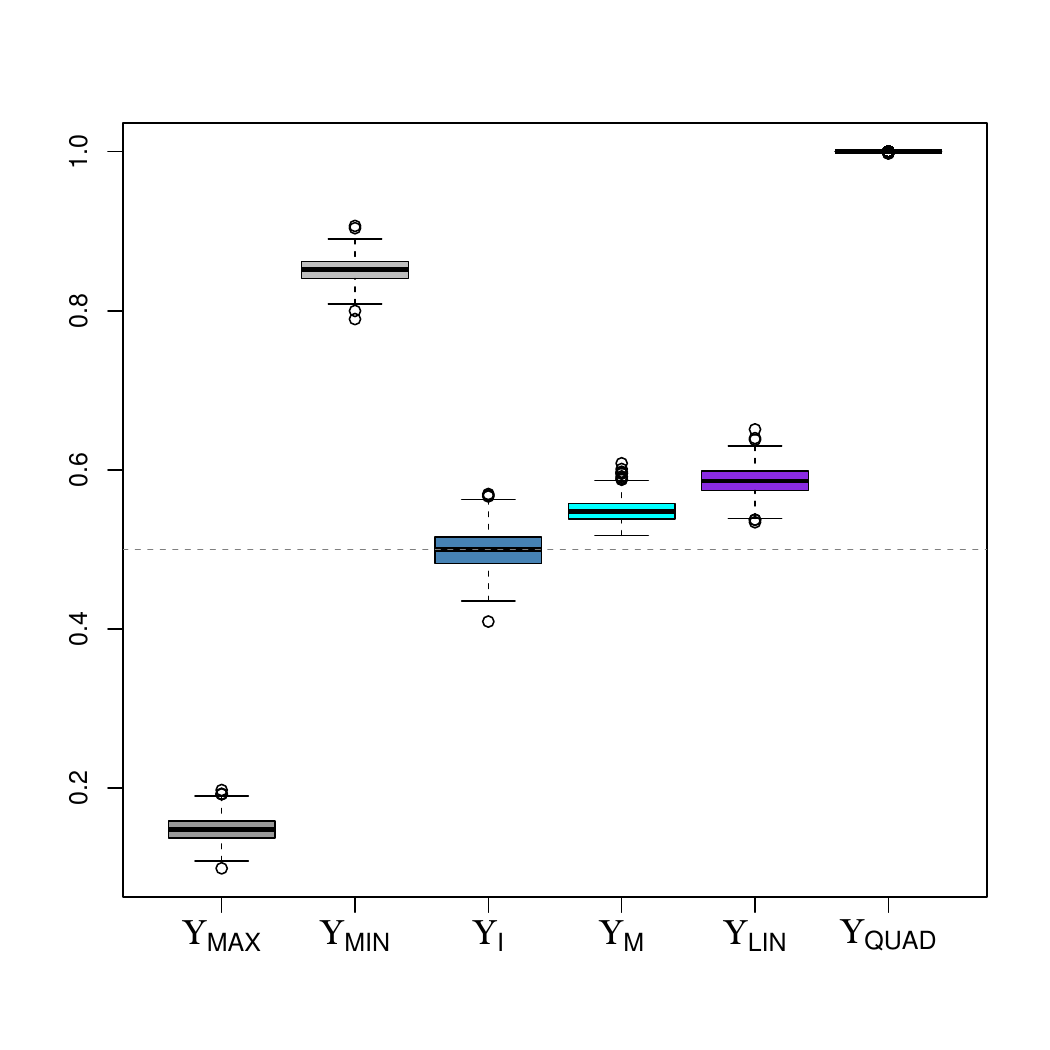}
\end{tabular}
\vskip-0.1in
\caption{Boxplots of the estimators of the AUC under scheme \textbf{DIFF}. The horizontal dashed line, when appearing, indicates 0.5.}
\label{fig:boxplot-DIFF}
\end{center} 
\end{figure}

To visualize the performance of the different estimators of the ROC curve, we used  the functional boxplots, as defined in \citet{sun:genton:2011}. The functional boxplots of the $n_R = 1000$ realizations of the different estimators of the ROC curve   under models \textbf{P0} and \textbf{P1} with $\rho=2$ are displayed in Figures \ref{fig:propor:Brownian} and \ref{fig:propor:varexp}.   In these plots, the magenta central box represents the 50\% inner band of curves, the solid black line indicates the central (deepest) function and the dotted red lines indicate outlying curves (in this case: outlying estimates $\widehat{\ROC}_j$  for some $1 \le j \le 1000$). The blue lines correspond  to the envelopes, that is, the whiskers in the univariate boxplot and demarcate the limits for a curve to be identified as atypical.  The diagonal, in gold color, is shown for comparison purposes.
Similarly, Figures \ref{fig:cpc:C1} and \ref{fig:cpc:C2}   display  the functional boxplots under the \textsc{fcpc} model and  Figures \ref{fig:DIFF-D1} and \ref{fig:DIFF-D2} depict the corresponding ones under schemes  \textbf{D1} and \textbf{D2}, respectively.   We do not show the results for $\Upsilon_{\mini}$ since it corresponds to the procedure with the worst performance.  

The behaviour observed in the functional boxplots is consistent with that of the AUC estimators. When the populations have equal mean, for the linear indexes,  $\Upsilon_{\inte}$,  $\wUps_{\media}$ and $\wUps_{\lin}$, the central region containing the 50\% deepest ROC curve estimators includes or crosses the identity function, except under the proportional model when considering the Exponential Variogram where the ROC curve estimators associated to $\wUps_{\media}$ and $\wUps_{\lin}$ exceed the diagonal for values of $p$ smaller than 0.4. The worst scenarios for these rules seem to be the \textsc{fcpc} model  under \textbf{C10}, when $\lambda_D=(2,0.30,0.05)\trasp$, and under scheme \textbf{D20} for which $X_D$  follows   a Brownian motion and $X_H$  an Exponential Variogram. The quadratic index results in the best discriminating index for the considered simulation schemes, providing a perfect rule under \textbf{D2}.

\begin{figure}[ht!]
 \begin{center}
 \footnotesize
 \renewcommand{\arraystretch}{0.2}
\begin{tabular}{p{2cm} cc}
 & \textbf{P1} &  \textbf{P0} \\[-2ex]
$\Upsilon_{\maxi}$  &
 \raisebox{-.5\height}{\includegraphics[scale=0.25]{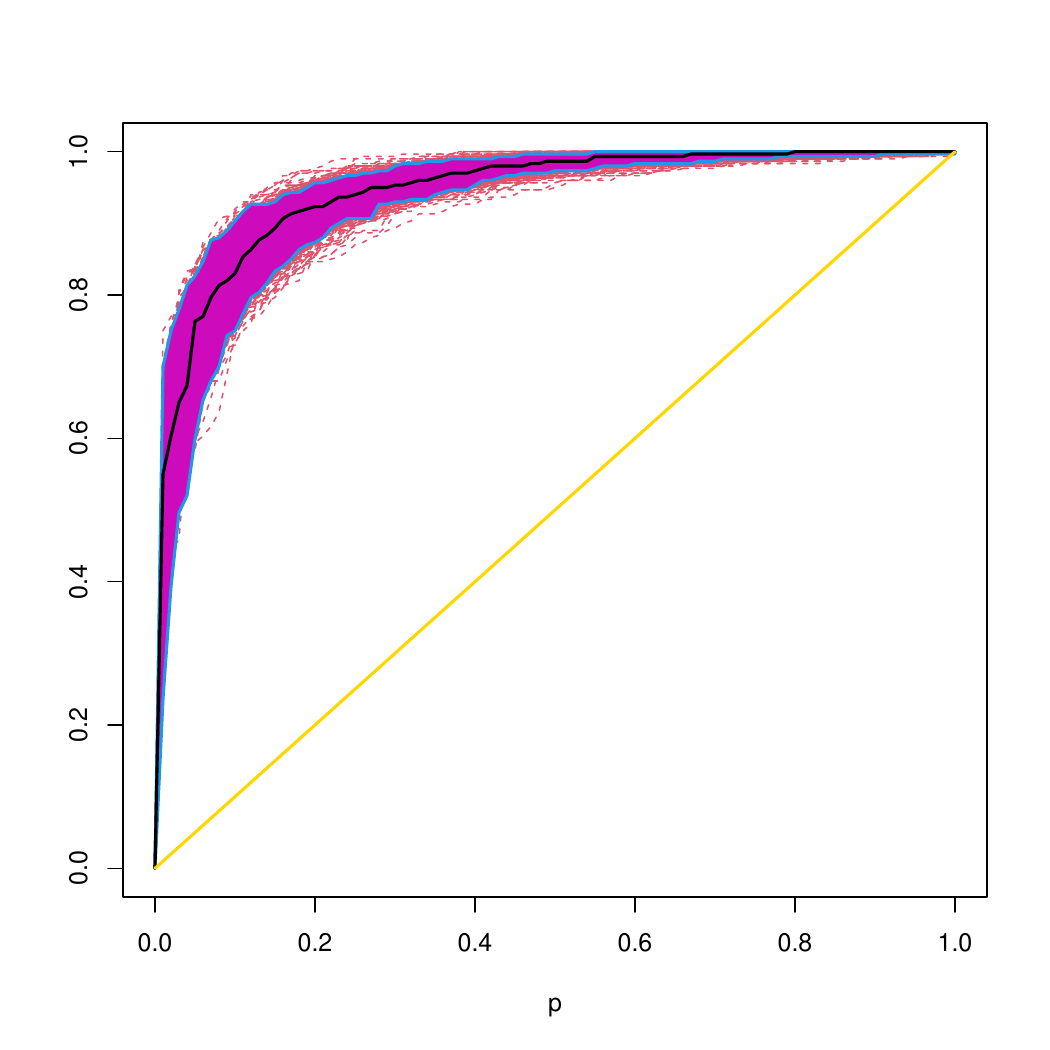}}
& \raisebox{-.5\height}{\includegraphics[scale=0.25]{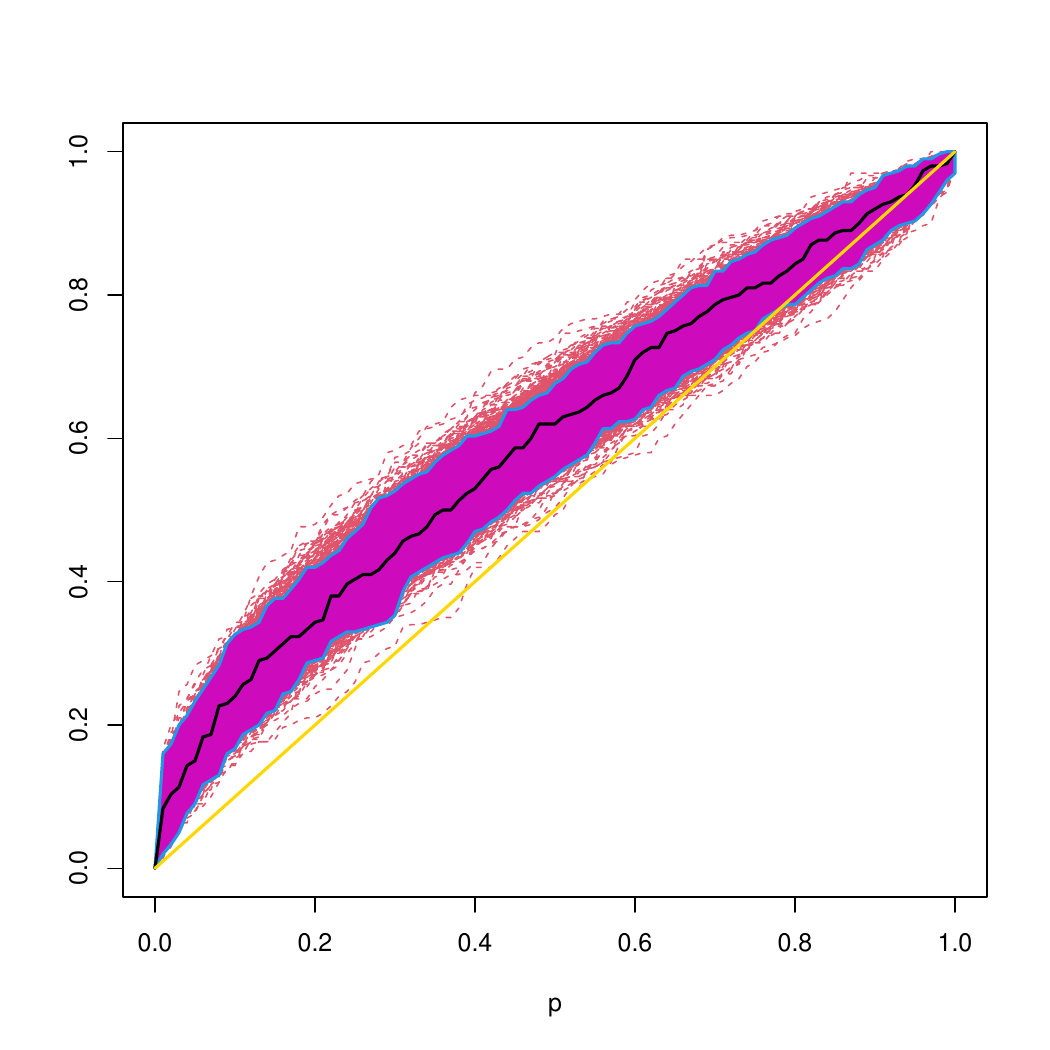}}
\\[-4ex]
$\Upsilon_{\inte}$  &
 \raisebox{-.5\height}{\includegraphics[scale=0.25]{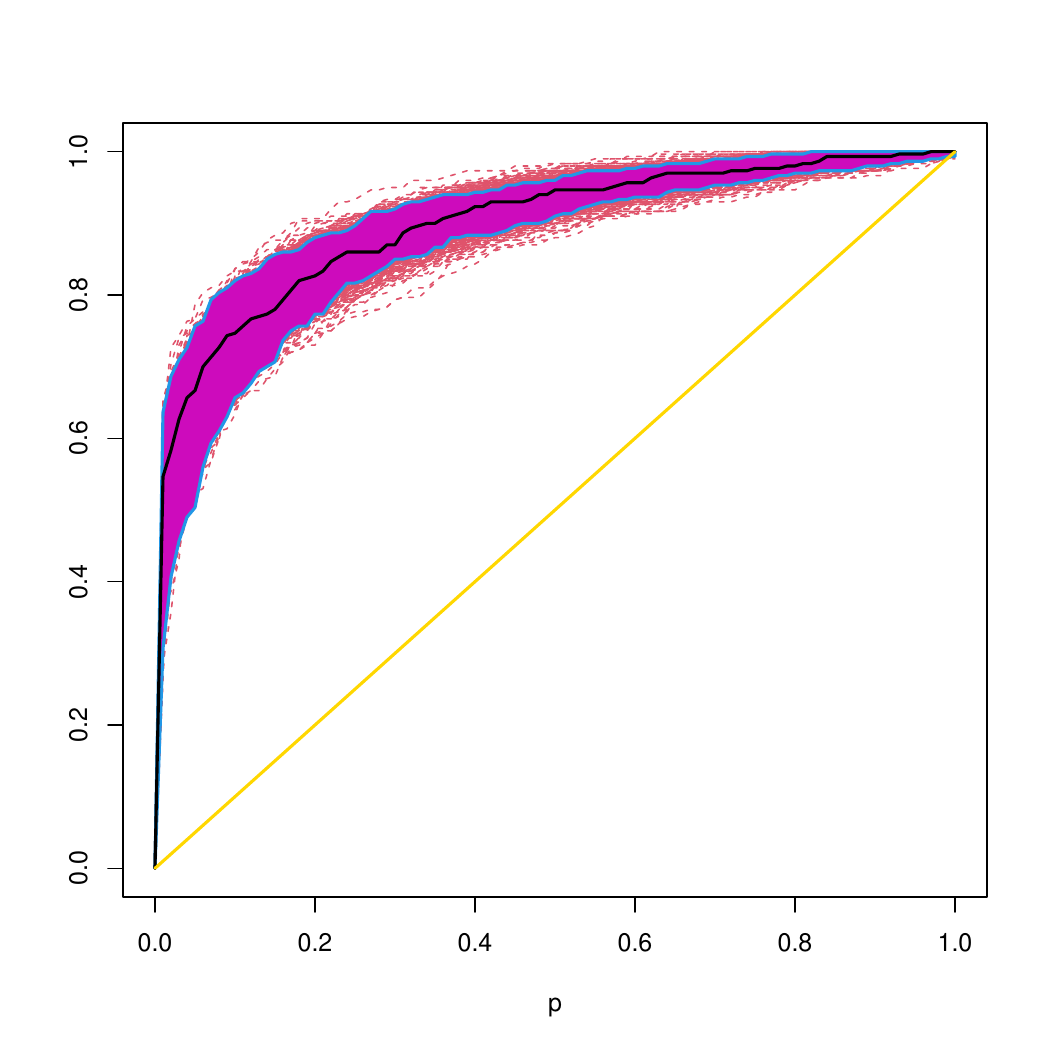}}
& \raisebox{-.5\height}{\includegraphics[scale=0.25]{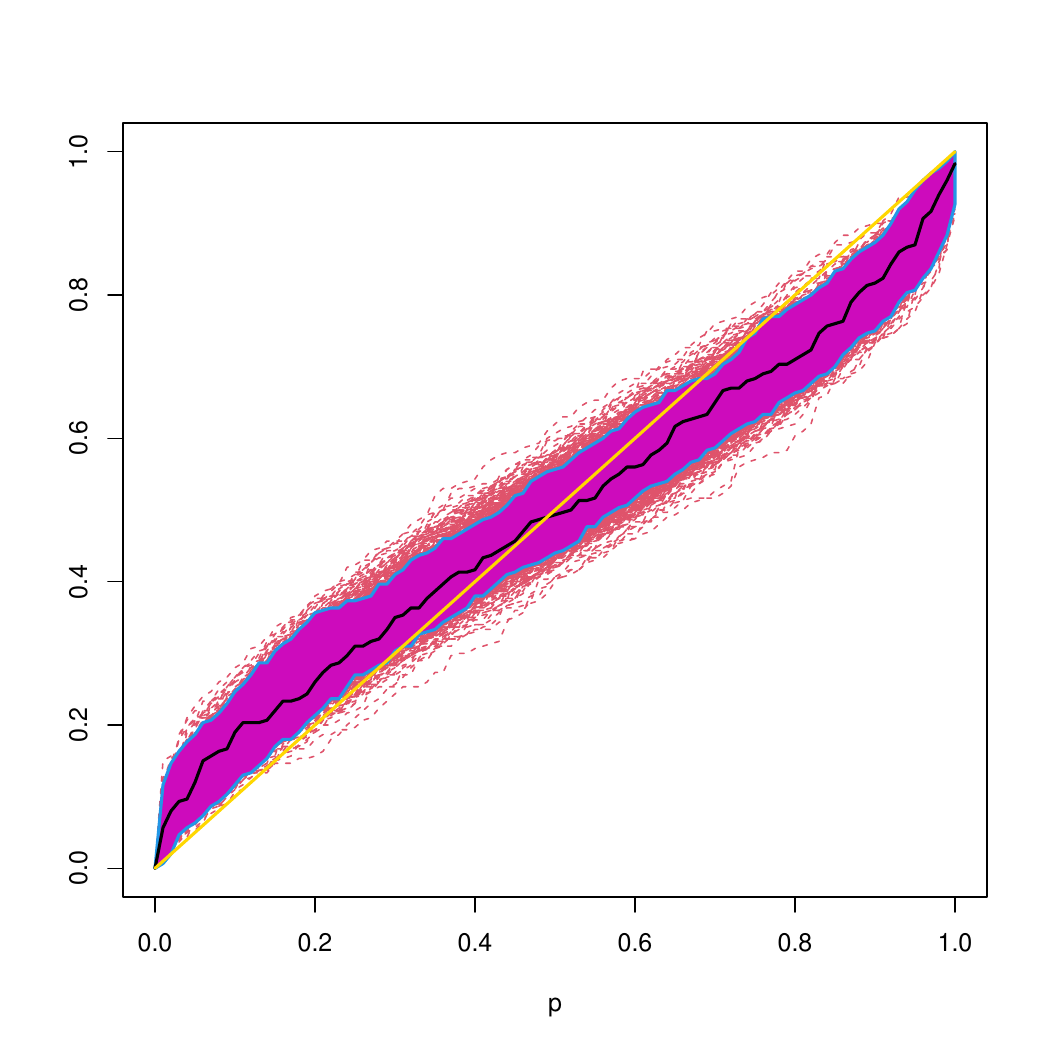}}
 \\[-4ex]
    
$\wUps_{\media}$ &
\raisebox{-.5\height}{\includegraphics[scale=0.25]{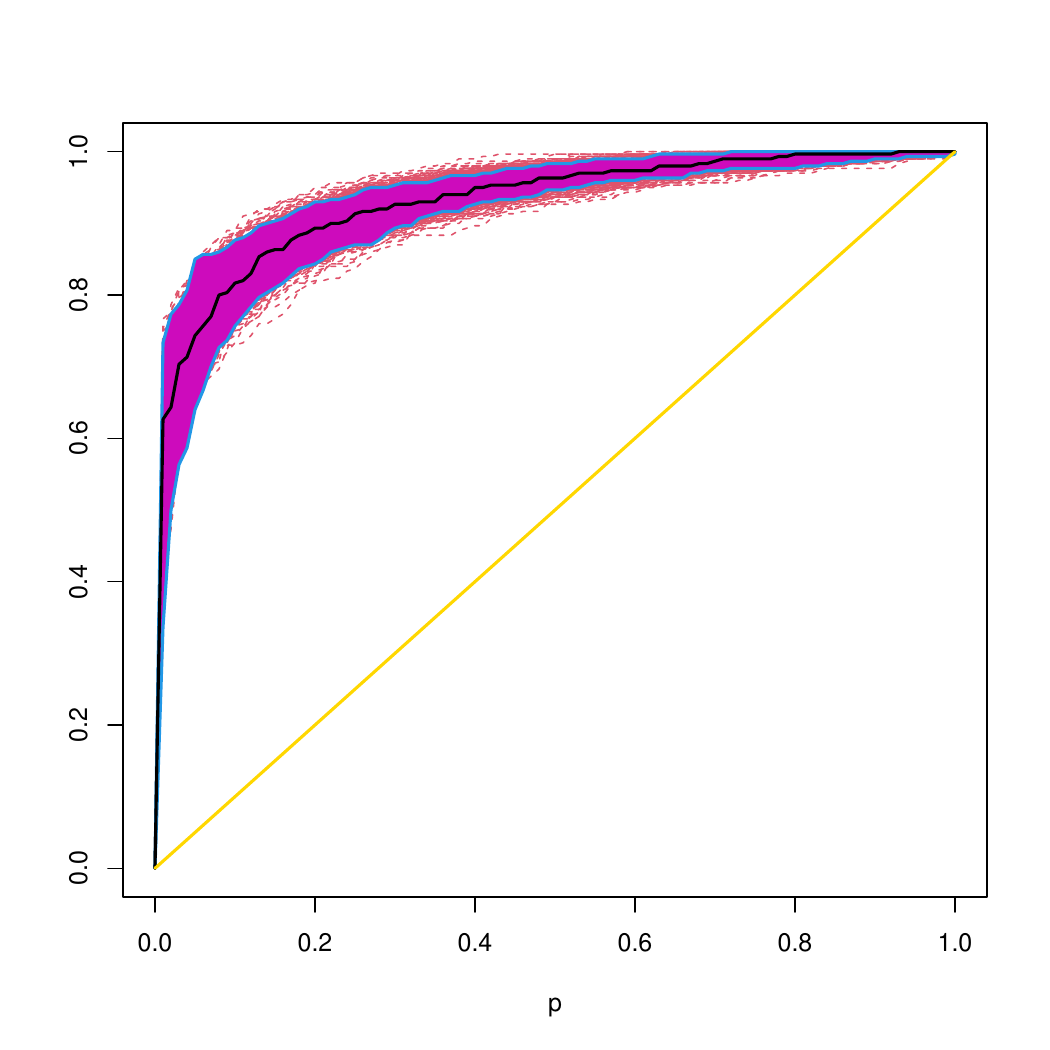}}
& \raisebox{-.5\height}{\includegraphics[scale=0.25]{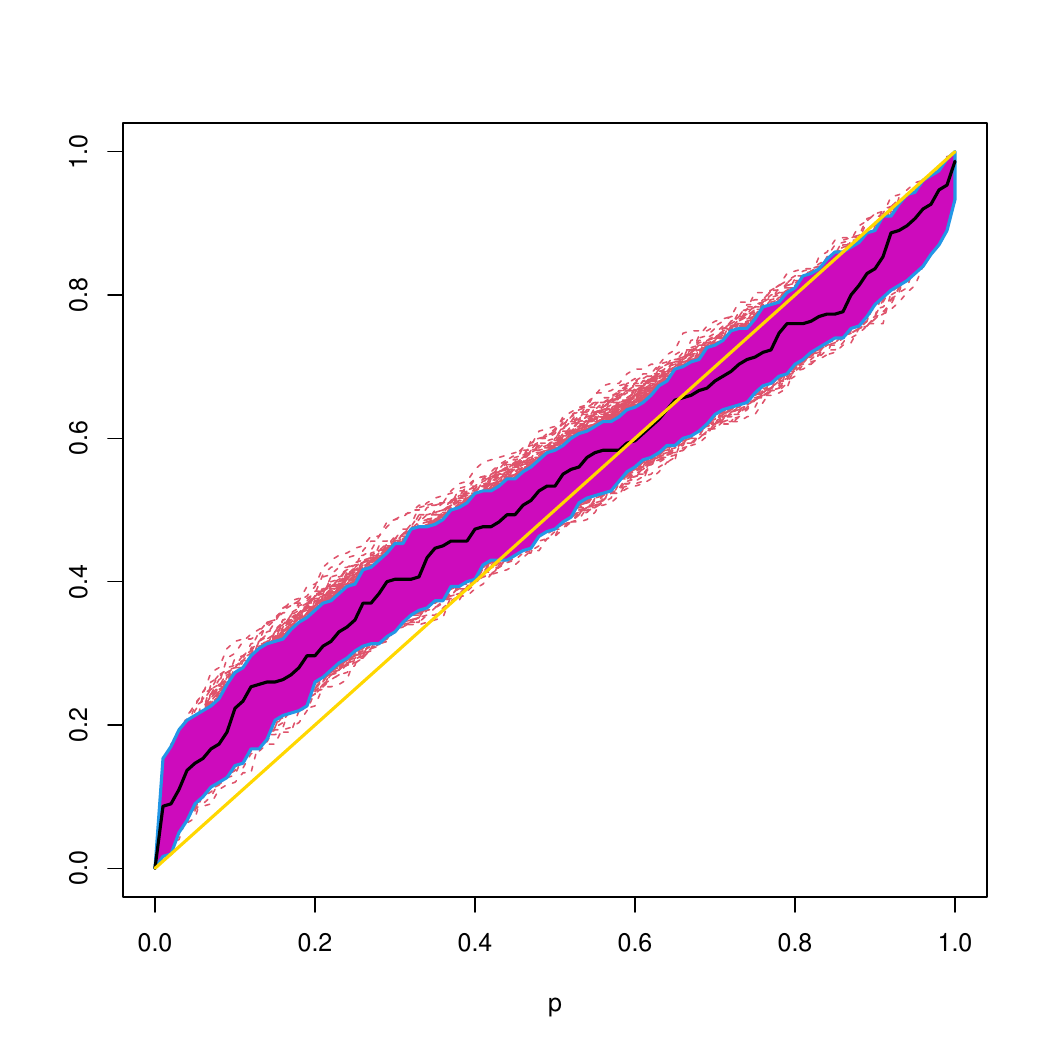}}
\\[-4ex]

$\wUps_{\lin}$  & 
\raisebox{-.5\height}{\includegraphics[scale=0.25]{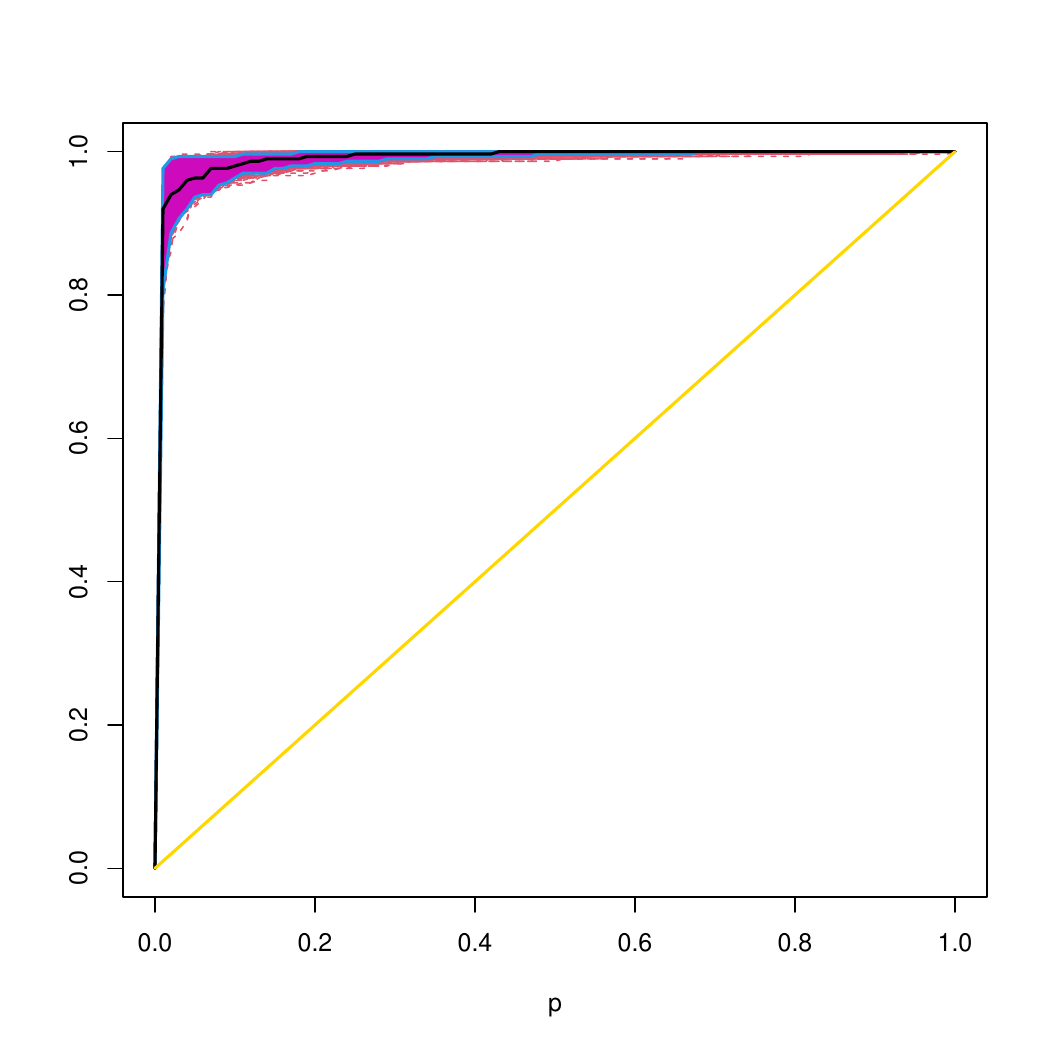}}
& \raisebox{-.5\height}{\includegraphics[scale=0.25]{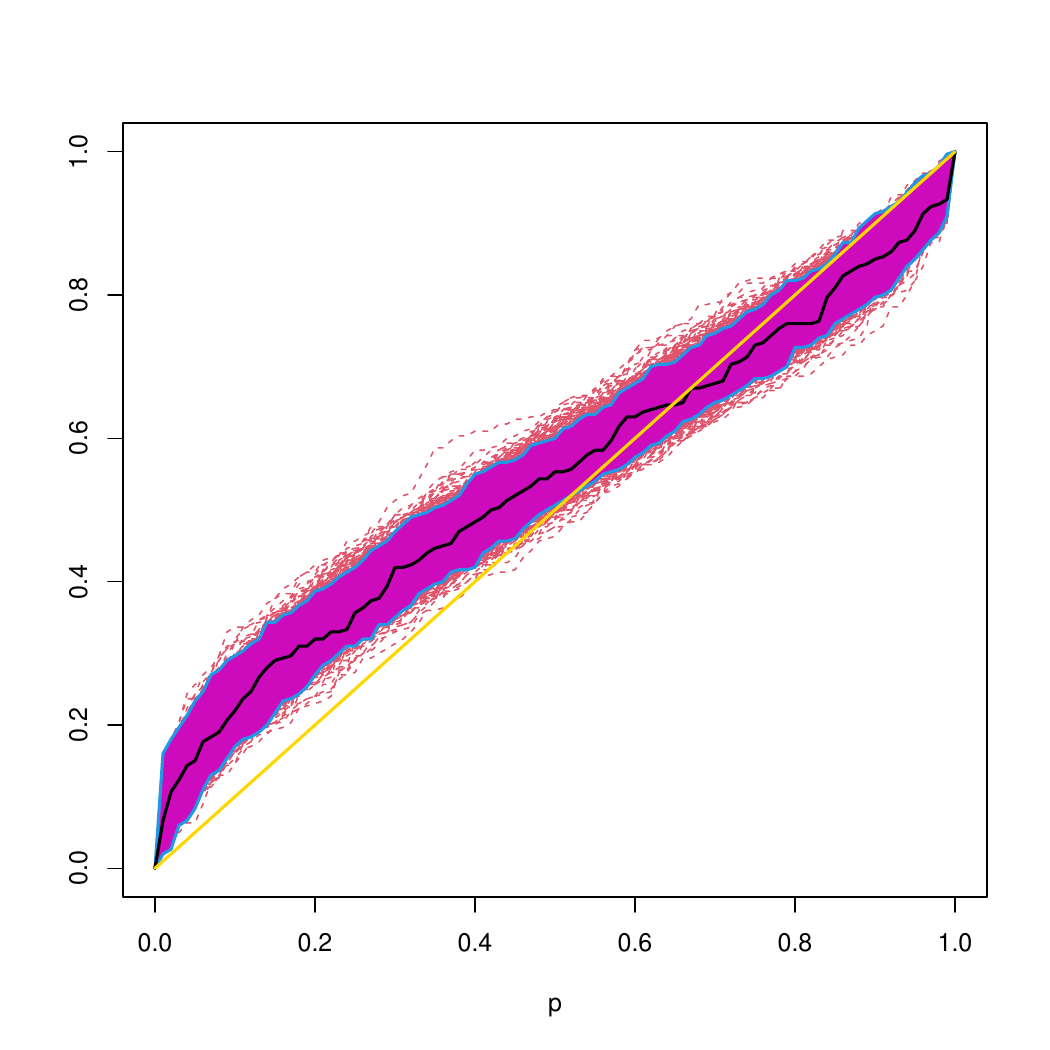}}
\\[-4ex]

$\wUps_{\cuad}$ & 
\raisebox{-.5\height}{\includegraphics[scale=0.25]{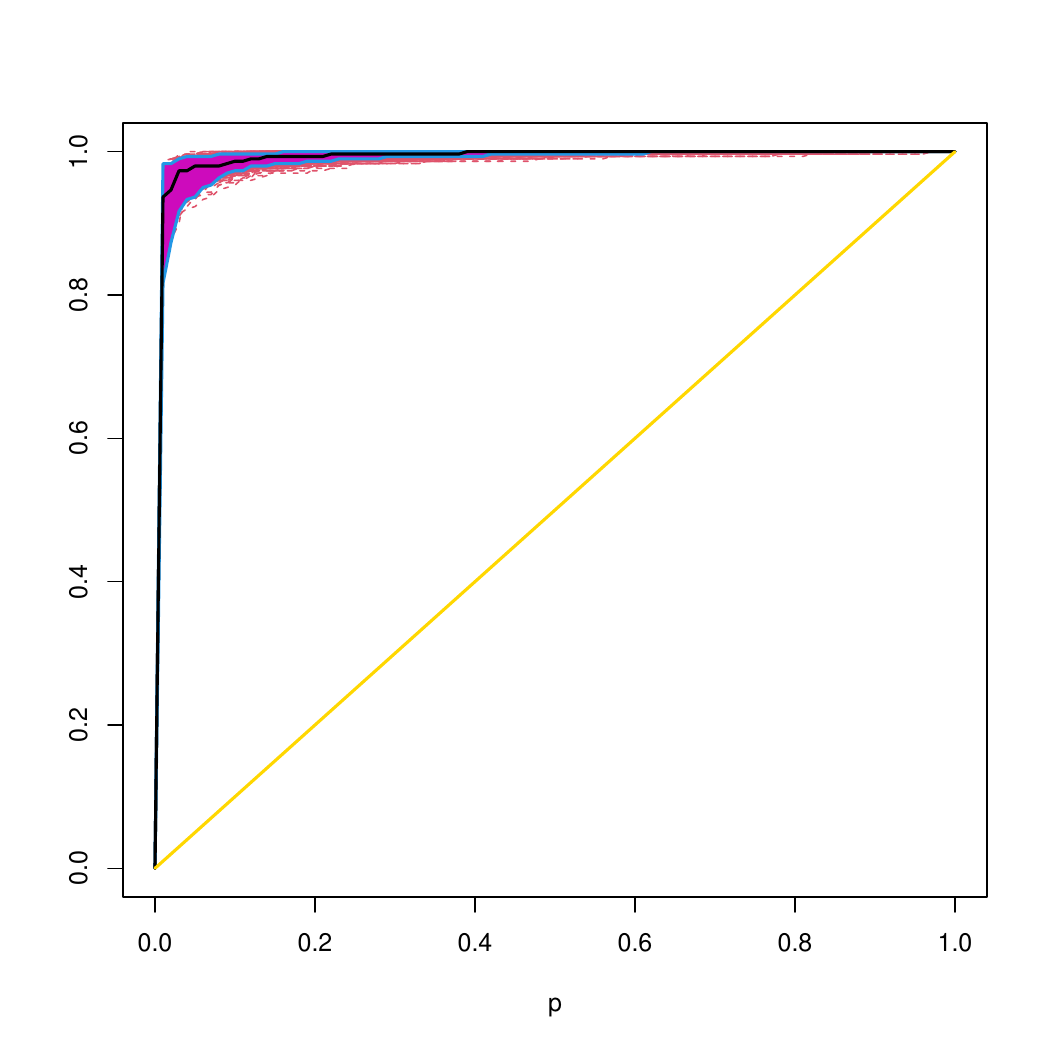}}
& \raisebox{-.5\height}{\includegraphics[scale=0.25]{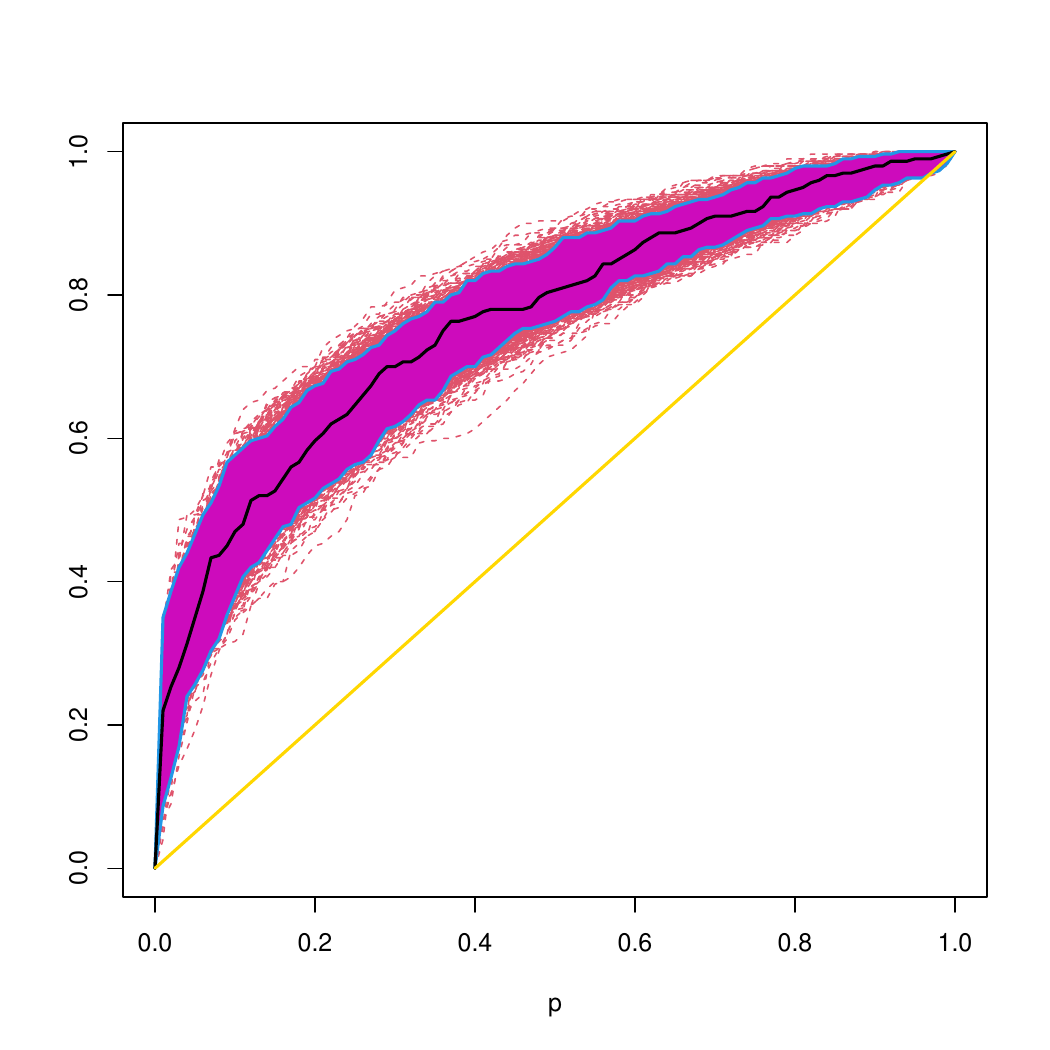}}

\end{tabular}
\caption{Functional boxplots of the estimators $\widehat{\ROC}$ under scenario \textbf{PROP} with  $\rho=2$ for the Brownian motion setting. Rows correspond to discriminating indexes, while columns to  $\mu_D(t)=2\, \sin(\pi  t)$ and $\mu_H=0$.}
\label{fig:propor:Brownian}
\end{center} 
\end{figure}

\begin{figure}[ht!]
 \begin{center}
 \footnotesize
 \renewcommand{\arraystretch}{0.2}

\begin{tabular}{p{2cm} cc}
 & \textbf{P1} &  \textbf{P0} \\[-2ex]  
$\Upsilon_{\maxi}$ &
 \raisebox{-.5\height}{\includegraphics[scale=0.25]{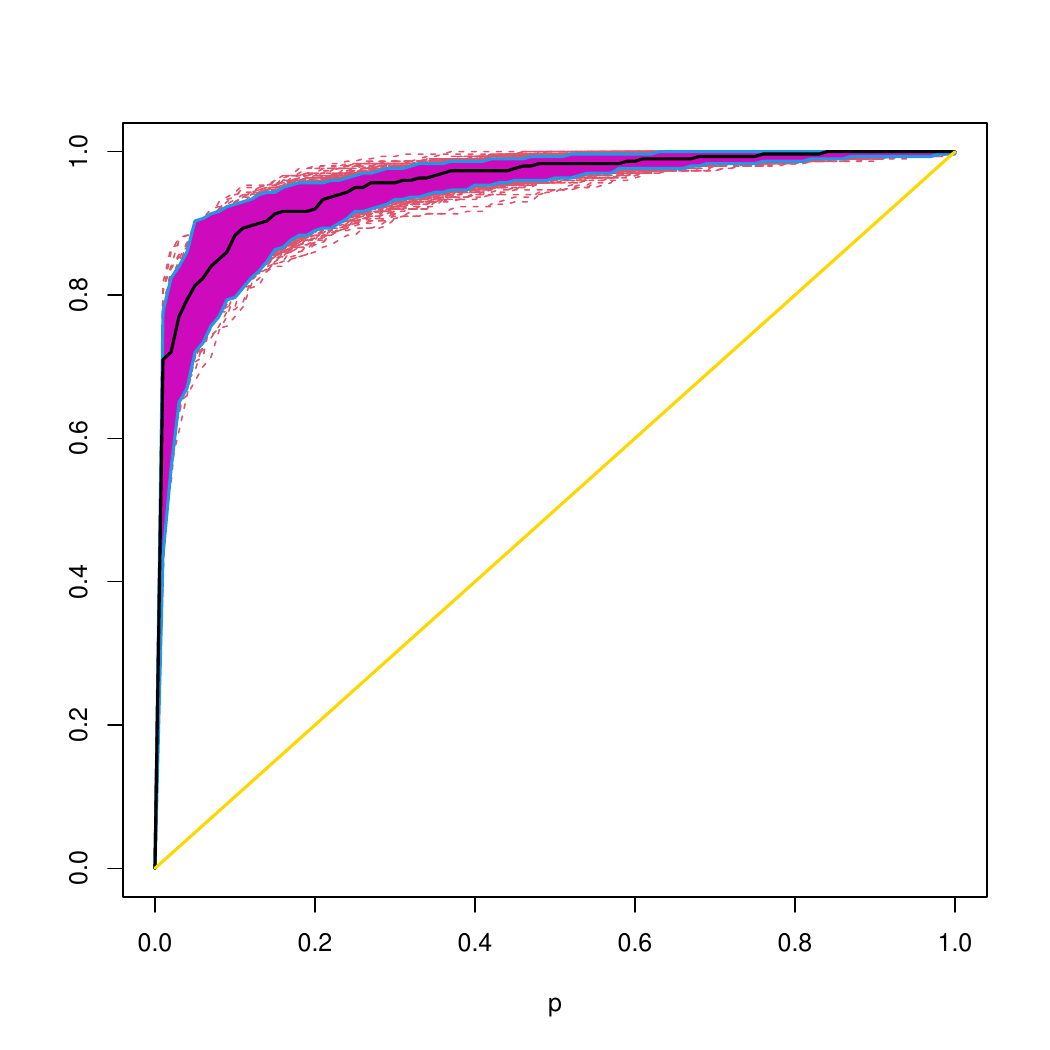}}
& \raisebox{-.5\height}{\includegraphics[scale=0.25]{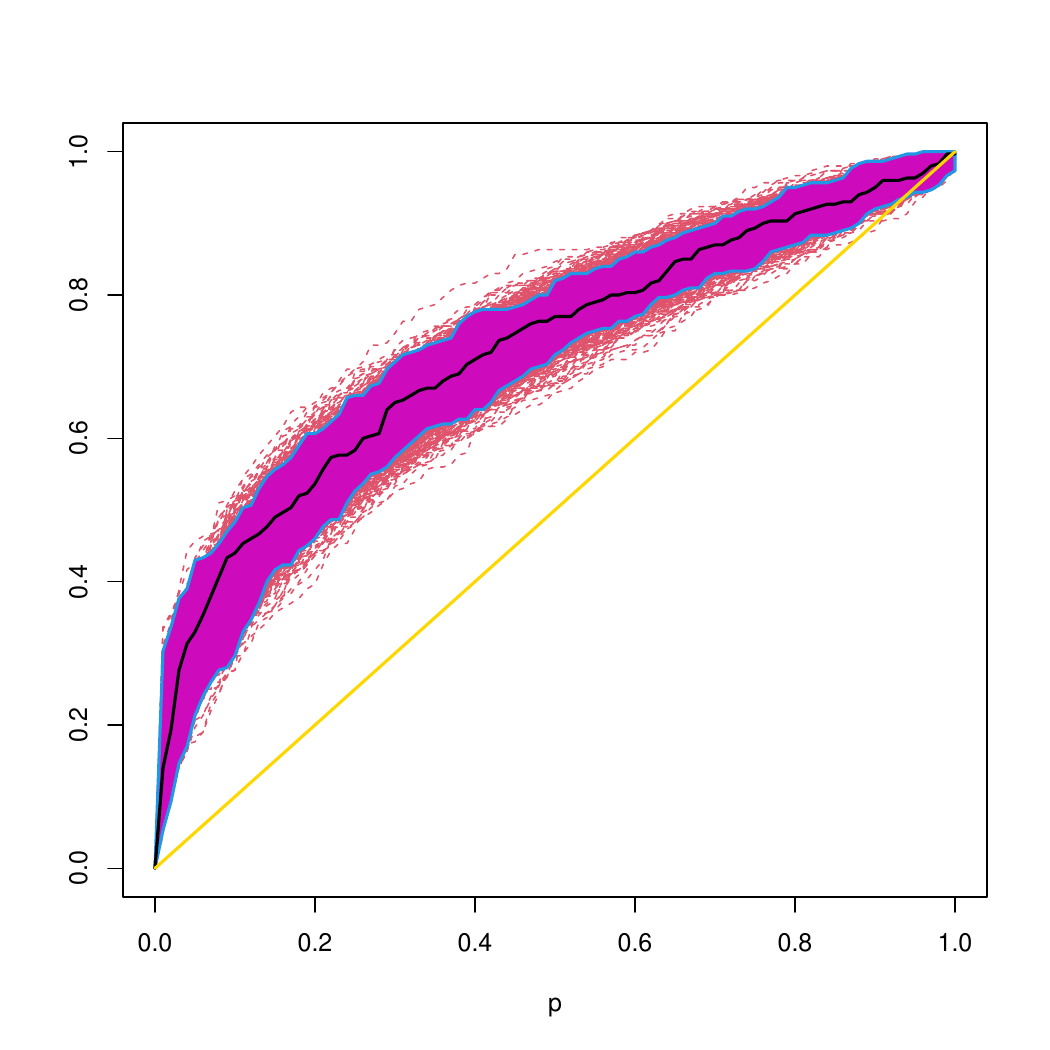}}
\\[-4ex]
 
$\Upsilon_{\inte}$  &
 \raisebox{-.5\height}{\includegraphics[scale=0.25]{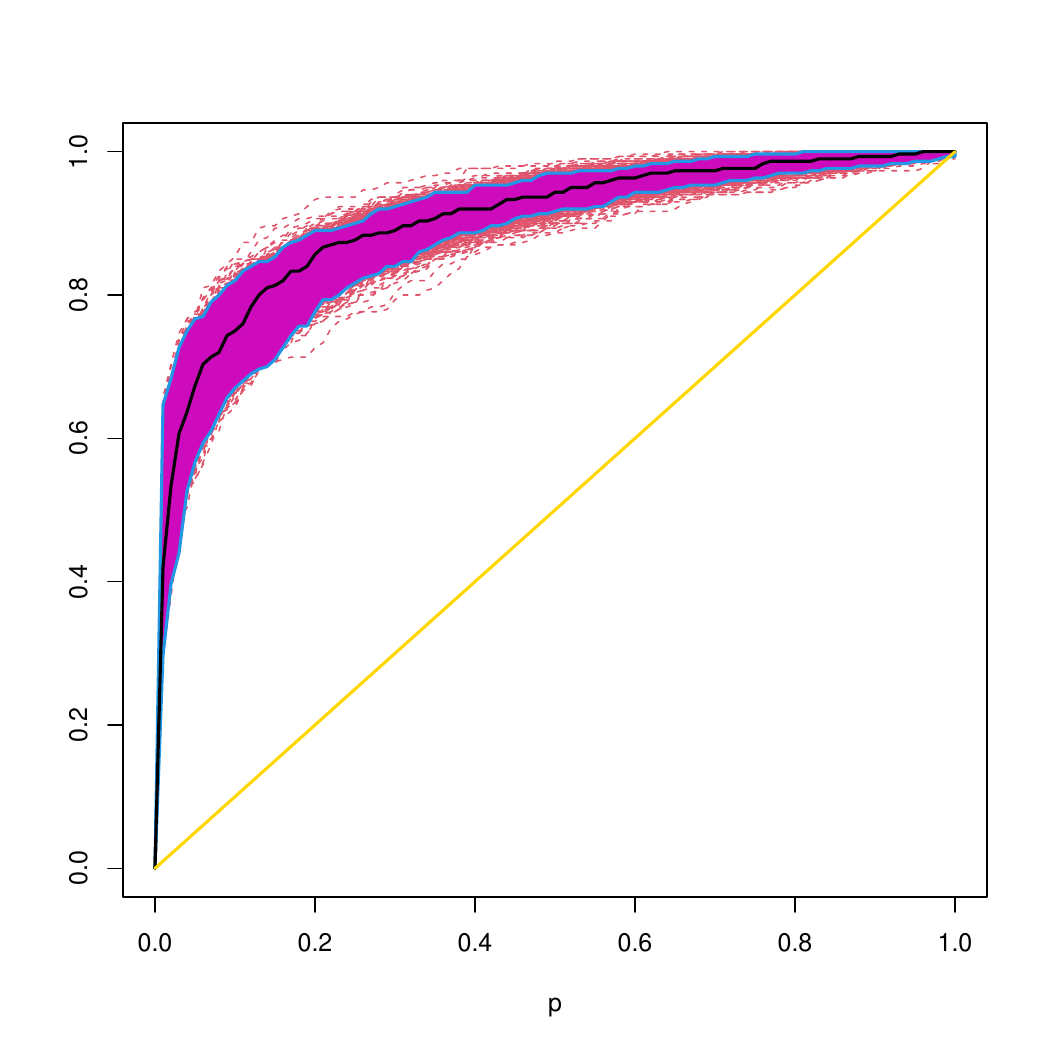}}
& \raisebox{-.5\height}{\includegraphics[scale=0.25]{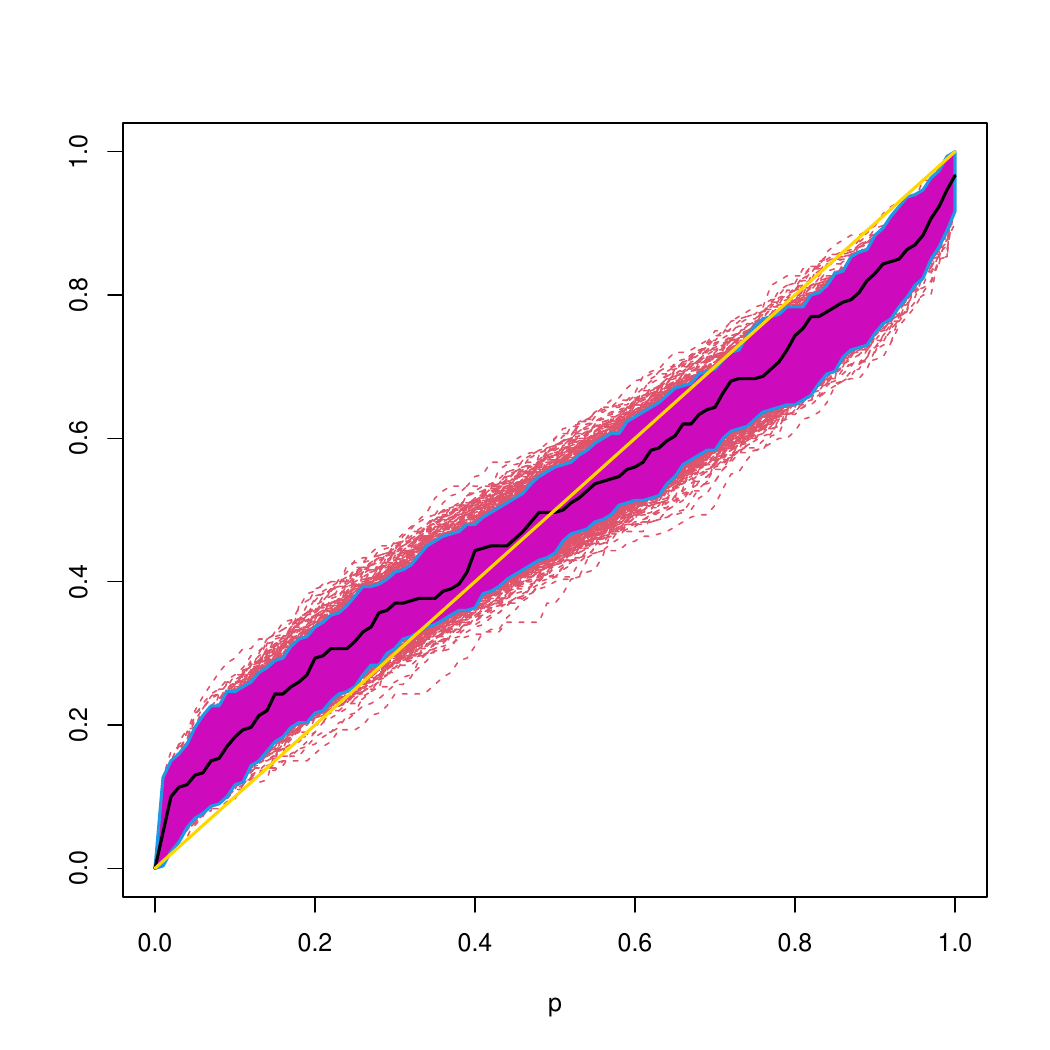}}
 \\[-4ex]
    
$\wUps_{\media}$ &
\raisebox{-.5\height}{\includegraphics[scale=0.25]{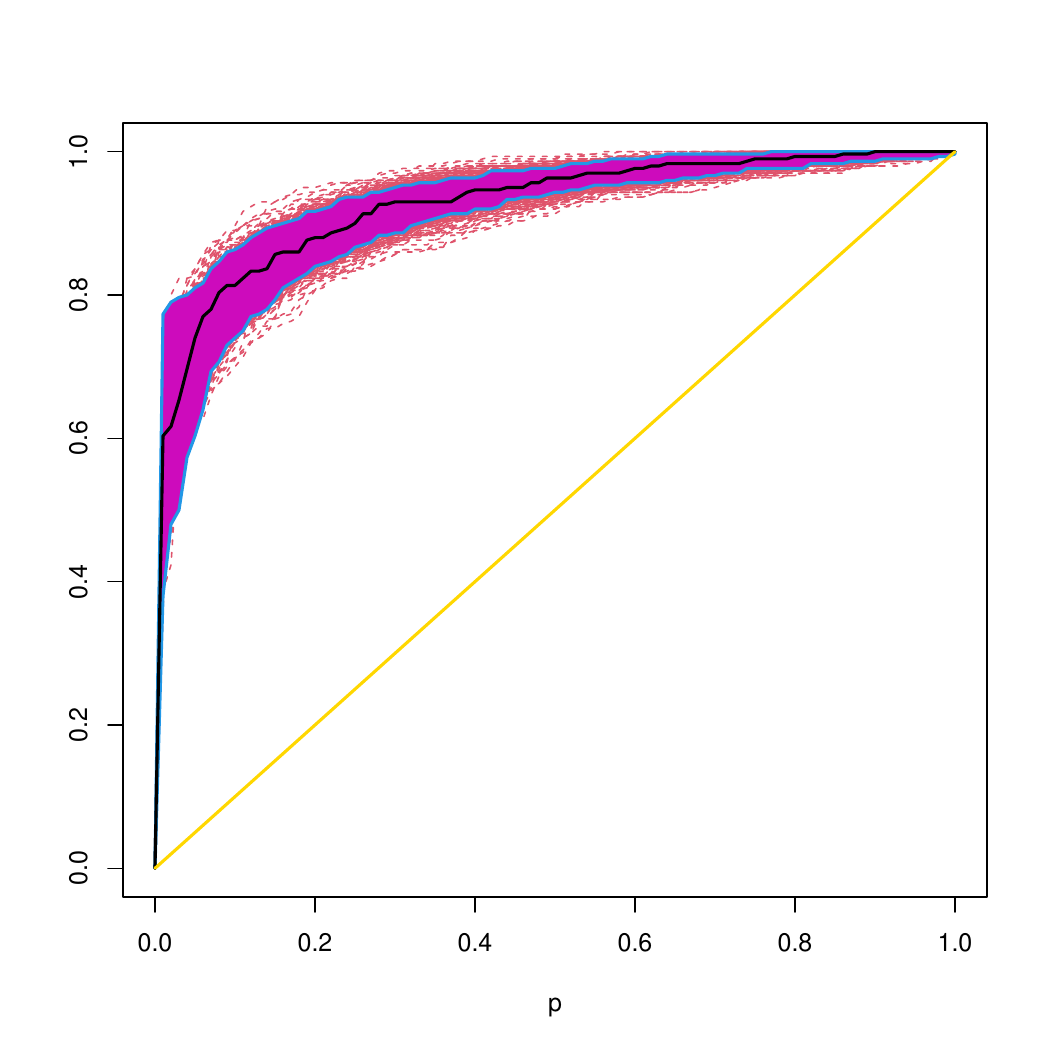}}
& \raisebox{-.5\height}{\includegraphics[scale=0.25]{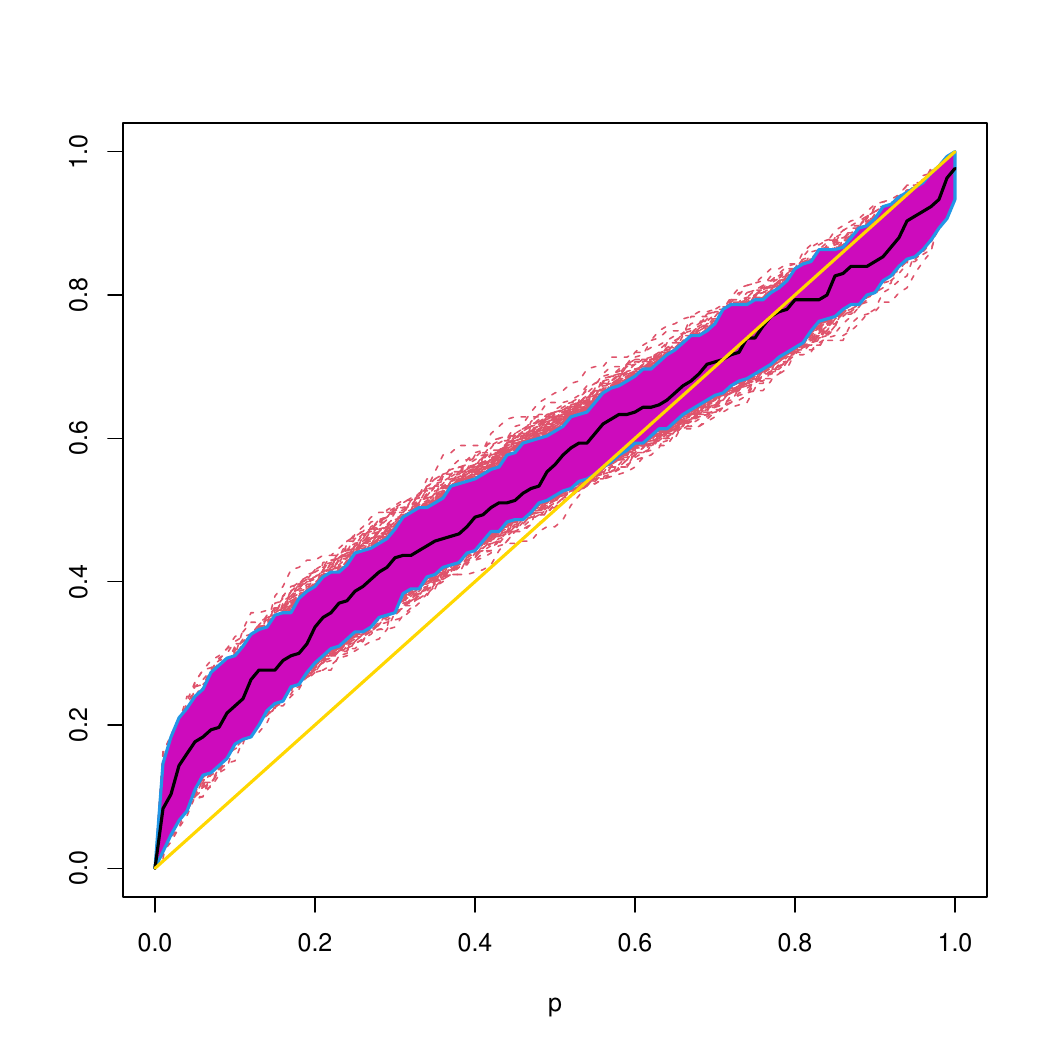}}
\\[-4ex]

$\wUps_{\lin}$  & 
\raisebox{-.5\height}{\includegraphics[scale=0.25]{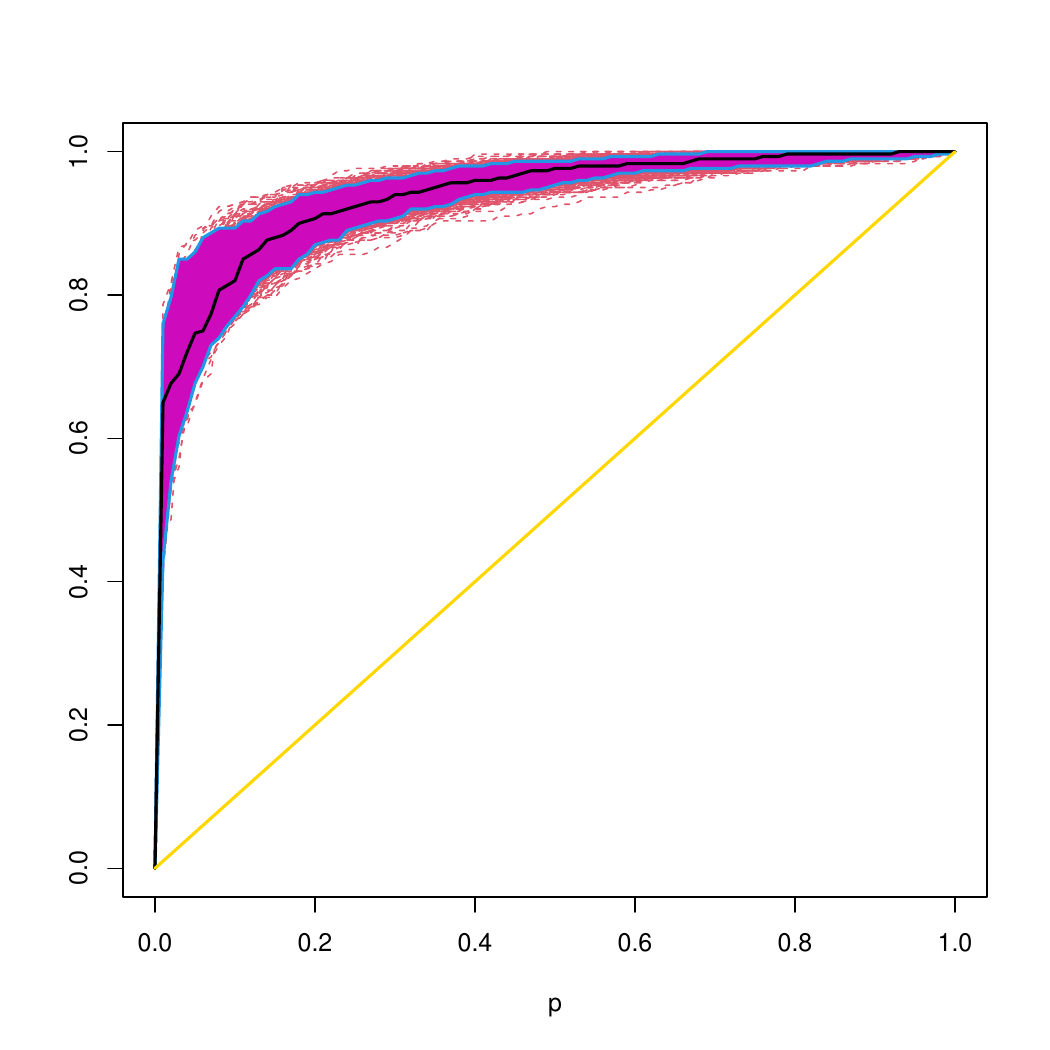}}
& \raisebox{-.5\height}{\includegraphics[scale=0.25]{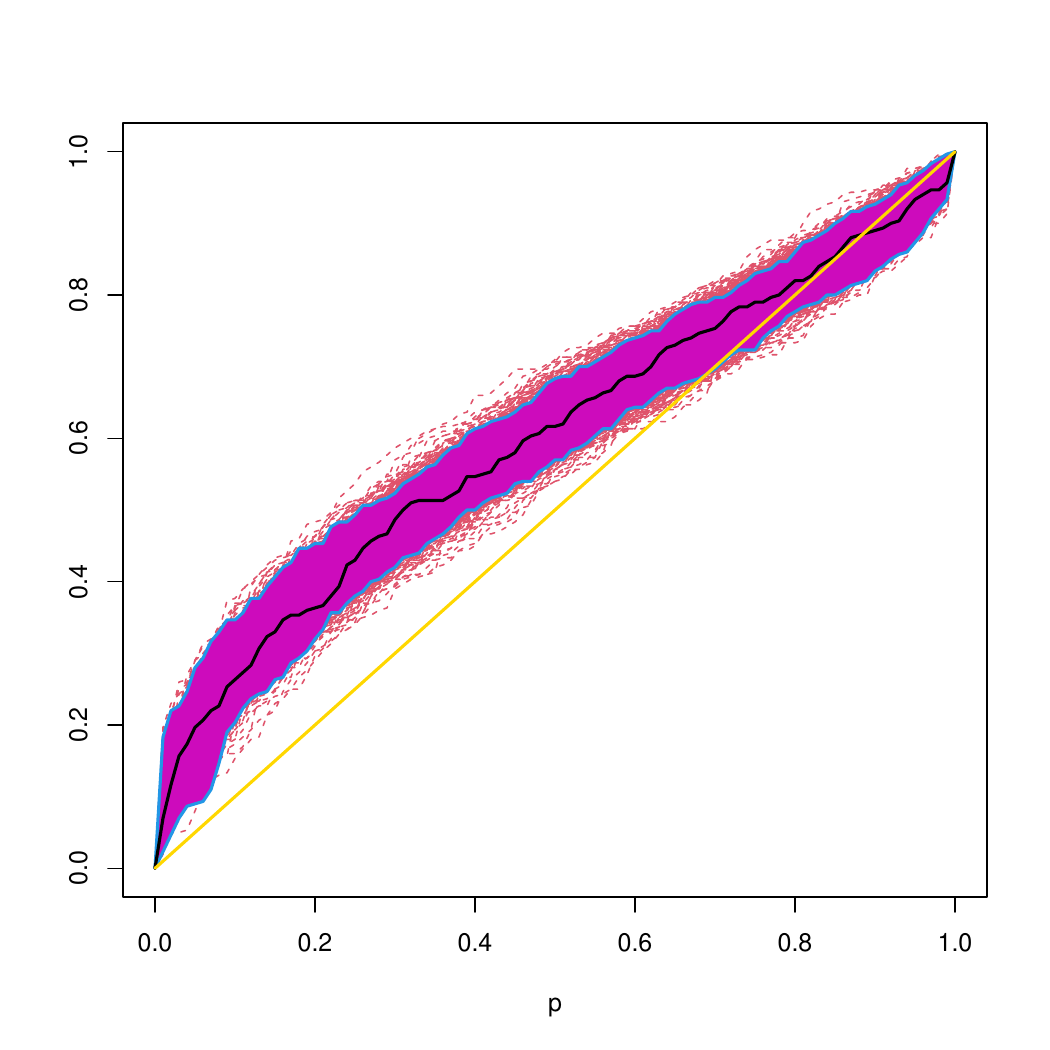}}
\\[-4ex]

$\wUps_{\cuad}$ 
& \raisebox{-.5\height}{\includegraphics[scale=0.25]{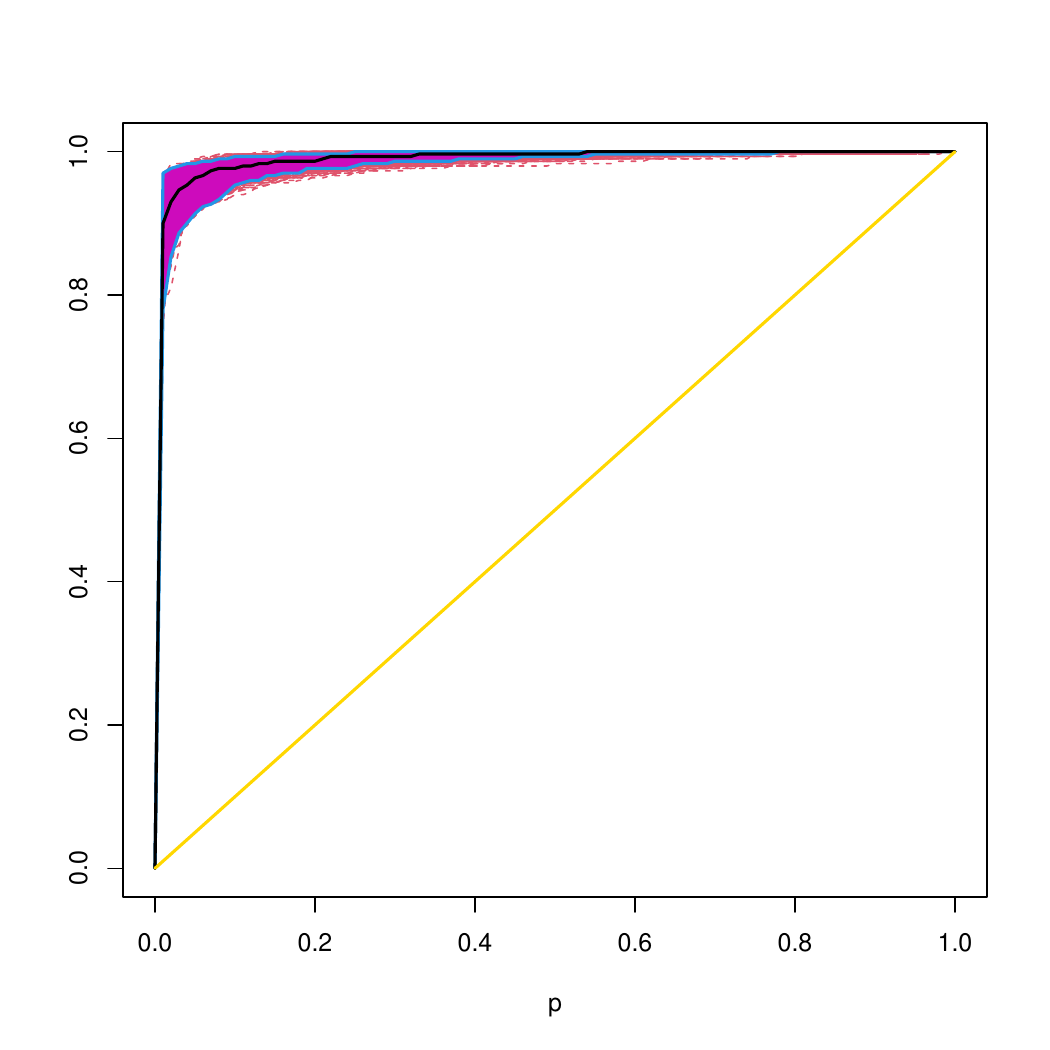}}
& \raisebox{-.5\height}{\includegraphics[scale=0.25]{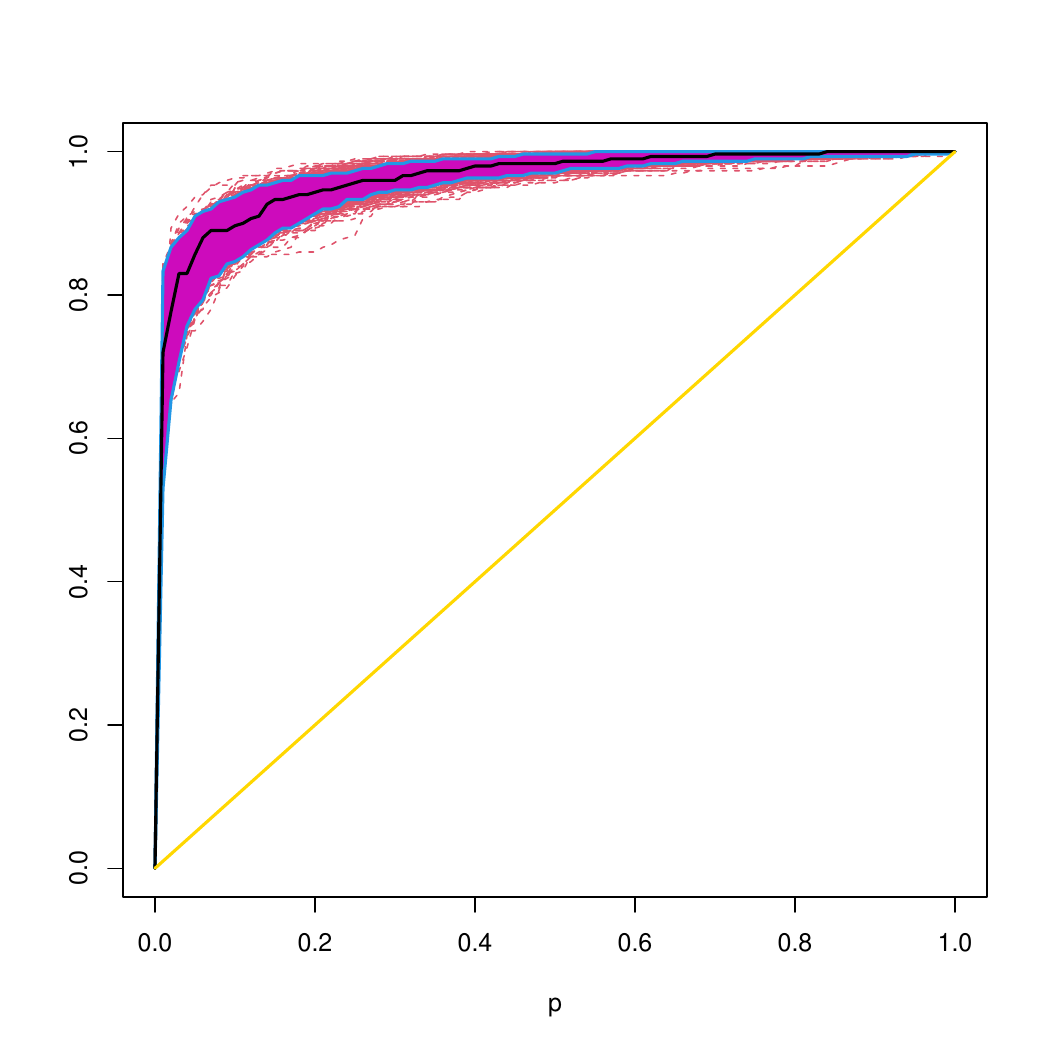}}

\end{tabular}
\caption{Functional boxplots of the estimators $\widehat{\ROC}$ under scenario \textbf{PROP} with  $\rho=2$ for the Exponential Variogram  process. Rows correspond to discriminating indexes, while columns to $\mu_D(t)=2\, \sin(\pi t)$ and $\mu_H=0$.}
\label{fig:propor:varexp} 
\end{center} 
\end{figure}


\begin{figure}[ht!]
 \begin{center}
 \footnotesize
 \renewcommand{\arraystretch}{0.2}
 \begin{tabular}{p{2cm} cc}
 & \textbf{C11} &  \textbf{C10} \\[-2ex]
$\Upsilon_{\maxi}$  &
 \raisebox{-.5\height}{\includegraphics[scale=0.25]{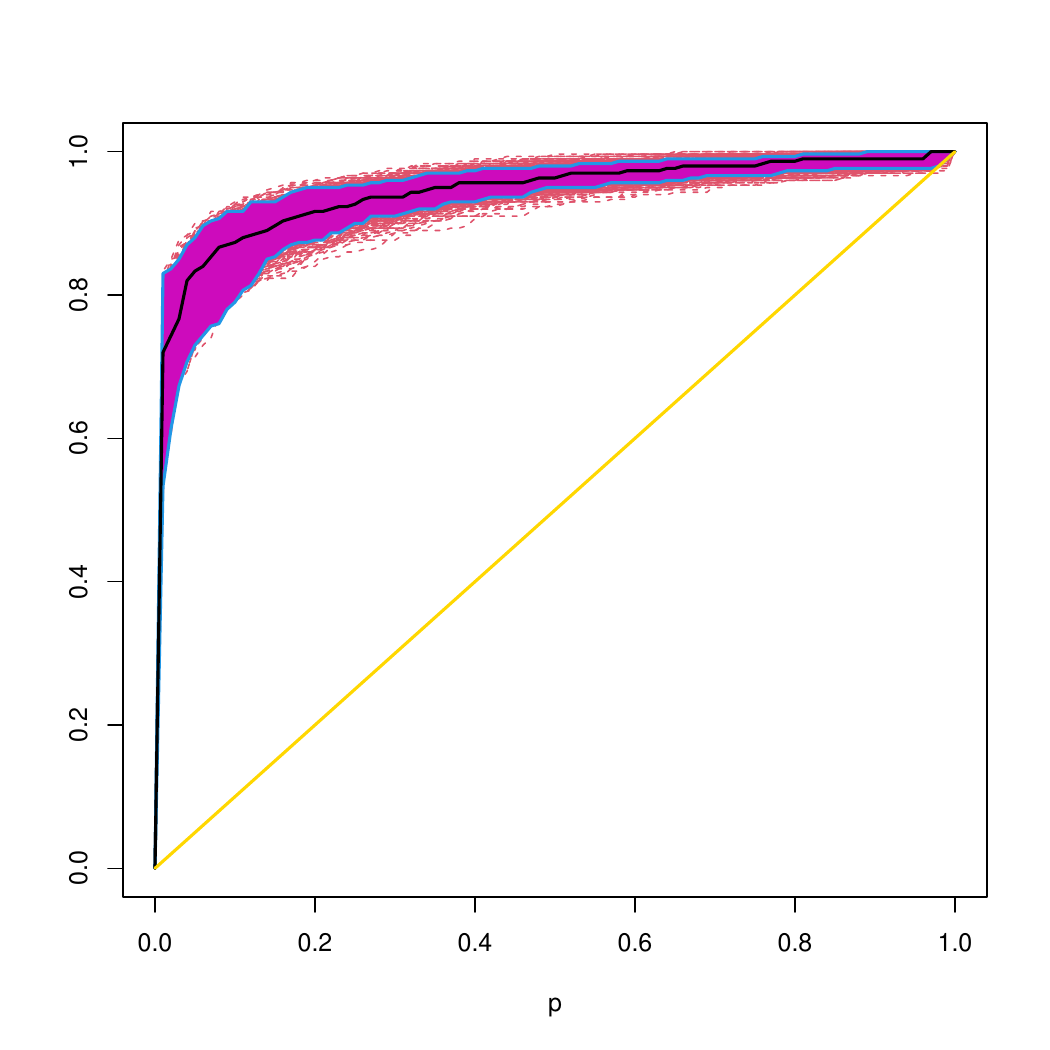}}
& \raisebox{-.5\height}{\includegraphics[scale=0.25]{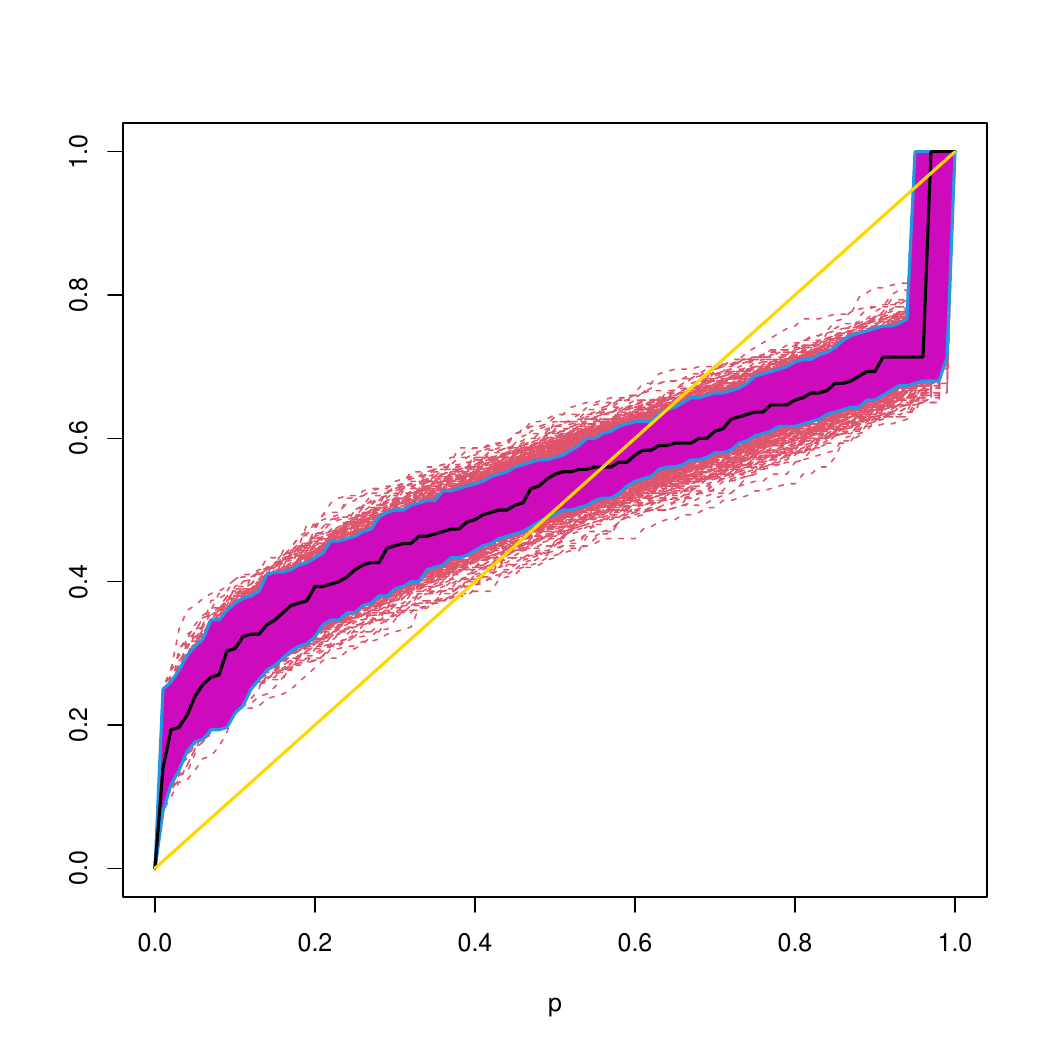}}
\\[-4ex]
 
$\Upsilon_{\inte}$  &
 \raisebox{-.5\height}{\includegraphics[scale=0.25]{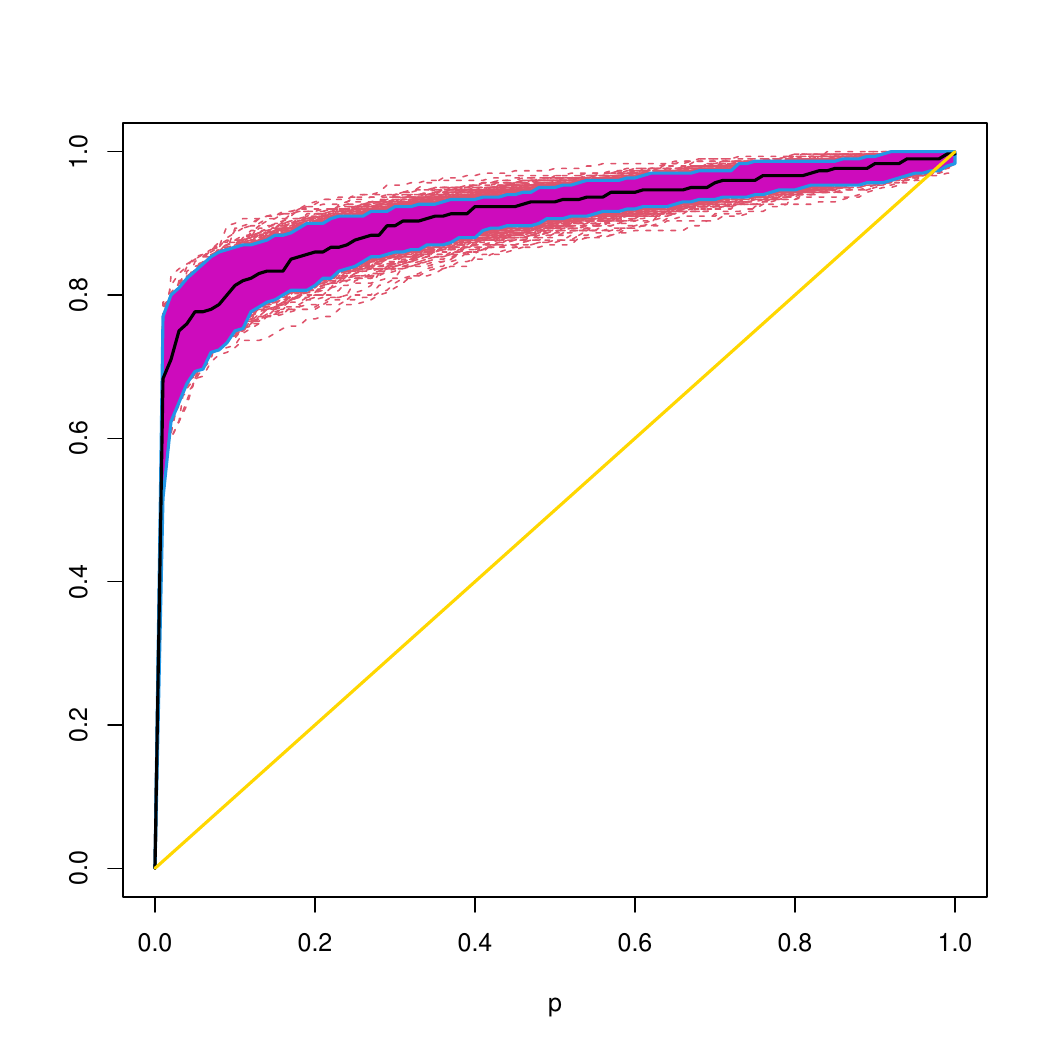}}
& \raisebox{-.5\height}{\includegraphics[scale=0.25]{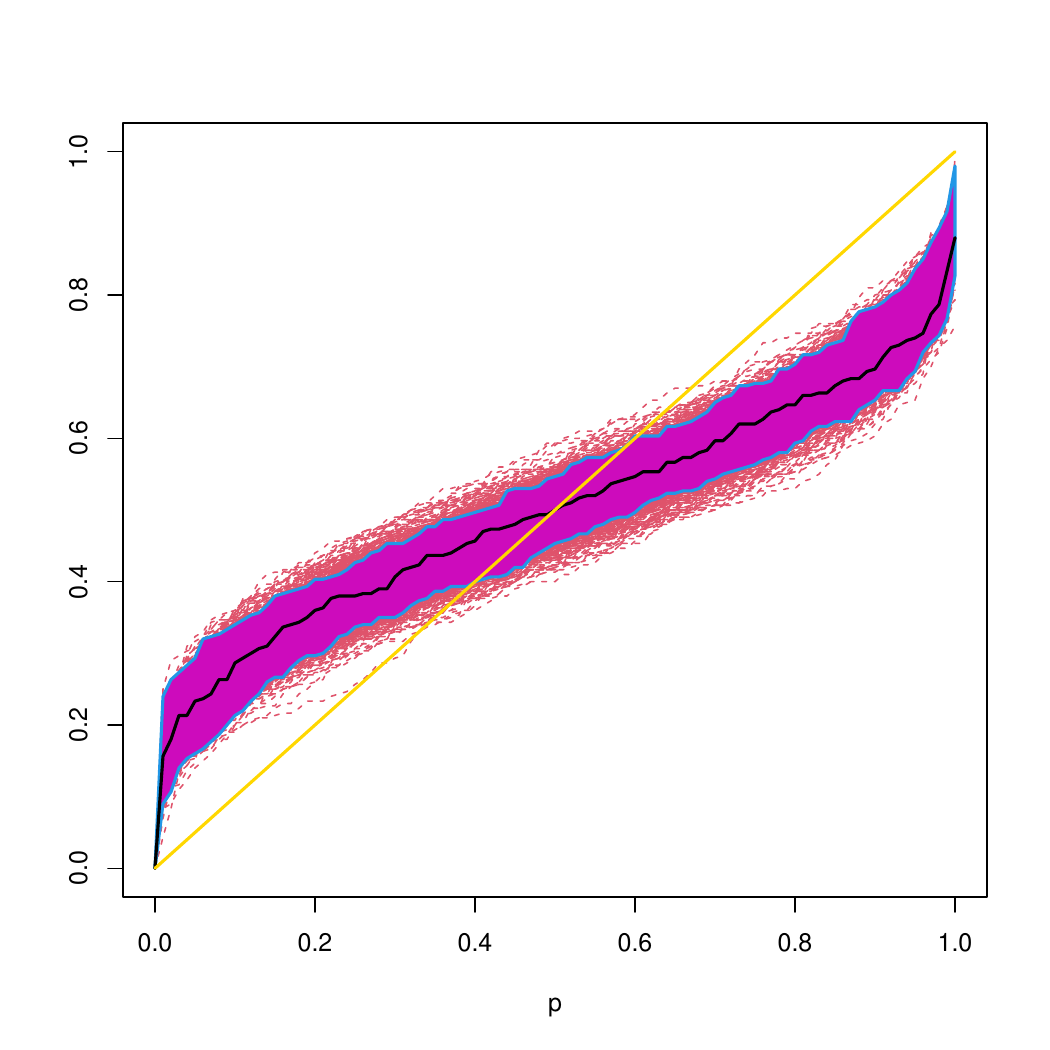}}
 \\[-4ex]
    
$\wUps_{\media}$ &
\raisebox{-.5\height}{\includegraphics[scale=0.25]{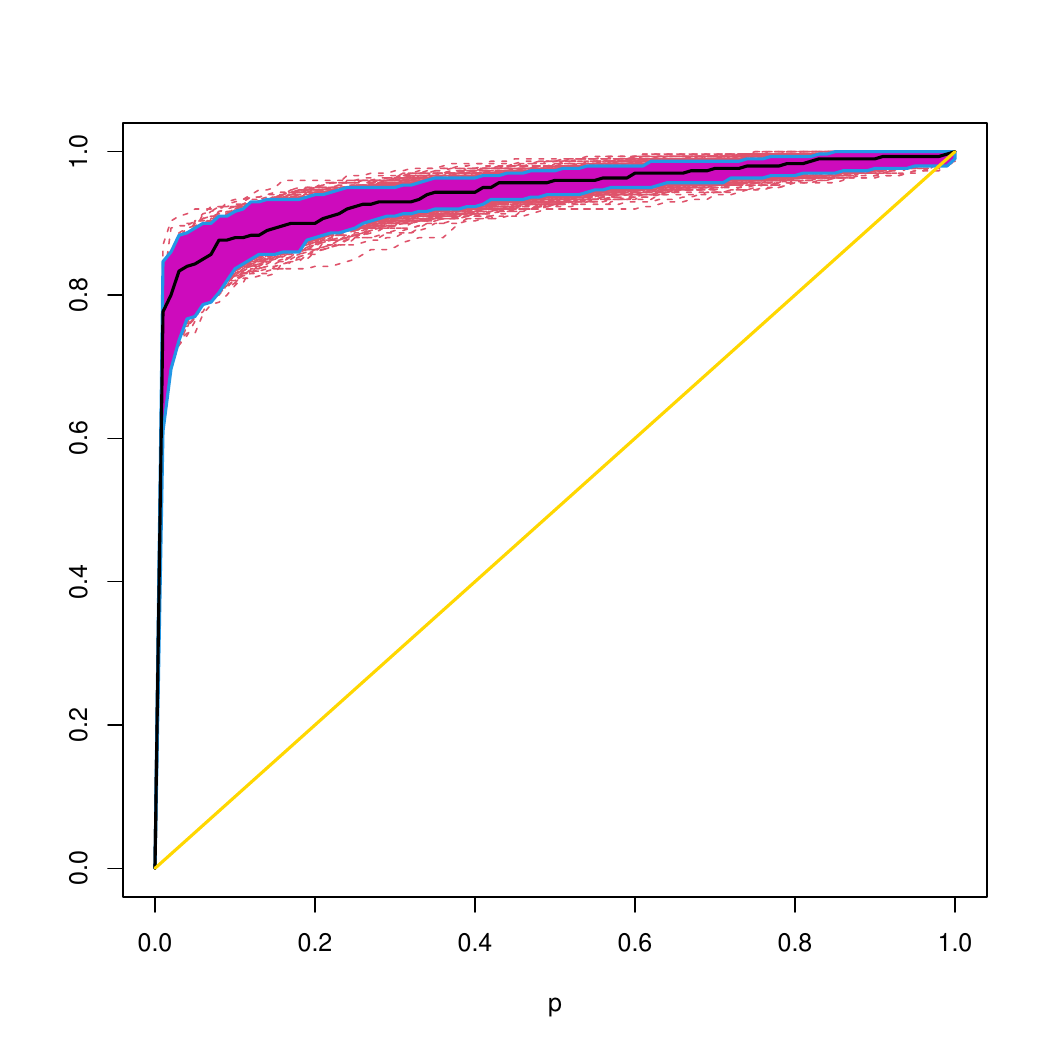}}
& \raisebox{-.5\height}{\includegraphics[scale=0.25]{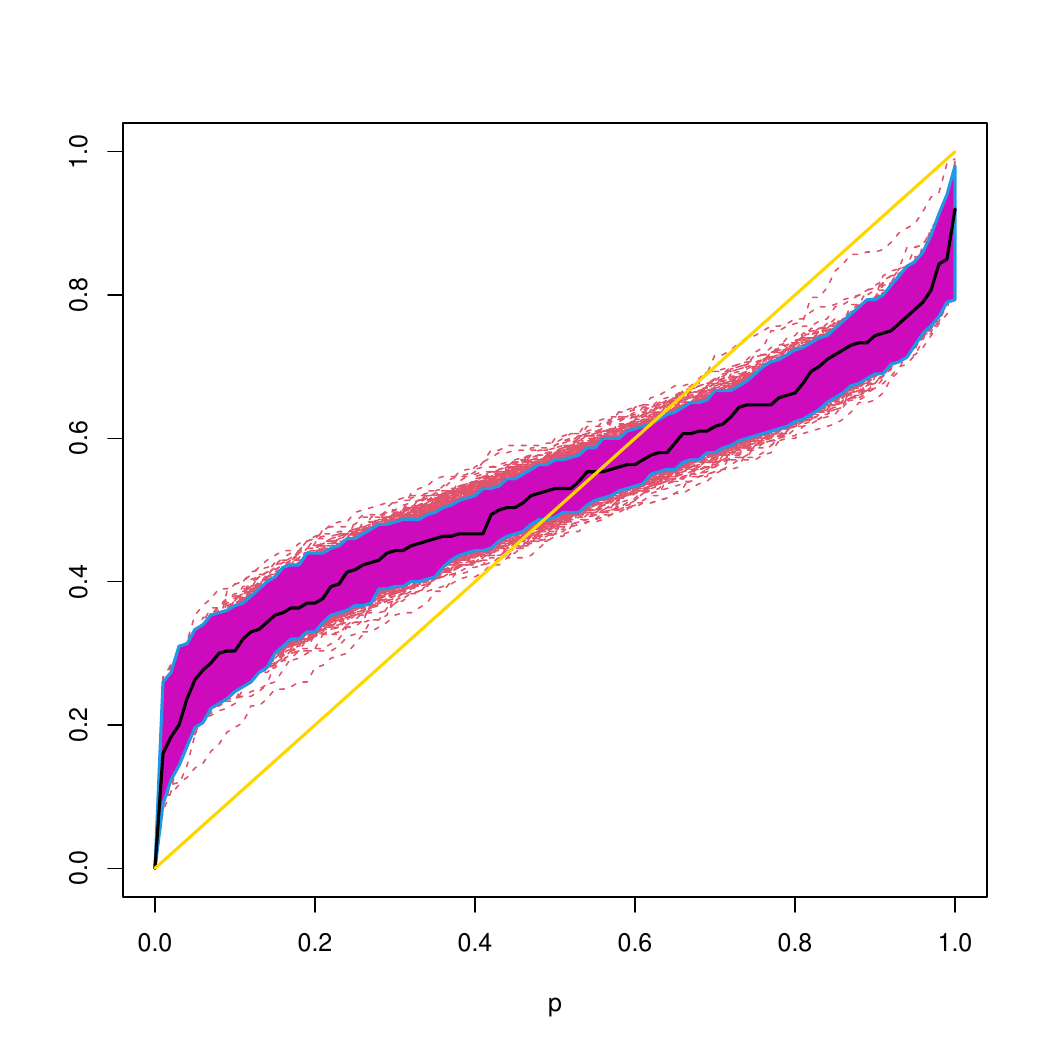}}
\\[-4ex]

$\wUps_{\lin}$  & 
\raisebox{-.5\height}{\includegraphics[scale=0.25]{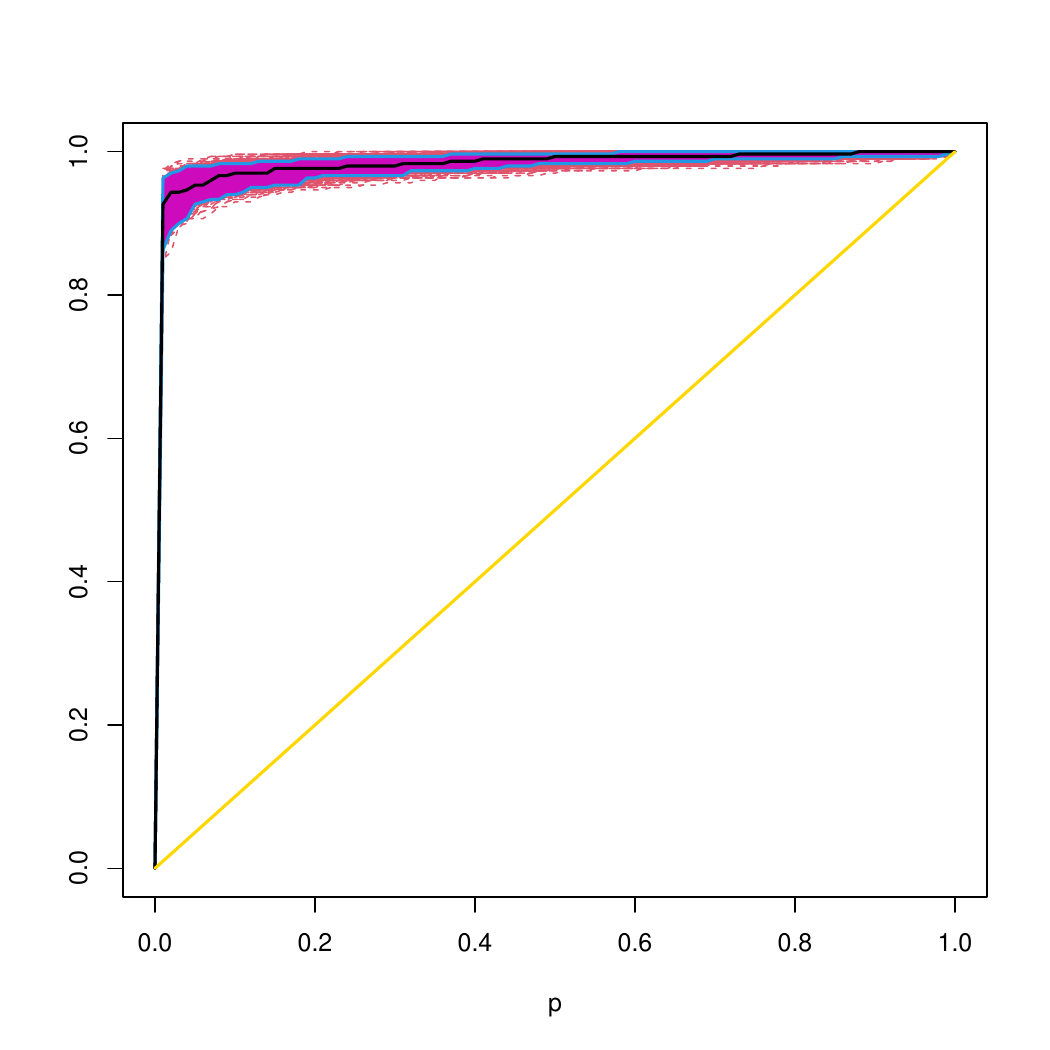}}
& \raisebox{-.5\height}{\includegraphics[scale=0.25]{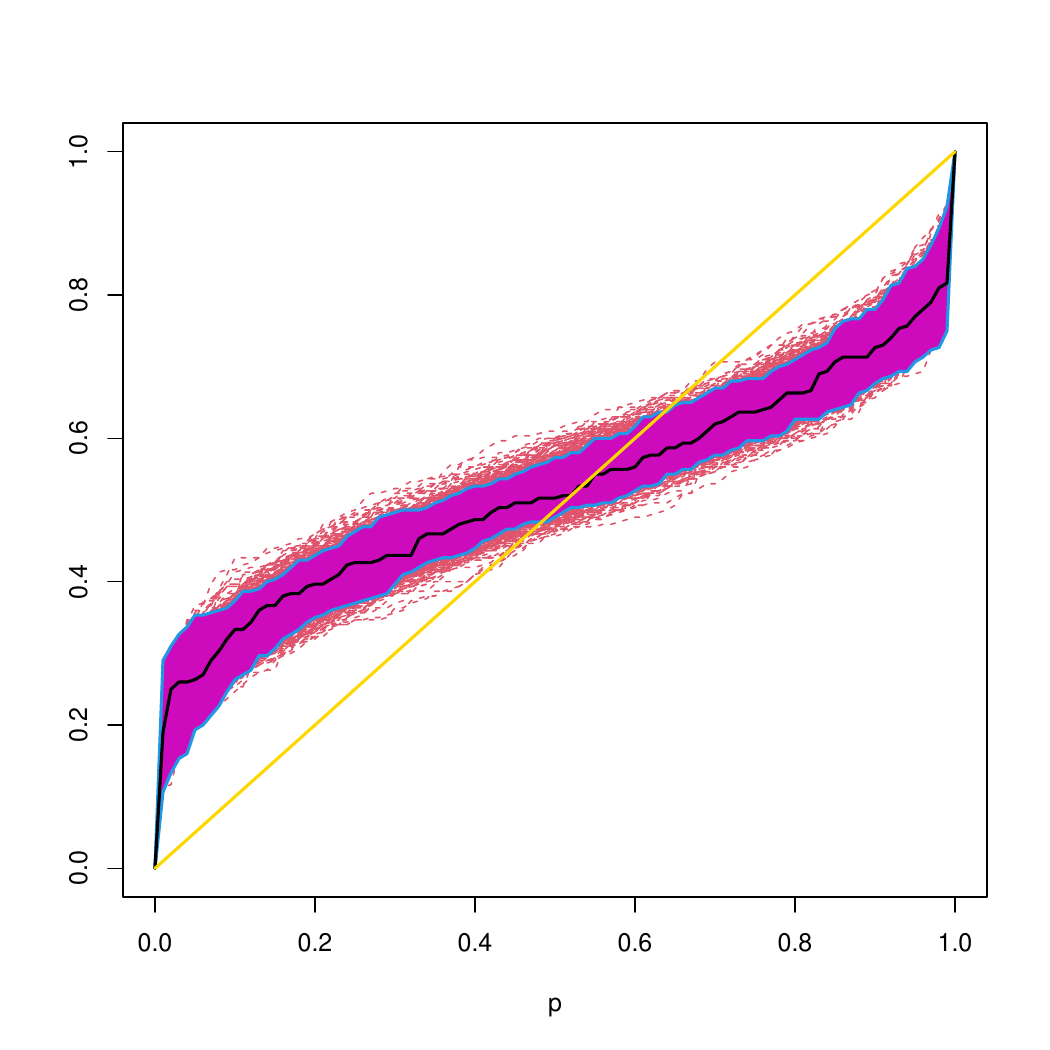}}
\\[-4ex]

$\wUps_{\cuad}$ & 
\raisebox{-.5\height}{\includegraphics[scale=0.25]{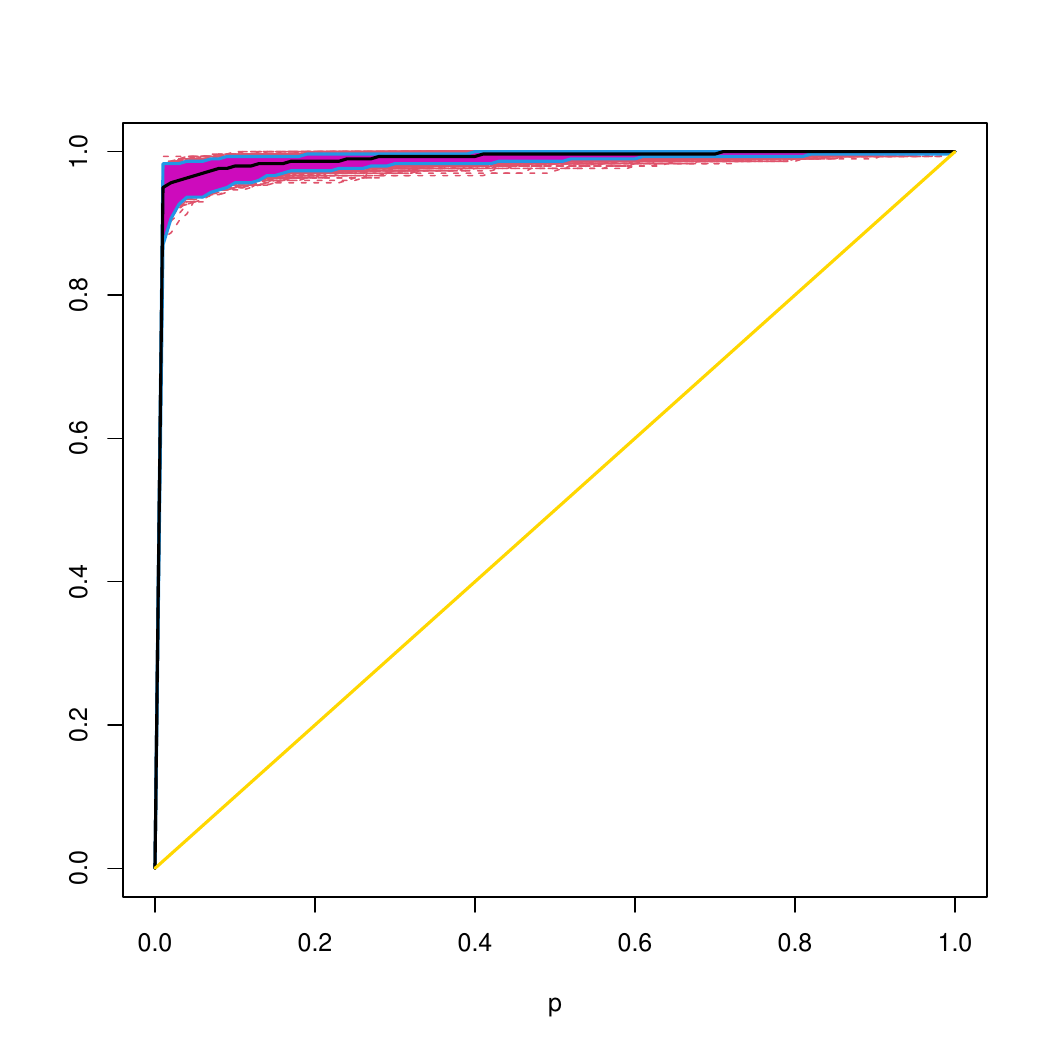}}
& \raisebox{-.5\height}{\includegraphics[scale=0.25]{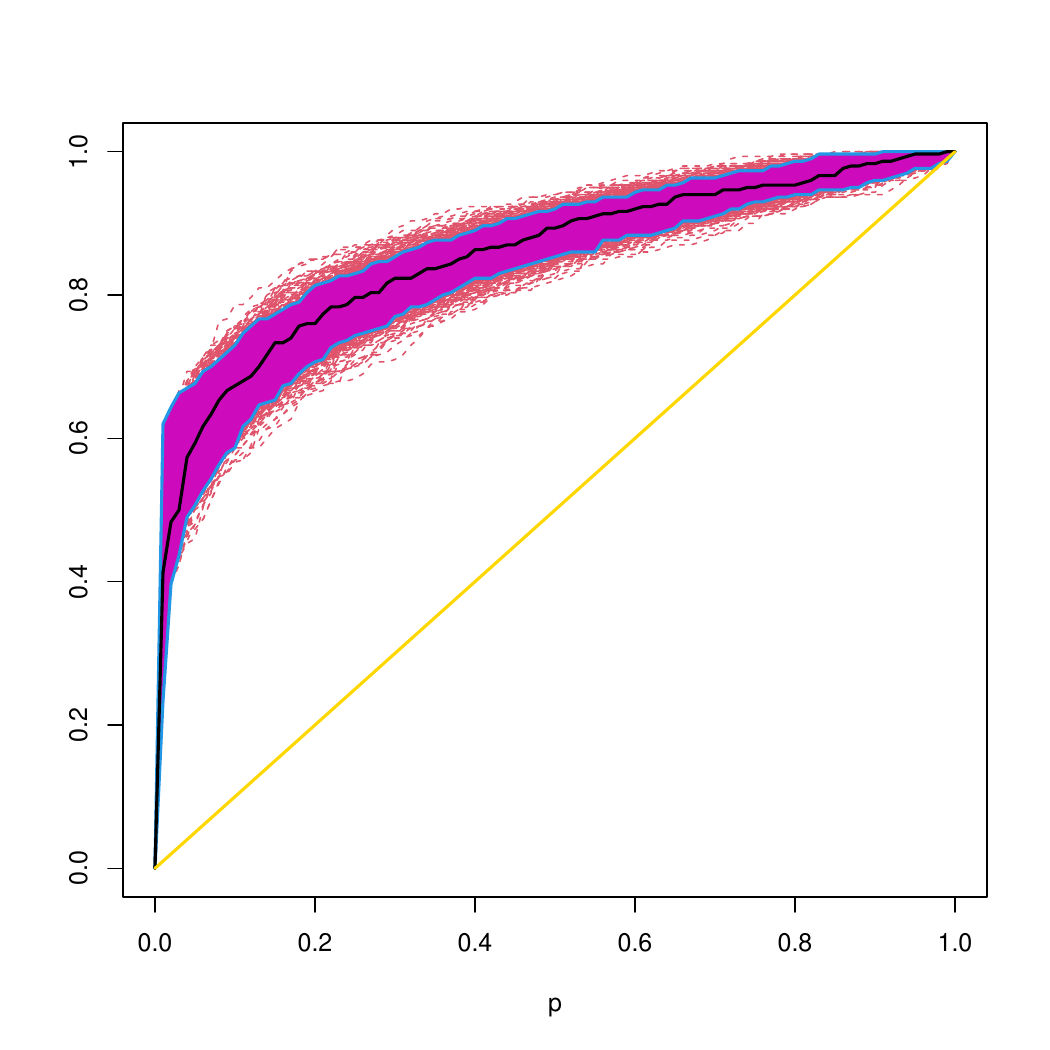}}

\end{tabular}
\caption{Functional boxplots of the estimators $\widehat{\ROC}$ under a \textsc{fcpc} model with  $\lambda_{D,1} = 2$, $\lambda_{D,2}=0.3$ and $\lambda_{D,3}=0.05$. Rows correspond to discriminating indexes, while columns to  $\mu_D(t)=2\, \sin(\pi  t)$ and $\mu_H=0$.}
\label{fig:cpc:C1} 
\end{center} 
\end{figure}

\begin{figure}[ht!]
 \begin{center}
 \footnotesize
 \renewcommand{\arraystretch}{0.2}
\begin{tabular}{p{2cm} cc}
 & \textbf{C21} &  \textbf{C20} \\[-2ex]  
$\Upsilon_{\maxi}$ &
 \raisebox{-.5\height}{\includegraphics[scale=0.25]{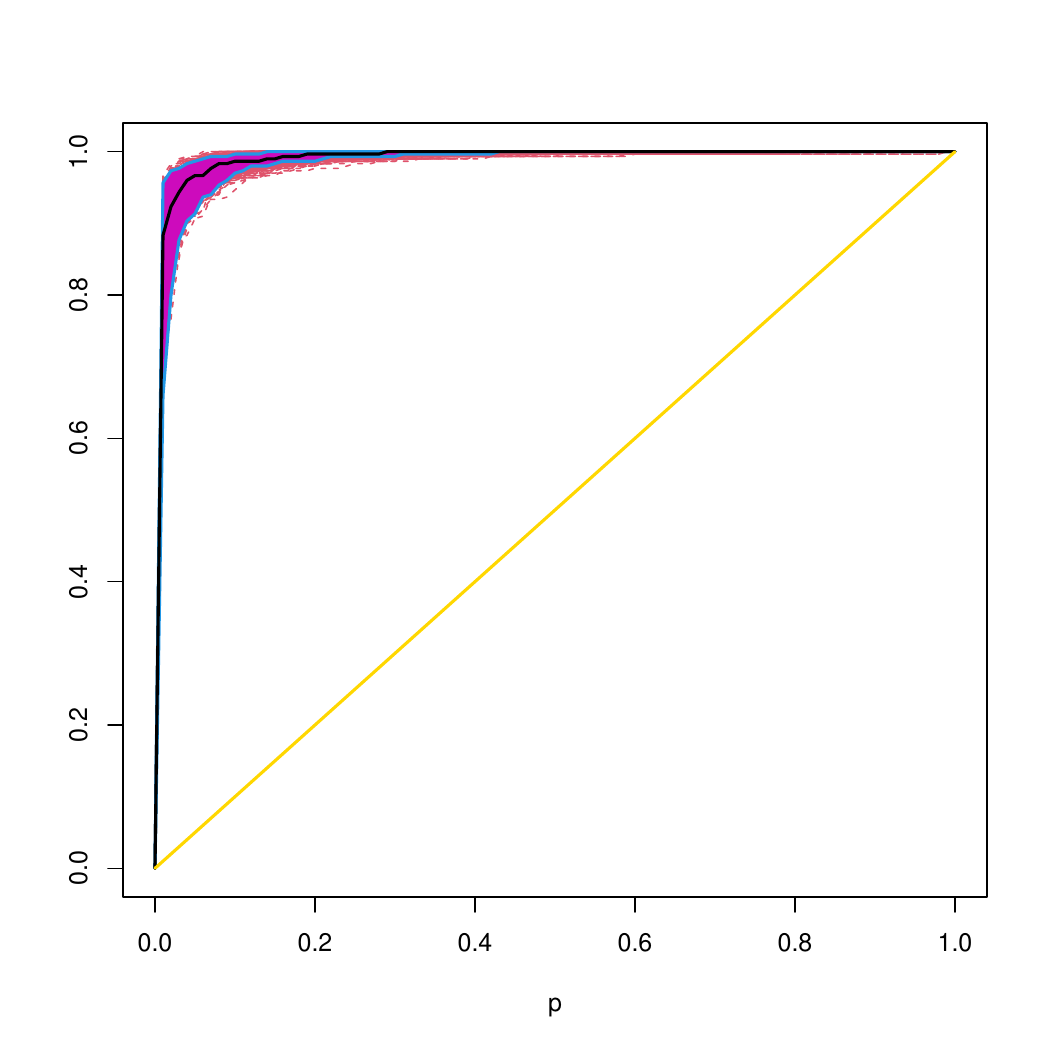}}
& \raisebox{-.5\height}{\includegraphics[scale=0.25]{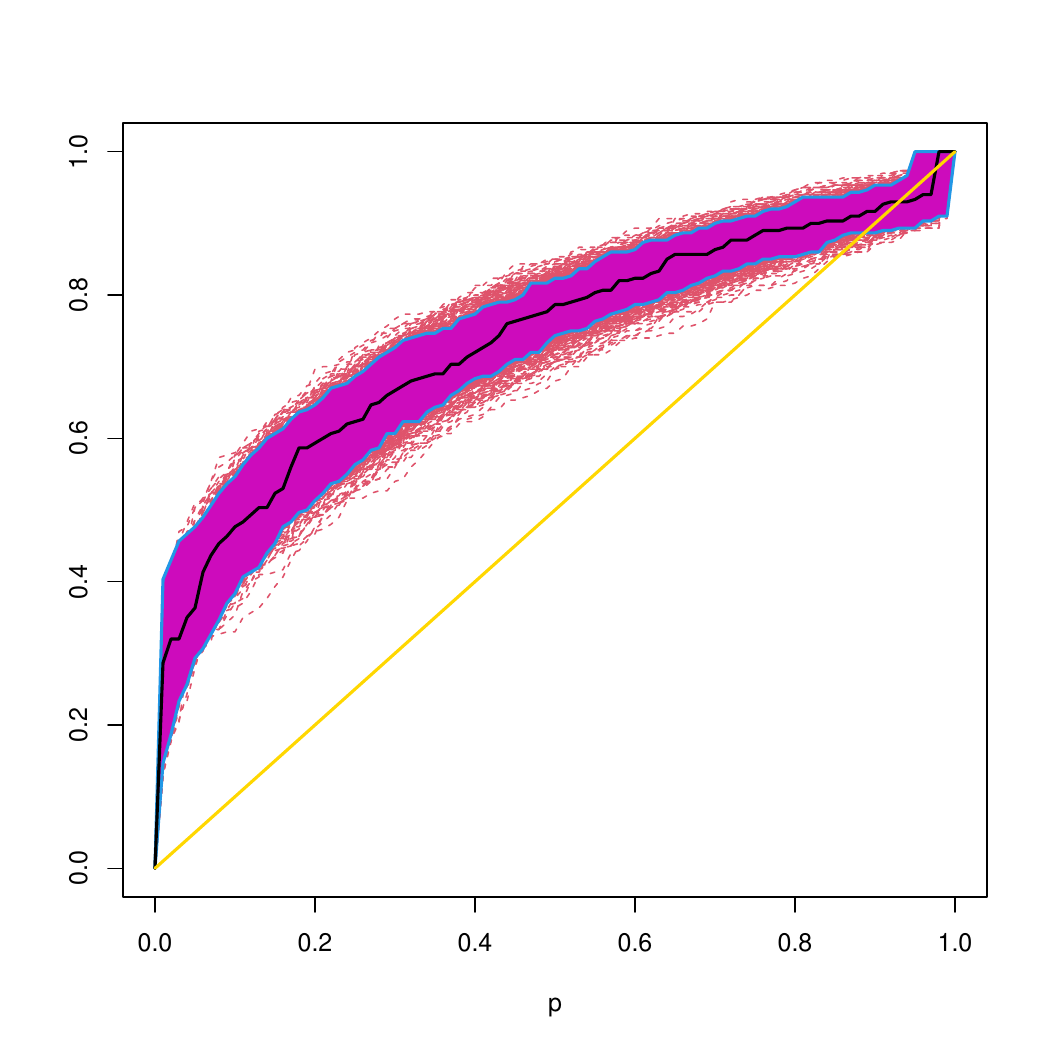}}
\\[-4ex]
    
$\Upsilon_{\inte}$  &
 \raisebox{-.5\height}{\includegraphics[scale=0.25]{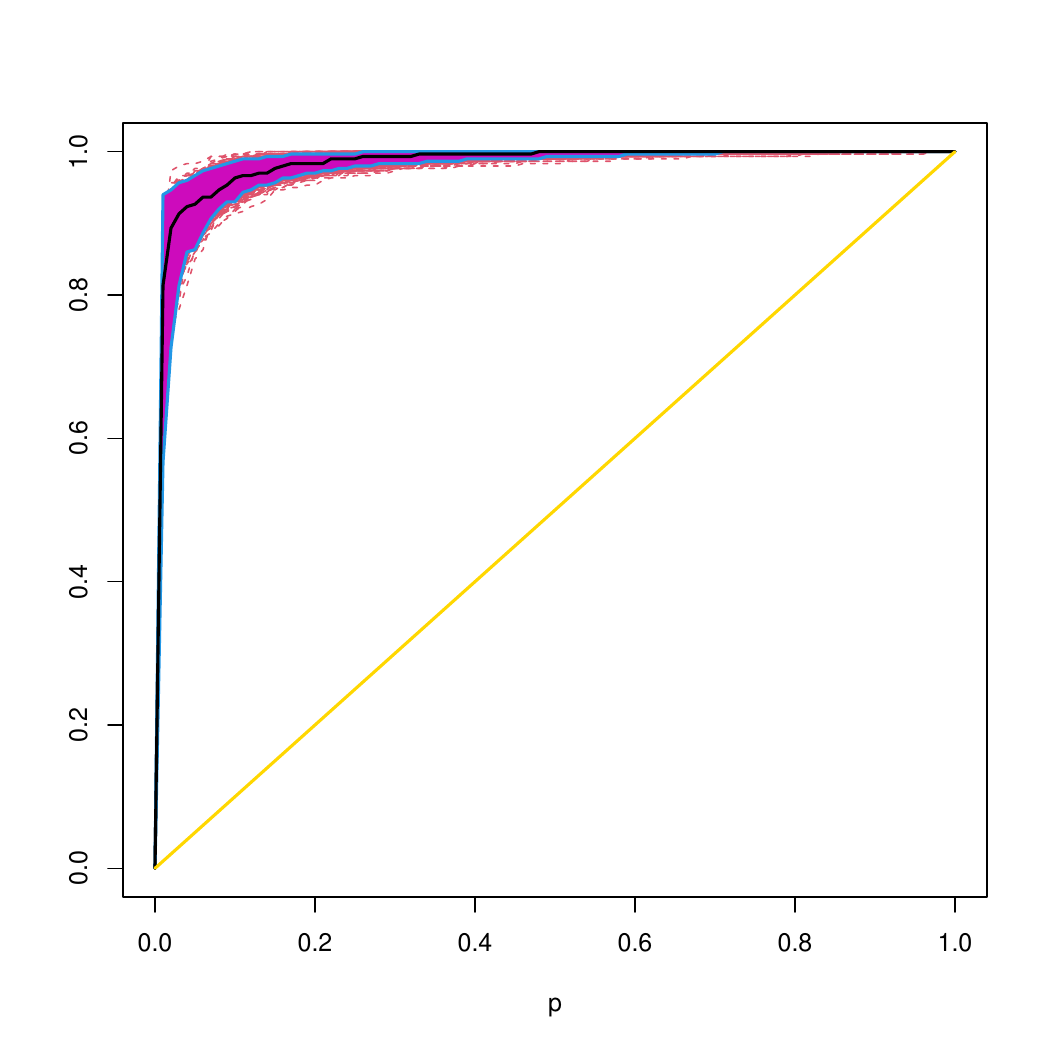}}
& \raisebox{-.5\height}{\includegraphics[scale=0.25]{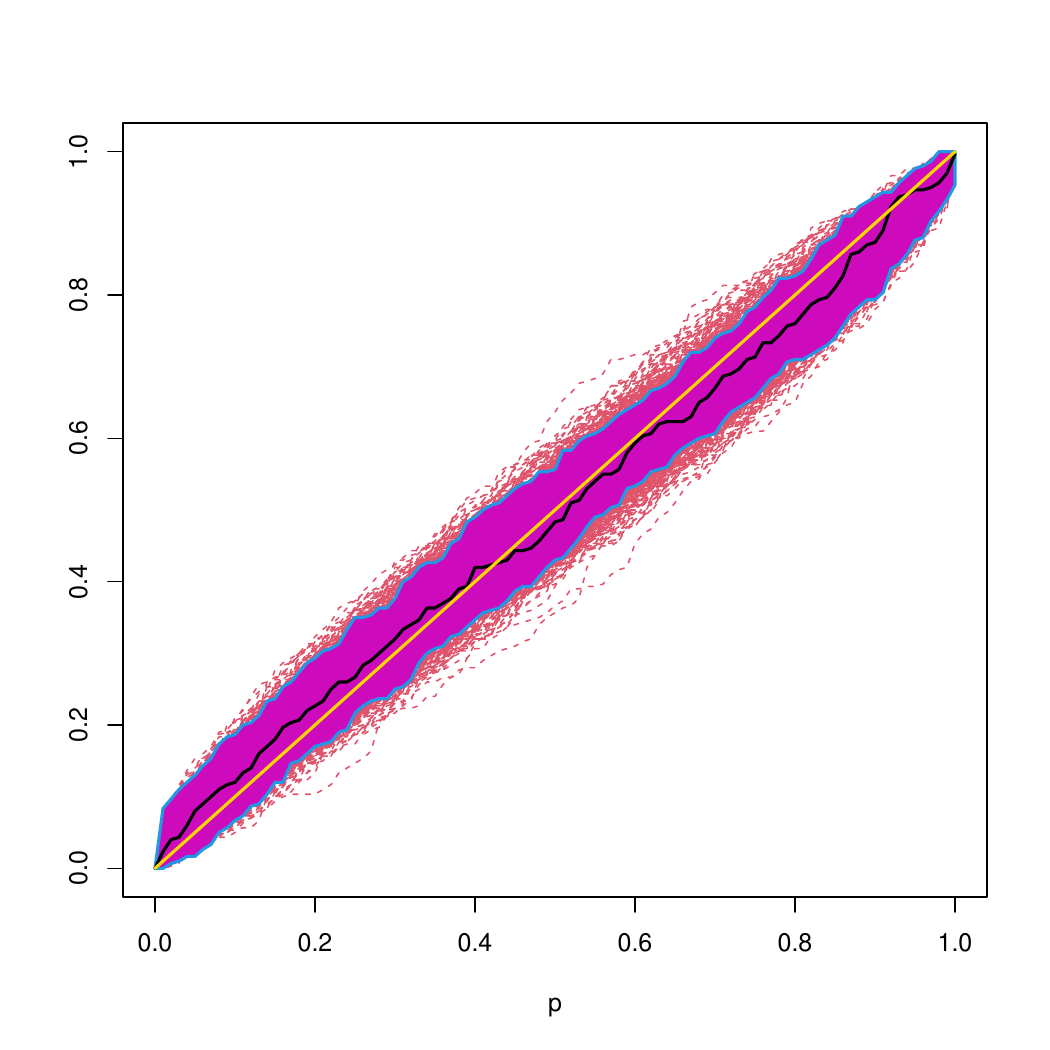}}
 \\[-4ex]
    
$\wUps_{\media}$ &
\raisebox{-.5\height}{\includegraphics[scale=0.25]{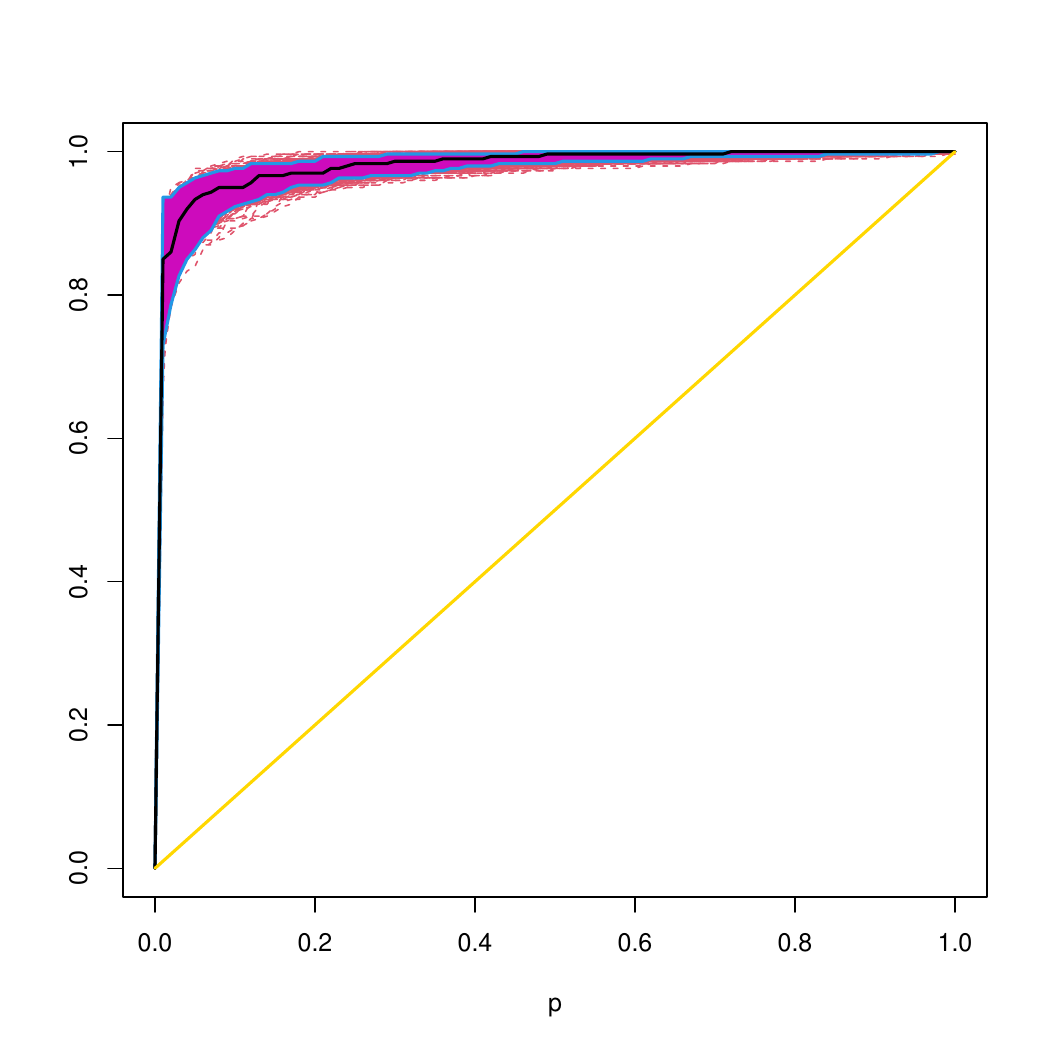}}
& \raisebox{-.5\height}{\includegraphics[scale=0.25]{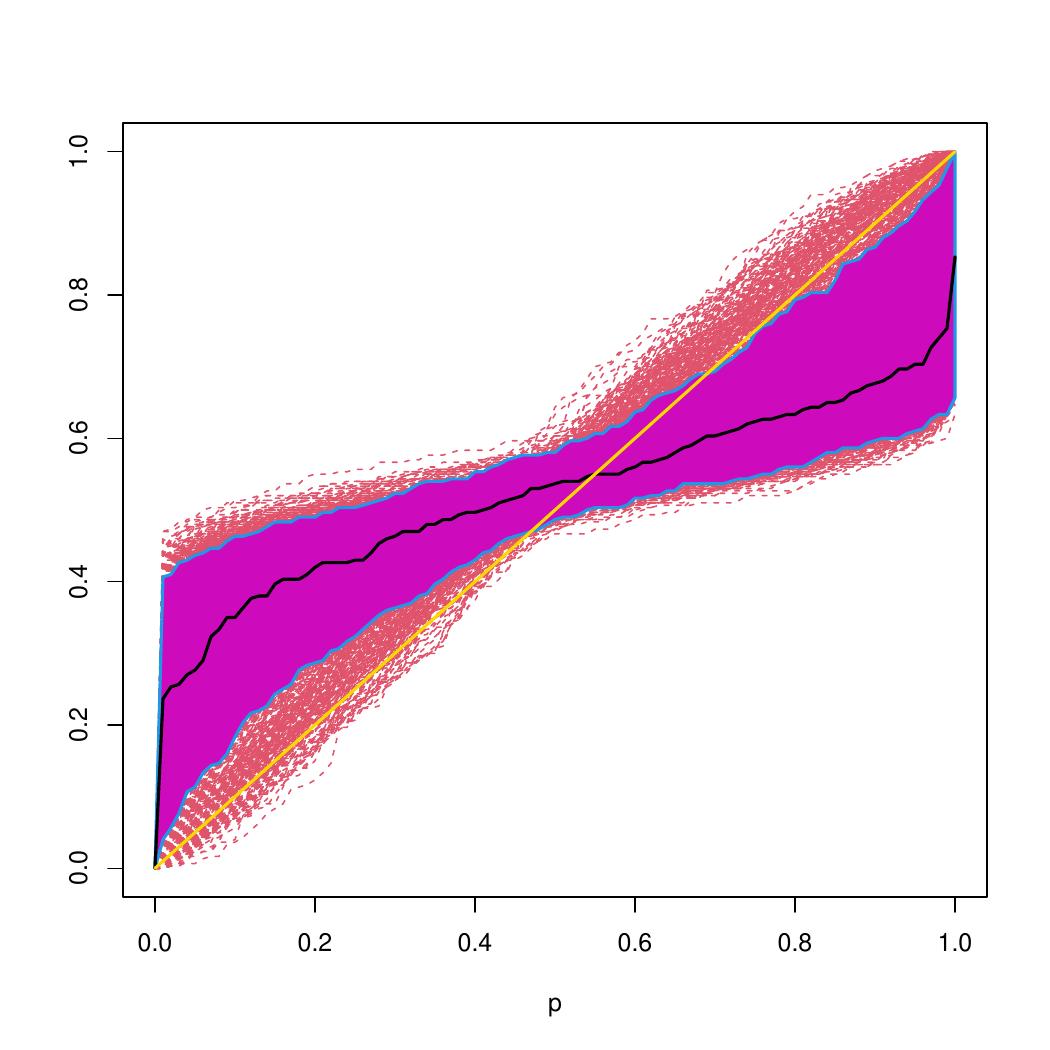}}
\\[-4ex]

$\wUps_{\lin}$  & 
\raisebox{-.5\height}{\includegraphics[scale=0.25]{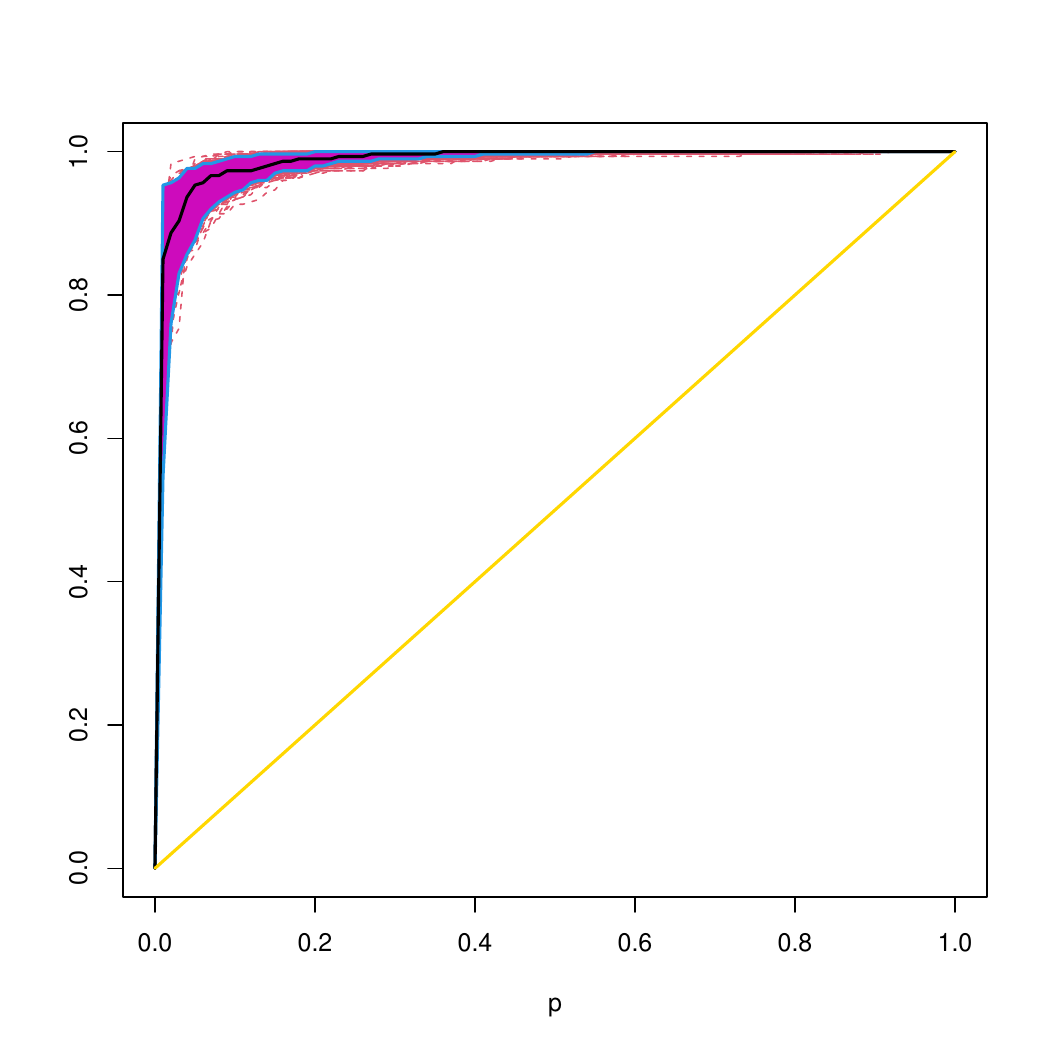}}
& \raisebox{-.5\height}{\includegraphics[scale=0.25]{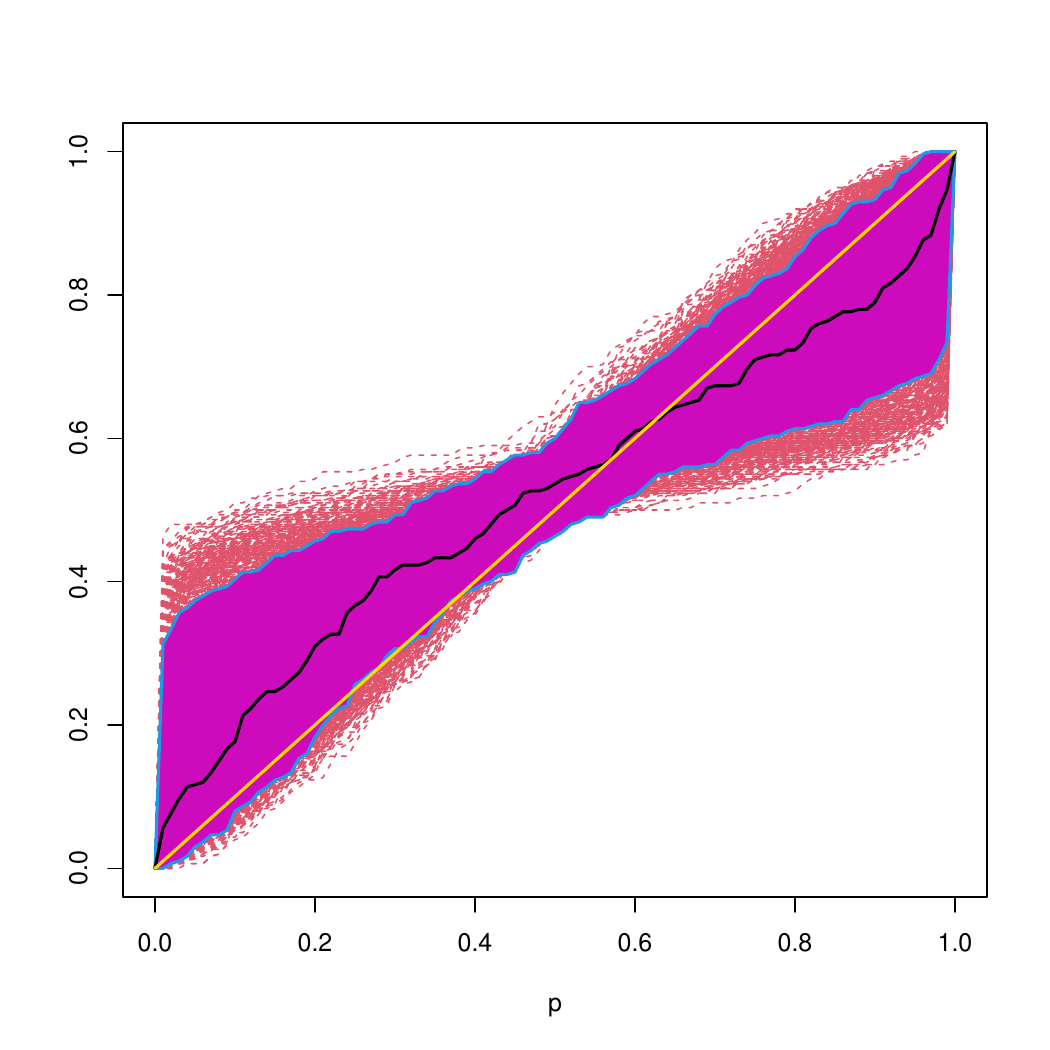}}
\\[-4ex]

$\wUps_{\cuad}$ 
& \raisebox{-.5\height}{\includegraphics[scale=0.25]{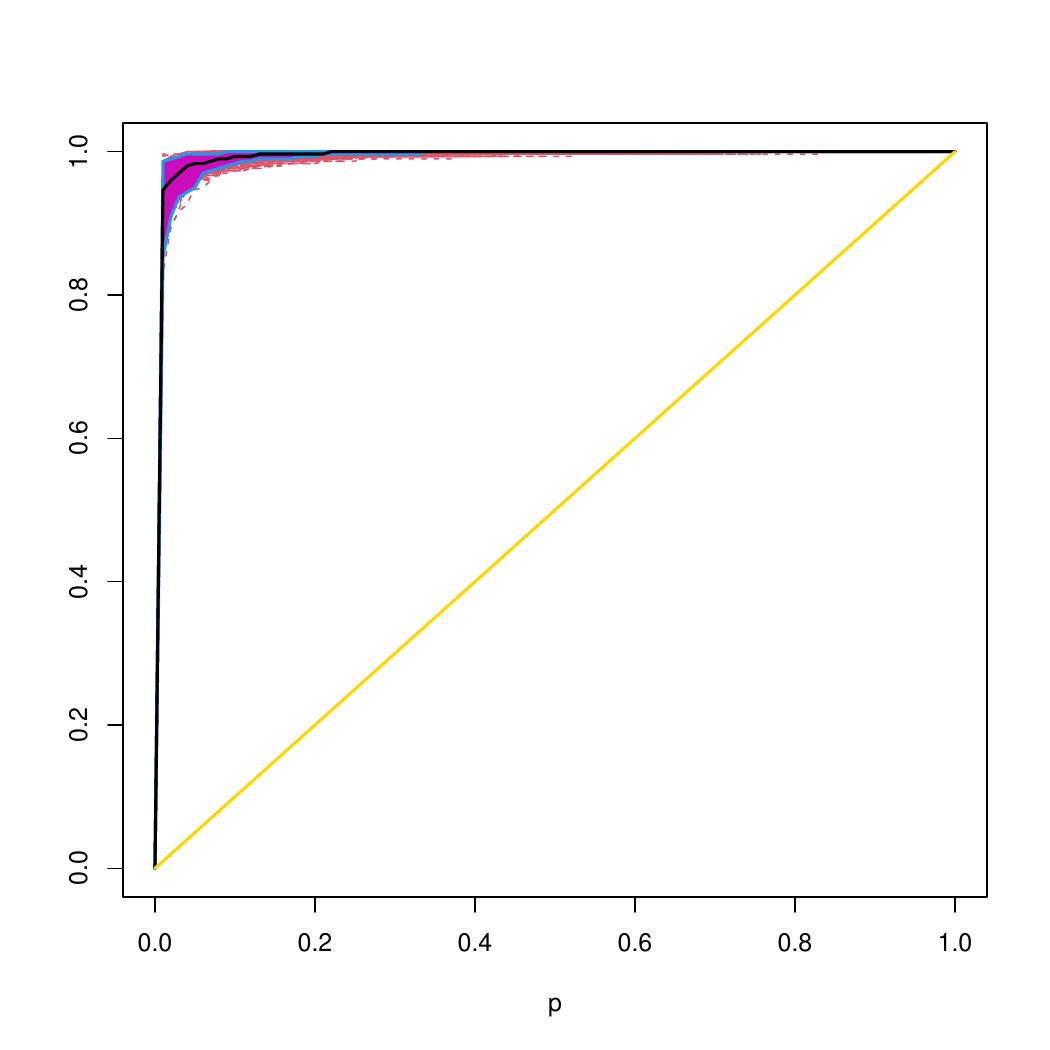}}
& \raisebox{-.5\height}{\includegraphics[scale=0.25]{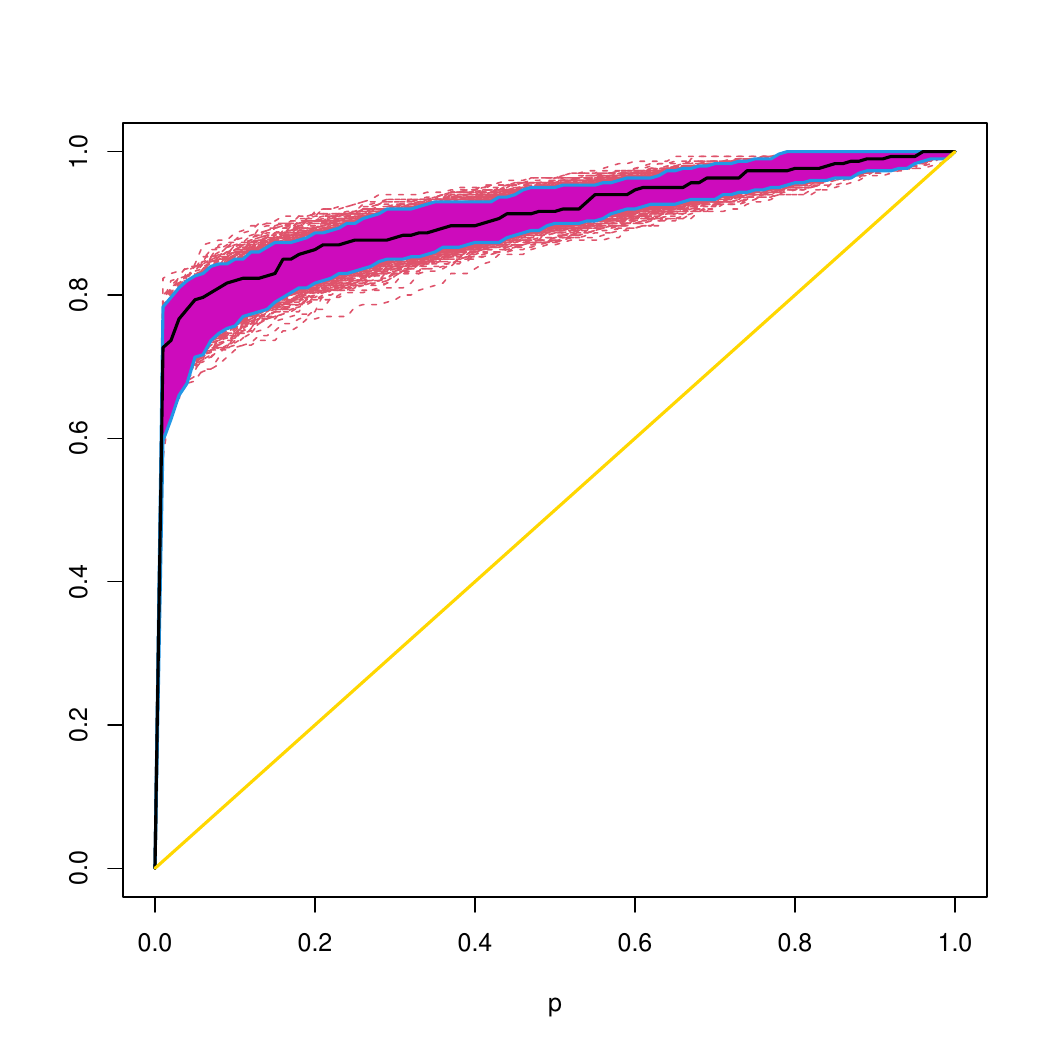}}

\end{tabular}
\caption{Functional boxplots of the estimators $\widehat{\ROC}$ under a \textsc{fcpc} model with  $\lambda_{D,1} = 0.3$, $\lambda_{D,2}=2$ and $\lambda_{D,3}=0.05$.  Rows correspond to discriminating indexes, while columns to $\mu_D(t)=2\, \sin(\pi  t)$ and $\mu_H=0$.}
\label{fig:cpc:C2} 
\end{center} 
\end{figure}

\clearpage

\begin{figure}[ht!]
 \begin{center}
 \footnotesize
 \renewcommand{\arraystretch}{0.2}

\begin{tabular}{p{2cm} cc}
 & \textbf{D11} &  \textbf{D10} \\[-2ex]
$\Upsilon_{\maxi}$ &
\raisebox{-.5\height}{ \includegraphics[scale=0.25]{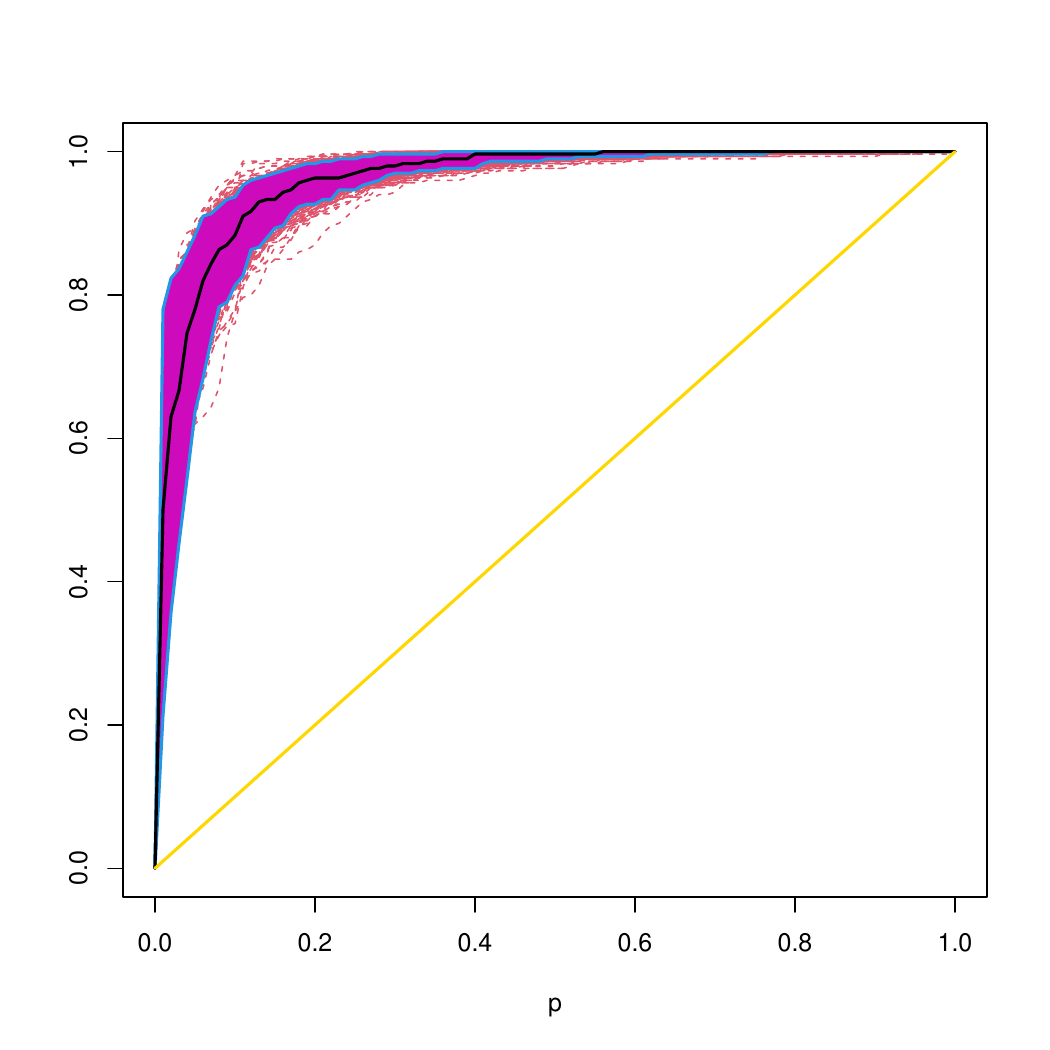}}
& \raisebox{-.5\height}{\includegraphics[scale=0.25]{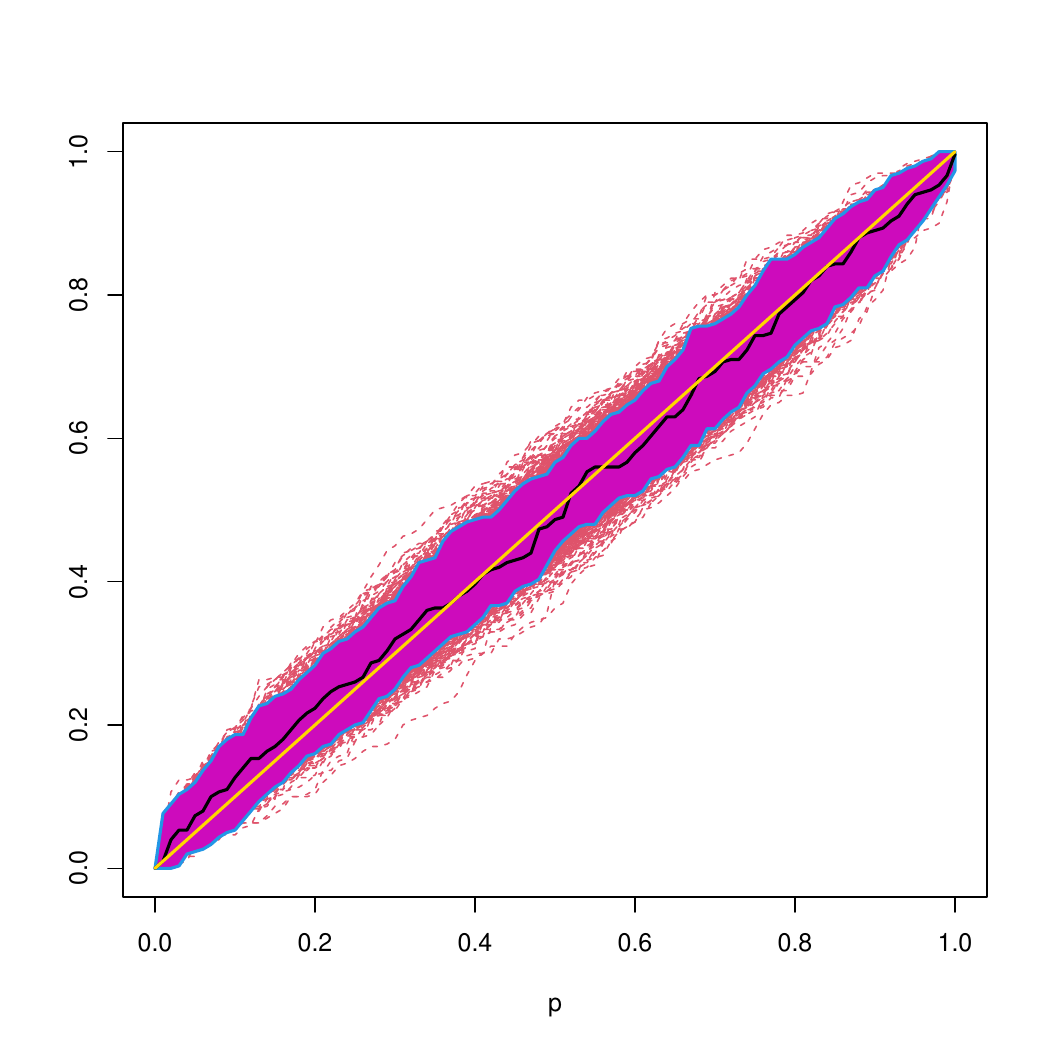}}
\\[-4ex]
    
$\Upsilon_{\inte}$  &
 \raisebox{-.5\height}{\includegraphics[scale=0.25]{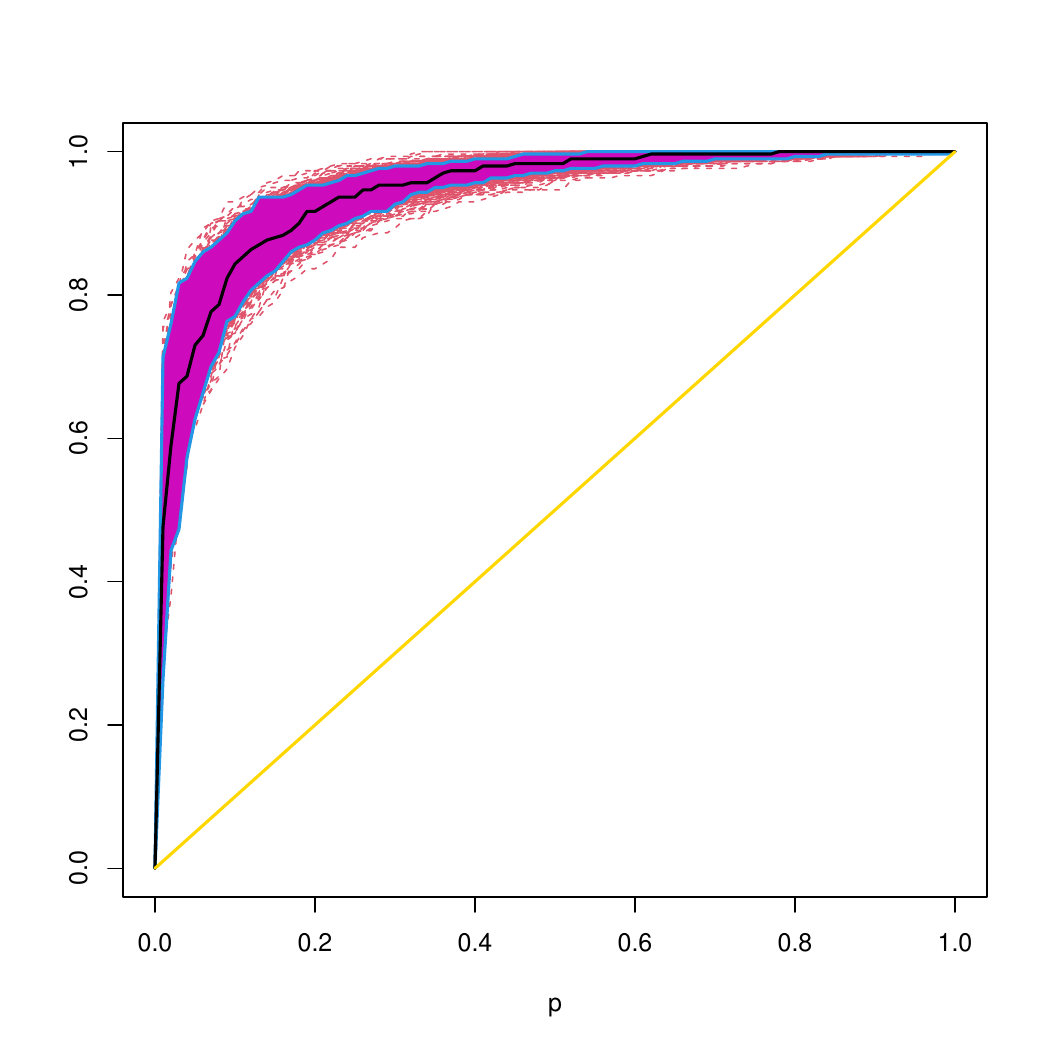}}
& \raisebox{-.5\height}{\includegraphics[scale=0.25]{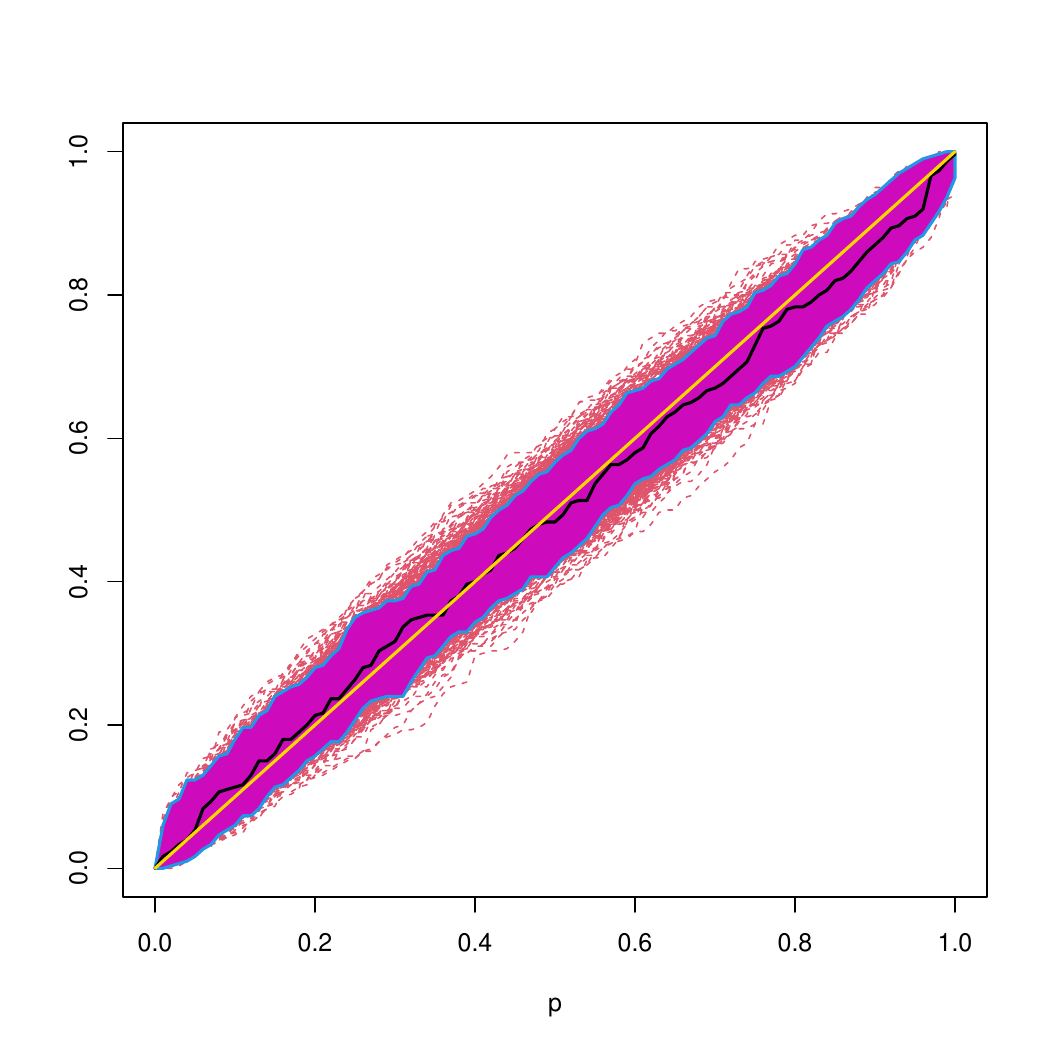}}
 \\[-4ex]
    
$\wUps_{\media}$ &
\raisebox{-.5\height}{\includegraphics[scale=0.25]{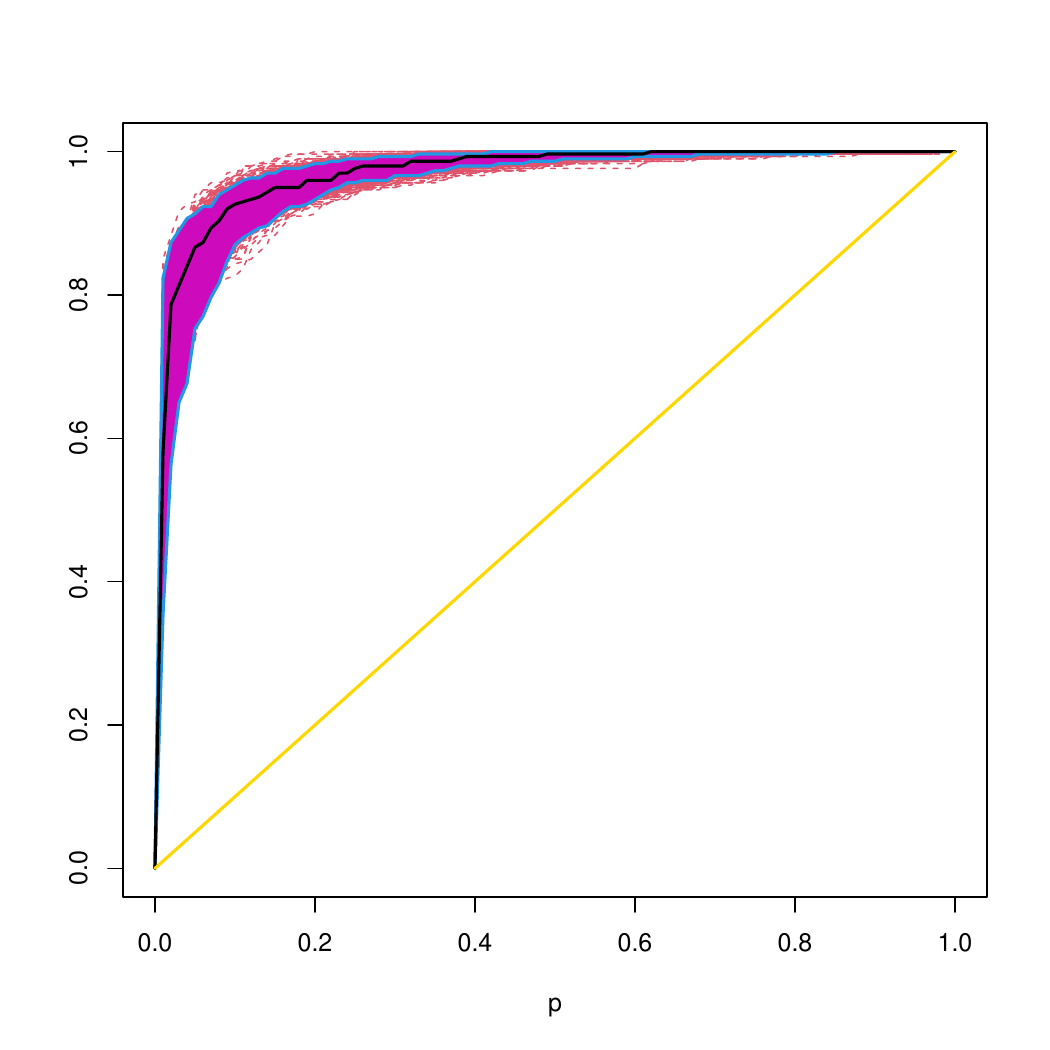}}
& \raisebox{-.5\height}{\includegraphics[scale=0.25]{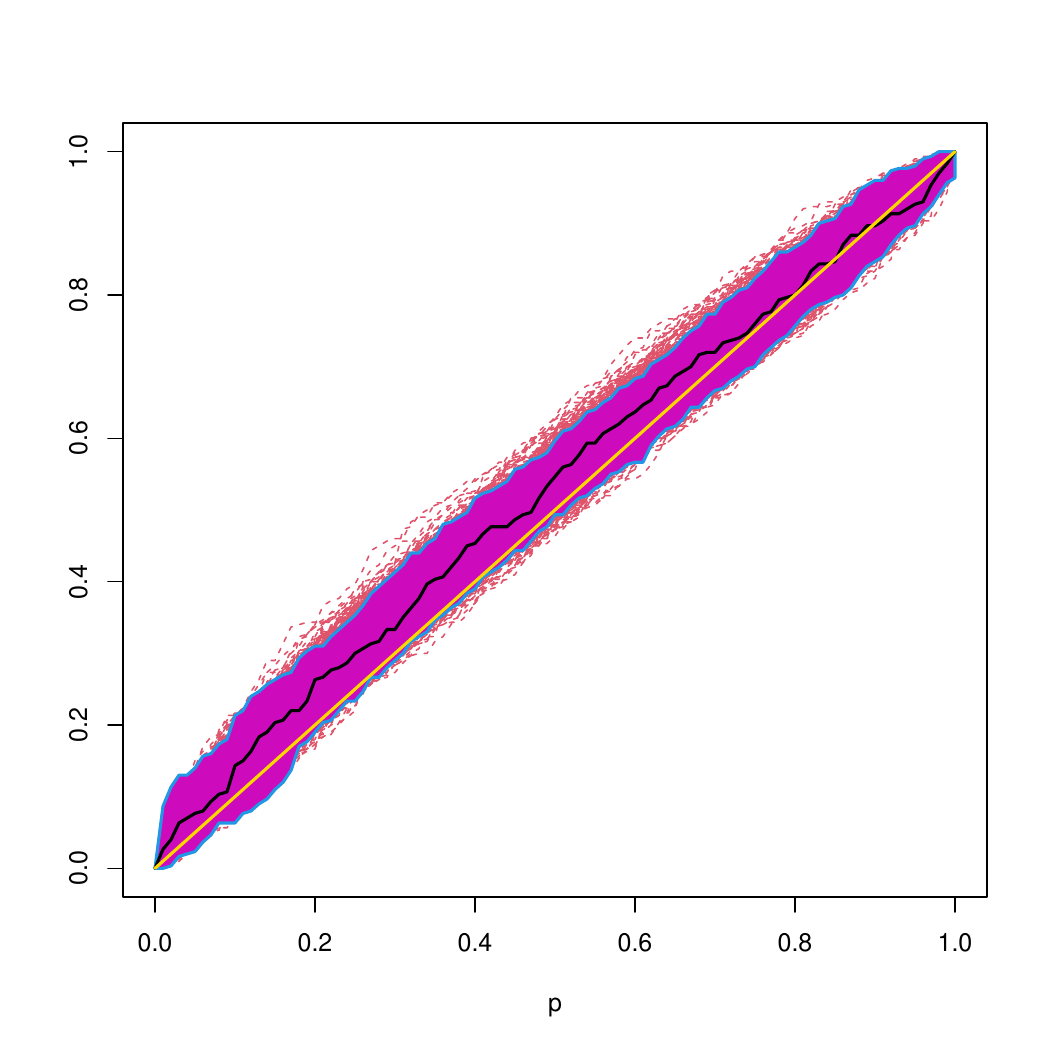}}
\\[-4ex]

$\wUps_{\lin}$  & 
\raisebox{-.5\height}{\includegraphics[scale=0.25]{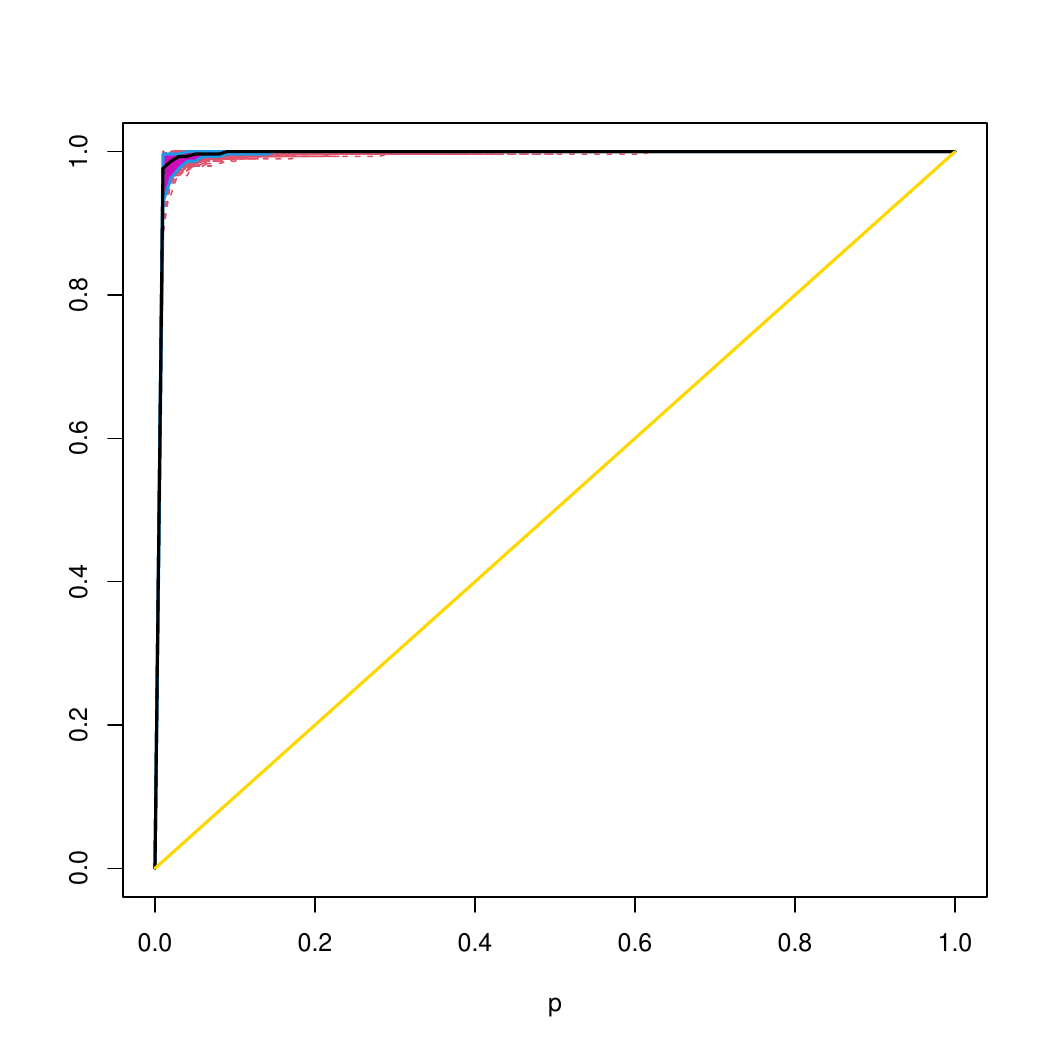}}
& \raisebox{-.5\height}{\includegraphics[scale=0.25]{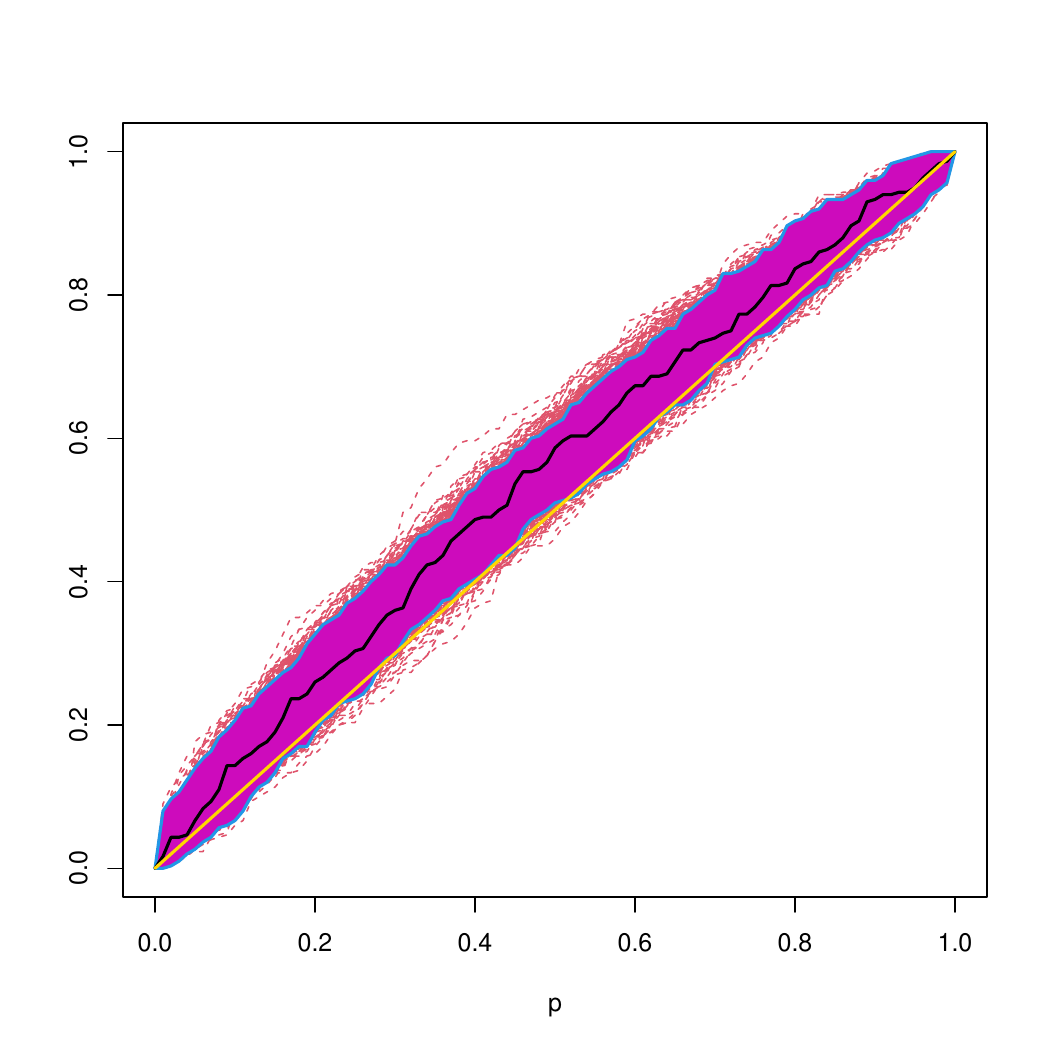}}
\\[-4ex]

$\wUps_{\cuad}$  
& \raisebox{-.5\height}{\includegraphics[scale=0.25]{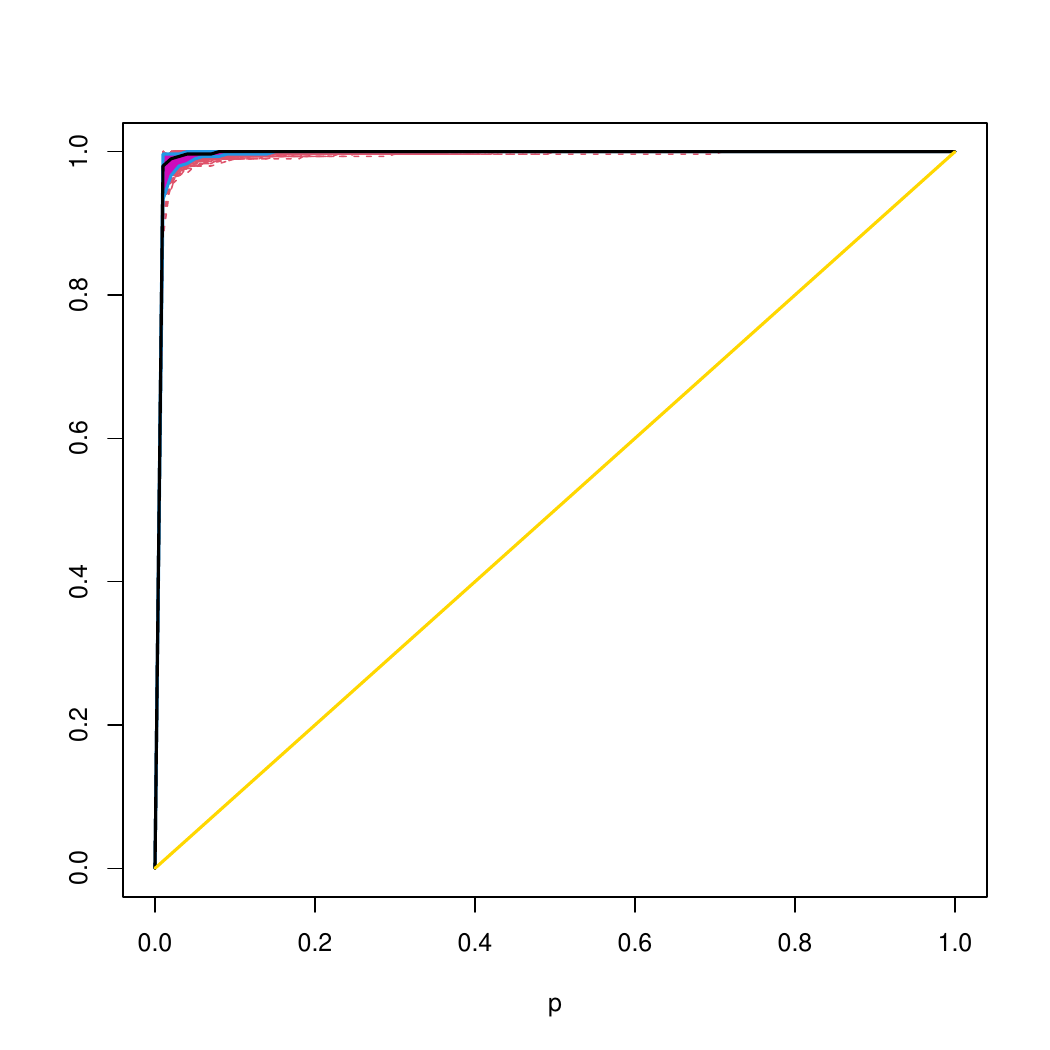}}
& \raisebox{-.5\height}{\includegraphics[scale=0.25]{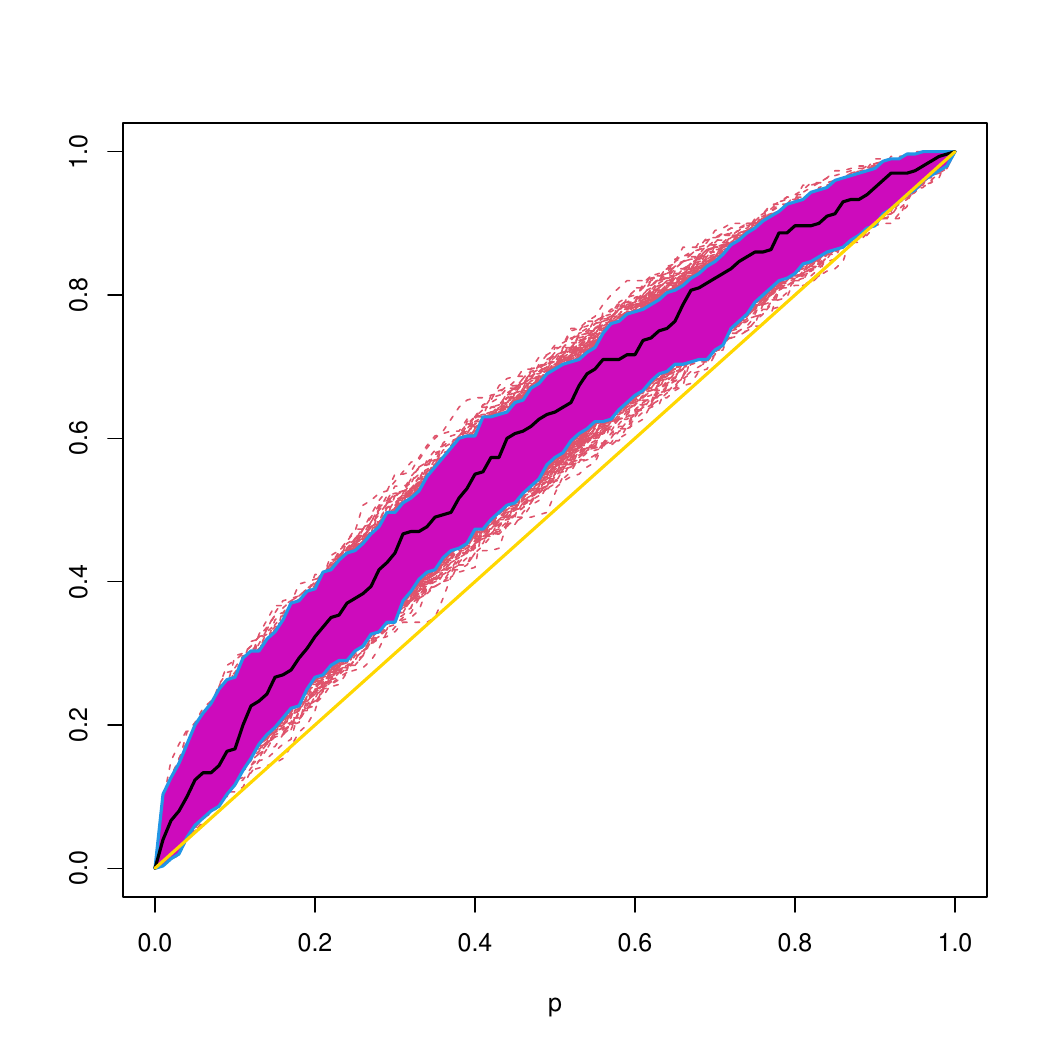}}

\end{tabular}
\caption{Functional boxplots of the estimators $\widehat{\ROC}$ under scheme \textbf{D1}. Rows correspond to discriminating indexes, while columns to  $\mu_D(t)=2\, \sin(\pi  t)$ and $\mu_H=0$.}
\label{fig:DIFF-D1} 
\end{center} 
\end{figure}

\begin{figure}[ht!]
 \begin{center}
 \footnotesize
 \renewcommand{\arraystretch}{0.2}

\begin{tabular}{p{2cm} cc}
 & \textbf{D21} &  \textbf{D20} \\[-2ex]  
$\Upsilon_{\maxi}$  &
 \raisebox{-.5\height}{\includegraphics[scale=0.25]{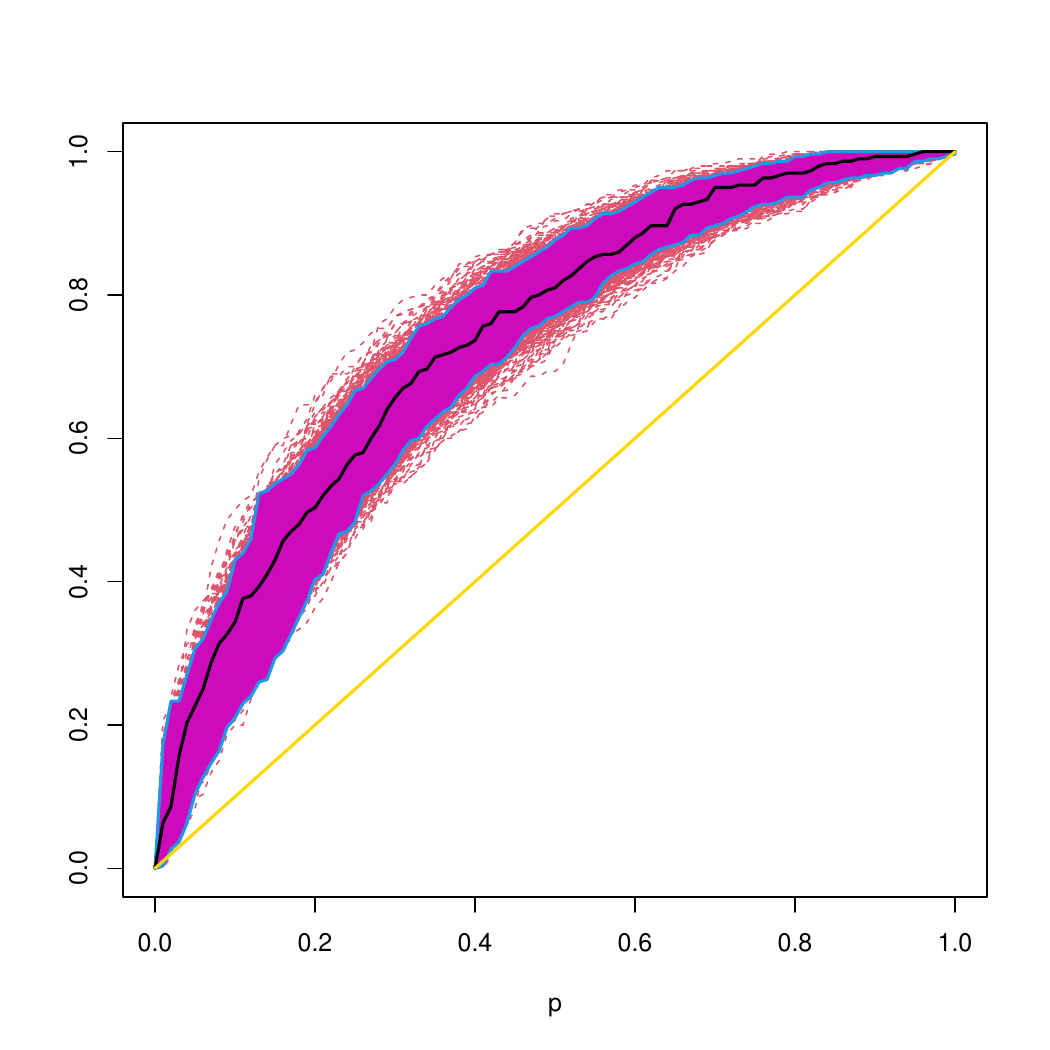}}
& \raisebox{-.5\height}{\includegraphics[scale=0.25]{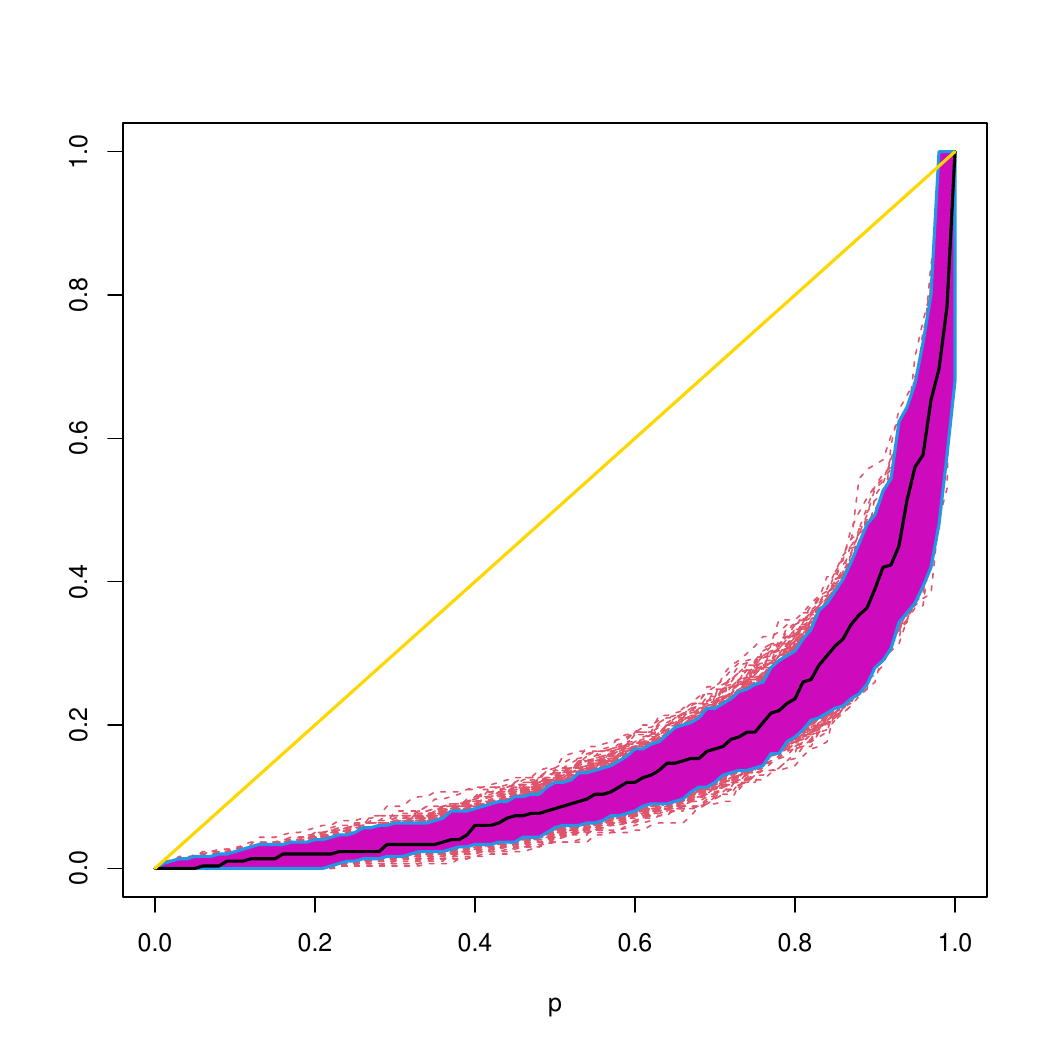}}
\\[-4ex]
 
$\Upsilon_{\inte}$  &
 \raisebox{-.5\height}{\includegraphics[scale=0.25]{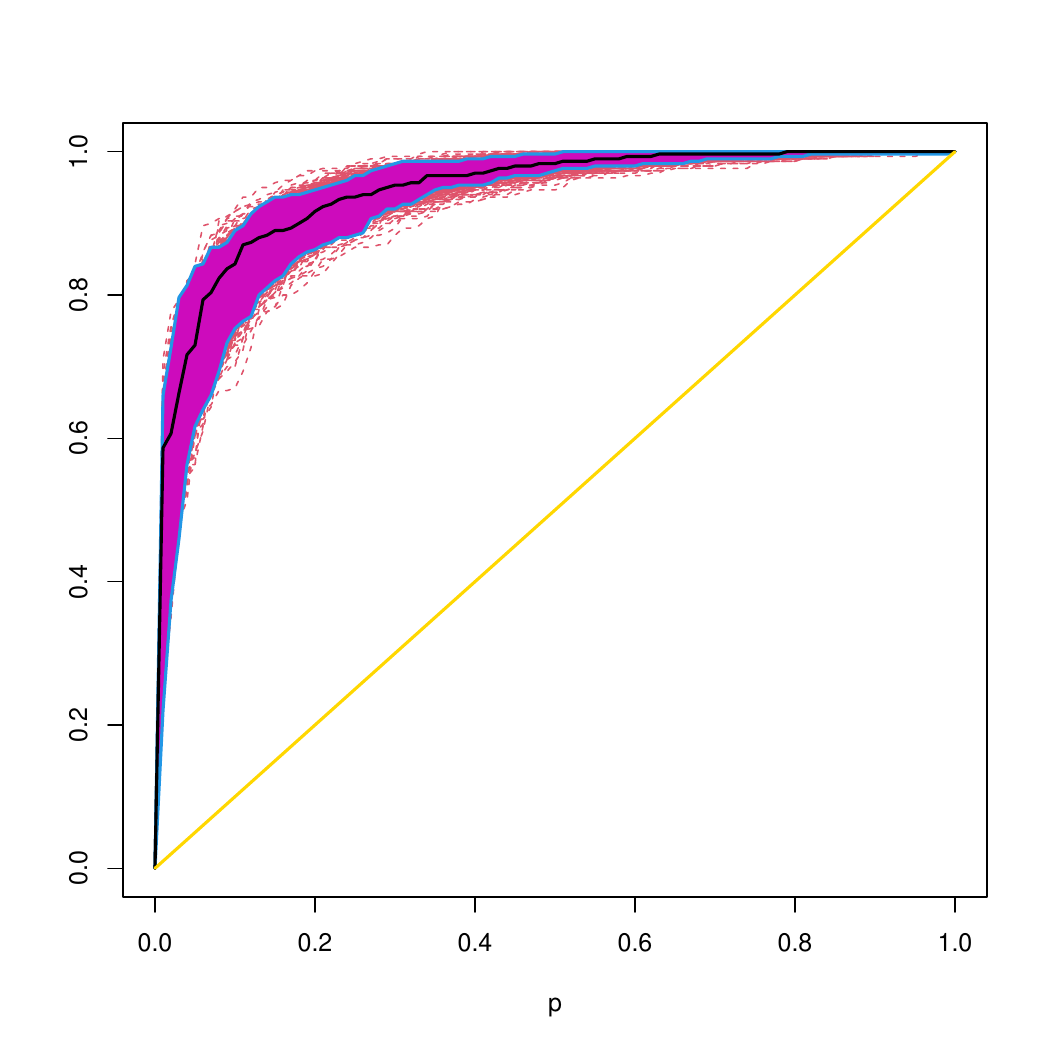}}
& \raisebox{-.5\height}{\includegraphics[scale=0.25]{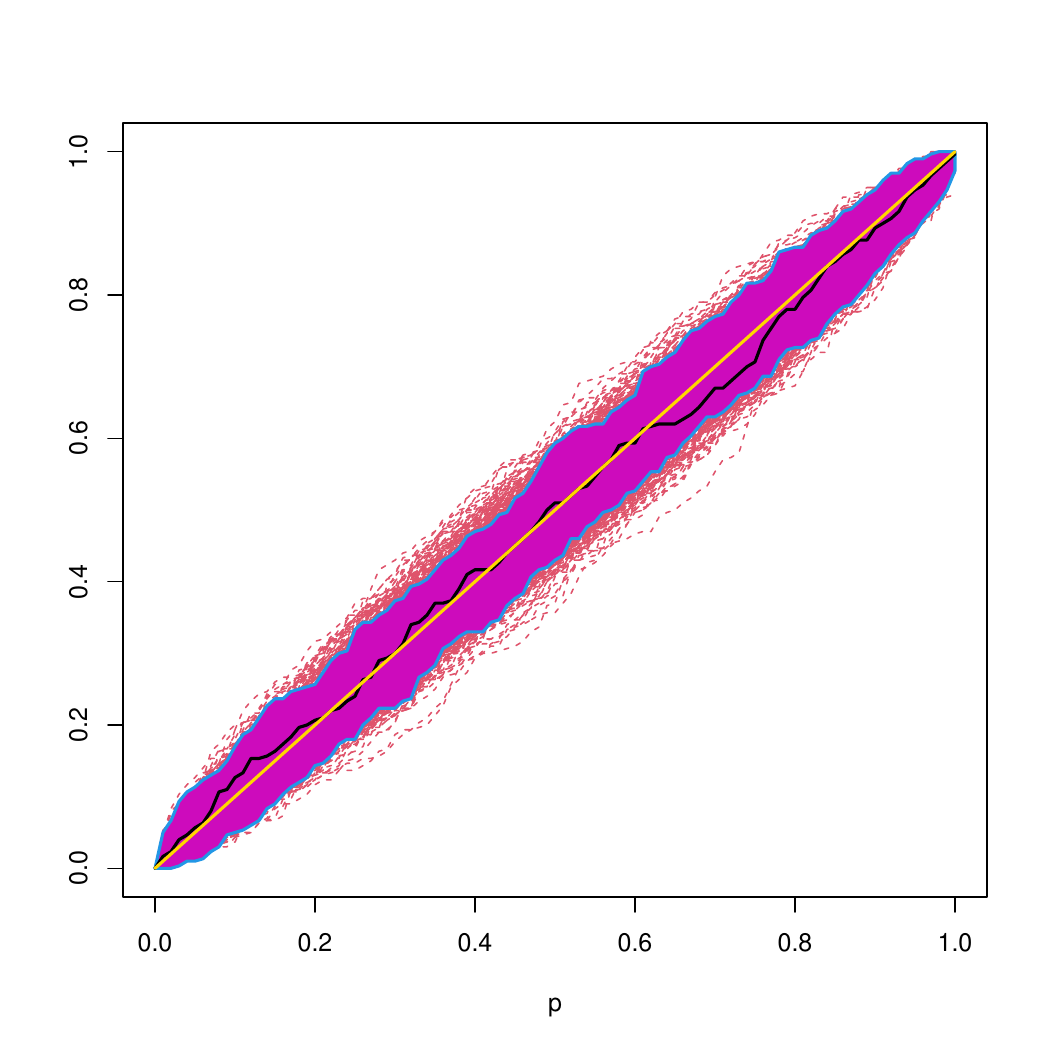}}
 \\[-4ex]
    
$\wUps_{\media}$ &
\raisebox{-.5\height}{\includegraphics[scale=0.25]{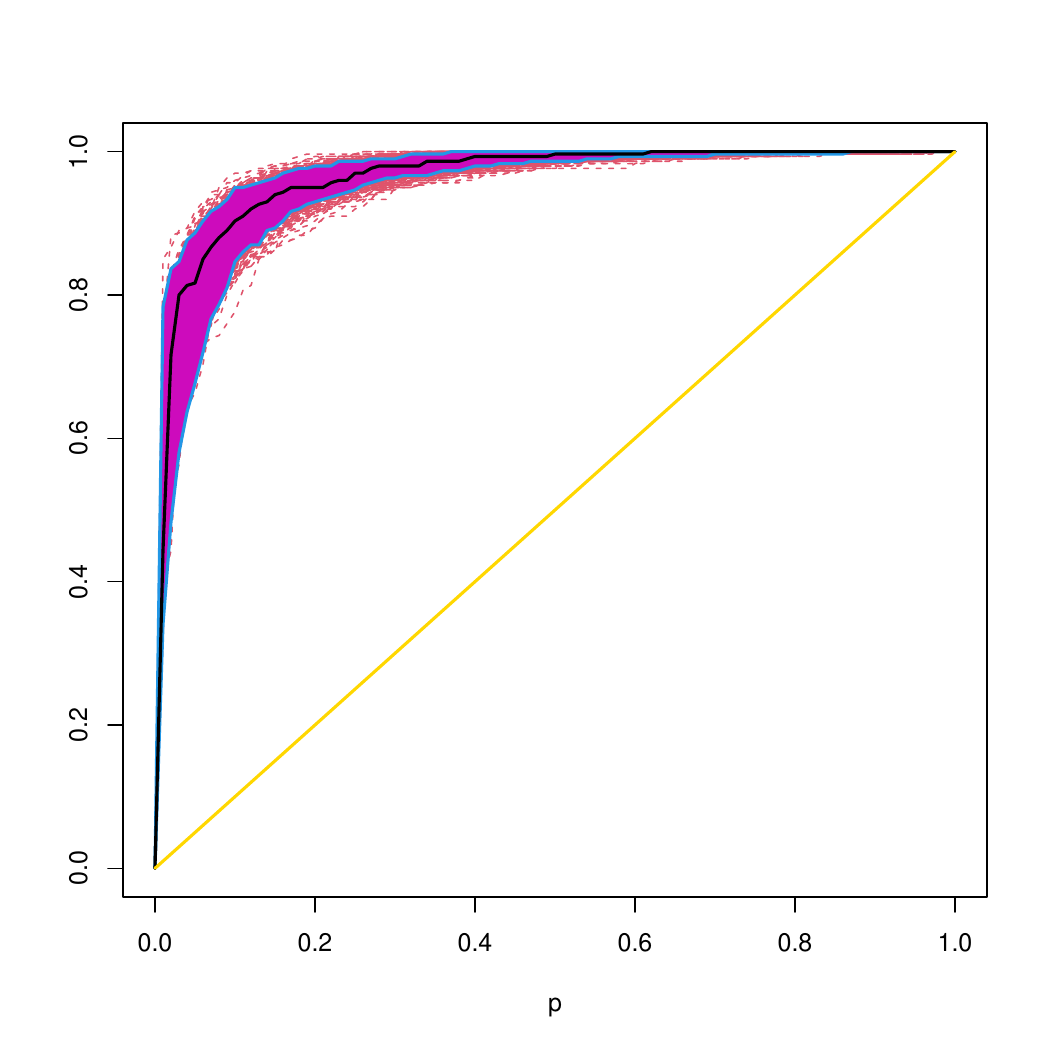}}
& \raisebox{-.5\height}{\includegraphics[scale=0.25]{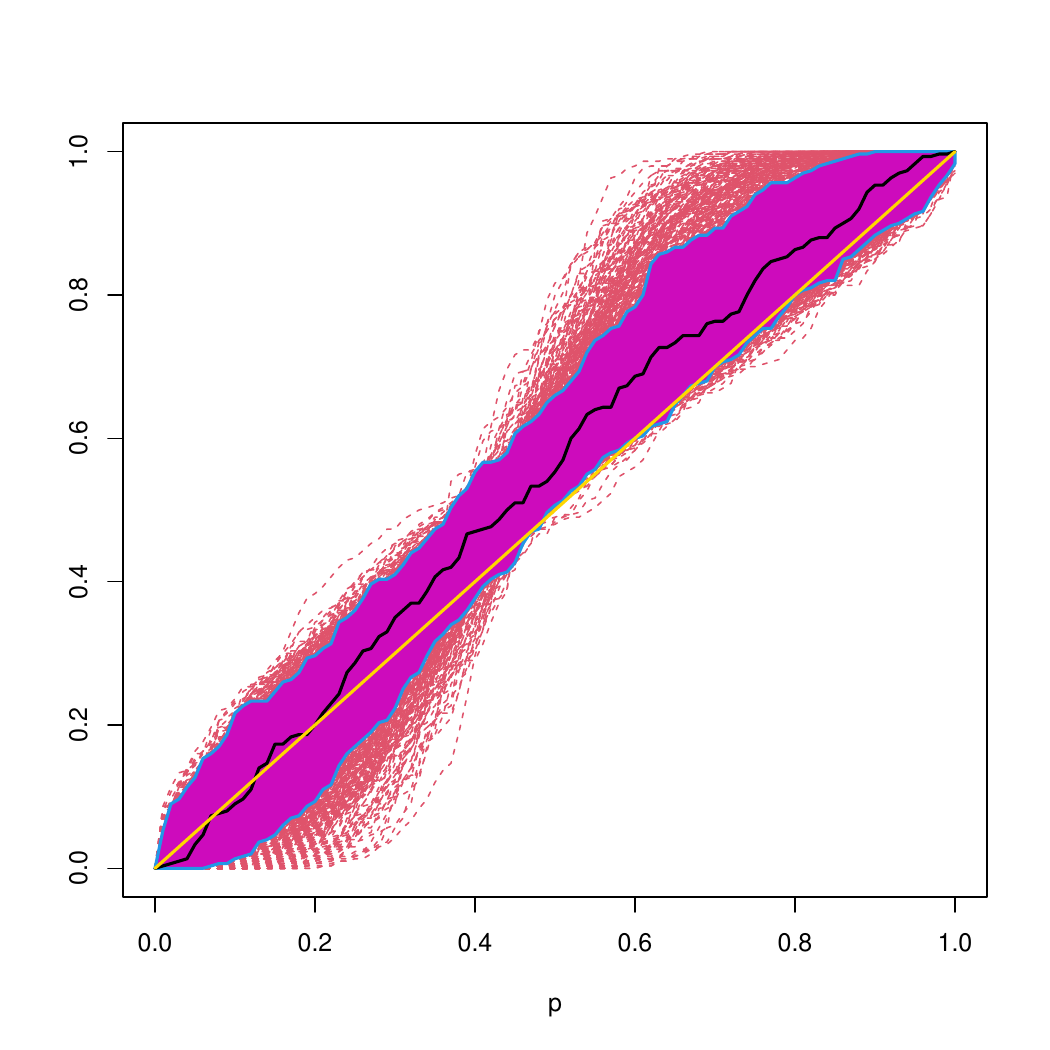}}
\\[-4ex]

$\wUps_{\lin}$  & 
\raisebox{-.5\height}{\includegraphics[scale=0.25]{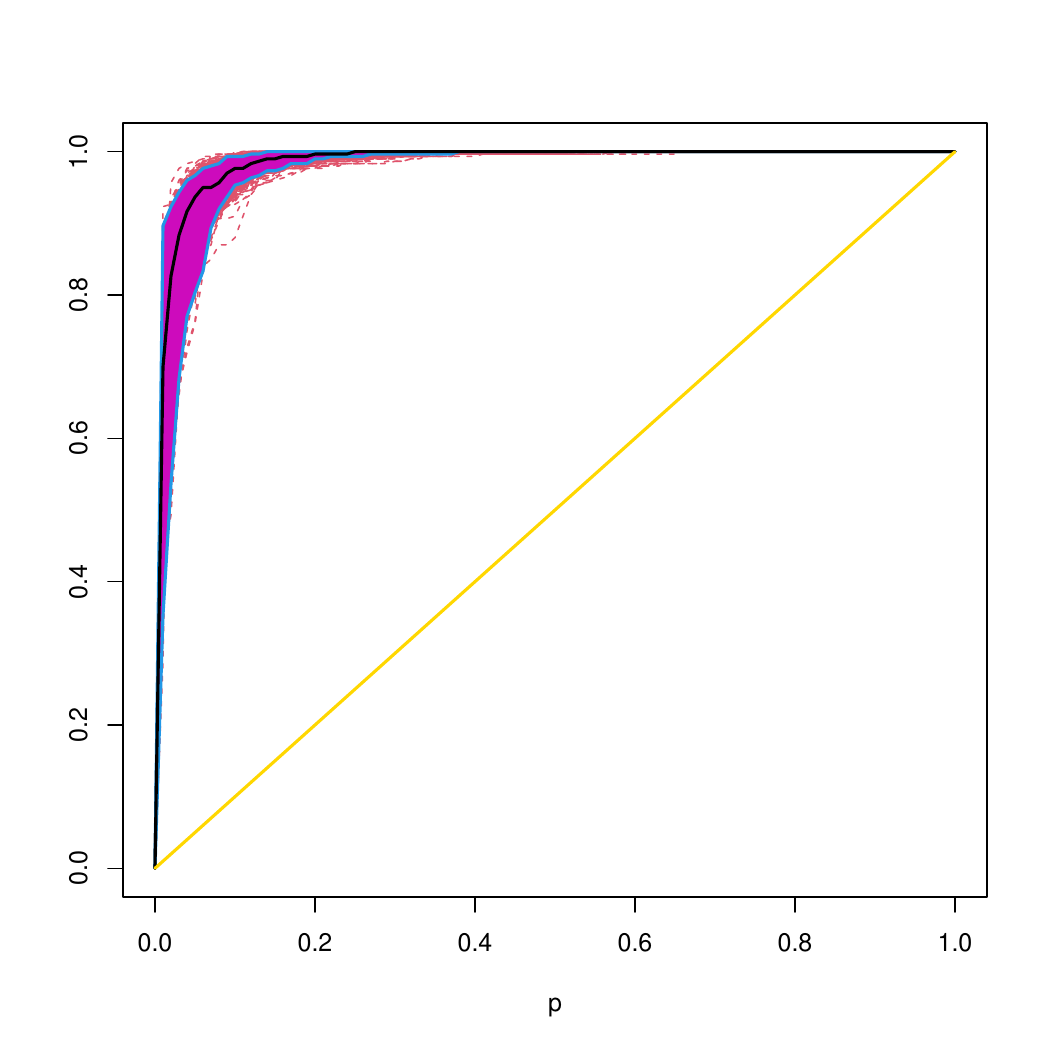}}
& \raisebox{-.5\height}{\includegraphics[scale=0.25]{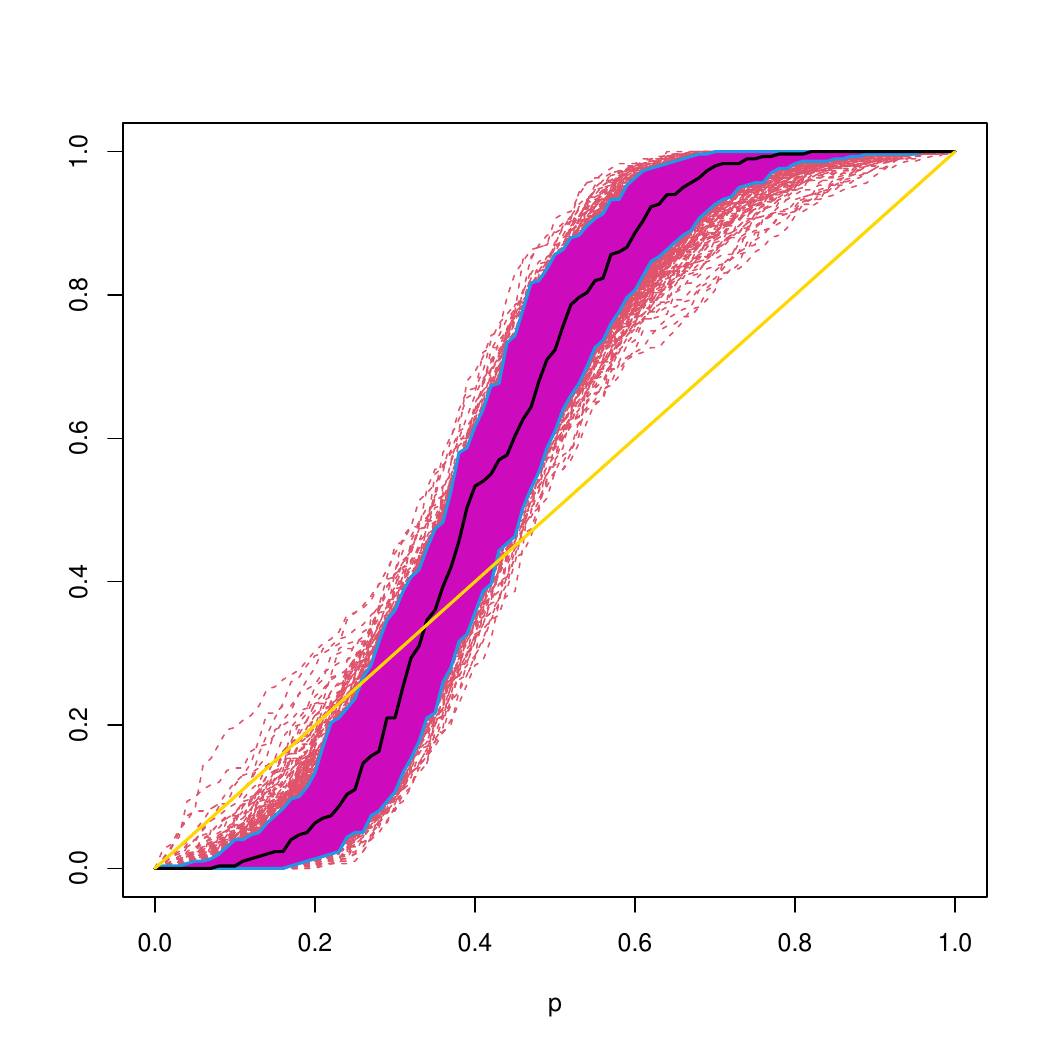}}
\\[-4ex]

$\wUps_{\cuad}$ & 
\raisebox{-.5\height}{\includegraphics[scale=0.25]{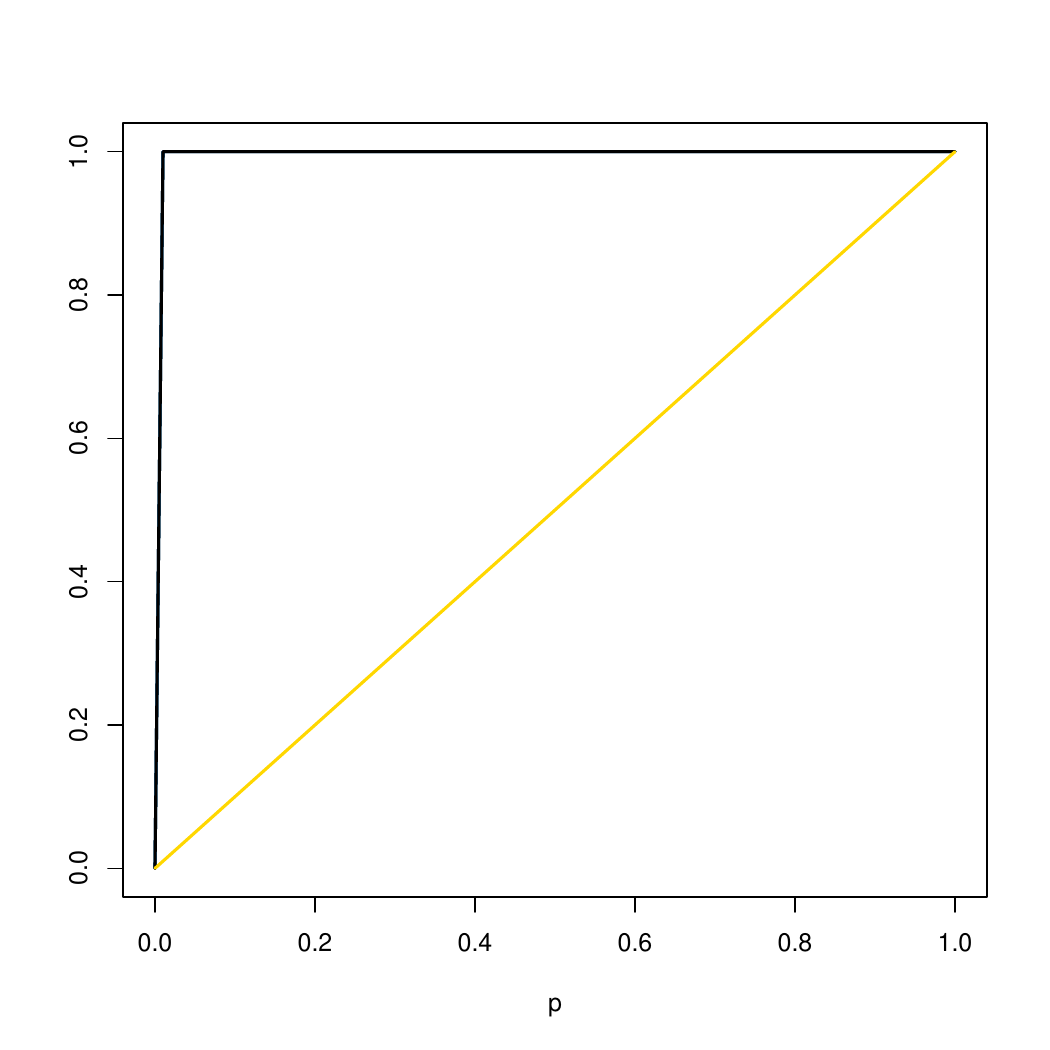}}
& \raisebox{-.5\height}{\includegraphics[scale=0.25]{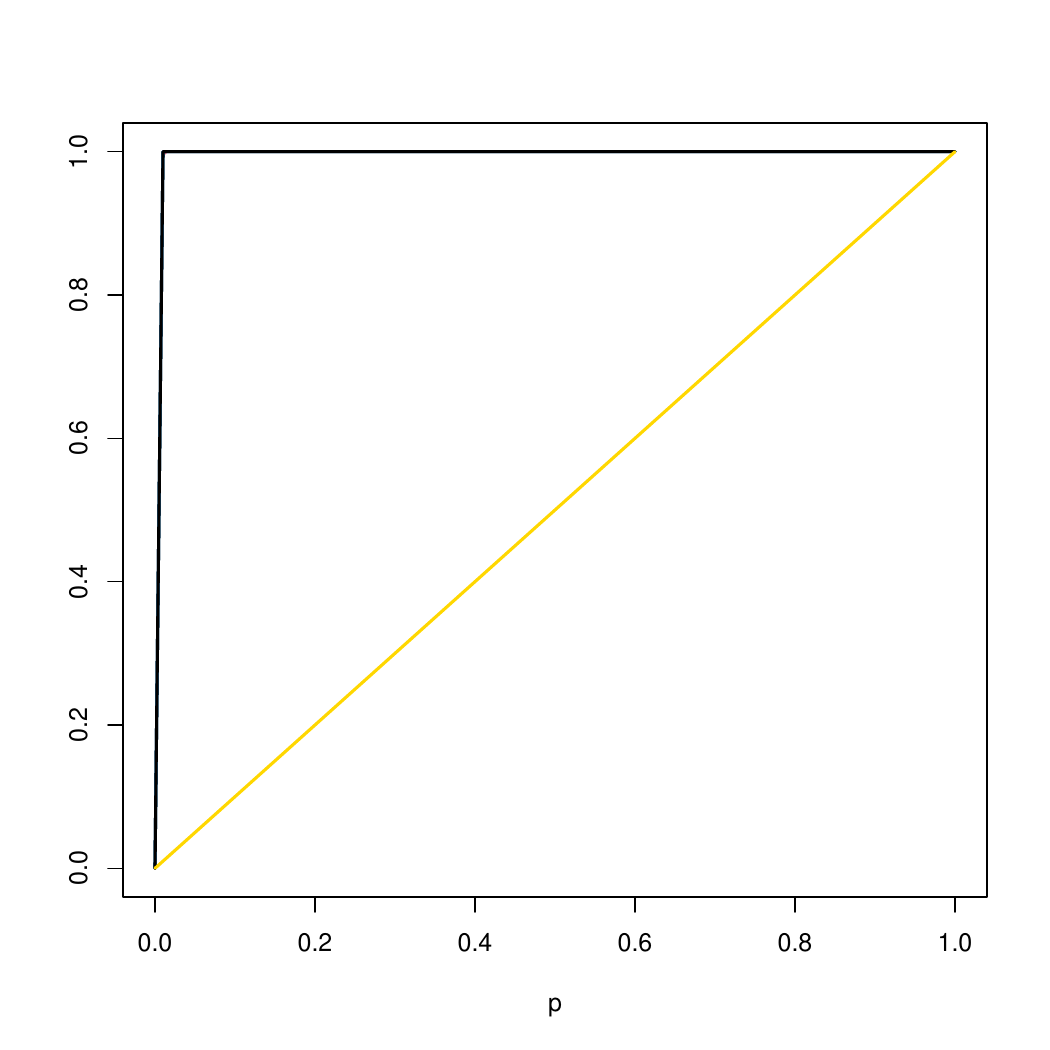}}

\end{tabular}
\caption{Functional boxplots of the estimators $\widehat{\ROC}$ under scheme \textbf{D2}.  Rows correspond to discriminating indexes, while columns to $\mu_D(t)=2\, \sin(\pi  t)$ and $\mu_H=0$.}
\label{fig:DIFF-D2}
\end{center} 
\end{figure}

\subsection{Numerical results for unbalanced designs}{\label{sec:desbalance}}
In this Section, we report the results of a numerical study conducted to evaluate the effect of  unbalanced sample sizes. To consider a framework similar to the cardiotoxicity data set, we chose $n_D=30$ and $n_H=250$ under a proportional model with $\rho=2$.  Table  \ref{tab:summary-propor-nD30-nH-250} displays the mean and standard deviations of the AUC estimators, while Figures \ref{fig:propor:Brownian-nD30-nH250} and \ref{fig:propor:varexp-nD30-nH250}  display  the functional boxplots corresponding to the estimates of the ROC. The obtained results reveal that, as for the situation where the two samples have the same size, the quadratic rule $\wUps_{\cuad}$ outperforms the other competitors, even when  this setting is not so harmful for the linear rule as the one where the covariance operators follow a functional common principal component model. This suggests that the quadratic rule should be taken into account in frameworks where equality of the covariance operators may be doubtful.

\begin{table}[ht!]
\begin{center}
\footnotesize
\renewcommand{\arraystretch}{1.2}
\setlength{\tabcolsep}{2pt}
\caption{\label{tab:summary-propor-nD30-nH-250} Mean and standard deviation of the $\widehat{\AUC}$, under scenario \textbf{PROP}, that is, under a proportional model  $\gamma_D(t,s)=\rho\;\gamma_H(t,s)$ with equal (\textbf{P0}) or different mean functions (\textbf{P1}), $\mu_H(t)=0$ and $\mu_D(t)=2\, \sin(\pi  t)$). In all cases, $n_H=250$ and $n_D=30$.} 
{\begin{tabular}{c c cccccc@{\extracolsep{1cm}}    c@{\extracolsep{3pt}}ccccc }
   \hline \\[-2ex]
$\hskip0.1in\rho\hskip0.1in$  & & $\Upsilon_{\maxi}$ & $\Upsilon_{\mini}$ & $\Upsilon_{\inte}$ & $\wUps_{\media}$ & $\wUps_{\lin}$  & $\wUps_{\cuad}$
& $\Upsilon_{\maxi}$ & $\Upsilon_{\mini}$ & $\Upsilon_{\inte}$ & $\wUps_{\media}$ & $\wUps_{\lin}$  & $\wUps_{\cuad}$
\\
\hline
& & \multicolumn{6}{c}{\textbf{P1}} & \multicolumn{6}{c}{\textbf{P0}}\\ \hline 
& & \multicolumn{12}{c}{Brownian Motion}\\
\hline 
 2 & Mean &  0.9435 &  0.6187 & 0.8965 & 0.9330 & \textit{0.9937} & \textbf{0.9951}
          &  0.5986 &  0.4031 & 0.5005 & 0.5816 & \textit{0.6216} & \textbf{0.7917}
  \\
  & SD & 0.0212 & 0.0655 & 0.0346 & 0.0235 & 0.0067 & 0.0057 
       & 0.0585 & 0.0581 & 0.0635 & 0.0370 & 0.0458 & 0.0462
 \\
\hline 
& & \multicolumn{12}{c}{Exponential Variogram}\\
\hline 
 2 & Mean & 0.9521 & 0.5520 & 0.9011 & 0.9295 & \textit{0.9683} & \textbf{0.9988}
 	      & 0.7193 & 0.2828 & 0.5006 & 0.6518 & \textit{0.7670} & \textbf{0.9943}
 \\
  & SD & 0.0223 & 0.0658 & 0.0342 & 0.0274 & 0.0173 & 0.0017
       &  0.0557 & 0.0546 & 0.0637 & 0.0363 & 0.0402 & 0.0052
 \\
\hline
\end{tabular}}
\end{center}
\end{table}


\begin{figure}[ht!]
 \begin{center}
 \footnotesize
 \renewcommand{\arraystretch}{0.2}
  
\begin{tabular}{p{2cm} cc}
 & \textbf{P1} &  \textbf{P0} \\[-2ex]
$\Upsilon_{\maxi}$  &
 \raisebox{-.5\height}{\includegraphics[scale=0.25]{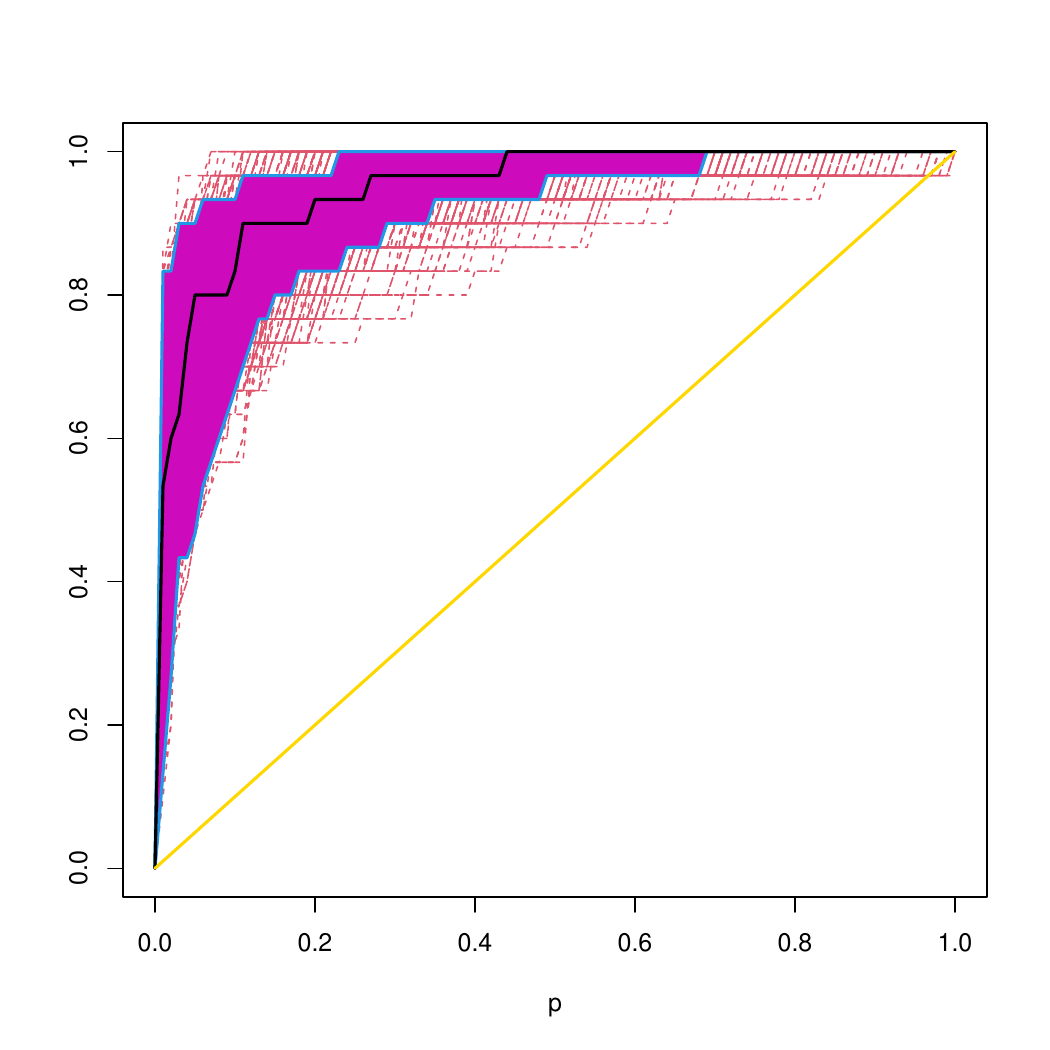}}
& \raisebox{-.5\height}{\includegraphics[scale=0.25]{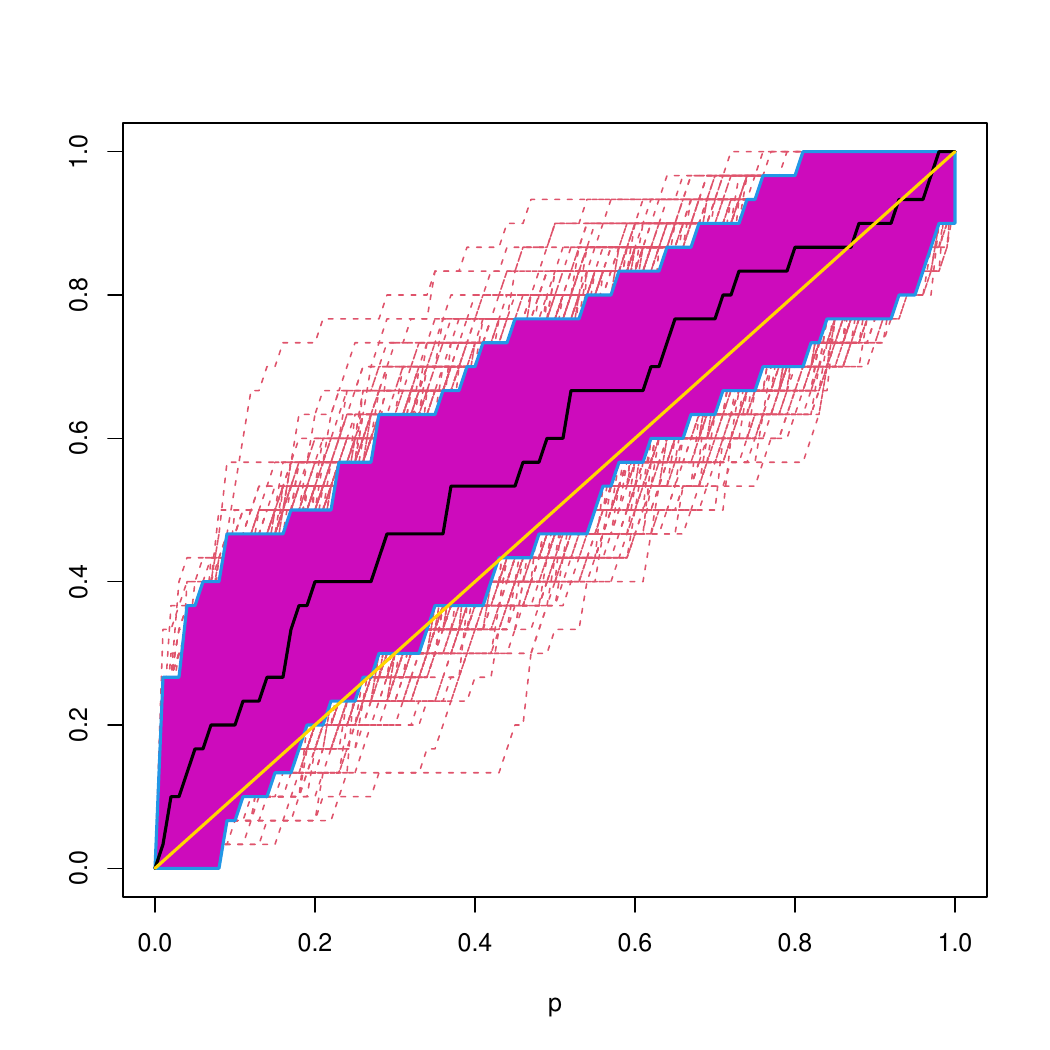}}
\\[-4ex]
$\Upsilon_{\inte}$  &
 \raisebox{-.5\height}{\includegraphics[scale=0.25]{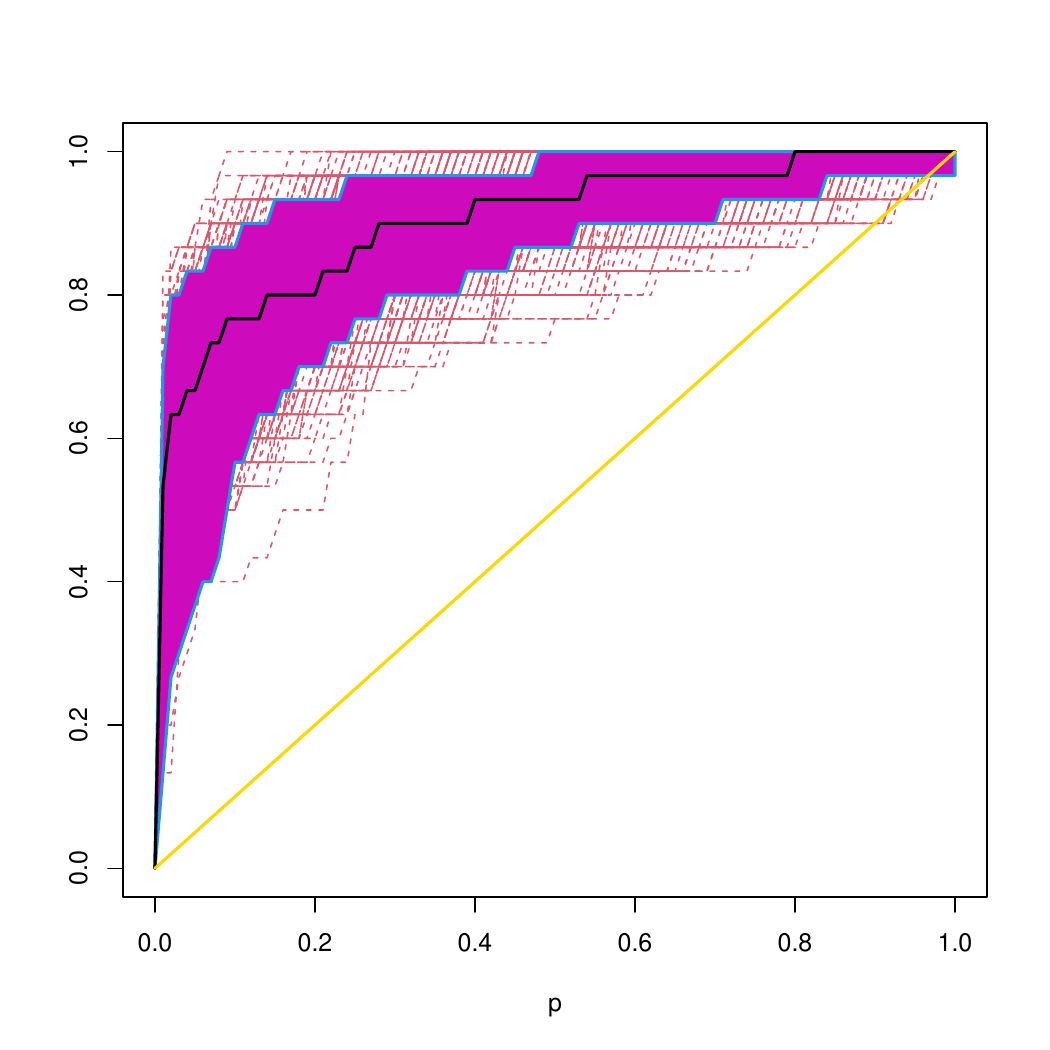}}
& \raisebox{-.5\height}{\includegraphics[scale=0.25]{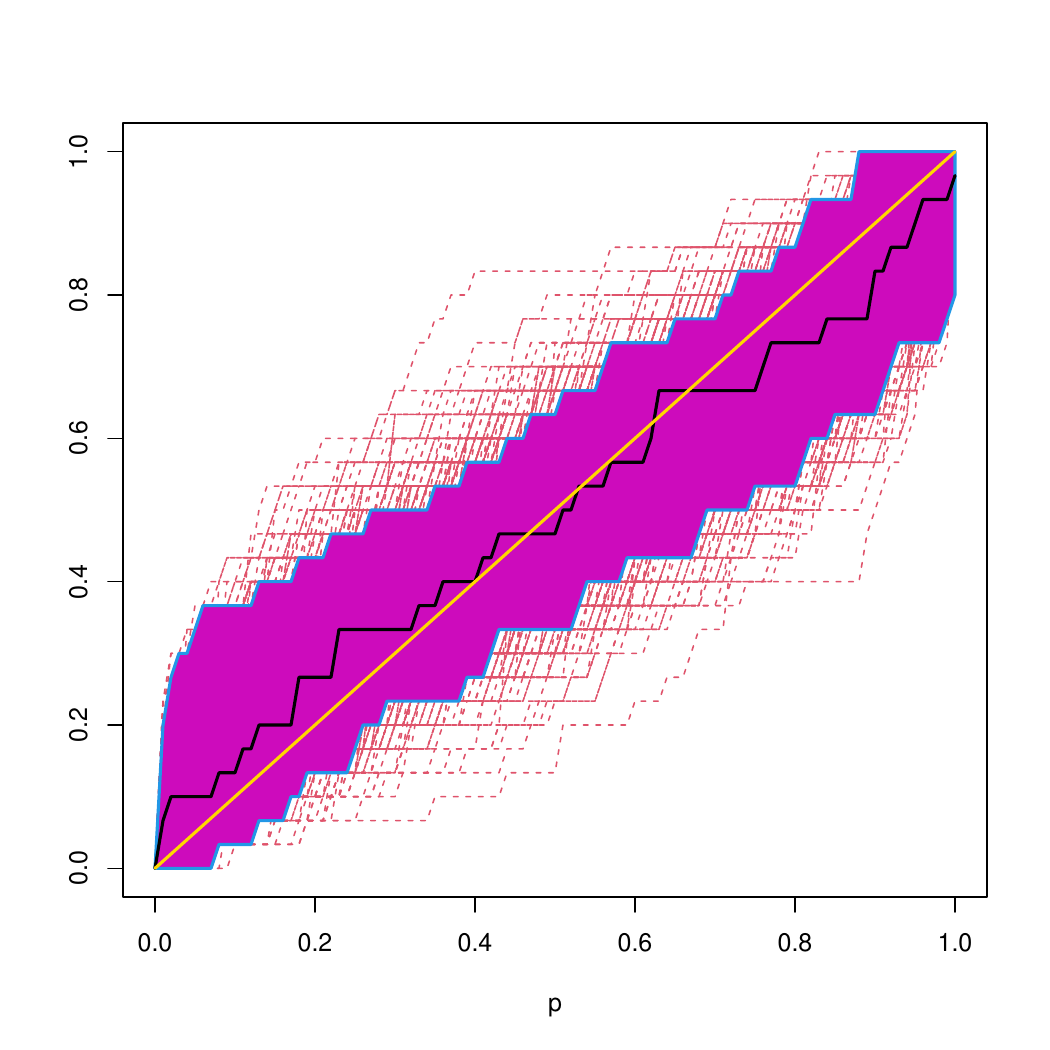}}
 \\[-4ex]
    
$\wUps_{\media}$ &
\raisebox{-.5\height}{\includegraphics[scale=0.25]{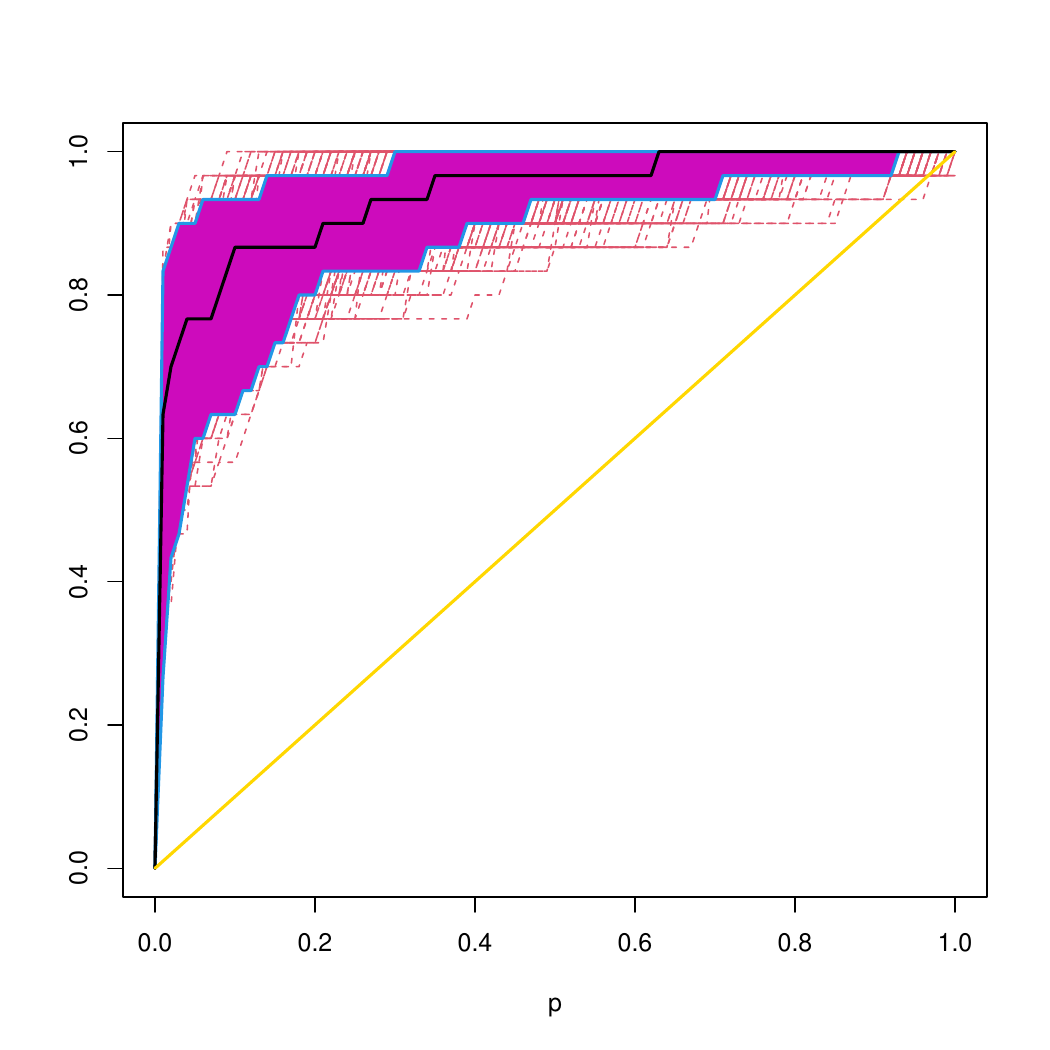}}
& \raisebox{-.5\height}{\includegraphics[scale=0.25]{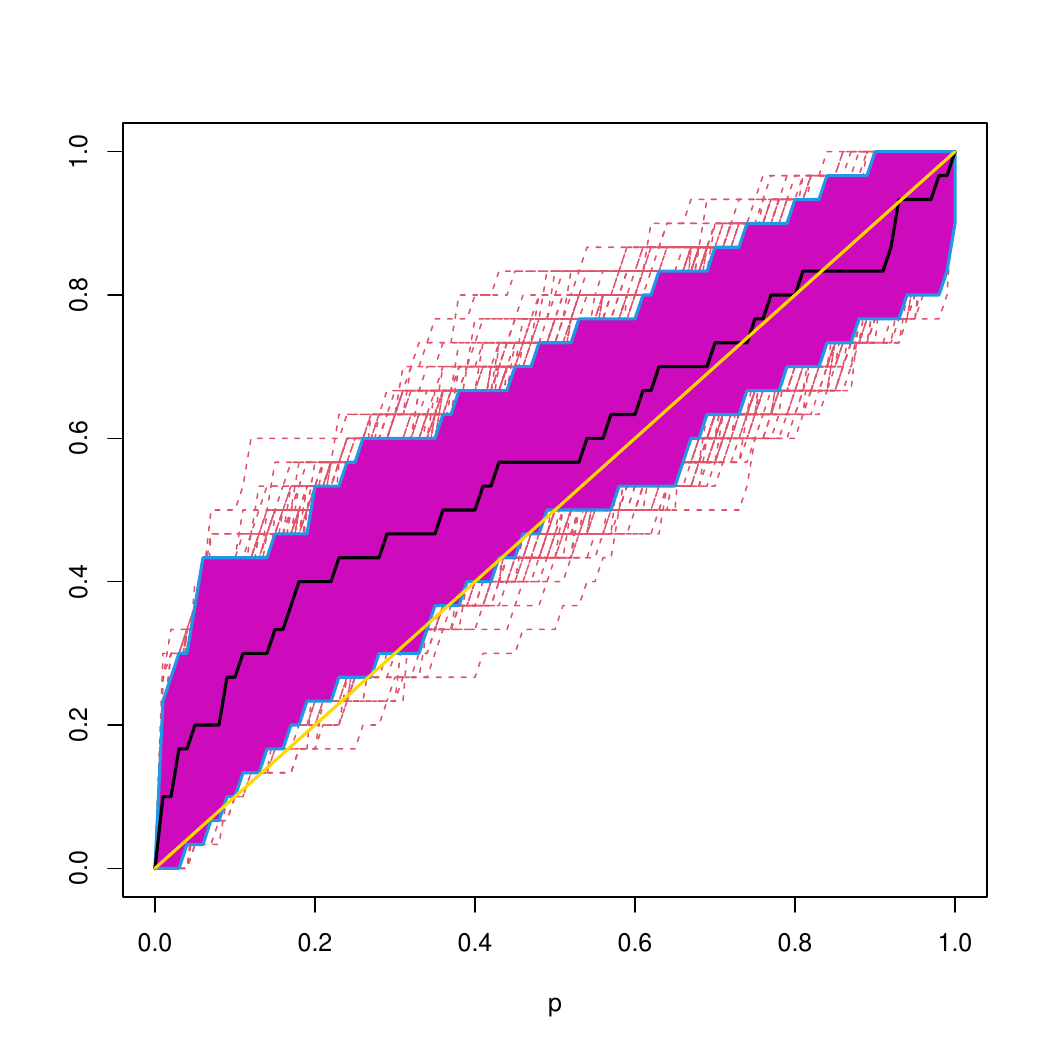}}
\\[-4ex]

$\wUps_{\lin}$  & 
\raisebox{-.5\height}{\includegraphics[scale=0.25]{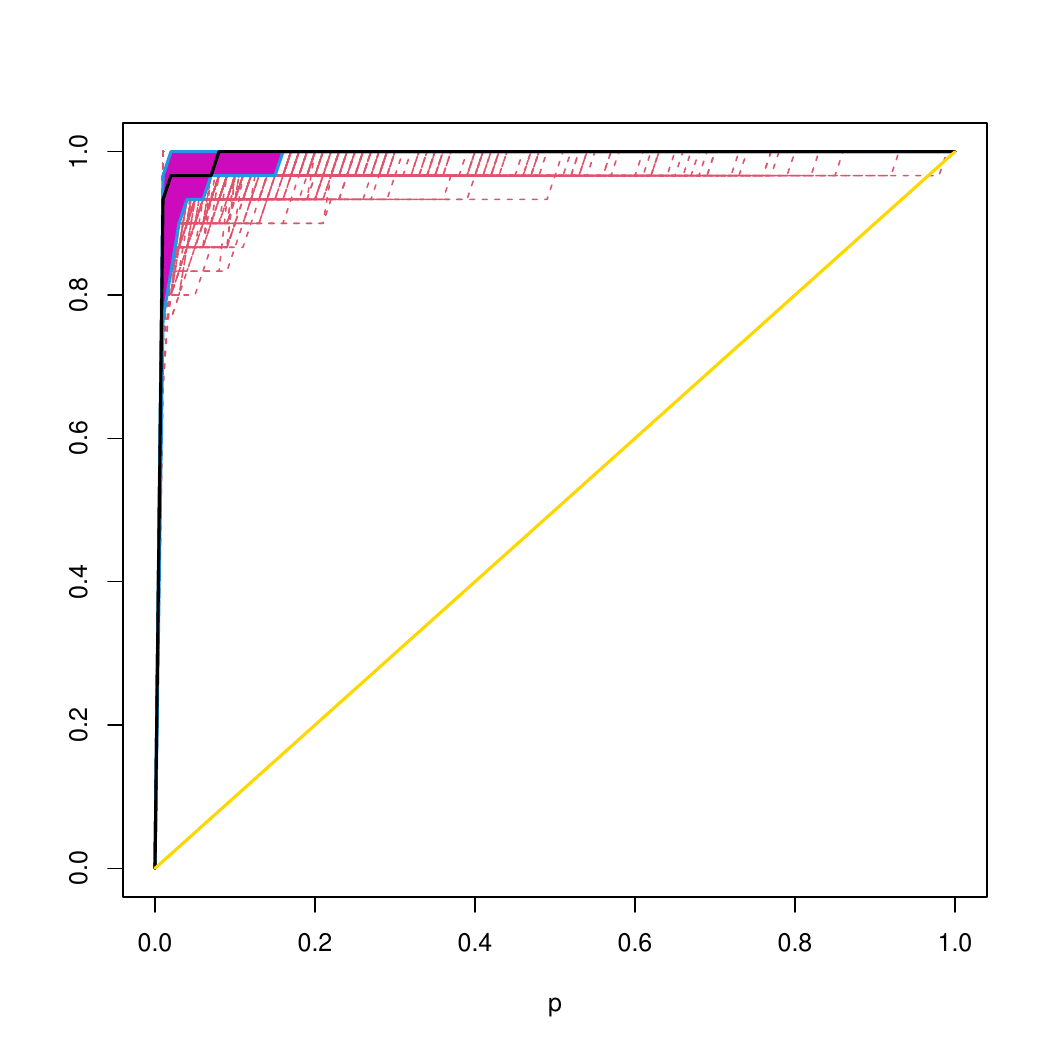}}
& \raisebox{-.5\height}{\includegraphics[scale=0.25]{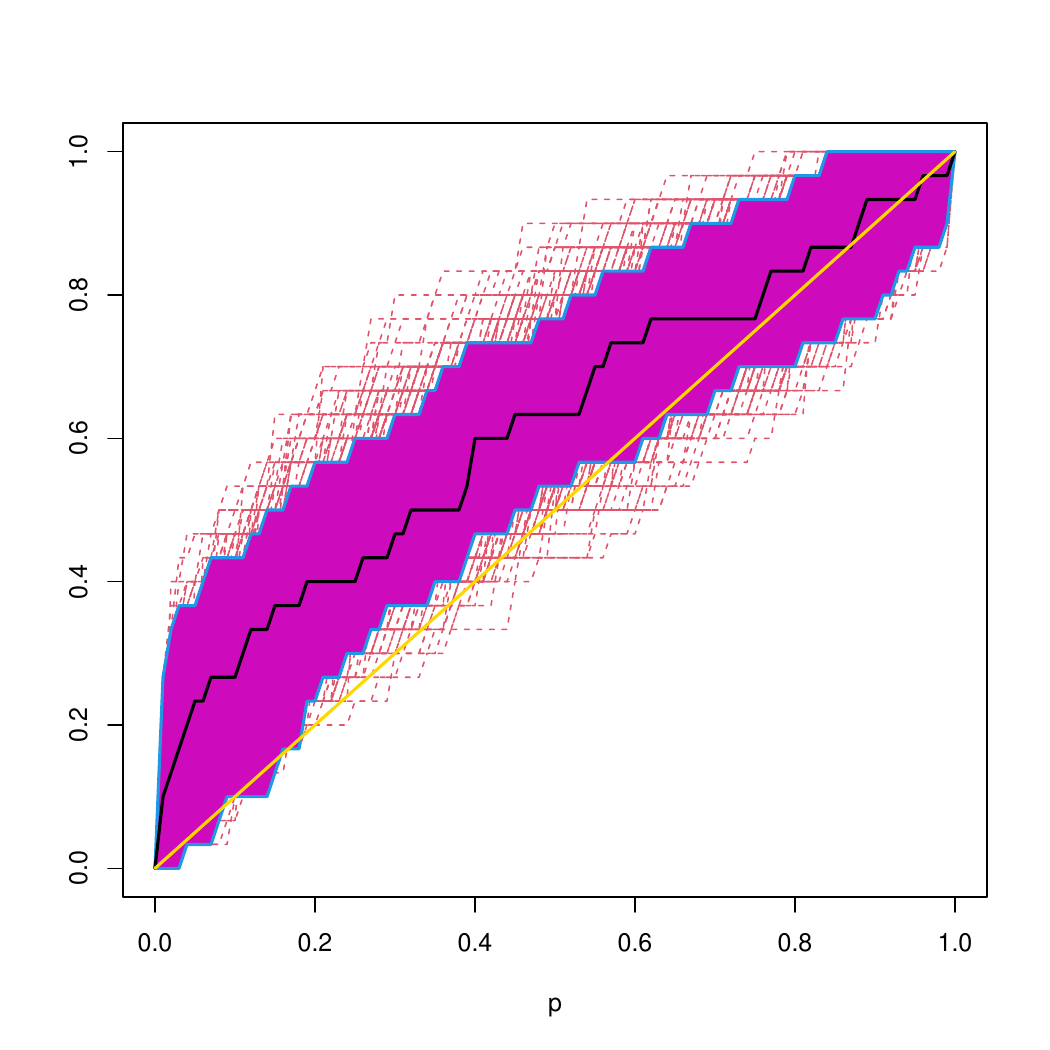}}
\\[-4ex]

$\wUps_{\cuad}$ & 
\raisebox{-.5\height}{\includegraphics[scale=0.25]{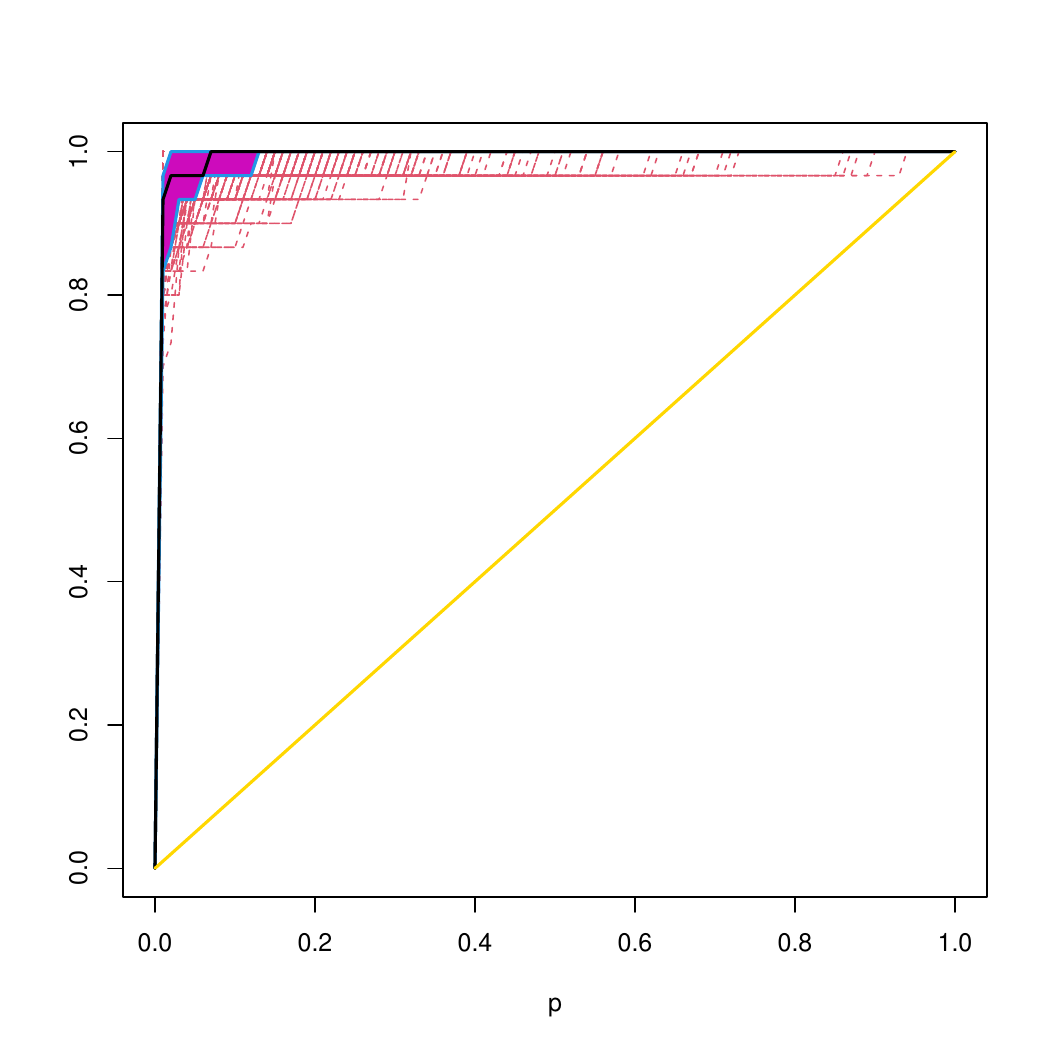}}
& \raisebox{-.5\height}{\includegraphics[scale=0.25]{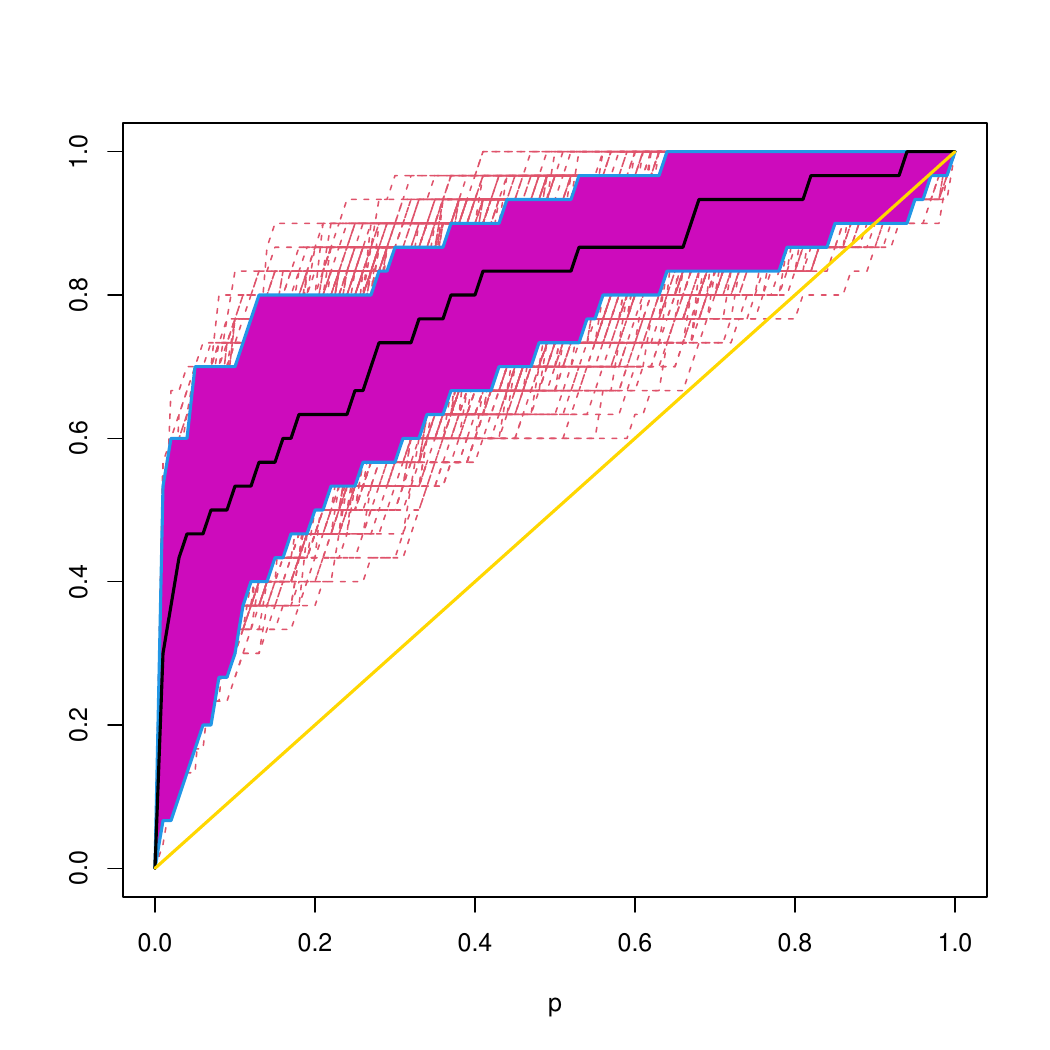}}

\end{tabular}
\caption{Functional boxplots of the estimators $\widehat{\ROC}$ under scenario \textbf{PROP} with  $\rho=2$ for the Brownian motion setting. Rows correspond to discriminating indexes, while columns to  $\mu_D(t)=2\, \sin(\pi  t)$ and $\mu_H=0$. The sample sizes are $n_D=30$ and $n_H=250$.}
\label{fig:propor:Brownian-nD30-nH250} 
\end{center} 
\end{figure}

\begin{figure}[ht!]
 \begin{center}
 \footnotesize
 \renewcommand{\arraystretch}{0.2}

\begin{tabular}{p{2cm} cc}
 & \textbf{P1} &  \textbf{P0} \\[-2ex]  
$\Upsilon_{\maxi}$ &
 \raisebox{-.5\height}{\includegraphics[scale=0.25]{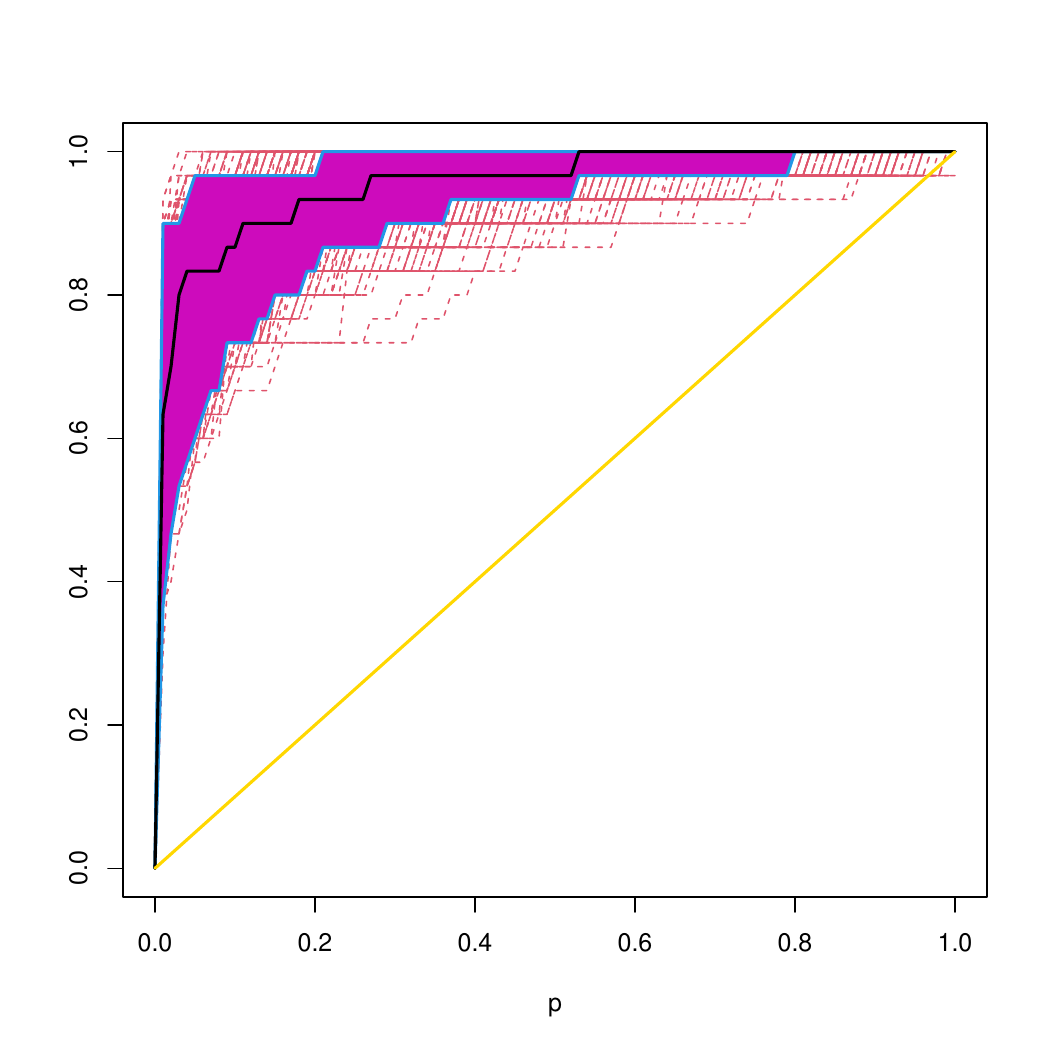}}
& \raisebox{-.5\height}{\includegraphics[scale=0.25]{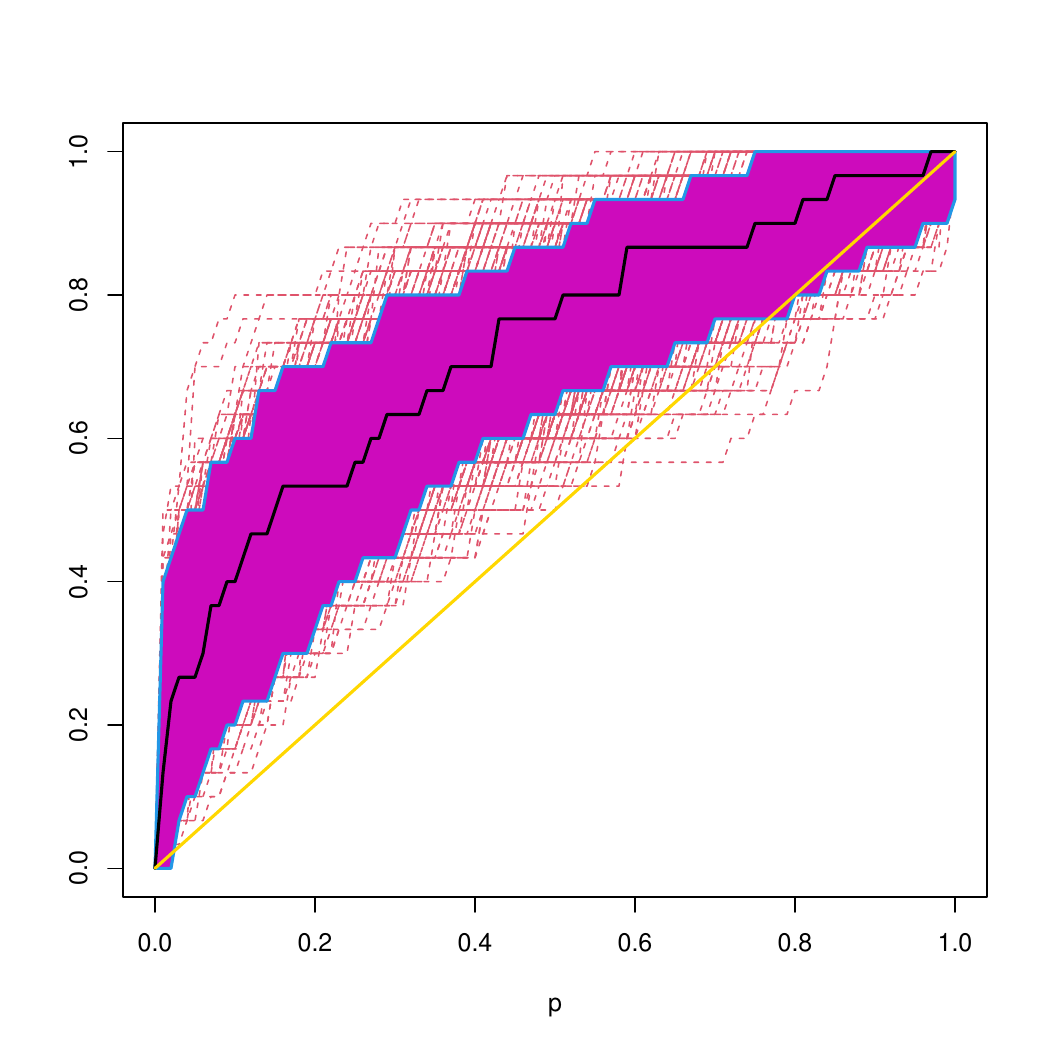}}
\\[-4ex]
 
$\Upsilon_{\inte}$  &
 \raisebox{-.5\height}{\includegraphics[scale=0.25]{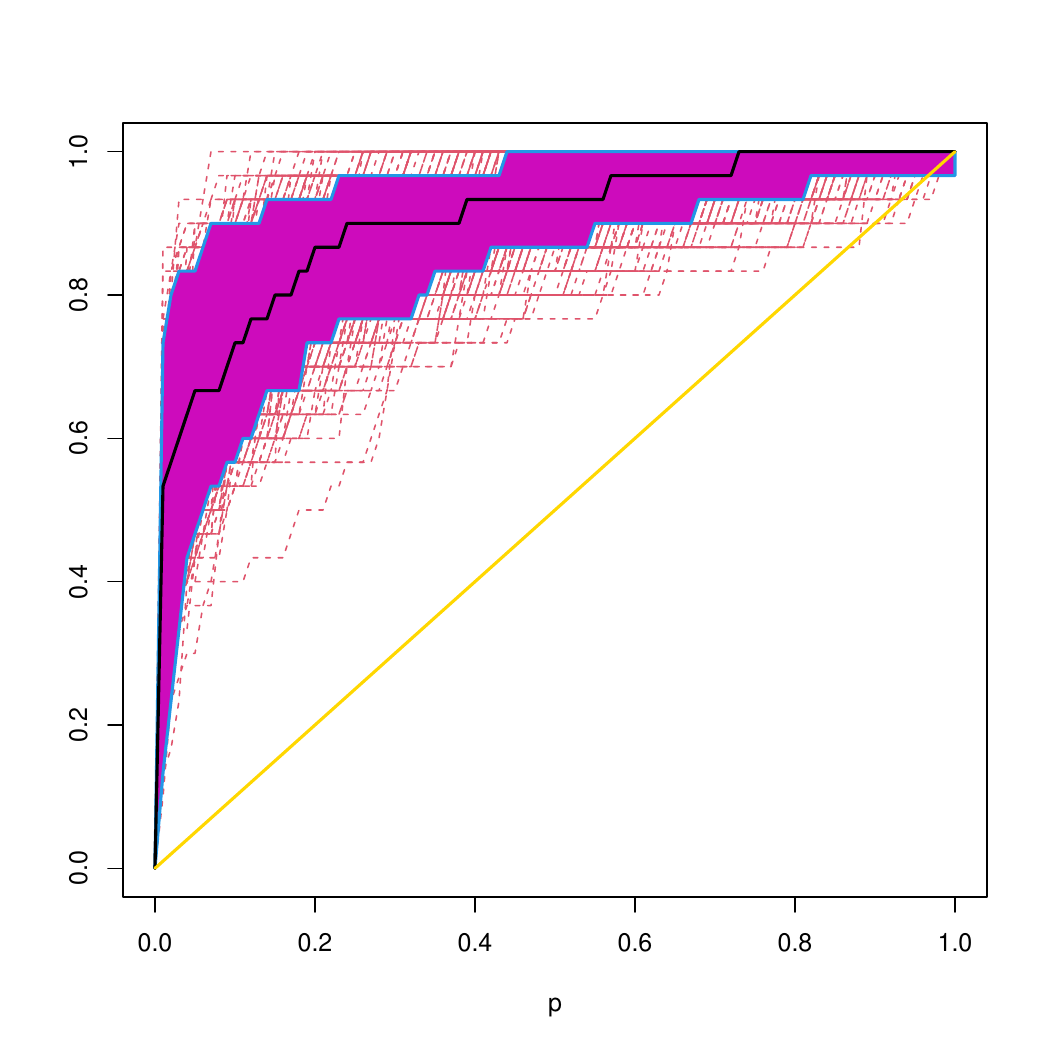}}
& \raisebox{-.5\height}{\includegraphics[scale=0.25]{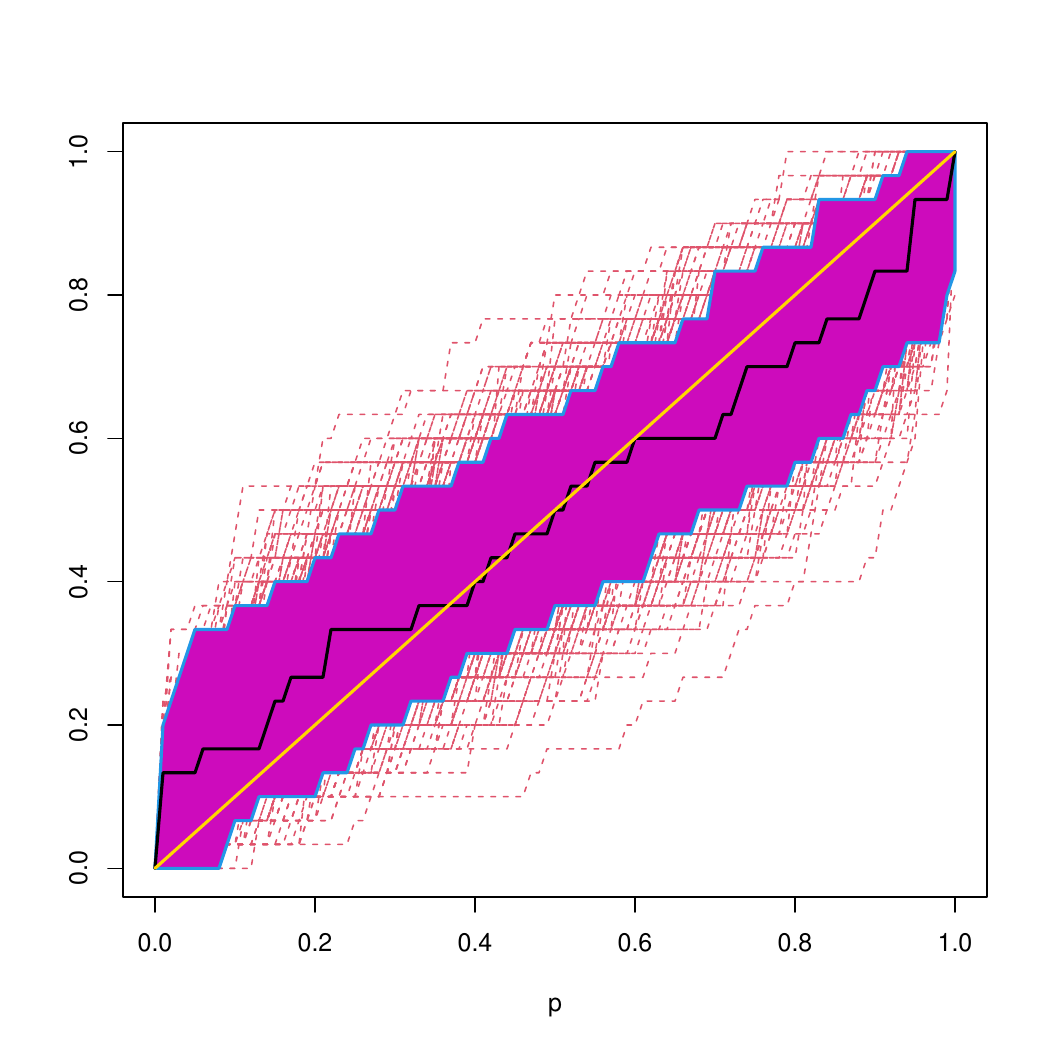}}
 \\[-4ex]
    
$\wUps_{\media}$ &
\raisebox{-.5\height}{\includegraphics[scale=0.25]{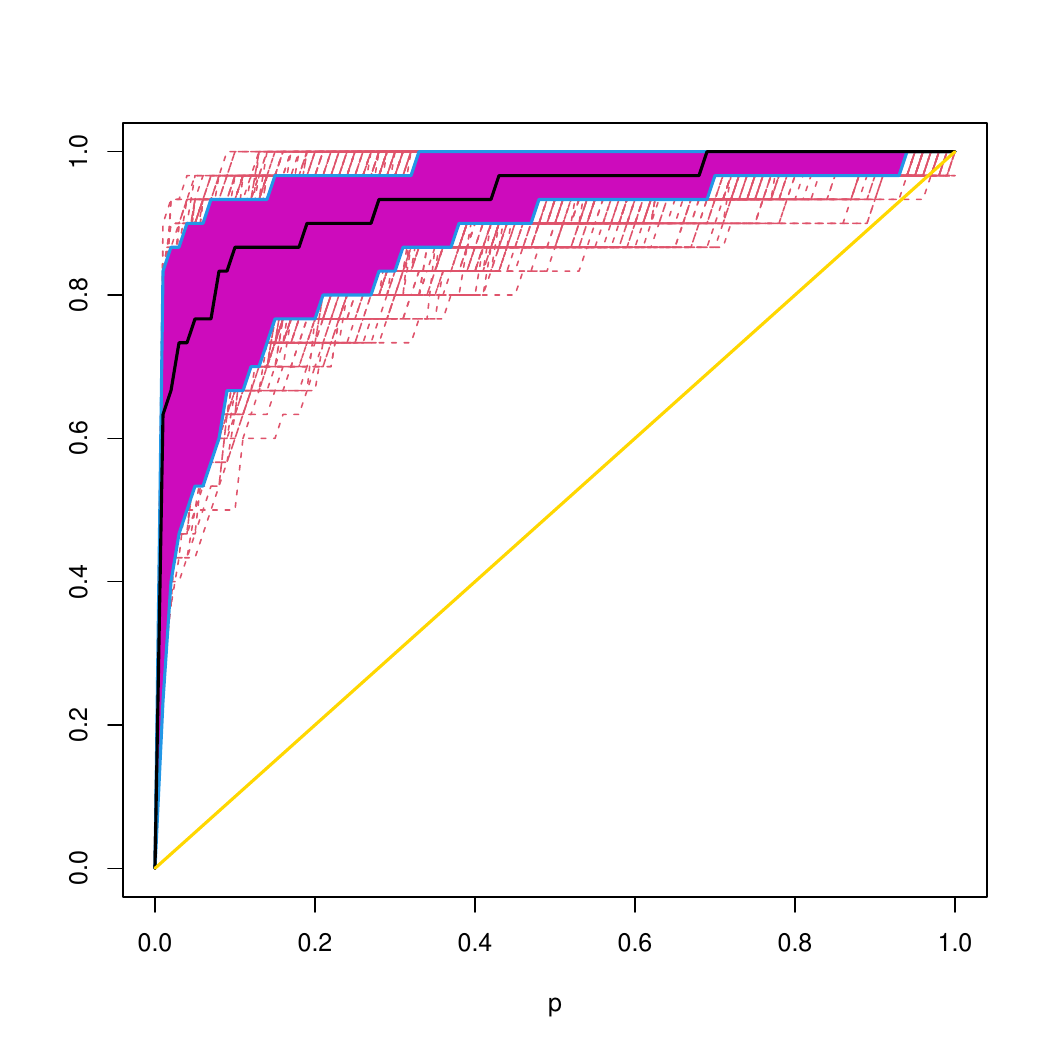}}
& \raisebox{-.5\height}{\includegraphics[scale=0.25]{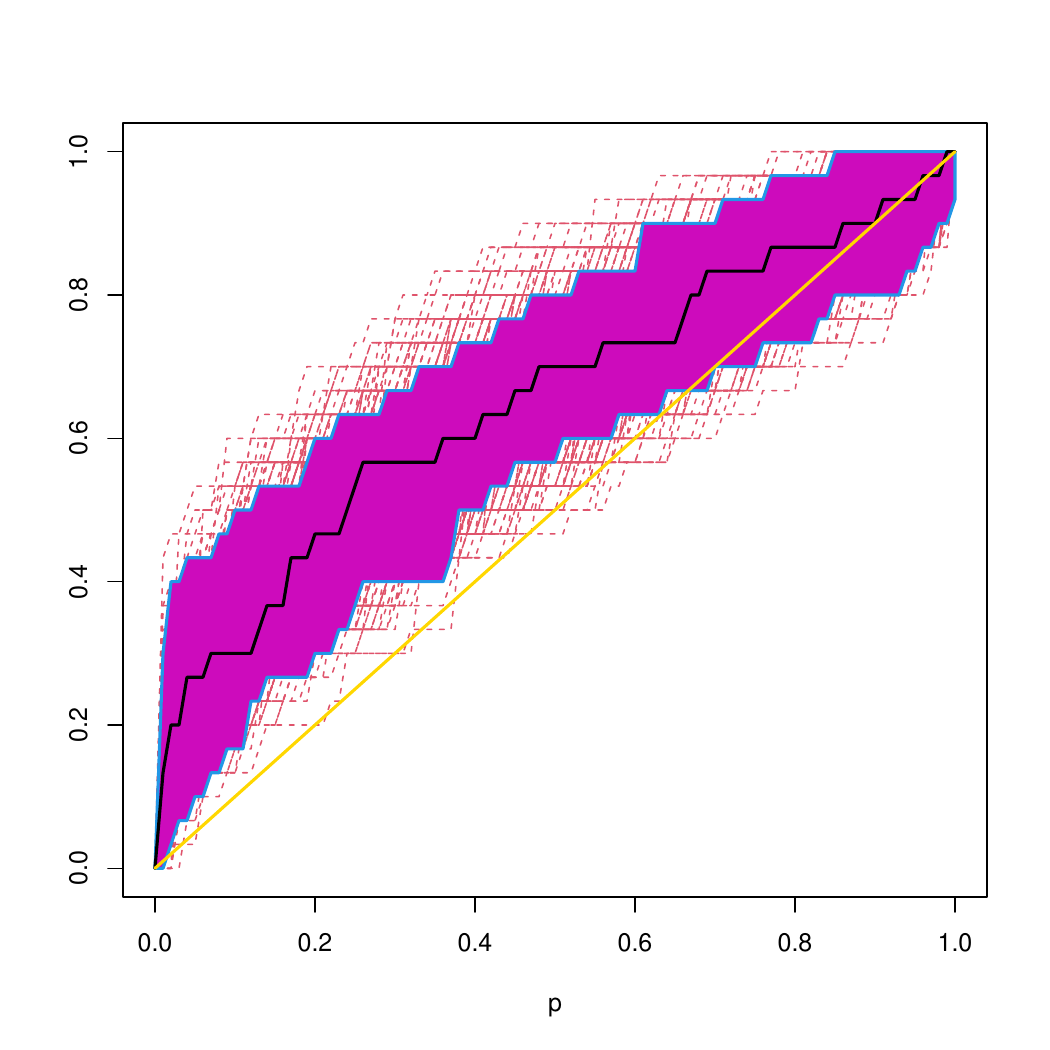}}
\\[-4ex]

$\wUps_{\lin}$  & 
\raisebox{-.5\height}{\includegraphics[scale=0.25]{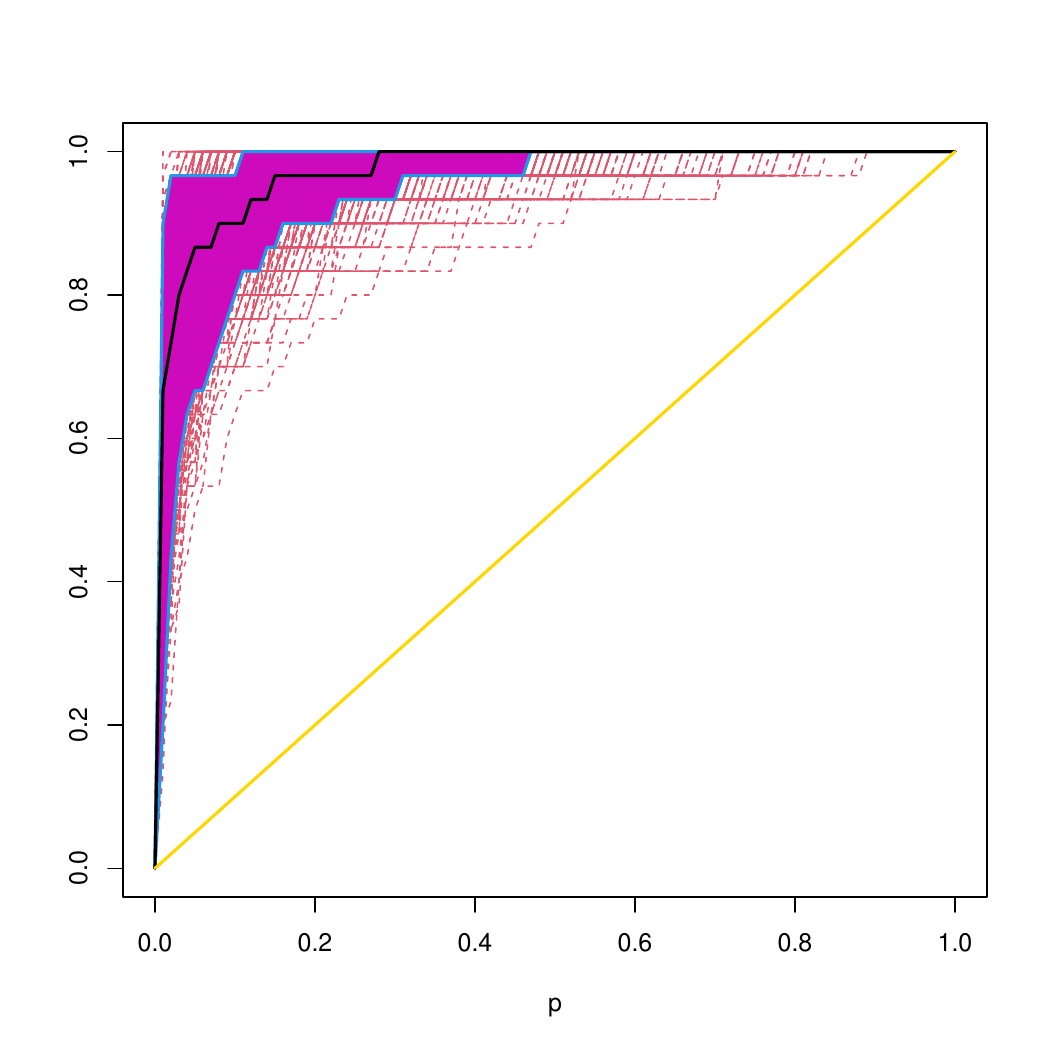}}
& \raisebox{-.5\height}{\includegraphics[scale=0.25]{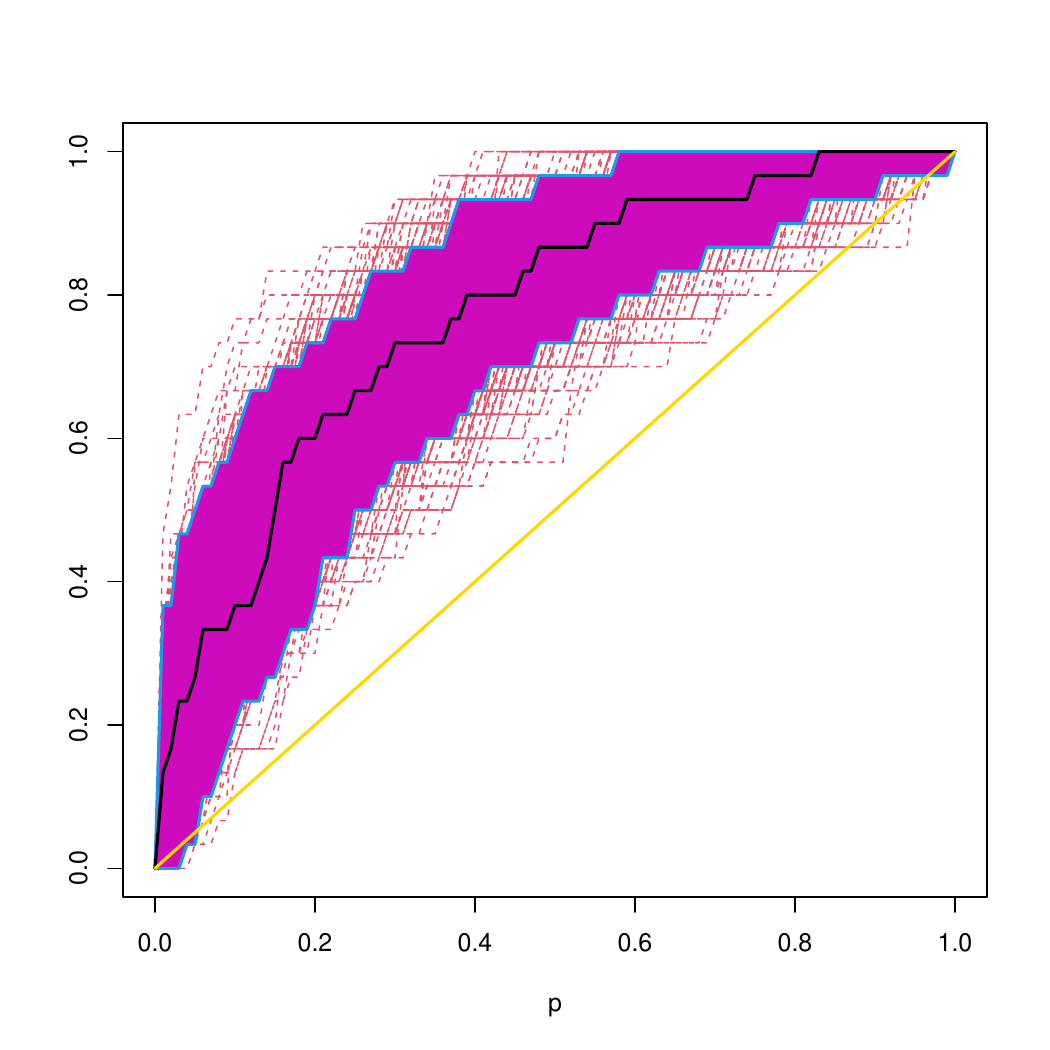}}
\\[-4ex]

$\wUps_{\cuad}$ 
& \raisebox{-.5\height}{\includegraphics[scale=0.25]{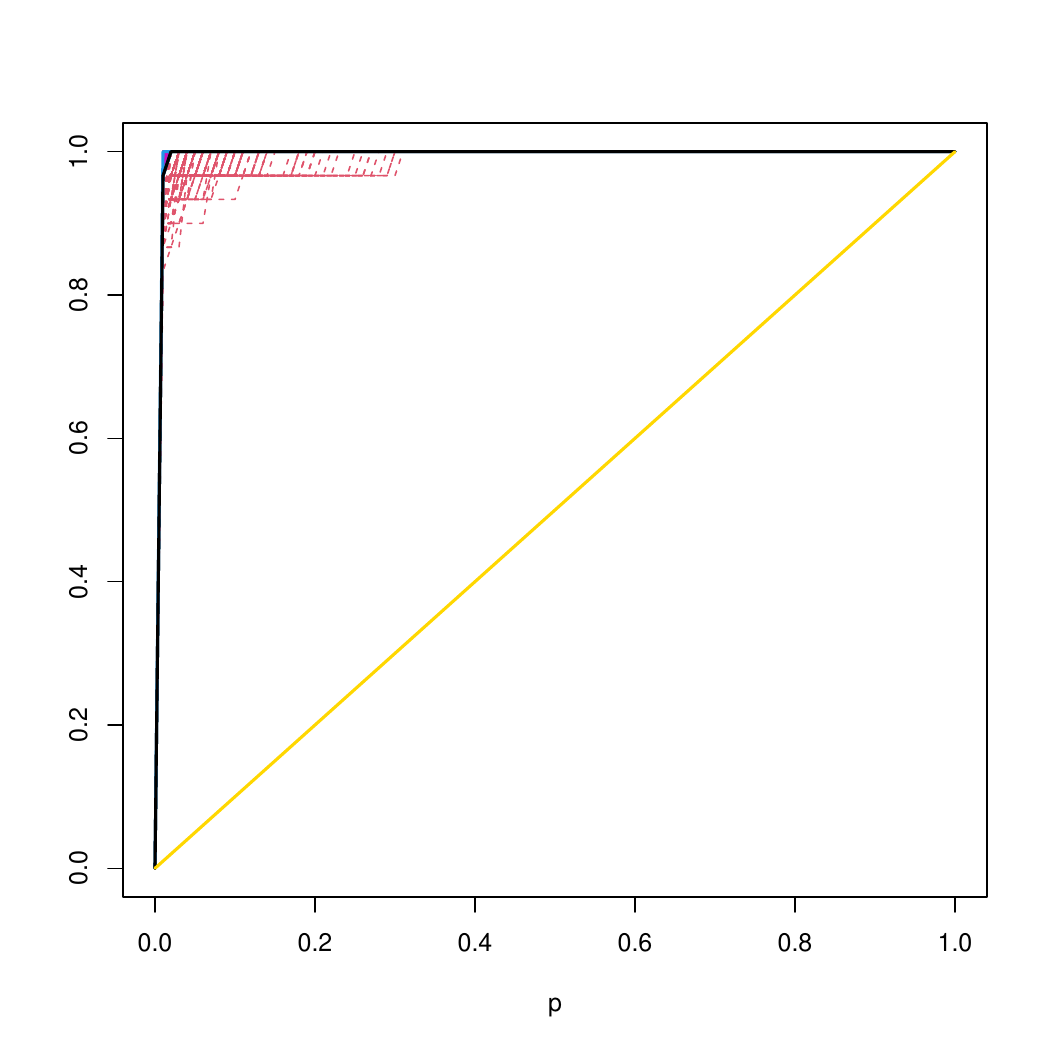}}
& \raisebox{-.5\height}{\includegraphics[scale=0.25]{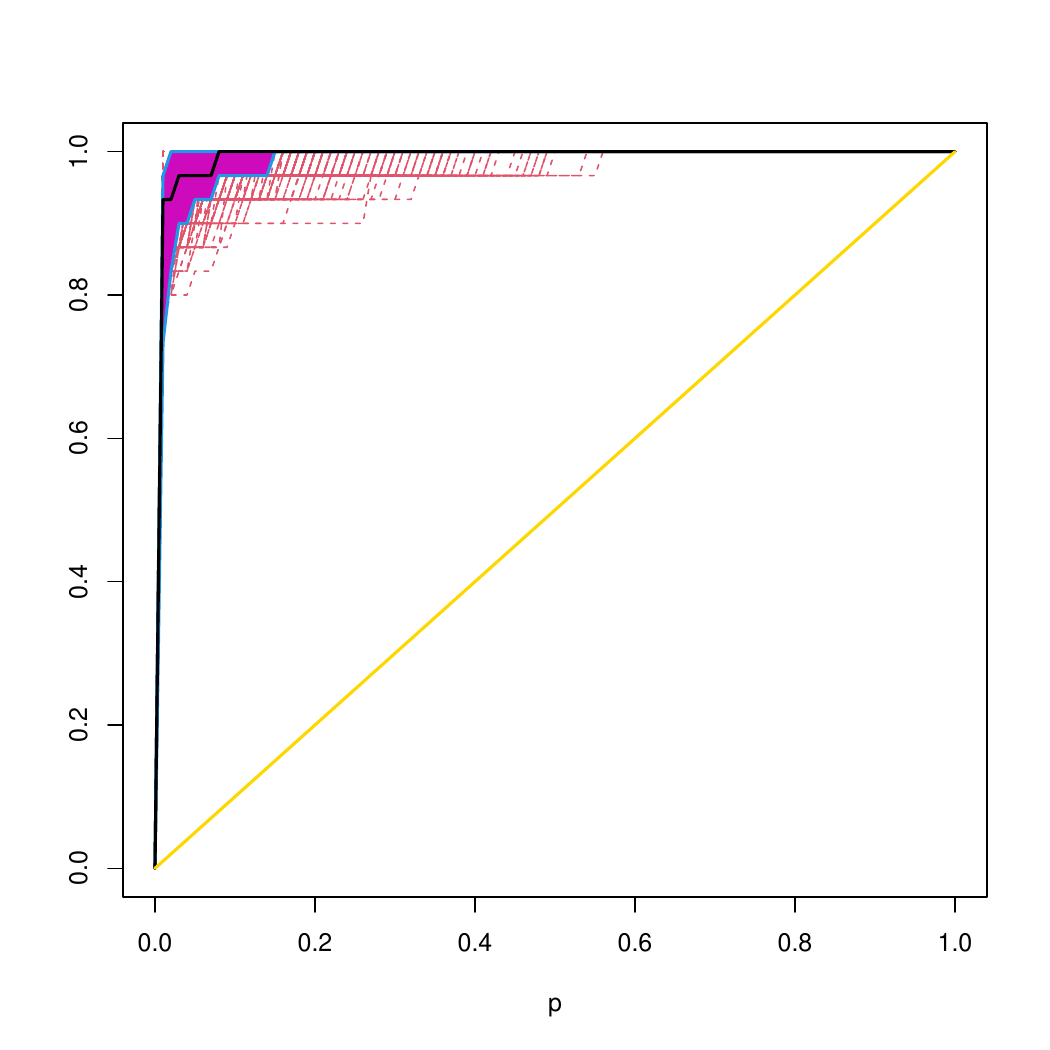}}

\end{tabular}
\caption{Functional boxplots of the estimators $\widehat{\ROC}$ under scenario \textbf{PROP} with  $\rho=2$ for the Exponential Variogram  process. Rows correspond to discriminating indexes, while columns to $\mu_D(t)=2\, \sin(\pi  t)$ and $\mu_H=0$. The sample sizes are $n_D=30$ and $n_H=250$.}
\label{fig:propor:varexp-nD30-nH250} 
\end{center} 
\end{figure}

\clearpage

\section{Real dataset analysis} \label{sec:realdata}

We illustrate the application of the developed methodology to a real dataset reported in \citet{Pineiro:etal:2023} related to  the  study of cardiotoxicity in breast cancer patients  mentioned in the Introduction.

Breast cancers related to high  levels of the protein human epidermal growth factor receptor 2
(HER2) are much more likely to respond to treatments with drugs that target the HER2 protein.
In fact, therapies that aim specifically HER2 have a strong anti--tumoral effect, improving the overall response of the patient and therefore, the survival expectancy. However, this kind of therapies may have side effects such as  cardiotoxicity. 
In this context, the detection of the cancer therapy-related cardiac dysfunction (CTRCD) is relevant with respect to the prognosis and hence,  it is recommended to follow--up the appearance of CTRCD through cardiac imaging tests, among  other  clinical tests. The availability of good markers to predict CTRCD is important to prevent cardiac problems.
The Tissue Doppler Imaging (TDI) is an echocardiographic technique that shows the velocity of myocardial motion. It may be helpful to early identify CTRCD if a study of the heart condition is performed before treatment.  TDI shows velocity as a function of time, thus it may be preprocessed  to obtain a functional datum, see \citet{Pineiro:etal:2023} for more details. 

\begin{figure}[ht!]
	\begin{center}
		\footnotesize
		\begin{tabular}{cc}
			CTRCD=0 & CTRCD=1\\[-2ex]
			\includegraphics[scale=0.35]{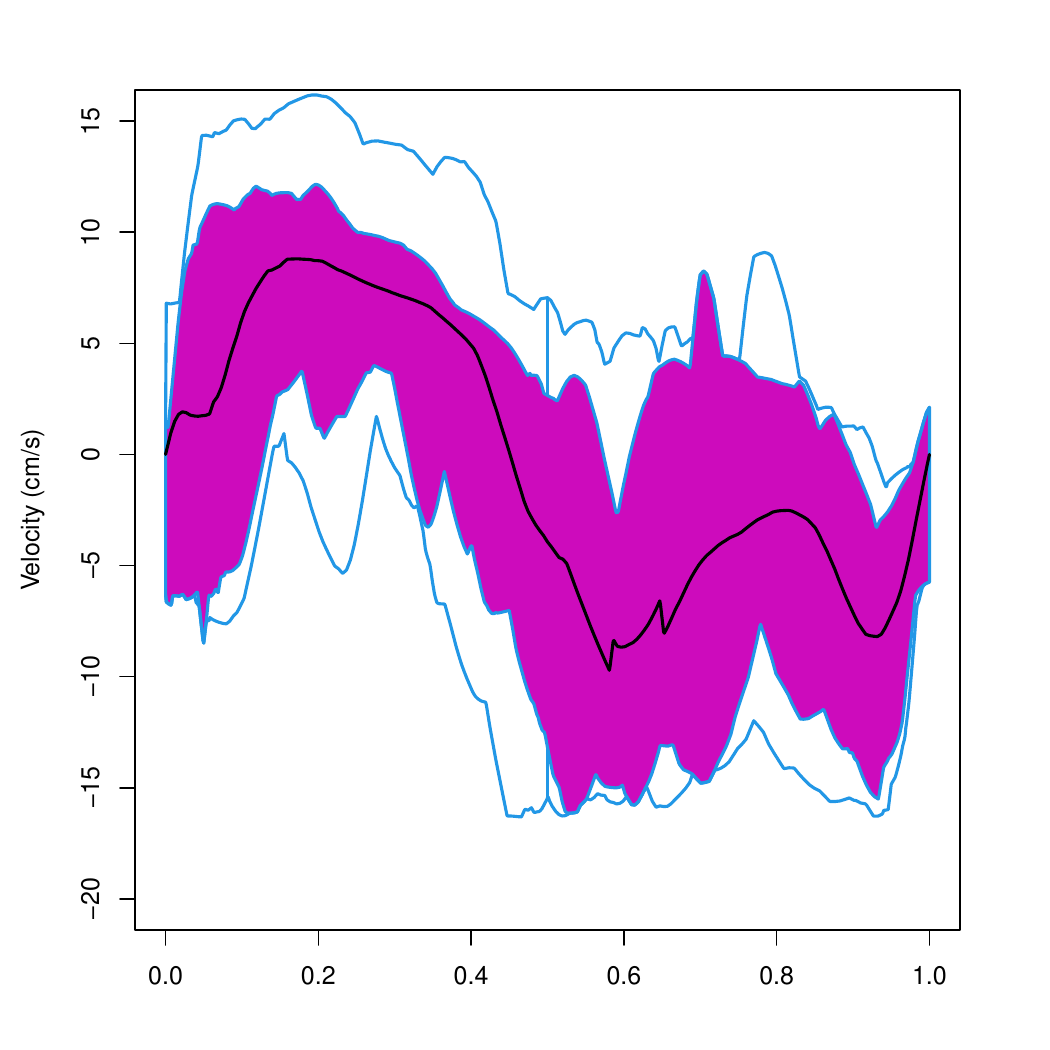} & 
			\includegraphics[scale=0.35]{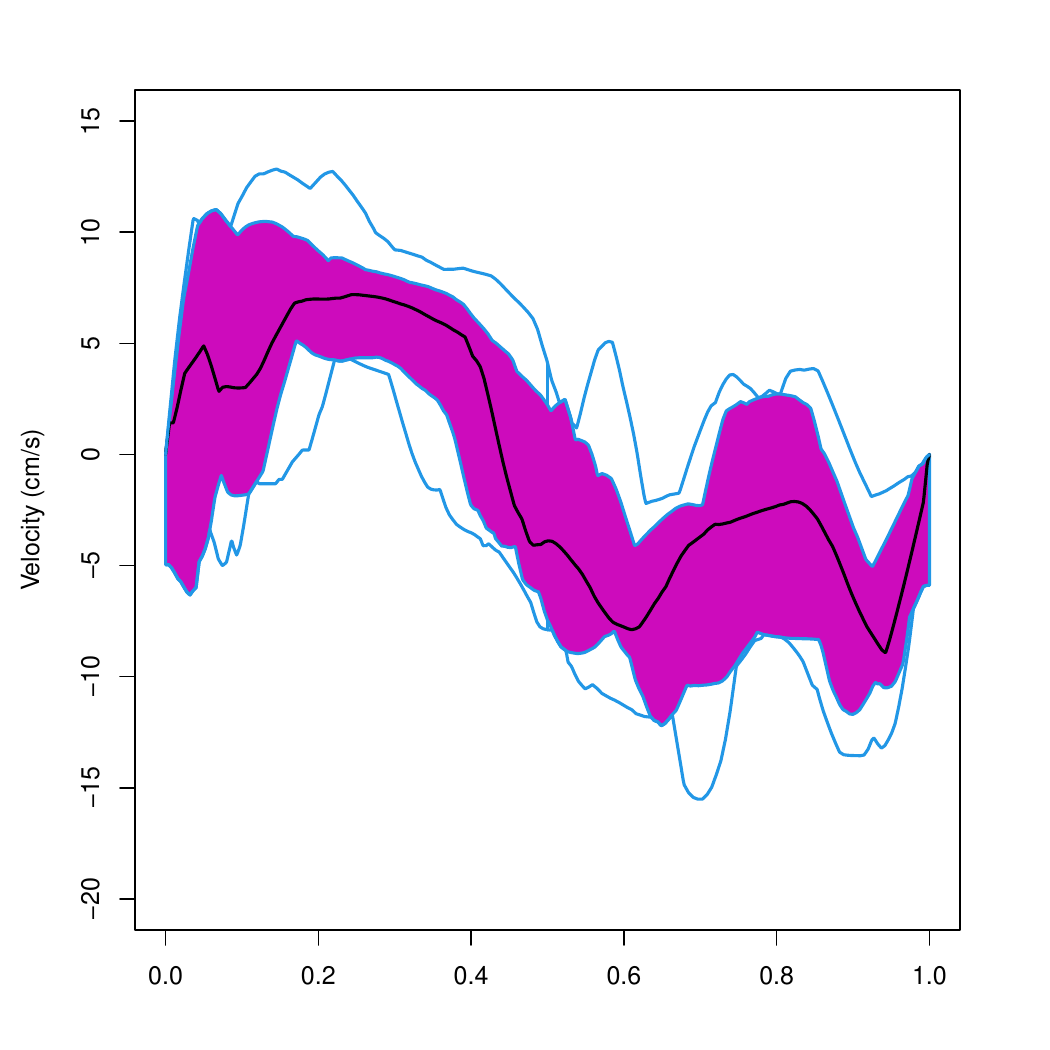}
		\end{tabular}
		\vskip-0.1in 
		\caption{\label{fig:ciclos_fbplot} Cardiotoxicity data. Left panel corresponds to the functional boxplot of the cycles of patients without CTRCD, while the right one to women with CTRCD.}
		 \end{center} 
\end{figure}

The data, displayed in Figure  \ref{fig:ciclos}, correspond to 270 women diagnosed with HER2+ breast cancer, 27 of them suffer from CTRCD. For each patient  the 
cycle extracted from the TDI discretized in 1001 equispaced points in the interval [0,1] is registered together with their CTRCD status. 
To have a deeper insight of the cycles in each status of CTRCD, in Figure \ref{fig:ciclos_fbplot} we display the functional boxplot of each group. No outlying cycles were detected in either group.

\begin{table}[ht!]
	\begin{center}
		\renewcommand{\arraystretch}{1.2}
		\caption{\label{tab:auc_example} Cardiotoxicity data. Estimated AUC of each method.}
		 \begin{tabular}{cccccc}
			\hline \\[-2ex]
			$\Upsilon_{\maxi}$ & $\Upsilon_{\mini}$  & $\Upsilon_{\inte}$ & $\wUps_{\media}$ &  $\wUps_{\lin}$ & $\wUps_{\cuad}$ \\
			\hline
			0.4547 & 0.6770  & 0.5328  & 0.6819  & 0.7034 & 0.8877\\
			\hline
		\end{tabular} 		
	\end{center}
\end{table}

\begin{figure}[ht!]
	\begin{center}
		\footnotesize
		\includegraphics[scale=0.40]{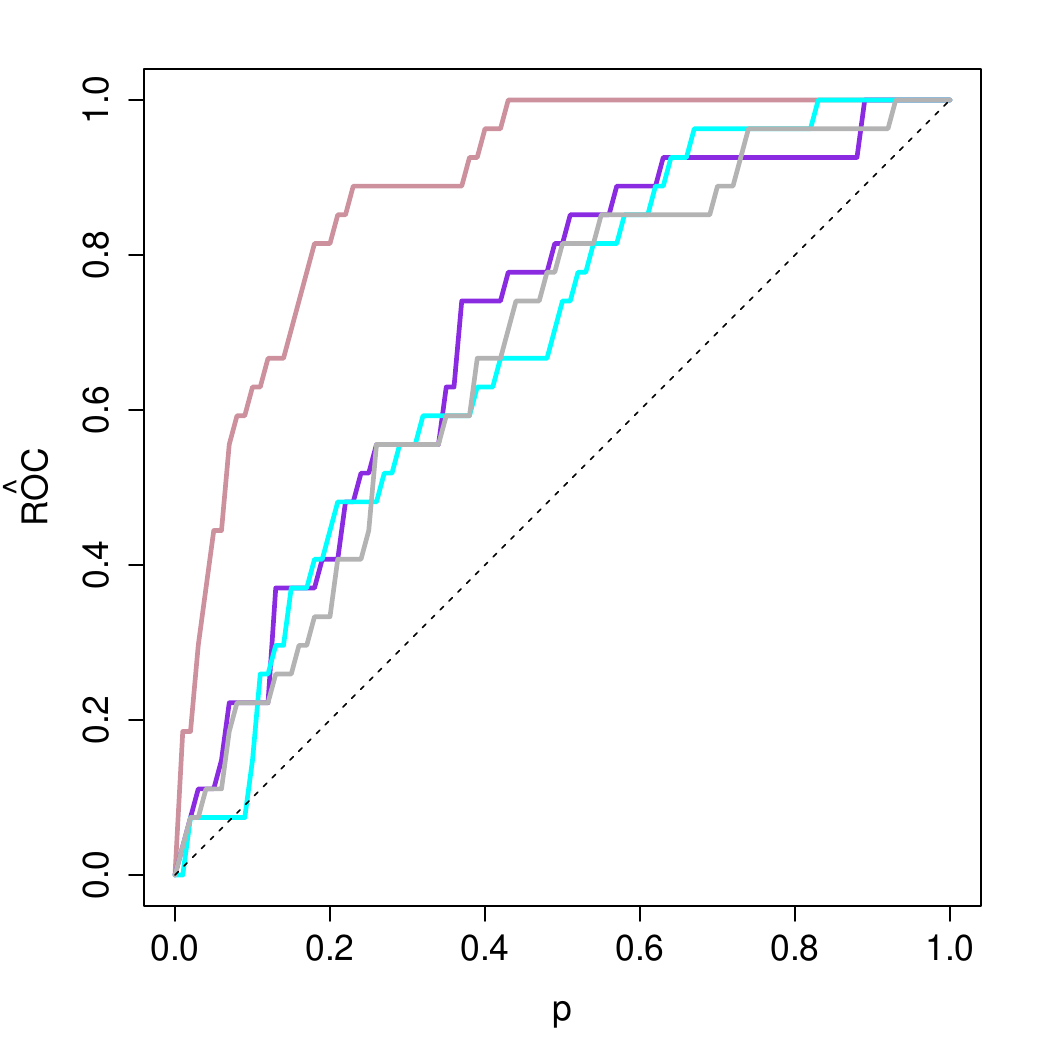}
		\vskip-0.1in 
		\caption{\label{fig:rocs} Cardiotoxicity data. Estimated ROC curves with $\mbox{AUC} \ge 0.65$. The dark pink, violet, cyan and gray lines correspond to the estimators related to  the quadratic method, the linear rule, the procedure based on the difference of means and  the one based on the minimum, respectively. }	  
	\end{center} 
\end{figure}

In order to assess the performance of the functional biomarker to distinguish between the two categories of CTRCD, we apply the discriminating indexes described in the previous sections. We computed the indexes based on the minimum 	($\Upsilon_{\mini}$), the maximum ($\Upsilon_{\maxi}$), the integral ($\Upsilon_{\inte}$),  the difference of means ($\wUps_{\media}$), and the linear ($\wUps_{\lin}$) and the quadratic ($\wUps_{\cuad}$) criteria taking the number of components that explain at least 95\% of the variability, which in this data is attained for $k=11$. Table \ref{tab:auc_example} collects the AUC for each method.  In Figure \ref{fig:rocs}  the estimates of the ROC curve with AUC greater or equal to 0.65 are depicted. The estimator obtained from the quadratic method is plotted in dark pink, in violet the  one corresponding to the linear rule, in cyan the estimate based on the difference of means and in gray that related to the minimum. It is evident from this figure that the better performance is achieved for the quadratic method. The better discriminating capability of the quadratic method is in some sense expectable due to the particular structure of the data, which makes difficult to distinguish the groups just taking into account either the minimum or the maximum or any linear rule.

\section{Final Comments}{\label{sec:finalcoment}}
In this paper, in order to construct a suitable ROC curve, we address the problem of defining proper univariate indexes when functional biomarkers are used to distinguish between two populations.  The defined  indexes allow  to construct a ROC curve to measure their discriminating capability. 
One of our goals was to provide a thorough insight of the difficulties arising at   population level to define and estimate a proper ROC curve in the functional setting. In particular, we have discussed the limitations of considering fully known operators to define the discriminant index and to solve this problem we have introduced two methods that require estimation of some unknown parameters.
In particular, we have proposed a linear  index  with the property of maximizing the AUC, when both populations have the same covariance operator. In order to estimate it and to circumvent  the curse of dimensionality we have suggested to use either a sieve or a penalized approach. The situation of different covariance structures has been also contemplated by means of a quadratic rule constructed  projecting the data over a finite--dimensional space. As discussed in Section \ref{sec:quadfun}, the need of this finite--dimensional approximation is justified by the fact that the limiting operator is defined only over  the intersection  of the squared--root covariance operator ranges.

Consistency results for the estimators of the ROC curve and its related summary measures were derived for both linear and quadratic discriminating indexes, under general assumptions. The results of our numerical experiments illustrate the advantages of using a quadratic rule in presence of different covariance operators. The application of  our proposals to a real data set confirms that, when differences between populations arise in covariances more than between means, as revealed in Figure \ref{fig:ciclos}, the quadratic index outperforms the linear one and the indexes constructed from known operators.

\setcounter{equation}{0}
\renewcommand{\theequation}{A.\arabic{equation}}

\section{Appendix: Proofs}{\label{sec:proofs}}

\subsection{Proof of the results in Section \ref{sec:binormal}}

\begin{proof}[Proof  of \eqref{eq:optYIbeta}]
To derive \eqref{eq:optYIbeta}, first note that   the value   $c$ maximizing $\Delta_{\bbech}(c)$ equals
$$c_{\bbech}= \frac{\bbe\trasp (\bmu_D+\bmu_H)}{2}\,,$$
giving the following expression for the Youden index
\begin{align*}
\YI(\bbe) & =\left|\Phi\left\{ \frac{\bbe\trasp (\bmu_H-\bmu_D)}{2 \left(\bbe\trasp \bSi   \bbe\right)^{1/2}}\right\}
-\Phi\left\{\frac{\bbe\trasp (\bmu_D-\bmu_H)}{2 \left(\bbe\trasp \bSi   \bbe\right)^{1/2}}\right\}\right|
=\left|1-2\,\Phi\left\{\frac{\bbe\trasp (\bmu_D-\bmu_H)}{2 \left(\bbe\trasp \bSi   \bbe\right)^{1/2}}\right\}\right| \\
& =\left|1-2\,\Phi\left\{\frac{1}{\sqrt{2}}L(\bbe)\right\}\right|\,.
\end{align*}
To simplify the notation let $\sigma_{\bbech}^2=\bbe\trasp \bSi  \bbe$ and $\bmu=(\bmu_{D}-  \bmu_H)/2$. Then,
$$\YI(\bbe)= \left|1-2\,\Phi\left(\frac{\bbe\trasp \bmu}{\sigma_{\bbech} }\right)\right|\,.$$
Taking into account that multiplying $\bbe$ by a constant does not change the value of the Youden index, to maximize it, we can search for the maximum of   $\YI^2(\bbe)$ under the constraint that $\sigma_{\bbech}^2=1$. Let
$$H(\bbe)=\left\{1-2\,\Phi\left( {\bbe\trasp \bmu} \right)\right\}^2+\lambda(\sigma_{\bbech}^2-1)\,.$$
Then, if $\varphi=\Phi^{\prime}$, we get that
\begin{align}
 \frac{\partial H}{\partial \bbe} &= \,-\, 4 \left\{1-2\,\Phi\left( {\bbe\trasp \bmu}\right)\right\}\, \varphi\left( {\bbe\trasp \bmu}\right)\bmu + 2\lambda \bSi \bbe\,.
 \label{eq:gradiente}
\end{align}
Multiplying \eqref{eq:gradiente} by $\bbe\trasp$ and using that the value maximizing $H(\bbe)$ has null gradient and that  $\sigma_{\bbech}^2=1$, we obtain that
\begin{align*}
0 &= \,-\, 4 \left\{1-2\,\Phi\left( {\bbe\trasp \bmu}\right)\right\}\, \varphi\left( {\bbe\trasp \bmu}\right)\bbe\trasp\bmu + 2\lambda \bbe\trasp\bSi \bbe
\\
 &= \,-\, 4 \left\{1-2\,\Phi\left( {\bbe\trasp \bmu}\right)\right\}\, \varphi\left( {\bbe\trasp \bmu}\right)\bbe\trasp\bmu + 2\lambda \,,
\end{align*}
which entails that
\begin{equation}
\label{eq:lambda}
2\, \lambda=   4 \left\{1-2\,\Phi\left( {\bbe\trasp \bmu}\right)\right\}\, \varphi\left( {\bbe\trasp \bmu}\right)\bbe\trasp\bmu \,.
\end{equation}
Therefore, if $\bbe\trasp \bmu=0$, we have that $\lambda=0$, $\bcero= {\partial H}/{\partial \bbe}$ and $H(\bbe)=\bcero$ meaning that the maximum is not reached in directions orthogonal to $\bmu$.

Assume that  $\bbe\trasp \bmu\ne 0$ and let $\nu_{\bbech}= 4 \left\{1-2\,\Phi\left( {\bbe\trasp \bmu}\right)\right\}\, \varphi\left( {\bbe\trasp \bmu}\right)\ne 0$. Using  \eqref{eq:lambda}, we conclude that  $2\, \lambda=  \nu_{\bbech}\; \bbe\trasp\bmu$. Besides, taking into account that at any critical point  $ {\partial H}/{\partial \bbe}=\bcero$, from \eqref{eq:gradiente} we conclude that $2\lambda \bSi \bbe   =  4 \left\{1-2\,\Phi\left( {\bbe\trasp \bmu}\right)\right\}\, \varphi\left( {\bbe\trasp \bmu}\right)\bmu =\nu_{\bbech} \bmu$, so
\begin{align*}
\nu_{\bbech} \; \bbe\trasp\bmu \; \bSi \bbe & = \nu_{\bbech} \bmu \,,
\end{align*}
or equivalently, $(\bbe\trasp\bmu)\,  \bSi \bbe = \bmu $.
Denoting   $a_{\bbech}=\bbe\trasp\bmu\ne 0$, we have 
$$\bbe= \bSi^{-1} \bmu\, \frac{1}{a_{\bbech}}\,,$$
which leads to $a_{\bbech}^2= \bmu\trasp \bSi^{-1} \bmu$ and 
$\bbe =  \bSi^{-1} \bmu/\left({\bmu\trasp \bSi^{-1} \bmu}\right)^{1/2}  $, concluding the proof.
\end{proof}

\subsection{Proof of the results in Section \ref{sec:functional}}

\begin{proof}[Proof of Lemma \ref{lema:Cancor}]
a) Note that as in the multivariate setting
\begin{align}
\cov(\langle \beta, X\rangle, G) &=\esp \left(G\; \langle \beta, X\rangle \right)- \esp \langle \beta, X\rangle \; \esp G
 = \esp \left\{G \;\esp \left(\langle \beta, X\rangle  \mid G \right)\right\}- \pi_D \esp \left\{\esp \left(\langle \beta, X\rangle \mid G \right) \right\}
\nonumber\\
&= \pi_D \esp \left(\langle \beta, X\rangle \mid G=1 \right) - \pi_D \left\{\pi_D  \esp \left(\langle \beta, X\rangle \mid  G=1\right)
+ \pi_H  \esp \left(\langle \beta, X\rangle \mid G=0\right)\right\}
\nonumber\\
&= \pi_D \left[\esp \left(\langle \beta, X_D\rangle \right)  -   \left\{\pi_D  \esp \left(\langle \beta, X_D\rangle  \right)
+ \pi_H  \esp \left(\langle \beta, X_H\rangle  \right)\right\}\right]
\nonumber\\
&= \pi_D \left\{  \langle \beta, \mu_D\rangle   -   \left(\pi_D   \langle \beta, \mu_D\rangle  +\pi_H   \langle \beta, \mu_H\rangle \right)\right\}
= \pi_D \left(\pi_H  \langle \beta, \mu_D\rangle   -   \pi_H   \langle \beta, \mu_H\rangle \right)
\nonumber\\
&= \pi_D \,\pi_H \, \langle \beta, \mu_D-\mu_H\rangle\,, 
\label{eq:covXG}
\end{align}
while 
\begin{align*}
\var(\langle \beta, X\rangle )&= \esp\left(\langle \beta, X\rangle^2\right) -\left( \esp \langle \beta, X\rangle\right)^2 
= \esp\left(\langle \beta, X\rangle^2\right) -\left(\pi_D   \langle \beta, \mu_D\rangle  +\pi_H   \langle \beta, \mu_H\rangle   \right)^2 \\
&= \esp\left\{\esp\left(\langle \beta, X\rangle^2 \mid G \right) \right\}
-\left(\pi_D   \langle \beta, \mu_D\rangle  +\pi_H   \langle \beta, \mu_H\rangle   \right)^2\\
&= \left\{\pi_D\esp\left(\langle \beta, X_D\rangle^2\right) 
+\pi_H\esp\left(\langle \beta, X_H\rangle^2\right)\right\}-\left(\pi_D   \langle \beta, \mu_D\rangle  +\pi_H   \langle \beta, \mu_H\rangle   \right)^2\\
&=  \pi_D\left(\langle \beta, \Gamma_D \beta\rangle +\langle \beta, \mu_D\rangle^2 \right) +\pi_H\left(\langle \beta, \Gamma_H \beta\rangle +\langle \beta, \mu_H\rangle^2 \right) \\
& -\left(\pi_D^2   \langle \beta, \mu_D\rangle^2  +\pi_H^2   \langle \beta, \mu_H\rangle^2 +2 \pi_D\, \pi_H  \langle \beta, \mu_D\rangle\, \langle \beta, \mu_H\rangle  \right) \\
&=  \langle \beta, \Gamma_{\pool} \beta\rangle +\pi_D \pi_H\langle \beta, \mu_D\rangle^2  +\pi_D \pi_H\langle \beta, \mu_H\rangle^2   - 2 \pi_D\, \pi_H  \langle \beta, \mu_D\rangle\, \langle \beta, \mu_H\rangle   \;.
\end{align*}
 Therefore, 
$$\var(\langle \beta, X\rangle )=   \langle \beta, \Gamma_{\pool} \beta\rangle + \pi_D \pi_H \langle \beta, \mu_D- \mu_H\rangle^2\,,$$
which together with the fact that $\var(G)=\pi_D\pi_H$ entails that
$$\mbox{corr}(\langle \beta, X\rangle, G)= 
\frac{\pi_D \,\pi_H \, \langle \beta, \mu_D-\mu_H\rangle }{\left\{\pi_D\pi_H\;\left( \langle \beta, \Gamma_{\pool} \beta\rangle + \pi_D \pi_H \langle \beta, \mu_D- \mu_H\rangle^2\right) \right\}^{1/2}}
=\frac{\pi_D^{1/2} \,\pi_H^{1/2} \,L_{\pool}(\beta)}{\left\{    1 +L_{\pool}^2(\beta) \right\}^{1/2}}\,,$$
 concluding the proof of a).
 
 b) Follows immediately  noting that $\Gamma_{\pool}= \Gamma$ when $\pi_D=1/2$ or when  $\Gamma_H =\Gamma_D$, since $\pi_D+\pi_H=1$ and the analogy between maximizing the AUC which corresponds to maximizing $L(\bbe)$ and that of maximizing $\mbox{corr}(\langle \beta, X\rangle, Z)$, which corresponds to $L_{\pool}(\beta)$. 
 \end{proof}

\begin{proof}[Proof of Proposition \ref{prop:expresionbeta}] 
Analogously to $\itR(\Gamma)$, define 
$$\itR(\Gamma^{1/2})=\left \{y\in \itH:\quad \sum_{\ell \ge 1} \frac 1{\lambda_{\ell}}\; \langle y, \phi_{\ell}\rangle^2 <\infty \right\}\,, $$
and  the inverse of $\Gamma^{1/2}$, which is well defined over $\itR(\Gamma^{1/2})$, as
$$\Gamma^{-1/2} (y)=\sum_{\ell \ge 1} \frac 1{ \lambda_{\ell}^{1/2}}\; \langle y, \phi_{\ell}\rangle\; \phi_{\ell}\,.$$
The fact that {$\lambda_{\ell}\to 0$ as $\ell\to \infty$} entails that for $\ell$ large enough $\lambda_\ell^2 <\lambda_\ell$, so taking into account that   $\mu_D-\mu_H\in \itR(\Gamma)$, we get that $\mu_D-\mu_H\in \itR(\Gamma^{1/2})$.

Let $R = (\pi_D\pi_H)^{-1/2}\Gamma^{-1/2} \Gamma_{XG}:\real \to L^2(0,1)$ with $\Gamma_{XG}$ being the covariance operator between $X$ and $G$, that is, the operator $\Gamma_{XG}:\real \to \itH$ such that for any $a\in \real$, $  \cov(\langle u, X\rangle, a\, G)=\langle u, \Gamma_{XG}(a)  \rangle$. Denoting $ \gamma_{XG}=\Gamma_{XG}(1)$ we have that  $  \cov(\langle u, X\rangle,   G)=\langle u, \gamma_{XG} \rangle$.
Analogous   arguments to those considered in Theorem 4.8 of \citet{he:etal:2003} allow to show that   the value $\beta_0$ maximizing $\mbox{corr}^2(\langle \beta, X\rangle, G)$ (respectively, the AUC)  equals $\beta_0=\Gamma^{-1/2} \psi_0$, where $\psi_0$ is the eigenfunction of the operator
$$R_0=R\,R^{*}: L^2(0,1)\to L^2(0,1)$$
related to its largest eigenvalue and $R^{*}$ stands for the adjoint operator of $R$.

From \eqref{eq:covXG}, we get that   $ \gamma_{XG}= \pi_D \,\pi_H   (\mu_D-\mu_H) $, then if   $\Delta_{DH}= (\pi_D\pi_H)^{1/2}\,\Gamma^{-1/2}  \left(\mu_D-\mu_H\right)\in \itH $, we get that $R a = a\, \Delta_{DH} $ and $R_0= \Delta_{DH} \left\{\Delta_{DH}\right\}^{*}$.
Note that $R^{*}:  L^2(0,1)\to \real $ satisfies $\langle u, R\, a\rangle=  a R^{*} \,u$, for any $a\in \real$, $u\in L^2(0,1)$, hence  we have that 
$$\langle u, R\, a\rangle= a\,  \;\langle u,   \Delta_{DH}\rangle = a\, \int_0^1 \Delta_{DH}(t) u(t) dt\,,$$
and $R^{*}$ is the linear operator with  representative $\Delta_{DH}$, i.e., $R^{*} u =\langle u, \Delta_{DH}\rangle$. Therefore, $R_0$ has only one eigenvalue different from 0, since for any $u\in \itH$  orthogonal to $\Delta_{DH}$, $R_0 u=0$ and $ R^{*}\Delta_{DH}=\|\Delta_{DH}\|^2$ meaning that
$$R_0 \Delta_{DH}= R \|\Delta_{DH}\|^2= \|\Delta_{DH}\|^2 \Delta_{DH}\,.$$
Thus, $\psi_0=\Delta_{DH}/\|\Delta_{DH}\|$ and $\beta_0=  (\pi_D\pi_H)^{1/2} \;\Gamma^{-1}\left(\mu_D-\mu_H\right)$, concluding the proof.
 \end{proof}

\begin{proof}[Proof of Proposition \ref{prop:cuad}]
a) follows easily noting that  $A^{*} A \alpha =\sum_{\ell=1}^k \langle  \alpha, \phi_\ell \rangle\,\phi_\ell$, $\balfa\trasp \bx=$ \linebreak  $\langle A \alpha, A X\rangle$   and the fact that
\begin{equation}
\Gamma_j^{-1} y = \sum_{\ell \ge 1} \frac 1{{\lambda_{j,\ell}}}\; \langle y, \phi_{\ell}\rangle\; \phi_{\ell}\,,\quad \mbox{ for any $y\in \itR(\Gamma_j)$}\, ,
\label{eq:gammajy}
\end{equation}
which implies that $A\,\Gamma_j^{-1}\mu_j=\bSi_j^{-1} \bmu_j$.

\noi b) Note that
\begin{align*}
 \bx\trasp \bLam \bx &=\sum_{\ell=1}^k \Lambda_\ell x_\ell^2= \sum_{\ell=1}^k \Lambda_\ell \langle X, \phi_\ell\rangle^2   =\sum_{\ell=1}^k \frac{1}{\lambda_{D,\ell}} \langle X, \phi_\ell\rangle^2 -\sum_{\ell=1}^k \frac{1}{\lambda_{H,\ell}} \langle X, \phi_\ell\rangle^2 \;.
 \end{align*}
 Taking into account \eqref{eq:gammajy} and that from the definition of the linear operator $A$, we have that
 $$A \Gamma_j^{-1} y=\left( \frac{\langle y, \phi_{1}\rangle}{\lambda_{j,1}}, \dots,  \frac{\langle y, \phi_{k}\rangle}{\lambda_{j,k}}\right)\trasp\,,$$
we easily obtain that,  for any $X\in \itR(\Gamma_D^{1/2})\cap \itR(\Gamma_H^{1/2})$,
 $$\sum_{\ell=1}^k \frac{1}{\lambda_{j,\ell}} \langle X, \phi_\ell\rangle^2  = \|A\,  \Gamma_j^{-1} X\|^2\,.$$
 The expression for $\Upsilon (X)$ follows easily from the convergence of the series $\sum_{\ell\ge 1}   \langle X, \phi_\ell\rangle^2/{\lambda_{j,\ell}}$, for $j=D,H$, and of the series  $\sum_{\ell\ge 1} \langle \alpha, \phi_\ell\rangle  \langle X, \phi_\ell\rangle=\langle\alpha, X\rangle$, concluding the proof.
\end{proof}

\subsection{Proof of the results in Section \ref{sec:consist}}
 

\begin{proof}[Proof of Theorem \ref{theo:consist.1}]
From Assumption \ref{ass:A3} and the continuity of the quantile function $F_{H}^{-1}=F_{H, \Upsilon}^{-1}:[0,1]\to \real$ when Assumption \ref{ass:A1} holds, we get that for each $0<p<1$, $\wF_{j}^{-1}(p)-F_j^{-1}(p)\convpp 0$. Therefore, using the continuity of $F_{D}$ stated in Assumption \ref{ass:A2}, the uniform convergence required in Assumption \ref{ass:A3} for  $j=D$ and the inequality
\begin{align*}
|\widehat{\ROC}(p) & - {\ROC}(p)|  =  \left|\wF_{D}\left\{  \wF_{H}^{-1}(1-p)\right\}- F_{D}\left\{  F_{H}^{-1}(1-p)\right\}\right|\\
&\le   \left|\wF_{D}\left\{  \wF_{H}^{-1}(1-p)\right\}- F_{D}\left\{  \wF_{H}^{-1}(1-p)\right\}\right| + \left| F_{D}\left\{  \wF_{H}^{-1}(1-p)\right\}- F_{D}\left\{  F_{H}^{-1}(1-p)\right\}\right|\\
&\le   \left\|\wF_{D}- F_D\right\|_{\infty} + \left| F_{D}\left\{  \wF_{H}^{-1}(1-p)\right\}- F_{D}\left\{  F_{H}^{-1}(1-p)\right\}\right|\; ,
\end{align*}
  we derive that
$ \widehat{\ROC}(p) \convpp {\ROC}(p) $. Moreover, the uniform convergence is obtained from the monotonicity of  $\widehat{\ROC}$ and  ${\ROC}$ and also from the continuity of ${\ROC}$.
\end{proof}

\begin{proof}[Proof of Theorem \ref{theo:consist.2}] 
From the monotonicity  of $\wF_j$ and $F_j$ and also the continuity of $F_j$, it will be enough to prove that $\wF_j(t)-F_j(t)\convpp 0$.  
	
Fix $j=D$ or $H$. Define $\eLe_t^{(j)}(\bbe)=\prob(\bx_j\trasp \bbe \le t)$  for any $t \in \real$, then
\begin{align*}
 \left|\wF_{j} (t)- F_{j}(t)\right| & \le  \left|\wF_{\wbbech, j} (t)- \eLe_t^{(j)}(\wbbe)\right|+ \left|\eLe_t^{(j)}(\wbbe)- F_{ \bbech_0 , j}(t)\right|= \left|\wF_{\wbbech , j} (t)- F_{\wbbech, j}(t)\right| \\
 & \hskip0.3in + \left| F_{\wbbech , j}(t)- F_{\bbech_0 , j}(t)\right|\,.
 \end{align*}
It suffices to prove that
\begin{align}
\sup_{t\in \real}\sup_{\bbech \in \real^k} \left|\wF_{\bbech, j} (t) -F_{ \bbech , j} (t) \right| & \convpp 0 
\label{eq:glivenkobeta}
\end{align}
and 
\begin{align}
\left| F_{ \wbbech , j}(t)- F_{ \bbech_0 , j}(t)\right| & \convpp 0 \,.
\label{eq:Fbeta}
\end{align}
To derive \eqref{eq:glivenkobeta}, let us consider the family of functions
  $$\itF=\{h_{\bbech,t}(\bx)=\indica_{\{ \bx\trasp\bbech\le   t\}} \mbox{ for } (\bbe, t )\in \real^k\times \real \} \,.$$ 
 Taking into account that $\{g(\bx)=\bx\trasp\bbe  -t ; (\bbe,  t) \in \real^k\times\real\}$ is a finite--dimensional space of functions with dimension $p+1$, from Lemmas 9.6  in \citet{kosorok:2008}  we get that $\itF$ is a VC-class  with index at most $p+3$. Hence, applying   Lemmas  9.8 and   9.9(iii) in \citet{kosorok:2008}, we get that the class of functions $\itF$ is a VC-class with index $V(\itF)$ smaller or equal than  $p+3$. Note that the envelope of $\itF$ equals $F\equiv 1$. Hence, Theorem 2.6.7 in \citet{vanderVaart:wellner:1996}  entails that, there exists a universal constant $K$ such that, for any measure $Q$
 $$N\{\epsilon, \itF, L_1(Q)\} \le K \;  V(\itF) \left(16 e\right)^{V(\itF)} \left(\frac{1}{\epsilon}\right)^{V(\itF)-1}\, ,$$
 which together with Theorem 2.4.3 in \citet{vanderVaart:wellner:1996}   or Theorem 2.4 in \citet{kosorok:2008}, leads to
 \begin{equation*}
 \label{eq:glivenko}
   \sup_{h\in \itF} | P_{n_j} h- P_j h| \convpp 0\,, 
 \end{equation*}
  where we have used the standard notation in empirical processes, i.e., $P h=\esp \{ h(\bX)\}$ and $P_n h=(1/n) \sum_{i=1}^n h(\bx_i)$. 
 Hence, we have that
 $$\sup_{t\in \real}\sup_{\bbech \in \real^k} \left|\wF_{\bbech , j} (t)- F_{\bbech , j}(t)\right| \convpp 0\,,$$
 which concludes the proof of  \eqref{eq:glivenkobeta}.

It remains to prove \eqref{eq:Fbeta}. To strengthen the dependence on the sample size denote $\wbbe_n=\wbbe$, where $n=n_D+n_H$. Then, from the fact that $\wbbe\convpp \bbe_0$, there exists $\itN\subset \Omega$ such that $\prob(\itN)=0$ and for $\omega\notin \itN$, $\wbbe_n(\omega)  \to  \bbe_0$. Take $\omega\notin \itN$, then  for any $\bx\in \real^k$,  $\bx\trasp\wbbe_n(\omega)  \to  \bx\trasp\bbe_0$. Hence, if $\bx_j\sim P_j$, the random variable $Z_n=\bx_j\trasp\wbbe_n(\omega)$ converges to $Z=\bx_j\trasp\bbe_0$ everywhere, so $F_{\wbbech_n(\omega) , j}(t)=\prob(Z_n\le t)\to \prob(Z\le t)=F_{j}(t)$, for any $t$. Using the fact that   $F_{j}(t)$ is  non-decreasing and continuous, we conclude that, for any $\omega\notin \itN$, 
$$ \left\| F_{ \wbbech_n(\omega) , j }- F_{j} \right\|_{\infty}\to 0\,,$$
concluding the proof of \eqref{eq:Fbeta}. 
\end{proof}

\begin{proof}[Proof of Theorem \ref{theo:consist.3}]
The proof is similar to that of Theorem \ref{theo:consist.2}. Denote $\itH_k$ the linear space spanned by $\phi_1, \dots, \phi_k$, that is,  $\itH_k=\{\beta=\sum_{s=1}^k  b_s \phi_s, \bb=(b_1,\dots, b_k)\trasp \in \real^k\}$.

Using that $\wF_j$ and $F_j$ are non--decreasing function and the continuity of $F_j$, it will be enough to show that  $\wF_j(t)-F_j(t)\convpp 0$.  
Fix $j=D$ or $H$ and define $\eLe_t^{(j)}(\beta)=\prob(\langle X_j,\beta\rangle \le t)$, then
\begin{align*}
 \left|\wF_{j} (t)- F_{j}(t)\right| & \le  \left|\wF_{ \wbeta , j} (t)- \eLe_t^{(j)}(\wbeta)\right|+ \left|\eLe_t^{(j)}(\wbeta)- F_{ \beta_0 , j}(t)\right|= \left|\wF_{ \wbeta , j} (t)- \eLe_t^{(j)}(\wbeta)\right| \\
 & \hskip0.3in + \left| F_{ \wbeta , j}(t)- F_{ \beta_0 , j}(t)\right|\,.
 \end{align*}
It is enough to show that
\begin{align}
\sup_{t\in \real}\sup_{\beta \in \itH_k} \left|\wF_{ \beta, j} (t) -F_{\beta , j} (t) \right| & \convpp 0 
\label{eq:glivenkobetafun}
\end{align}
and 
\begin{align}
\left| F_{ \wbeta, j}(t)- F_{ \beta_0, j}(t)\right| & \convpp 0 \;. 
\label{eq:Fbetafun}
\end{align}
The proof of \eqref{eq:Fbetafun} is similar to that of \eqref{eq:Fbeta} and for that reason  it is omitted.

To derive \eqref{eq:glivenkobetafun}, we will follow similar arguments to those considered in the proof of Proposition 1 in \citet{Bianco:Boente:2023}. Consider the family of functions
 $$\itF=\{h_{\beta,t}(x)=\indica_{\{ \langle x,\beta\rangle\le   t\}} \mbox{ for } (\beta, t )\in \itH_k\times \real \} \,.$$ 

Taking into account that for $\beta\in \itH_k$, $\langle x,\beta\rangle= \sum_{s=1}^k b_s \langle x,\phi_s\rangle= \bx\trasp \bb$, where $\bx=(\langle x,\phi_1\rangle, \dots, \langle x,\phi_k\rangle)\trasp$, we obtain that $\{g(x)=\langle x,\beta\rangle  -t  ; (\beta,  t) \in \itH_k\times\real\}$ is a finite--dimensional space of functions with dimension $k+1$. Thus, from Lemma  9.6  in \citet{kosorok:2008}  we get that $\itF$ is a VC-class  with index at most $k+3$. Hence, applying   Lemmas  9.8 and   9.9(iii) in \citet{kosorok:2008}, we conclude that the class of functions $\itF$ is a VC-class with index $V(\itF)$ smaller or equal than  $k+3$. Note that the envelope of $\itF$ equals $F\equiv 1$. Hence, Theorem 2.6.7 in \citet{vanderVaart:wellner:1996}   entails that there exists a universal constant $K$ such that, for any measure $Q$
\begin{equation}
	\label{eq:entropia}
	N\left\{\epsilon, \itF, L_1(Q)\right\} \le K \;  (k+3) \left(16 e\right)^{k+3} \left(\frac{1}{\epsilon}\right)^{k+2}\, .
	\end{equation}
 To show   $\sup_{h\in \itF} | P_{n_j} h- P_j h| \convpp 0$,  where  $P_j h=\esp \left\{h(X_j)\right\}$ and $P_{n_j} h=(1/{n_j}) \sum_{i=1}^{n_j} h(X_{j,i})$,  it will be enough to prove that 
\begin{equation}
	\label{eq:NPn}
	\frac{1}{n_j} \log N\left\{\epsilon, \itF, L_1(P_{n_j})\right\} \convprob 0\,.
\end{equation}
From \eqref{eq:entropia} we get that
$$\frac{1}{n_j} \log N\left\{\epsilon, \itF, L_1(P_{n_j})\right\}  \le C \frac{ k+2 }{n_j}\log \left(\frac{1}{\epsilon}\right)\,,$$
for some constant $C$. Thus, using that  $k_n/n\to 0$, we obtain that 
 $$\sup_{t\in \real}\sup_{\beta  \in \itH^k} \left|\wF_{ \beta , j} (t)- F_{\beta , j}(t)\right| \convpp 0\,,$$
 which concludes the proof of  \eqref{eq:glivenkobetafun}. 
 \end{proof}

\begin{proof}[Proof of Theorem \ref{theo:consist.4}]
Note that b) is a direct consequence of a) and Theorem \ref{theo:consist.1}, so we will only show a). For that purpose recall that, for any $\balfa \in \real^k$ and $\bLam \in \real^{k\times k}$, we have denoted
$\Upsilon_{\bLamch, \balfach} (X)= \,-\, \bx\trasp \bLam  \bx +  \balfa \trasp \bx$ where $\bx=A(X)=(\langle X, \phi_1\rangle, \dots,\langle X, \phi_k\rangle)\trasp$ and as in the proof of Theorem \ref{theo:consist.3} let $L_t^{(j)}(\bLam,\balfa)=\prob\left( \Upsilon_{\bLamch, \balfach}(X_j)\le t\right)$.
Then,
\begin{equation}
|\wF_j(t)-F_j(t)|\le |\wF_j(t)-L_t^{(j)}(\wbLam,\wbalfa)|+|L_t^{(j)}(\wbLam,\wbalfa)-L_t^{(j)}(\bLam_0,\balfa_0)| \;.
\label{eq:cotaquad}
\end{equation} 
As in the proof of \eqref{eq:Fbeta}, the second term on the right hand side in \eqref{eq:cotaquad} converges almost surely to $0$, since $(\wbLam, \wbalfa) \convpp (\bLam_0, \balfa_0)$.
Hence, we only have to show that $\wF_j(t)-L_t^{(j)}(\wbLam,\wbalfa) \convpp 0$, which follows as in  Theorem  \ref{theo:consist.2} using that the class of functions
 $$\itF=\left\{h_{\bLamch, \balfach, t}(X)=\indica_{\{\,-\, \bx\trasp \bLam  \bx +  \balfa \trasp \bx \le t\}}\;,  \bLam \in \real^{k\times k}, \balfa\in \real^k, t \in \real, \mbox{ where } \bx=A(X)\right\}\,,$$ 
is a finite--dimensional space of functions. 
\end{proof}

\small
\section*{Acknowledgement}
This research of A.~M.~Bianco and G.~Boente was partially supported by Grants 20020170100022\textsc{ba} from the Universidad de Buenos Aires and \textsc{pict} 2021-I-A-00260 from \textsc{anpcyt}, Argentina. The three authors were also supported by Grants PID2020-118101GB-I00 (MCIN/ AEI /10.13039/501100011033) and PID2023-148811NB-I00 (MICIU/AEI/10.13039/501100011033 and ERDF/EU), Spain. The research was begun while A.~M.~Bianco and G. Boente were visiting the Universidade de Vigo.

\bibliographystyle{apalike}
\bibliography{referencias2}

 \end{document}

%% file: definiciones.tex
\newcommand\bb {\mathbf b}

\newcommand\bt {\mathbf t}
\newcommand\bu {\mathbf u}

\newcommand\bx {\mathbf x}

\newcommand\bz {\mathbf z}

\newcommand\bA {\mathbf A}

\newcommand\bX {\mathbf X}

\newcommand\indica {\mathbb{I}}

\newcommand\wb {\widehat{{b}}}

\newcommand\wbx {\widehat{\bx}}

\newcommand\wF {\widehat{F}}

\newcommand\wL {\widehat{L}}

\newcommand\wY {\widehat{Y}}

\newcommand\itA {{\mathcal{A}}}

\newcommand\itE {{\mathcal{E}}}
\newcommand\itF {{\mathcal{F}}}
\newcommand\itG {{\mathcal{G}}}
\newcommand\itH {{\mathcal{H}}}
\newcommand\itI {{\mathcal{I}}}

\newcommand\itN {{\mathcal{N}}}

\newcommand\itR {{\mathcal{R}}}
\newcommand\itS {{\mathcal{S}}}

\newcommand\balfa {\mbox{\boldmath $\alpha$}}
\newcommand\balfach {\mbox{\scriptsize\boldmath $\alpha$}}

\newcommand\bbe {\mbox{\boldmath $\beta$}}

\newcommand\bbech {\mbox{\scriptsize${\bbe}$}}

\newcommand\bmu {\mbox{\boldmath $\mu$}}

\newcommand\bLam {\mbox{\boldmath $\Lambda$}}
\newcommand\bLamch {\mbox{\scriptsize \boldmath $\Lambda$}}

\newcommand\bSi {\mbox{\boldmath $\Sigma$}}

\newcommand\wbalfa {\widehat{\balfa}}

\newcommand\wbalfach {\widehat{\balfach}}
\newcommand\wbeta {\widehat{\beta}}
\newcommand\wbbe {\widehat{\bbe}}
\newcommand\wbbech {\mbox{\scriptsize$\wbbe$}}

\newcommand\wphi {\widehat{\phi}}

\newcommand\wmu {\widehat{\mu}}
\newcommand\wbmu {\widehat{\bmu}}

\newcommand\wGamma {\widehat{\Gamma}}

\newcommand\wbLam {\widehat{\bLam}}
\newcommand\wbLamch {\widehat{\bLamch}}

\newcommand\wUps  {\widehat{\Upsilon}}

\newcommand\wbSi {\widehat{\bSi}}

\def\real{\mathbb{R}}

\newcommand{\esp}{\mathbb{E}}
\newcommand{\prob}{\mathbb{P}}
\newcommand{\eLe}{\mathbb{L}}
\newcommand{\cov}{\mbox{\sc Cov}}
\newcommand{\var}{\mbox{\sc Var}}

\newcommand{\diag}{\mbox{\sc diag}}

\newcommand{\convpp}{ \buildrel{a.s.}\over\longrightarrow}

\newcommand{\convprob  }{ \buildrel{p}\over\longrightarrow}

\newcommand{\trasp}{^{\top}}

\newcommand\bcero {{\bf{0}}}

\def\dst{\displaystyle}

\def\argmax{\mathop{\mbox{argmax}}}

\def\ROC {\mathop{\mbox{ROC}}} 
\def\AUC {\mathop{\mbox{AUC}}}

\def\YI {\mathop{\mbox{\sc YI}}}

\newcommand\noi{\noindent}
\def\dst{\displaystyle}

\def\square{\ifmmode\sqr\else{$\sqr$}\fi}
\def\sqr{\vcenter{
         \hrule height.1mm
         \hbox{\vrule width.1mm height2.2mm\kern2.18mm
\vrule width.1mm}
         \hrule height.1mm}}

\newcommand\ave {\mbox{\scriptsize \sc a}}

\newcommand{\pool}{\mbox{\scriptsize \sc pool}}

\newcommand{\inte}{\mbox{\footnotesize \sc i}}

\newcommand{\mini}{\mbox{\footnotesize \sc min}}
\newcommand{\maxi}{\mbox{\footnotesize \sc max}}
\newcommand{\cuad}{\mbox{\footnotesize \sc quad}}
\newcommand{\cuadch}{\mbox{\scriptsize \sc quad}}
\newcommand{\lin}{\mbox{\footnotesize \sc lin}}
\newcommand{\media}{\mbox{\footnotesize \sc m}}